\documentclass[11pt]{article}
\usepackage{graphicx}
\usepackage{amsmath}

\usepackage{enumerate}
\usepackage{natbib}
\usepackage{url} 
\usepackage[compact]{titlesec}

\usepackage{mkolar_definitions,command}

\usepackage{multirow}
\usepackage{booktabs}
\usepackage{longtable}
\usepackage{mathrsfs}
\usepackage{wrapfig}
\usepackage{placeins}
\usepackage{bbm}

\usepackage{algorithm}
\usepackage[noend]{algpseudocode}
\usepackage{tabularx}
\usepackage{array}
\usepackage{mathabx}
\usepackage{bibunits}
\usepackage{caption}

\usepackage[usenames,dvipsnames,svgnames,table]{xcolor}
\usepackage[colorlinks=true,
            linkcolor=blue,
            urlcolor=blue,
            citecolor=blue]{hyperref}

\numberwithin{equation}{section}
\numberwithin{theorem}{section}

\usepackage{xargs}

\usepackage{enumitem}
\newlist{steps}{enumerate}{1}
\setlist[steps, 1]{label = Step \arabic*:}

\usepackage{thmtools}
\usepackage{thm-restate}
\usepackage{hyperref}
\usepackage{cleveref}
\usepackage{authblk}
\def\spacingset#1{\renewcommand{\baselinestretch}%
{#1}\small\normalsize} 

\begin{document}
\begin{bibunit}[my-plainnat]

\title{Latent Multimodal Functional Graphical Model Estimation}
\author[1]{Katherine Tsai}
\author[2]{Boxin Zhao}
\author[3]{Sanmi Koyejo}
\author[2]{Mladen Kolar}
\affil[1]{Department of Electrical and Computer Engineering, University of Illinois Urbana-Champaign}
\affil[2]{Booth School of Business, University of Chicago}
\affil[3]{Department of Computer Science, Stanford University}

\renewcommand\Authands{ and }
\date{\today}
\maketitle
\begin{abstract}
Joint multimodal functional data acquisition, where functional data from multiple modes are measured simultaneously from the same subject, has emerged as an exciting modern approach enabled by recent engineering breakthroughs in the neurological and biological sciences. One prominent motivation to acquire such data is to enable new discoveries of the underlying connectivity by combining multimodal signals. Despite the scientific interest, there remains a gap in principled statistical methods for estimating the graph underlying multimodal functional data. To this end, we propose a new integrative framework that models the data generation process and identifies operators mapping from the observation space to the latent space. We then develop an estimator that simultaneously estimates the transformation operators and the latent graph. This estimator is based on the partial correlation operator, which we rigorously extend from the multivariate to the functional setting. Our procedure is provably efficient, with the estimator converging to a stationary point with quantifiable statistical error. Furthermore, we show recovery of the latent graph under mild conditions. Our work is applied to analyze simultaneously acquired multimodal brain imaging data where the graph indicates functional connectivity of the brain. We present simulation and empirical results that support the benefits of joint estimation. 
\end{abstract}


\section{Introduction}

Recent engineering breakthroughs have enabled new ways to acquire rich multimodal data from individual subjects. For example, high-throughput sequencing enables the acquisition of genotype, and gene expression, among other signals, from the same set of subjects~\citep{hao2021integrated}. Our work is motivated by emerging technology that simultaneously acquires data from electroencephalogram (EEG) and functional magnetic resonance imaging (fMRI)~\citep{morillon2010neurophysiological}. FMRI and EEG data are multivariate time series, where, after standard preprocessing~\citep{wirsich2020concurrent}, each dimension represents a region of the brain. We study the case where both the EEG and fMRI data are parcellated into the same atlas, resulting in the same number of dimensions. We view the time series of each region as a function of time, namely the functional data~\citep{ramsay2005fitting}, given the continuous underlying brain signals and the high sampling rate of the measurements. Our goal is to estimate the functional connectivity of the brain by solving a graph estimation problem~\citep{qiu2016joint,qiao2019functional}. 

There are several challenges in integrating multimodal functional data for our problem. 
First, estimation is often performed in a high-dimensional setting, where the ambient dimension is much larger than the sample size. Furthermore, multimodal data are often highly correlated across modes, since they measure related features of the same subject. Thus, many high-dimensional methods, such as graphical lasso~\citep{yuan2007model}, are not easily applied, as they require restrictions on correlations to achieve the desired statistical properties. Furthermore, in our setting of interest, the EEG and fMRI data are noisy and confounded by non-neural activity~\citep{murphy2013resting,goto2015head}. Therefore, the graphs estimated by either modality alone are inaccurate. Due to the observed structured noise and confounders, the graphs estimated separately from the two modalities may be dissimilar, although they contain partial information about the same underlying brain network~\citep{wirsich2020multi}. To address these shortcomings, we present a novel generative framework for observed processes (EEG-fMRI measurements) and latent processes (the underlying brain networks). Then, we provide a framework to jointly estimate the inverse operator mapping from the observed to the latent space, along with the latent graph, which encodes conditional independence. 

Our framework uses functional graphical models~\citep{qiao2019functional, zhao2021high} as a building block. We study the setting in which multimodal temporal data are viewed as functions belonging to different Hilbert spaces or subspaces of the same Hilbert space. This allows us to encode temporal characteristics using functional scores, a vector of real numbers obtained by projection to specific bases, effectively circumventing the need to address temporal discrepancies between modalities; for example, the bottom right of Figure~\ref{fig:illustration} indicates that fMRI and EEG recordings have distinct frequency characteristics, belonging to different subspaces of the space of continuous functions. Our estimation approach handles such a challenging setting by first constructing linear operators that transform the functional data from the observation spaces to a shared latent space. Then, we estimate the latent graph using functional neighborhood regression. Our algorithm jointly estimates the graphical model and the linear operators through an alternating iterative procedure.  
 
Our work offers several contributions.
From the modeling perspective, we propose an integrative function-on-function regression framework to estimate a latent functional Gaussian graphical model. We develop a novel initialization method inspired by the equivalence of the maximum log-likelihood estimator of the latent model and canonical correlation analysis~\citep{bach2005probabilistic},
\textcolor{black}{ when there are two data modalities. 
While our focus is on two data modalities, the proposed initialization method can be generalized to $M\geq 2$ by merging the data modalities, as discussed in Appendix~\ref{ssec:generalization}.}
Subsequently, we devise an efficient iterative method that boasts a linear rate of convergence. Specifically, our optimization analysis draws inspiration from recent advances in nonconvex optimization~\citep{zhang2018unified}. Regarding the theoretical contributions, we demonstrate that, under mild conditions, we can recover the latent graph with high probability. We showcase the effectiveness of the model through simulations and the analysis of concurrent EEG-fMRI. Empirical results indicate improved prediction scores when using the estimated latent graph.

\begin{figure}[t]
     \centering     \includegraphics[width=.9\textwidth]{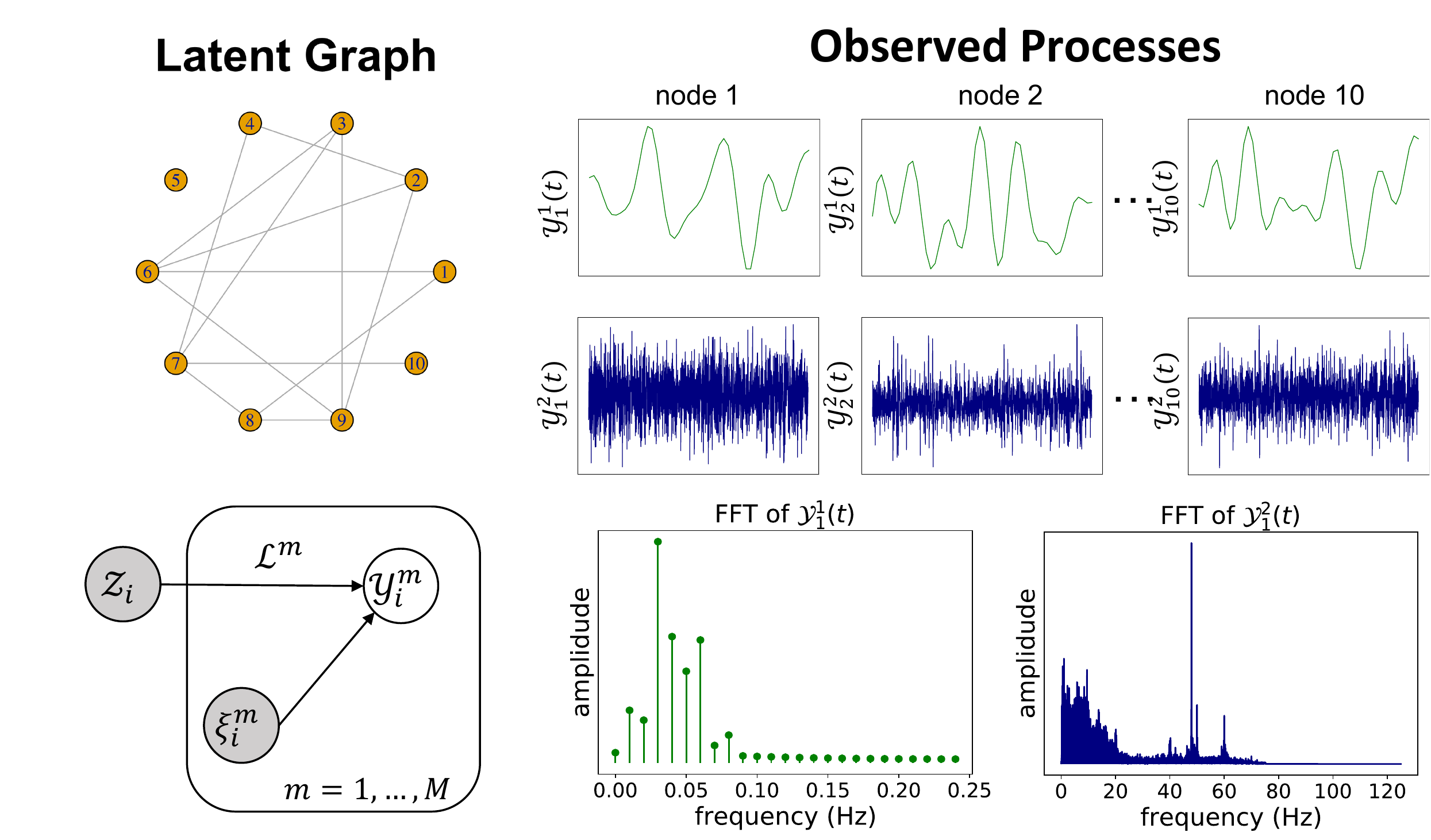}
     \caption{An illustration of the generative processes. 
     {\bf Top left}: The latent graph of interest. {\bf Bottom left}: The diagram of the generative model. {\bf Top right}: Observed processes from multiple sources. Simultaneous fMRI recordings, $\Ycal^1_i(t)$, and EEG recordings, $\Ycal_i^2(t)$, from~\citet{morillon2010neurophysiological}. 
     {\bf Bottom right}: The spectrum of fMRI dominates in the lower frequency domain while that of EEG spans to a higher frequency domain.}
     \label{fig:illustration}
\end{figure}

\section{Methodology}\label{sec:mathodology}

This section introduces a generative model for multimodal data. As our goal is to estimate the latent graphical model using regression, we establish the relation between function-on-function regression and conditional independence in the Gaussian setting.  We rigorously define the partial covariance operator and prove its equivalence to conditional independence. 

\subsection{Notation}

Given a separable Hilbert space $\HH$ of continuous functions with a Complete OrthoNormal System (CONS):  $\{\phi_\ell:\ell\in\NN\}$, the inner product is defined as $\dotp{f}{g}_{\HH}=\sum_{i=1}^\infty\dotp{f}{\phi_i}\dotp{g}{\phi_i}$, where $\dotp{f}{\phi_i}=\int_{t\in\Tcal} f(t)\phi_i(t)dt$. The induced norm is defined as $\norm{f}_{\HH}=(\dotp{f}{f}_{\HH})^{1/2}$ for any $f\in\HH$. We let $\HH^p$ be the Cartesian product of $\HH^{(1)}\times\cdots\times\HH^{(p)}$ where $\HH^{(j)} \equiv \HH$ for all $j=1,\ldots,p$ and its inner product is defined as $\dotp{f}{g}_{\HH^p}=\sum_{i=1}^p\dotp{f_i}{g_i}_{\HH}$ for any $f, g\in\HH^p$. Let $\mathfrak{B}(\HH_1,\HH_2)$ be the class of linear bounded operators from $\HH_1$ to $\HH_2$. 
A linear operator $\Kscr\in\mathfrak{B}(\HH_1,\HH_2)$ is in the equivalence class of a zero operator if $\Kscr\in\{\Ascr\in\mathfrak{B}(\HH_1,\HH_2):\dotp{f}{\Kscr g}=0,\;\forall g\in\HH_1, f\in\HH_2\}$. 
For any $\Tscr\in\mathfrak{B}(\HH_1,\HH_2)$, we define the Hilbert-Schmidt norm as $\hsnorm{\Tscr}^2=\sum_{\ell=1}^\infty\norm{\Tscr\phi_{1,\ell}}_{\HH_2}^2$, where $\{\phi_{1,\ell}:\ell\in\NN\}$ is a CONS for $\HH_1$ and $\opnorm{\Tscr}{}=\sup_{f\in\HH_1,f\neq 0}\norm{\Tscr f}_{\HH_2}/\norm{f}_{\HH_1}$.  $\Tscr$ is a Hilbert-Schmidt operator if $\hsnorm{\Tscr}<\infty$. Given an operator $\Tscr$, $\Image(\Tscr)$ is the image space of $\Tscr$.  
Let $\chi$ be a random element taking values in a measurable space $(\HH,\Bscr(\HH))$, where $\Bscr(\cdot)$ denotes the Borel $\sigma$-algebra. A random element $\chi$ in $L^2(\Omega,\Fscr,\mu)$ satisfies  
$\EE\norm{\chi}_\HH^2=\int_\Omega\norm{\chi}_\HH^2d\mu<\infty$.   Let $x_1,y\in\HH_1$ and $x_2\in\HH_2$, we define the tensor product $(x_2\otimes x_1):\HH_1\rightarrow \HH_2$ as $(x_2\otimes x_1)y=\dotp{x_1}{y}_{\HH_1}x_2$. The mean element is defined as $\mathfrak{m}=\int_{\Omega}\chi d\mu$ and the covariance operator $\Kscr\in\mathfrak{B}({\HH},{\HH})$ for $\chi$ is defined as $\Kscr=\EE[(\chi-\mathfrak{m})\otimes(\chi-\mathfrak{m})]=\int_\Omega(\chi-\mathfrak{m})\otimes(\chi-\mathfrak{m})d\mu$. Given a matrix $\Ab\in\RR^{m\times n}$ and for $p,q\geq 1$, we define $\norm{\Ab}_{p,q}=(\sum_{j=1}^n(\sum_{i=1}^m|a_{ij}|^p)^{q/p})^{1/q}$, $\norm{\Ab}_p=\sup_{\xb\neq 0}\norm{\Ab \xb}_p/\norm{\xb}_p$ and $\norm{\Ab}_F$ as the Frobenius norm of $\Ab$. 
We use $\mathbb{N}$ to denote positive integers.

\subsection{The Latent and Observed Processes}\label{ssec:observation}

Suppose that there are $M$ modalities of functional data, each living in a Hilbert space $\HH_m$, $\mseq$.  For modality $m$, we observe a $p$-dimensional random function $\Ycal^m=(\Ycal_{1}^m,\ldots,\Ycal_{p}^m)^\top\in\HH_m^p$, where $\HH_m$ is a separable Hilbert space of continuous functions defined on a closed interval $\mathcal{T}\subseteq \mathbb{R}$. Additionally, denote by $\HH$ the latent separable Hilbert space of continuous functions defined on the closed interval $\mathcal{T}$. Assume that $\Ycal^m$ is driven by some $p$-dimensional latent processes $\Zcal=(\Zcal_{1},\ldots,\Zcal_{p})\in\HH^p$ where $\Zcal_{i}:L^2(\Omega,\Fcal,\mu)\rightarrow(\HH,\Bscr(\HH))$ is a centered Gaussian random element. In neuroscience applications, we have $M=2$ modalities, ${\Ycal^m}$, $m=1,2$, which represent fMRI and EEG measurements, and $\Zcal$ is the latent functional brain process. 

Motivated by recent findings that EEG and fMRI measurements can be modeled as linear transformations of brain signals \citep{calhoun2012multisubject, chen2013dynamic}, we define ${\Lscr^m}\in\mathfrak{B}({\HH},{\HH_m})$, $\mseq$, as data generation operators that transform data from latent space to observed spaces. We assume that the transformation operator is the same across nodes; hence ${\Ycal_{i}^m}$ can be decomposed as:
\begin{align}\label{eq:multiview}
     \Ycal_{i}^m=\chi_{i}^m+\xi_{i}^m,\qquad \chi_{i}^m={\Lscr^m}\Zcal_{i},\qquad\pseq,\;\mseq,
\end{align}
where  $V=\{1,\ldots,p\}$ is the set of vertices, $\chi_{i}^m$ is obtained by a deterministic transformation of $\Zcal_{i}\in\HH$,
the noise $\xi_{i}^m$, with $\EE[\xi_{i}^m]=0$ and $\EE[\norm{\xi_{i}^m}_{\HH_m}^2]<\infty$, is the nuisance independent of $(\Zcal_{i})_{\pseq}$, $(\xi_{j}^m)_{j\in\pni}$, and $(\xi_{j}^{m'})_{m\neq m', j=1,\ldots,p}$.  
See Figure~\ref{fig:illustration} for an illustration of the model.
 
To avoid identifiability issues, 
we assume the following regularity conditions. 
\begin{assumption}\label{assumption:compact:L}
    Let ${\Lscr^m}$, $\mseq$, be compact operators. The set of eigenfunctions associated with nonzero distinct eigenvalues, i.e., $\lambda_1>\lambda_2>\ldots$, of ${\Lscr^m}^*{\Lscr^m}$ is denoted as $\{\phi_\ell\in\HH:\ell\in\NN\}$,  where ${\Lscr^m}^*$ is the adjoint operator of ${\Lscr^m}$. We assume that
\begin{equation}\label{eq:factor_assumption}
 z_{i,\ell}:=\dotp{ \Zcal_{i}}{\phi_{\ell}}_{\HH}\sim\Ncal(0,1),\qquad \ell\in\NN,
\end{equation}
and $z_{i,\ell}$ is not correlated with $z_{i,\ell'}$ for $\ell\neq\ell'$.
\end{assumption}
This assumption is common in factor model analysis~\citep[Chapter~10.4]{hsing2015theoretical}. Note that for $i\neq j$, $z_{i,\ell}$ and $z_{j,\ell'}$ may be correlated. The distinct eigenvalues assumption allows the set of eigenfunctions to be unique.

Let $V=\{1,\ldots,p\}$ be the set of vertices, and let $E=\{(i,j):i,j\in V,i\neq j\}$
be the set of edges. Our goal is to estimate the undirected latent graph $G=(V,E)$ underlying $\Zcal$ from the observed processes $\Ycal^{m}$, $\mseq$. 
The conditional independence for functional graphical models is defined as in~\citet{qiao2019functional}.
\begin{definition}\label{def:conditional_graph}
 A centered Gaussian random vector $\Zcal\in\HH^p$ follows a functional graphical model with respect to an undirected graph $G=(V, E)$ if we have
 \begin{align*}
 \Zcal_i\indep\Zcal_j\mid\Zcal_{-(i,j)} \text{ if and only if } (i,j)\not\in E,
 \end{align*}
 where $\Zcal_{-(i,j)} = \{\Zcal_{j'} : j' \in V\backslash\{i,j\}\}$ denotes the components of $\Zcal$ indexed by $V\backslash\{i,j\}$. 
\end{definition}
 
\citet{qiao2019functional} estimate $E$ by studying the inverse covariance operator of $\Zcal$, which generally does not exist in infinite dimensional space and can only be well-approximated under some restrictions on the eigenvalues of the covariance operator. In contrast, in the multivariate Gaussian setting, it is well known that conditional covariance is zero if and only if the partial covariance is zero, which motivates neighborhood regression as an estimator of conditional independence~\citep{meinshausen2006high, peng2009partial}. We develop a neighborhood regression estimator for the functional data setting, bypassing the need to compute the inverse covariance operator directly.

 \subsection{Functional Partial Covariance Operator}\label{ssec:latent_process}

Although the relation between the partial and conditional correlation has been implicitly stated in the RKHS setting~\citep{fukumizu2009kernel}, to our knowledge, the formal notion of partial covariance operators has not been clearly established in the literature on functional Gaussian graphical models. The form of the regression model that establishes the equivalence of conditional independence is unclear: when does a zero \emph{regression element} imply conditional independence? To answer this, we introduce the partial covariance operator and establish the corresponding equivalence with respect to conditional independence. 

The regression of $\Zcal_i$ on $\Zcal_{-(i,j)}$, is defined as 
\begin{align}\label{eq:rgression_b}
\{\tilde\beta_{ij'}\}_{j'\in V \backslash\{i,j\}} = \argmin_{\substack{\beta_{ij'}\in\mathfrak{B}(\HH,\HH)\\j'\in V \backslash\{i,j\}}}\EE\sbr{\bignorm{\Zcal_i-\sum_{j'\in V \backslash\{i,j\}}\beta_{ij'}\Zcal_{j'}}_{\HH}^2},
\end{align}
where $\beta_{ij}\in\mathfrak{B}(\HH,\HH)$ is a bounded linear operator. 
For any $f,g\in\HH$, we define the partial cross covariance operator $\Kscr_{ij\sbullet}$ as:
\begin{equation}\label{eq:partialcov_multi}
 \left\langle{\Kscr_{ij\sbullet}f}, {g}\right\rangle_\HH = \EE\sbr{
 \left\langle{\Zcal_i-\sum_{j'\in V \backslash\{i,j\}}\tilde\beta_{ij'}\Zcal_{j'}},{g}\right\rangle_\HH
 \left\langle{\Zcal_j-\sum_{j''\in V\backslash\{i,j\}}\tilde\beta_{jj''}\Zcal_{j''}},{f}\right\rangle_\HH
 },
\end{equation}
that is, the cross covariance operator of $\Zcal_{i}$ and $\Zcal_j$ after removing the effect of $\Zcal_{-(i,j)}$.
The partial cross covariance operator $\Kscr_{ij\sbullet}$ is a bounded linear operator, which can be shown using the same technique as in Lemma~\ref{lemma:bdd_partial}. 
We show the following result.
\begin{theorem}\label{theorem:partialcovariance_multi}   
The partial cross covariance operators
$\Kscr_{ij\sbullet}$ and $\Kscr_{ji\sbullet}$ are in the equivalence class of the zero operator if and only if $\Zcal_i\indep\Zcal_j\mid\Zcal_{-(i,j)}$ or, equivalently, if $(i,j)\not\in E$. 
\end{theorem}

This result allows us to estimate the partial covariance operator to measure the conditional independence of two nodes given the remaining nodes. An immediate corollary to Theorem~\ref{theorem:partialcovariance_multi} is that we can write $\Zcal_{i}$ as a linear combination of $\Zcal_j$, $j\in\pni$.

\begin{corollary}\label{corollary:partialcovariance_multi}
There exists $\beta_{ij}^\star\in\mathfrak{B}(\HH,\HH)$, $j\in\pni$, such that 
\begin{align}\label{eq:zlinear_relation}
\Zcal_{i}=\sum_{j\in\pni}\beta_{ij}^\star\Zcal_{j}+\Wcal_i,
\end{align}
and $\Wcal_i\in\HH$ is a Gaussian random function independent of $\Zcal_{j}$, $j\in\pni$. Furthermore, $\beta_{ij}^\star$ is in the equivalence class of the zero operator if and only if $\Zcal_i\indep\Zcal_j\mid\Zcal_{-(i,j)}$.
\end{corollary}

Based on Corollary~\ref{corollary:partialcovariance_multi}, specifically the relationship~\eqref{eq:zlinear_relation}, we can apply functional regression to measure conditional independence. However, since the space of bounded linear operators $\mathfrak{B}(\HH,\HH)$ is infinite dimensional, it is computationally intractable to estimate the parameters in~\eqref{eq:zlinear_relation}. Thus, it is important to study a subclass of models that can be well-approximated by finite-rank operators. As a consequence, we study the case where the true regression operators $\beta_{ij}^\star$ are compact and $\norm{\beta_{ij}^\star}_{\text{HS}}<\infty$.
\begin{assumption}\label{assumption:beta:HSoperator}
The operators $\beta^\star_{ij}$, $i,j\in V$ and $i\neq j$, are Hilbert-Schmidt operators.
\end{assumption}
Our formulation of functional neighborhood regression closely follows the formulation in~\citet{zhao2021high}.  Since $\beta_{ij}^\star\in\bophs{\HH}{\HH}$ is a Hilbert-Schmidt operator, it admits a singular system~\citep[Theorem 4.3.1]{hsing2015theoretical}. Hence, we are able to represent the singular system in terms of $\{\phi_\ell,\ell\in\NN\}$ and obtain
\begin{equation}\label{eq:bij} \beta_{ij}^\star=\sum_{\ell,\ell'=1}^\infty b_{ij,\ell\ell'}^\star \phi_{\ell} \otimes \phi_{\ell'},\quad i,j\in V,i\neq j,
\end{equation}
where $b_{ij,\ell\ell'}^\star\in\RR$. Under Theorem~\ref{theorem:partialcovariance_multi}, one can verify that $\Zcal_{i}\indep\Zcal_{j}\mid\Zcal_{-(i,j)}$ if and only if $\norm{\beta_{ij}^\star}_{HS}=0$.  Therefore, the set of neighbors of $\Zcal_{i}$ 
is defined as 
\begin{align}\label{eq:def_pNi}
\pNi = \{j\in\pni:\norm{\beta_{ij}^\star}_{HS}>0\}. 
\end{align}

In the high-dimensional setting, where $p$ is larger than the sample size $N$, we often consider that the coefficients are sparse~\citep{meinshausen2006high}, rendering a sparse network structure. Such structure is often observed in biological and clinical experiments~\citep{tsai2022joint}. In the functional data scenario, a sparse functional network is equivalent to assuming that many of $\beta_{ij}^\star$ are zero operators:
\begin{assumption}\label{assumption:s-sparse}
    The set $\{\beta_{ij}^\star\}_{j\in\pni}$ has at most $s^\star$ nonzero operators. 
\end{assumption}
Assumption~\ref{assumption:s-sparse} implies that 
the neighborhood $\pNi$, defined in~\eqref{eq:def_pNi}, satisfies
$\abr{\pNi}\leq s^\star$. 

For $\beta_{ij}^\star$ not a zero operator, we expect that $\beta_{ij}^\star$ can be well approximated by some finite-rank operator: consider a $k$-dimensional subspace of the original space, denoted as $\{\phi_\ell:\ell=1,\ldots,k\}$, we want $\beta_{ij}^{k\star}:=\sum_{\ell,\ell'=1}^kb_{ij,\ell\ell'}^k\phi_\ell\otimes\phi_{\ell'}\approx\beta_{ij}^\star$. 
Observe that $\norm{\beta_{ij}^\star}_{HS}=\norm{\beta_{ij}^\star-\beta_{ij}^{\star {k}}+\beta_{ij}^{\star {k}}}_{HS}\leq \norm{\beta_{ij}^\star-\beta_{ij}^{\star {k}}}_{HS}+\norm{\beta_{ij}^{\star {k}}}_{HS}$. 
Hence, if there exists a positive $\epsilon_0$ such that  $\min_{i\in V}\min_{j\in\pNi}\norm{\beta_{ij}^\star}_{HS}\geq2\epsilon_0$,
then we can choose a ${k}$ such that the truncated signal still has half the magnitude of the original signal: $\norm{\beta_{ij}^{{k\star}}}_{HS}\geq\epsilon_0$.  We define
\[
\Lambda(k) :=\min_{\pseq}\min_{j\in\pNi}\norm{\beta_{ij}^{k \star}}_{HS}.
\]

In addition, for $\beta_{ij}^\star$ not a zero operator, we expect that most of $b_{ij,\ell\ell'}^\star$, $\ell,\ell'\in\NN$, are zero or have a small magnitude.  We make the following assumption on $\beta^\star_{ij}$.
\begin{assumption}\label{assumption:disperse}
The rank $k$ is selected such that $\Lambda(k)\geq(1/2)\min_{\pseq}\min_{j\in\pNi}\norm{\beta_{ij}^{ \star}}_{HS}=:\epsilon_0$.
There exists a constant $\alpha^\star\in(0,1]$ such that the sets $\{b_{ij,\ell 1}^\star,\ldots ,b_{ij,\ell k}^\star\}$ and $\{b_{ij, 1\ell'}^\star,\ldots ,b_{ij,k\ell'}^\star\}$ each have at most $\alpha^\star k$ non-zero coefficients,  for $\ell,\ell'=1,\ldots,k$.
\end{assumption}
The first assumption implies that~\eqref{eq:def_pNi} is equivalent as
\begin{align}\label{eq:def_pNi2}
    \pNi = \{j\in\pni:\norm{\beta_{ij}^{k \star}}_{HS}>\epsilon_0\}.    
\end{align}
The second structural assumption implies that each basis function is only correlated with at most $\alpha k-1$ other basis functions but we allow the rank to grow with $k$.  From the neuroscience perspective~\citep{olsen2012hippocampus, vidaurre2017brain}, the first level of sparsity, Assumption~\ref{assumption:s-sparse}, corresponds to the sparse connectivity between sub-networks of the whole connectome.
\textcolor{black}{We consider $\phi_\ell$ to be the $\ell$-th cognitive process~\citep{posner1988localization}, e.g., visual imaging, word reading, shifting visual attention, and etc. Assumption~4 implies two connectivity structures, as outlined in the following. First, each node $i$ only has a few ongoing cognitive processes. Secondly, each connectivity process $\ell$ in node $i$ is sparsely correlated with the other cognitive process $\ell'$ in node $j$~\citep{park2013structural, vidaurre2017brain}. This assumption further implies that the operator $\beta_{ij}^\star$ is approximately low-rank, because many of the $\{b_{ij,\ell\ell'}^\star\}$ are zero. 
    Additionally, the variability of sparsity patterns imposed by Assumption~3--4 are subject to the cognitive tasks performed. For example,  ~\citet{cohen2016segregation} have found that during the motor cognitive process, within sub-network connectivity is highly activated, meaning that $\alpha^\star$ may be high.  In contrast, during the memory task, between-network communication is more active, meaning that $s^\star$ might be high.
}
While the former assumption is well recognized in functional graphical model literature~\citep{qiao2019functional}, we argue that adding the second level sparsity assumption provides more informative interpretability of the underlying functional processes.

 \section{Estimation}\label{sec:estimation}
 
We describe the estimation of the parameters for the generative model described in Section~\ref{sec:mathodology}. We first introduce an estimation procedure for the transformation operator, which we formulate as an inverse problem. Then, the objective function in infinite dimensional space is outlined in Section~\ref{ssec:infinite_estimator}. Finally, the objective function in a finite dimensional space along with the structured assumptions are discussed in Section~\ref{ssec:OF_SA}.

\subsection{Transformations to the Latent Space}

Our problem is to estimate the graphical model $G=(V, E)$ in the latent space through data from multiple modalities $\Ycal^{1},\ldots,\Ycal^{M}$. While the mathematical formulation in~\eqref{eq:multiview} outlines the generative model, from the estimation point of view, we focus on its inverse, that is, the linear operators that transform the functional processes from the observation spaces to the latent space. In general, the inverse operator of ${\Lscr^m}$ does not exist; even if it does exist, it might not be continuous. However, the inverse operator is well defined when the operator domain is restricted to ${\Image({\Lscr^m})}$, giving us hope to recover the inverse of ${\Lscr^m}$ on the restricted domain. 
Following Definition 3.5.7 in~\citet{hsing2015theoretical}, we define the Moore-Penrose (generalized) inverse of ${\Lscr^m}$.
\begin{definition} 
 Let $\tilde{\Lscr}^m$ be the operator ${\Lscr^m}$ restricted to $\ker({\Lscr^m})^\perp$, $\mseq$.
 The domain of ${\Lscr^m}^\dagger$ is defined as $\Domain({\Lscr^m}^\dagger)=\Image({\Lscr^m})\oplus\Image({\Lscr^m})^\perp$. Then, for any $y\in\Domain({\Lscr^m}^\dagger)$, ${\Lscr^m}^\dagger y=(\tilde{\Lscr}^m)^{-1} y$ if $y\in\Image({\Lscr^m})$ and ${\Lscr^m}^\dagger y=0$ if $y\in\Image({\Lscr^m})^\perp$.
\end{definition}

We use ${\Ascr^m}$ to denote the pseudoinverse operator ${\Lscr^m}^\dagger$ for simplicity of notation. 
Since ${\chi_{i}^m}\in\Image({\Lscr^m})$, we have the following relationship:
 \begin{equation}\label{eq:relationofA}
 \Zcal_{i}={\Ascr^m}(\chi_{i}^m),\qquad \mseq,\;\pseq.
 \end{equation}
The operator ${\Ascr^m}$ maps the random elements from the observation spaces back to the latent space, where we can jointly estimate the latent graph from multiple views $\Ycal^{1},\ldots,\Ycal^{M}$.

The following representation of ${\Ascr^m}$ will be useful for estimation from the observed data. Let $\{\phi_\ell^m\in\HH_m:\ell\in\NN\}$ be the set of eigenfunctions of $\Lscr^m\Lscr^{m*}$. Since ${\Lscr^m}\in\mathfrak{B}({\HH},{\HH_m})$ is a compact operator, when the domain is restricted to $\Image({\Lscr^m})\oplus\Image({\Lscr^m})^\perp$ we have:
\begin{equation}\label{eq:Aform}
{\Ascr^m}=\sum_{\ell=1}^\infty\sum_{\ell'=1}^\infty a_{\ell\ell'}^{m\star}\phi_{\ell}\otimes\phi_{\ell'}^m,\quad {\Ascr^m}\chi = \sum_{\ell=1}^\infty\sum_{\ell'=1}^\infty a_{\ell\ell'}^{m\star}\dotp{\chi}{\phi_{\ell'}^m}_{\HH_m}\phi_{\ell},
\end{equation}
where $a_{\ell\ell'}^{m\star}\in\RR$ are coefficients 
and $\chi\in\Domain({\Ascr^m})$. 

Although ${\Ascr^m}$ is well defined on the restricted domain, it is not necessarily a bounded operator. Since $\Lscr^m$ is a compact operator, to have a bounded inverse ${\Ascr^m}$, $\Lscr^m$ must be finite dimensional~\citep[Corollary~1.3.3]{groetsch1984theory}. The assumption is mild in our application because the rank of ${\Lscr^m}$ can be arbitrarily large, so long as it is finite.

\subsection{Approximation of Infinite Dimensional Estimator}\label{ssec:infinite_estimator}

When the noise ${\xi_{i}^m}$ in~\eqref{eq:multiview} is small in magnitude compared to ${\chi_{i}^m}$, then $\Zcal_i={\Ascr^m}{\chi_{i}^m}\approx {\Ascr^m}{\Ycal_{i}^m}$. We combine~\eqref{eq:rgression_b} and the estimation of ${\Ascr^m}$ and optimize the objective:
\begin{equation}\label{eq:infinite_obj}
\min_{\{{\Ascr^m}\},\{\beta_{ij}\}}\qquad \sum_{i=1}^p\sum_{m=1}^M\EE\sbr{\bignorm{{\Ascr^m}{\Ycal_{i}^m}-\sum_{j\in\pni}\beta_{ij}{\Ascr^m}\Ycal_{m,j} }_\HH^2},
\end{equation}
where ${\Ascr^m}$ is of the form~\eqref{eq:Aform} and $\beta_{ij}$ is of the form~\eqref{eq:bij}. 

The temporal realization of random elements $\Zcal_{i}(\omega,\cdot)$ and ${\Ycal_{i}^m}(\omega,\cdot)$, $\omega\in\Omega$, are bivariate functions $\Zcal_{i}(\omega,t)$ and ${\Ycal_{i}^m}(\omega,t)$ that we assume are jointly measurable with respect to the product $\sigma$-field $\Fcal\times\Bscr(\Tcal)$.
Since both $\Zcal_i$ and ${\Ycal_{i}^m}$ are mean-squared integrable random elements taking values in spaces of continuous functions, by the Karhunen–Lo\`eve theorem~\citep{karhunen1947uber, loeve1945fonctions}, we can uniquely represent them as
\begin{equation}\label{eq:randomelement_p}
 \Zcal_{i}(\omega, t)=\sum_{\ell=1}^\infty z_{i,\ell}(\omega)\phi_{\ell}(t),\qquad
 {\Ycal_{i}^m}(\omega, t)=\sum_{\ell=1}^\infty y_{i,\ell}^m(\omega)\phi_{\ell}^m(t)
 ,\qquad \omega\in\Omega,\;t\in\Tcal,
\end{equation}
where we recall that $z_{i,\ell}(\omega)$ is defined in~\eqref{eq:factor_assumption} and $y_{i,\ell}^m(\omega)=\dotp{{\Ycal_{i}^m}(\omega,\cdot)}{\phi_{\ell}^m(\cdot) }_{\HH_m}$. 
For conciseness, we omit $\omega$ from the notation on $y_{i,\ell}^m(\omega)$ and $z_{i,\ell}(\omega)$, and $t$ on $\phi_\ell(t)$ and $\phi_{\ell}^m(t)$.

\begin{proposition}\label{prop:infinite_program}
The optimization problem in~\eqref{eq:infinite_obj} is equivalent to 
 \begin{equation}\label{eq:obj_infinite}
 \min_{\{a_{\ell\ell'}^m\},\{b_{ij,\ell\ell'}\}}\qquad \sum_{m=1}^M\sum_{i=1}^p\EE\sbr{\sum_{\ell\in\NN}
 \rbr{\sum_{\ell''\in\NN}a_{\ell\ell''}^m y_{i,\ell''}^m
 - 
 \sum_{j\in\pni}\sum_{\ell',\ell''\in\NN} b_{ij,\ell\ell'}a_{\ell'\ell''}^my_{j,\ell''}^m
}^2},
\end{equation}
where ${\Ascr^m}=\sum_{\ell,\ell'}a_{\ell\ell'}^m\phi_{\ell}\otimes\phi_{\ell'}^m$ and $\beta_{ij}=\sum_{\ell,\ell''}b_{ij,\ell\ell''}\phi_{\ell}\otimes\phi_{\ell''}$.
\end{proposition}

Based on Proposition~\ref{prop:infinite_program}, the optimization over the operators ${\Ascr^m}$ and $\beta_{ij}$  in~\eqref{eq:infinite_obj} is transformed to optimization over sequences of real numbers $\{a_{\ell\ell'}^{m}\}_{\ell,\ell'\in\NN}$ for $\mseq$ and $\{b_{ij,\ell\ell'}\}_{\ell,\ell'\in\NN}$ for $i\in V$ and $j\in\pni$.
While functional data live in infinite dimensional spaces, in practice, one is faced with finite computational resources and is forced to truncate both the latent functions and observed functions to perform estimation. In addition,  Assumption~\ref{assumption:s-sparse}--\ref{assumption:disperse} indicates that most of $b_{ij,\ell\ell'}^\star$ are zero. 
 Hence, using the result from~\eqref{eq:randomelement_p}, the ${k}$-dimensional approximation of $\Zcal_{i}$ via projection to $\{\phi_1,\ldots,\phi_{k}\}$  is $\sum_{\ell=1}^{ k}z_{i,\ell}\phi_{\ell}$. 
For the observed space, we assume that the signals are smooth enough such that the $k_m$-dimensional projection $\sum_{\ell=1}^{ k_m}y_{i,\ell}^m\phi_{\ell}^m$ can well approximate ${\Ycal_{i}^m}$. These observations imply that the following optimization problem can recover the model parameters sufficiently well:
\begin{equation}\label{eq:population:finite:program}
\Ab^\star, \Bb^\star:=\arg\min_{\Ab,\Bb} \sum_{m=1}^M\sum_{i=1}^p\EE\sbr{\bignorm{\Amk\ykmfree{i}-\sum_{j\in\pni}\Bijk\Amk\ykmfree{j}}_F^2},
\end{equation}
where $\Ab^\star=\{\pAmk\}_{m=1}^M$, 
$\Bb^\star=\{\pBik\}_{i=1}^p$, 
$\yb_{i}^m=(y_{i,1}^m,\ldots,y_{i,k_m}^m)^\top\in\RR^{k_m}$, 
$\Amk=(a_{\ell\ell'}^m)_{\ell=1,\ell'=1}^{k,k_m}$, and 
$\Bijk=(b_{ij,\ell\ell'})_{\ell=1,\ell'=1}^{k,k}\in\RR^{k\times k}$. { The parameters $k$ and $k_m$ are suitably chosen -- we discuss the selection of $k$ and $k_m$ in more detail in Section~\ref{sec:parameterselection}.}

 

\subsection{Empirical Objective Function and Structured Assumptions}\label{ssec:OF_SA}

Throughout the section, we assume the sets of basis $\{\phi_1^m,\ldots,\phi_{k_m}^m\}$, $\mseq$ are provided. We are given independent observations from $N$ individuals, where for each individual $n$ and each vertex $i$, we have data from $M$ modalities, $\{\mathcal{Y}_i^{m, (n)}\}_{m=1}^M$. The vector $\ykmnfree{i}{n}$ is obtained by projecting $\mathcal{Y}_i^{m, (n)}$ onto the set of basis $\{\phi_1^m,\ldots,\phi_{k_m}^m\}$, and we let $\Ykmnfree{i}=(\ykmnfree{i}{1},\ldots, \ykmnfree{i}{N})\in\RR^{k_m\times N}$. 
The sample version of the objective in~\eqref{eq:population:finite:program} is:
\begin{align}\label{eq:f_n}
 f(\setAk,\setBk) &= \sum_{i=1}^p\sum_{m=1}^M\frac{1}{2N}\bignorm{\Amk \Ykmnfree{i}-\sum_{j\in\pni}\Bijk\Amk \Ykmnfree{j}}_F^2.
\end{align}
To enforce the sparsity of $\setBk$, we define the following norm.
\begin{definition}\label{def:rnorm}
Let $\Ub=(\Ub_{1}\cdots\Ub_{p})\in\RR^{d\times kp}$, where $\Ub_i\in\RR^{d\times k}$. 
We define the norm
\[
\norm{\Ub}_{r(k,\ell)}
=
\left\|\begin{pmatrix}\VEC\rbr{\Ub_1}\cdots\VEC\rbr{\Ub_p}\end{pmatrix}\right\|_{2,\ell},
\quad \ell\geq 1,\]
and let $\norm{\Ub}_{r(k,0)}=\sum_{i=1}^p\mathbbm{1}(\norm{\VEC\rbr{\Ub_i}}_2>0)$, where $\mathbbm{1}(\cdot)$ denotes the indicator function.
\end{definition}
Let $\pBik = (\pBikfree{1}\cdots\pBikfree{(i-1)}\pBikfree{(i+1)}\cdots\pBikfree{p})\in\RR^{k\times k(p-1)}$. 
Under Assumption~\ref{assumption:s-sparse}, we have that $\norm{\pBik}_{r(k,0)}\leq s^\star$. 
Assumption~\ref{assumption:disperse} implies that each nonzero submatrix $\pBikfree{j}$ of $\pBik$ has columns and rows with at most $\alpha$-fraction of nonzero entries.  We define the following constraint set
\[
\Kcal_B(s,\alpha) = \{ \Ub:\norm{\Ub}_{r(k,0)}\leq s,\norm{\Ub_{j,m\cdot}}_0, \norm{\Ub_{j,\cdot n}}_{0}\leq \alpha k, m,n=1,\ldots,k,j=1,\ldots,p-1\},
\]
where $\Ub_{j}\in\RR^{k\times k}$ denotes the sub-matrix of $\Ub\in\RR^{k\times k(p-1)}$.
We assume that $\Amk\in\RR^{k\times k_m}$ satisfies the following constraints. Let $\Wb_{j\cdot}$ be the $j$-th row of a matrix $\Wb$, we define
\begin{equation}\label{eq:constraint_a}
\Kcal_m(\tau_1,\tau_2)=\cbr{\Wb:\tau_1^{1/2}\sinmin{\Ab^{m(0)}}\leq
\norm{\Wb_{ j\cdot}}_2
\leq(\tau_2/k_m)^{1/2}\sinmax{\Ab^{m(0)}},\;j=1,\ldots k},
\end{equation}
where $\Ab^{m(0)}$ is the initial guess. 
We will later verify that there exist some $\tau_1,\tau_2>0$ such that $\pAmk $ lies in $\Kcal_m(\tau_1,\tau_2)$. 
The lower bound $\tau_1$ ensures that any row of $\Amk$
is always bounded away from zero.\footnote{We need to make sure that $\Amk$ is nonzero so that the solution is nontrivial, since $\Amk={\bf 0}$ and $\Bik={\bf 0}$ for $\mseq$ and $\pseq$ is always an optimal solution in the unconstrained optimization problem.} 
If $\tau_2$ is proportional to $k^{1/2}$, then the upper bound constraint can be viewed as the incoherence condition~\citep{ candes2011robust}.  
We want to optimize the following objective function:
 \begin{align}\label{eq:mainobj}
     \argmin_{\setAk,\setBk}&\quad f(\setAk,\setBk),\\
     &\quad\text{s.t.}\quad\Bik\in\Kcal_B(s,\alpha),\; \pseq;\notag\\
     &\quad\quad\quad\;\Amk\in\Kcal_m(\tau_1,\tau_2),\;\mseq\notag.
 \end{align}

\section{Algorithm}\label{sec:algo}

\begin{algorithm}[t!]
\spacingset{1}
 \caption{Latent Graph Estimation}\label{alg:update}
\begin{algorithmic}
\State Input: { $\{\Yb_m\}_{\mseq}$,$\{\Ab^{m(0)}\}_{\mseq}$, $\{\Bb_i^{(0)}\}_{\pseq}$,$\tau=(\tau_1,\tau_2)$, $s$, $\alpha$, $\eta_A$, $\eta_B$; }
\State Output: {$\{\hat\Ab^m\}_{\mseq}$, $\{\hat{\Bb}_i \}_{\pseq}$ }
 \While{Not converged}
 \For{\mseq}
 \State$\Ab^{m(t+1)}\leftarrow \Pcal_{m,\tau}(\Ab^{m(t)} - \eta_A\nabla_{\Amk }f$)\;
 \EndFor
 \For{$\pseq$}
 \State$\Bb_i^{(t+.5)}\leftarrow \Tcal_{s}(\Bb_i^{(t)} - \eta_B\nabla_{\Bik}f$)\;
     \For{$j\in\pni$}
\State$\Bijk^{(t+1)}\leftarrow\Hcal_{\alpha}(\Bijk^{(t+.5)})$\;
   \EndFor
 \EndFor
 \EndWhile
 \end{algorithmic}
\end{algorithm}

We propose a two-stage algorithm to minimize~\eqref{eq:mainobj} efficiently. 
We start by introducing the second stage. 
Given a suitably chosen initial tuple $(\Ab^{(0)},\Bb^{(0)})$ from the first stage, 
we use the alternating projected gradient descent with the group-sparse hard-thresholding operator. We first define the $s$-group sparse truncation operator as
 \begin{align*}\label{eq:Tcal}
 [\Tcal_s(\Bik)]_{\cdot,k(j-1)+1:kj}=\left\{\begin{array}{cc}
      \Bijk\in\RR^{k\times k},&  \text{the magnitude $\norm{\Bijk}_F$ is among the top-s of $\{ \Bijk \}$};\\
      {\bf 0}\in\RR^{k\times k},& \text{otherwise}.
 \end{array}\right.
 \end{align*}
The operator $\Tcal_s$ keeps the largest $s$ submatrices of $\Bik\in\RR^{k\times k(p-1)}$.
Next, we define the $\alpha$-truncation operator as 
\begin{equation*}\label{eq:Hcal}
 [\Hcal_{\alpha}(\Bijk)]_{u,v}=\left\{\begin{array}{ll}
    [\Bijk]_{u,v}, & \text{$[\Bijk]_{u,v}$ is one of the $\alpha k$ largest elements in}\\
    &\text{magnitude of $[\Bijk]_{u,\cdot}$ and $[\Bijk]_{\cdot,v}$.}\\
    0,  & \text{otherwise}.
 \end{array}\right.
\end{equation*}
The operator $\Hcal_\alpha$ keeps the largest $\alpha$-fraction of entries in each row and column of $\Bijk\in\RR^{k\times k}$.
Then, the projection to the set $\Kcal_B(s,\alpha)$ can be implemented by the composition of $\Tcal_s$ and $\Hcal_\alpha$, as shown in Algorithm~\ref{alg:update}. The operator $\Pcal_{m,\tau}$ as the projection operator to the set $\Kcal_{m}(\tau_1,\tau_2)$ is defined as 
\begin{align*}
     \Pcal_{m,\tau}(\Amk)=\argmin_{\Wb\in\Kcal_m(\tau_1,\tau_2)}\norm{\Wb-\Amk}_F.
\end{align*}
After Algorithm~\ref{alg:update} converges, we select the edges of the graph either using the AND or OR operation with a threshold $\epsilon_0>0$:
 \begin{equation}\label{eq:edge_thre}
\hat{\Nscr}_i(\epsilon_0)=\{(i,j):\norm{\hat{\Bb}_{ij}}_F\geq\epsilon_0\text{ and (or) }\norm{\hat{\Bb}_{ji}}_{F}\geq\epsilon_0\}.
 \end{equation}

\subsection{Initialization Procedure}\label{ssec:initialization}

We describe a procedure to find a good initial tuple $(\{\Ab^{m(0)}\}_{m=1}^M,\{\Bb_{i}^{(0)}\}_{i=1}^p)$. 
We focus on the case when $M=2$, as this is the case for the motivating example. 
To find initial estimates, we connect the problem with probabilistic canonical correlation analysis.
We first construct $\Amk^{(0)}$ by finding the canonical correlation between two views and then compute $\Bik^{(0)}$ with an iterative method by fixing $\Amk$.

Let $\Lb^{m\star}\in\RR^{k_m\times k}$ be such that the $\ell\ell'$-th entry of  $\Lb^{m\star}$ is $L_{\ell\ell'}^{m\star }=\dotp{{\Lscr^m}\phi_{\ell'}^m}{\phi_{\ell}^m}$.  Let ${\qb}_{i}^m=(q_{i,1}^m,\ldots,q_{i,k_m}^m)^\top\in\RR^{k_m}$, where $q_{i,j}^m=\dotp{\xi_{i}^m}{\phi_{j}^m}$, be the truncation of $\xi_{i}^m$; ${\mu}_{i}^m=(\mu_{i,1}^m,\ldots,\mu_{i,k_m}^m)^\top\in\RR^{k_m}$ be the bias induced by finite-dimensional truncation and $\mu_{i,\ell'}^m=\sum_{\ell=k+1}^\infty\dotp{{\Lscr^m} (z_{i,\ell}\phi_{\ell})}{\phi_{\ell'}^m}$ for $\ell'=1,\ldots,k_m$. Recalling that $\zb_i \sim \Ncal({\bf 0},\Ib_k)$, it follows that 
\begin{align}\label{eq:y_z}
    \yb_{i}^m|\zb_i
    &\sim\Ncal(\Lb^{m\star}\zb_i, \Sigmab^{m,\qb}_{i,i}
    +
    \Sigmab^{m,{\bm\mu}}_{i,i}
    ),
\end{align}
where $\Sigmab^{m,{\bm\mu}}_{i,i}$ is the covariance of $\mu_i^{m}$ and $\Sigmab^{m,{\qb}}_i$ is the covariance of $\qb_i^m$. 
From~\eqref{eq:factor_assumption} we have that $\bm\mu_{i}^m$ is uncorrelated with $\zb_i$. Recall that we assume that $\Lb^{m\star}$ is the same across $\pseq$.
We briefly discuss how extend to the setting where $\Lb_{i}^{m\star}$ differ across the nodes in Appendix~\ref{ssec:extension}.
To estimate $\Lb^{m\star}$ from data, we pick the first node $i=1$ for convenience.  
We discuss an alternative to this strategy in Appendix~\ref{sec:appendix:aggregate}
The log-likelihood is defined as
\begin{align*}
    \ell({\bf S}) = 
    (k_1+k_2)\log2\pi
    &+\log\abr{{\bf S}}+\tr({\bf S}^{-1}\hat\Sigmab),\\
{\bf S} = \begin{pmatrix}
\Lb^1{\Lb^1}^\top+\Sigmab^{1,\qb}_{11}&\Lb^1{\Lb^2}^\top\\
\Lb^2{\Lb^1}^\top &\Lb^2{\Lb^2}^\top+\Sigmab^{2,\qb}_{11}
\end{pmatrix},&
\quad\hat\Sigmab = \frac{1}{N}\sum_{n=1}^N
\begin{pmatrix}
\yb_{1}^{1,(n)}\yb_{1}^{1,(n)\top}&
\yb_{1}^{1,(n)}\yb_{1}^{2,(n)\top}\\
\yb_{1}^{2,(n)}\yb_{1}^{1,(n)\top}&
\yb_{1}^{2,(n)}\yb_{1}^{2,(n)\top}
\end{pmatrix}.
\end{align*}
Let $\Sigmab^m_{11}$, $\hat{\Sigmab}^m_{11}$ be the covariance matrix and sample covariance of node $1$ for modality $m$, respectively. Let $\Sigmab^{12}_{11}$, $\hat{\Sigmab}^{12}_{11}$ be the cross-covariance matrix and sample cross-covariance matrix for two modalities at node $1$, respectively. By Theorem 2 in~\citet{bach2005probabilistic}, the maximum log-likelihood estimator for $\Lb^{m\star}$ is
$
\hat\Lb^m=\sCovfree{1}{1}\hat\Vb^{m} \hat\Gammab^{1/2}
$, where
$\hat\Gammab$ is a diagonal matrix whose diagonal entries are the top-k singular values of $\hat\Rb^{12}=(\hat\Sigmab^1_{11})^{-1/2}\hat\Sigmab^{12}_{11}(\hat\Sigmab^2_{11})^{-1/2}$, columns of $\hat\Vb^{1}$ are the corresponding top-k left singular vectors and columns of $\hat\Vb^{2}$ are the corresponding right singular vectors. The initial estimate $\Amk$ is obtained by taking the pseudo-inverse of $\hat\Lb^m$, denoted as
\begin{align}
\Amk^{(0)} = \Gammab^{-1/2}\hat\Vb^{m\top}(\sCovfree{1}{1})^{-1/2}\label{eq:cca_model}.
\end{align}
After obtaining $\Amk^{(0)}$, we compute the initial estimate for $\defsetBk$. We solve the following constrained optimization problem:
\[
        \min_{\substack{\Bik\in\Kcal_B(s,\alpha)}}h_i(\Bb_i)=
        \min_{\substack{\Bik\in\Kcal_B(s,\alpha)}}
        \sum_{m=1}^M\frac{1}{2MN}
        \bignorm{
            \Amk^{(0)}\Ykmnfree{i}  
            -
            \Bik\bigAd{p-1}{\Amk^{(0)}}\Ykmnfree{\noi}
        }_F^2
\]
using the projected gradient descent described in Algorithm~\ref{alg:initializationB}.
\begin{algorithm}[t!] 
\spacingset{1}
\caption{Initialization of $\setBk$}\label{alg:initializationB}
\begin{algorithmic}
\State Input: { $\{\Yb_m\}_{\mseq}$,$\{\Ab^{m(0)}\}_{\mseq}$, $\{\Bb_{0,i}={\bf 0}_{i\in V}\}$, $s$, $\alpha$, $\eta_{B_0}$;}
     \For{$\pseq$, $t=1,\ldots,L-1$}
    \State $\Bb_{0,i}^{(t+.5)}\leftarrow \Tcal_{s}(\Bb_{0,i}^{(t)} - \eta_{B_0}\nabla_{\Bb_{0,i}}h)$\;
    \For{$j\in\pni$}
    \State $\Bb_{0,ij}^{(t+1)}\leftarrow\Hcal_{\alpha}(\Bb_{0,ij}^{(t+.5)})$\;
   \EndFor
    \EndFor
    \State $\Bb_i^{(0)}\leftarrow \Bb_{0,i}^{(L)}$
\end{algorithmic}
\end{algorithm}

\subsection{Selection of Basis Parameters}\label{sec:parameterselection}

We select $k_m$ based on the projection score to the basis, where we employ the elbow method to decide $k_m$. Specifically, we compute the mean-squared error of the projected signals with the original signals and pick the elbow point. Then, we compute the canonical correlation of two views and select $k$ based on the canonical correlation score using the elbow method. 
\textcolor{black}{Although Theorem~\ref{theorem:main} requires knowledge of $\varepsilon_0$, in practice we find setting $\varepsilon_0$ to be a small constant $10^{-3}$ suffices to give stable and reproducible results in various simulations and real data}.
The parameters $s$, $\alpha$, $\tau_1$, $\tau_2$ are selected based on the $5$-fold cross-validation with the BIC score function used in the functional regression~\citep{zhao2021high}:
\begin{align*}
&\BIC(s,\alpha, \tau_1,\tau_2) =
\sum_{i=1}^p\cbr{\sum_{m=1}^MN\log\det\rbr{\frac{1}{2N}\Gb_{i}^m\Gb_{i}^{m\top}}
+|\hat\Nscr_i(s,\alpha, \tau_1,\tau_2)|\log N},
\end{align*}
where $\Gb_{i}^m=\Amk \Ykmnfree{i}-\sum_{j\in\pni}\Bijk\Amk \Ykmnfree{j}$ for $\pseq$ and $\mseq$.

\section{Theory}\label{sec:theory}

We show that Algorithm~\ref{alg:update}--\ref{alg:initializationB} can recover the underlying latent graphs with high probability under mild conditions.
The quality of the graph estimate depends on how well we can estimate $\setpAk$ and $\setpBk$.
We start by showing the convergence guarantee for both parameters. 
The convergence analysis provides the rate of the algorithmic convergence, along with the statistical error at the stationary points.

In Section~\ref{sec:theory_initial}, we assess the quality of the initial estimates. The theory suggests the consistency of the estimator with a small truncation error. The convergence guarantee of Algorithm~\ref{alg:update} is provided in Section~\ref{sec:theory_convergence} along with the quantification of the statistical error at the convergence points. Section~\ref{ssec:graph_recovery} presents the guarantee of latent graph recovery. 

\subsection{Analysis of the Initialization}\label{sec:theory_initial}

We quantify the initial error of $\Amk^{(0)}$ to $\pAmk$ and $\Bik^{(0)}$ to $\pBik$ in terms of the sample size $N$. When the truncation error is negligible, the result suggests that $\Amk^{(0)}$ converges in Frobenius norm at a rate $O(\sqrt{k(k_1+k_2)/N})$. Furthermore, we show that $\Bik^{(0)}$ converges in Frobenius norm at a rate  $O(\sqrt{k(k_1+k_2)/N}+\sqrt{s^\star k^2/N})$, where the first term is induced by the estimation error of $\Amk^{(0)}$ and the second term is is the statistical error of $\Bik^{(0)}$.

We make the following assumption on the matrix $\Rb^{12}=(\Sigmab^1_{1,1})^{-1/2}\Sigmab^{12}_{1,1}(\Sigmab^{2}_{1,1})^{-1/2}$.

\begin{assumption}\label{assumption:distinctcca}
Suppose that $\min\{k_1,k_2\}\geq k$ and the top-k singular values of $\Rb^{12}$ satisfy: $\gamma_{1}> \gamma_{2}>\ldots>\gamma_{k}>0$.
\end{assumption}
\textcolor{black}{By definition, the singular values of $\Rb^{12}$ are the canonical correlations of two views. Assumption~\ref{assumption:distinctcca} is for identifiability purpose. It implies that the canonical vectors are unique, up to sign changes. }
Let the sets of left and right singular vectors of $\hat\Rb^{12}$ and $\Rb^{12}$ be $\{(\hat{\vb}_{i}^1,\hat{\vb}_{i}^2)\}_{i=1,\ldots, k}$ and $\{({\vb}_{i}^{1\star},{\vb}_{i}^{2\star})\}_{i=1,\ldots, k} $, respectively. Hence we define $\Qb$ as a diagonal matrix whose $i$-th diagonal entry $q_{i}\in\{-1,1\}$ satisfies the condition that $\dotp{q_{i}\hat{\vb}_{i}^1}{\vb_{i}^{1\star}}\geq0$.

\begin{assumption} \label{assumption:cov}
The covariance $\Sigmab^m = \frac{1}{N}\EE[\Ykmn\YkmnT]\in\RR^{pk_m\times pk_m}$ satisfies
\[
0 < \nu_{x} \leq \sinmin{\Sigmab^m} < \sinmax{\Sigmab^m}\leq \rho_{x}<\infty,\quad\mseq.
\]
\end{assumption}

We now establish the theory of the distance of $\Amk^{(0)}$ and $\pAmk$ under sign matrix $\Qb$.
\begin{theorem}\label{theorem:initialcca}
 Let $\defcdnm$. Suppose that $N=O(\max_{m=1,2}\kappa_m^2 k_m)$ and Assumption~\ref{assumption:distinctcca}--\ref{assumption:cov}  hold. Let $\Lscr^{m,r}=\sum_{(\ell,\ell')\in I}\dotp{\Lscr^m\phi_{\ell'}}{\phi_{\ell}^m}_{\HH_m}\phi_{\ell}^m\otimes\phi_{\ell'}$, with $I=\{(\ell,\ell')\in\NN\times\NN:(\ell,\ell')\neq (a,b), a=1,\ldots,k_m, b=1,\ldots,k\}$,  denote the truncation term. Then
\begin{equation}\label{eq:errorofA}
\norm{\Amk^{(0)} -\Qb\pAmk}_F
\leq 
C_{\gamma_k,\nu,\rho}\rbr{1+\max_{j\neq i}\frac{1}{|\gamma_j -\gamma_i|}}
\sqrt{k\frac{k_1 + k_2}{N}}+\sqrt{8k}\opnorm{\Ascr^m}{}^2\opnorm{\Lscr^{m,r}}{},
\end{equation}
with probability at least $1-5\exp(-\sum_{m=1}^2 k_m)$, where
$C_{\gamma_k,\nu,\rho}>0$ is a constant depending on $\gamma_k$, $\nu_x$, and $\rho_x$ defined in Appendix~\ref{ssec:notation}.
\end{theorem}

\begin{remark}\label{remark:truncation}
The second term in~\eqref{eq:errorofA} is the truncation error, which becomes small as $k$ and $k_m$ grow. To obtain a non-trivial upper bound, we must have $\opnorm{{\Ascr^m}}{}<\infty$. Therefore, the smallest non-zero singular value of ${\Lscr^m}$ should be bounded away from zero, suggesting that assuming ${\Lscr^m}$ to be a finite-rank operator~\citep[Theorem~4.2.3]{hsing2015theoretical} is a necessary condition. 
\textcolor{black}{
If the truncation error scales the same as the statistical error, the first term of~\eqref{eq:errorofA}, the theorem tells us that $\Amk^{(0)}$ is reasonably close to $\pAmk$, up to a sign change, at the rate of $O(\sqrt{k(k_1+k_2)/N})$.
We can achieve this by choosing the appropriate $k_m$ because $\opnorm{\Lscr^{m,r}}{}$ decreases with increasing $k_m$. Therefore, there exists a $k_m$ such that $\opnorm{\Lscr^{m,r}}{}=O(\sqrt{(k_1+k_2)/N})$.
Additionally, the first term of~\eqref{eq:errorofA} is inversely proportional to the difference between two consecutive canonical correlations. Hence, the canonical correlations must be distinct, as stated in Assumption~\ref{assumption:distinctcca}, and a gap between two consecutive canonical correlations is required to obtain a non-trivial upper bound.}
\end{remark}

Note that multiplications with a signed diagonal matrix do not break structural assumptions; that is, $\Qb\pAmk\in\Kcal_m(\tau_1,\tau_2)$
and $(\Qb\Bb_{i1}^{\star}\Qb\cdots\Qb\Bb_{ip}^{\star}\Qb){\top}\in\Kcal_B(s,\alpha)$. 
Furthermore, we have $f(\Ab,\Bb)=f(\{\Qb\Amk\},\{\Qb\Bik(\Ib_{p-1}\otimes\Qb^\top)\})$ and $h_i(\Bik)=h_i(\Qb\Bik(\Ib_{p-1}\otimes\Qb^\top))$ for any $\Qb$. Therefore, it would be cumbersome to write out $\Qb$ explicitly for the rest of the analysis. We abuse the notation here by writing $\Qb\pAmk$ as $\pAmk$  and $\Qb\pBik\Qb^\top$ as $\pBik$. 
Define $\ptBik=(-\Bb_{i1}^{\star}\cdots-\Bb_{ip}^{\star})\in\RR^{k\times pk}$, where $\Bb_{ii}^\star = \Ib_k$. 
\begin{assumption}\label{assumption:AB}
     For every $\pseq$, $\ptBik $ satisfies
     \[
     0<\nu_b^{1/2}\leq\sinmin{\ptBik }<\sinmax{\ptBik }\leq\rho_b^{1/2}<\infty.
     \]
     For every $\mseq$, $\pAmk $ satisfies
     \[
     0<\nu_a^{1/2}\leq\sinmin{\pAmk }<\sinmax{\pAmk }\leq\rho_a^{1/2}<\infty.
     \]
\end{assumption}
Assumptions~\ref{assumption:cov}--\ref{assumption:AB} guarantee that $f(\setAk,\setpBk)$ is strongly convex and smooth with respect to $\Amk$, $\mseq$, and $f(\setpAk,\setBk)$ is strongly convex and smooth with respect to $\Bik$, $\pseq$. Furthermore, combined with  the result of Theorem~\ref{theorem:initialcca}, it follows that $h_i(\setBk)$ is strongly convex and smooth with respect to $\Bik$ for $\pseq$ with high probability. 

Let $s=\vartheta_1s^\star/2$, $\alpha=\vartheta_2\alpha^\star/2$ for some constants $\vartheta_1$ and $\vartheta_2$. Our theory requires more stringent conditions:  $\vartheta_1$,$\vartheta_2$ are some constants greater than $2$, which implies that $\pBik\in\Kcal_{B}(s^\star,\alpha^\star)\subset\Kcal_{B}(s,\alpha)$ for $\pseq$. 
We state the convergence rate of Algorithm~\ref{alg:initializationB}. 
\begin{lemma}\label{lemma:initial_stepB} 
Suppose that Assumptions~\ref{assumption:s-sparse}--\ref{assumption:disperse} and \ref{assumption:cov}--\ref{assumption:AB} hold.
Let 
$m'=\argmin \sinmintwo{\Amk^{(0)}}$, $\eta_{B_0}= \{\max_{\mseq}\sinmaxtwo{\Ab^{m(0)}}\norm{\hat\Sigmab^{m}}_2\}^{-1}$, 
$N=O(\max_{\mseq} \kappa_m^2 (k_m s^\star +\log M + \log p))$ and $\pi_0=C_\vartheta\cbr{1-{\sinmintwo{\Ab^{m'(0)}}\eta_{B_0}\nu_x}/2}$,
where $C_\vartheta$, $C_\gamma$ and $C_{\delta,i}$ 
are constants whose explicit values are given in Appendix~\ref{ssec:notation}.  
Then, for each $1\leq i\leq p$, after $L$ iterations of Algorithm~\ref{alg:initializationB} we have
\begin{align*}
\norm{\Bb_i^{(L)}-\pBik}_F
\leq \pi_0^L\norm{\pBik}_F
+   \frac{C_\gamma C_\vartheta}{1-\pi_0}\frac{\eta_{B_0}}{M}\sum_{m=1}^M\norm{\Amk^{(0)} - \pAmk}_F
+   \frac{C_{\delta,i}C_\vartheta\eta_{B_0}}{1-\pi_0}
    \sqrt{\frac{\alpha^\star s^\star  k^2}{N}}    ,
\end{align*}
with probability at least $1-3\max_{\mseq}\exp(-k_m s^\star)$. 
\end{lemma}

\begin{remark}
The quality of the initial guess of $\pBik$ depends on  $M^{-1}\sum_{m=1}^M\norm{\Amk^{(0)} - \pAmk}_F$. While $s^\star k^2$ in the second term seems large at first glance, it is worth pointing out that the maximum degree of a node in functional graphical model scales with $s^\star k^2 \ll k^2p$. 
\end{remark}


\subsection{Convergence Analysis of Algorithm~\ref{alg:update}}\label{sec:theory_convergence}
We show the minimum number of iterations required for Algorithm~\ref{alg:update} to produce a useful result. Additionally, Algorithm~\ref{alg:update} converges in a linear rate under mild assumptions. 

The following assumption constrains the space of the local region where Algorithm~\ref{alg:update} takes place. The condition is further fulfilled by Theorem~\ref{theorem:initialcca} and Lemma~\ref{lemma:initial_stepB}.
\begin{assumption}\label{assumption:coef}
There exist constants $C_A$ and $C_B$ such that the initial guesses satisfy $\norm{\Ab^{m(0)}-\pAmk}_2\leq C_A\norm{\pAmk}_2$, $\norm{\Bb_i^{(0)}-\pBik}\leq C_B\norm{\pBik}$ for any unitarily invariant norm and $\norm{\Bb_i^{(0)\top}-\pBikT}_1\leq C_B\norm{\pBikT}_1$
for $\mseq$ and $\pseq$.
\end{assumption}
Let $C_\alpha, C_\beta$, $C_{\eta_A}, C_{\eta_B}$ be constants depending on the data, as detailed in Appendix~\ref{ssec:notation}. 
We introduce the statistical error $\Xi:=
    \sqrt{N^{-1}(\Xi_1+\Xi_2)}$ with 
 \begin{align*}
    &\Xi_1 
    := 
    C_\alpha
    \max_{\mseq}{k(k_m+k+\log M)},
    \quad
    \Xi_2 
    :=
    C_\beta
    {s^\star(k^2+\log p)},
\end{align*}
where $\Xi_1$ is induced by the statistical error of estimating $\Amk$ and $\Xi_2$ is by the statistical error of estimating $\Bik$. 
Under Assumption~\ref{assumption:coef}, we define $R_0$ as the smallest positive real number satisfying
$
    \max_{\mseq}\norm{\Ab^{m(0)} - \Amk}_F^2 
    + 
    \max_{\pseq}\norm{\Bb_i^{(0)} - \Bik^{\star}}_F^2
    \leq
        R_0^2, 
$ and
\begin{align*}
\delta_0&:=\exp\sbr{-\{(s^\star+1)\min_{\mseq} k_m\wedge (\alpha k)^2\}},\\
\pi &:= \sbr{4
    \cbr{
                    1
                    -
                    p(s^\star+1)C_{\eta_A}\eta_A
                }
    \vee
    C_\vartheta\rbr{
                1
                -
                MC_{\eta_B}\eta_B
            }
    }.
\end{align*}
We now state the convergence guarantee for Algorithm~\ref{alg:update}, with proof in Appendix~\ref{ssec:theorem_convergence}.
\begin{theorem}\label{theorem:convergence}
Suppose that Assumptions~\ref{assumption:s-sparse}--\ref{assumption:disperse},\ref{assumption:cov}--\ref{assumption:coef} hold and that, $\Xi^2\leq C_1 (C_A^2\rho_a+C_B^2\rho_b)$, and $\SCtwo$. 
The step sizes $\eta_A$ and $\eta_B$ are selected to satisfy
(i) $\eta_A\leq{C_2}\{p(s^\star+1)\}^{-1}$ and
$\eta_B\leq{C_3}{M}^{-1}$, 
(ii) $p\eta_A=C_\vartheta M\eta_B$, and (iii) $\pi<1$, 
where $\{C_i\}_{i=1}^3$ are positive  constants.
After $L$ iterations of Algorithm~\ref{alg:update}, we have
\begin{align*}
    \max_{\mseq}\norm{\Ab^{m(L)} - \Amk^{\star}}_F^2 
    + 
    \max_{\pseq}\norm{\Bb_i^{(L)} - \Bik^{\star}}_F^2\leq
    \pi^LR_0^2+
    \frac{1}{1-\pi} \Xi^2,
\end{align*}
with probability at least $1-10\delta_0$.

\end{theorem}

\begin{remark}
From the theorem it follows that Algorithm~\ref{alg:update} converges to a local optimum that is close to the population parameters with distance of up to $2(1-\pi)^{-1}\Xi^2$ if the number of iterations $L$ exceeds $o((-\log(1-\pi)+2\log\Xi-2\log R_0)/\log\pi)$.  
\end{remark}

\subsection{Graph Recovery}\label{ssec:graph_recovery}

The final step is to combine Theorem~\ref{theorem:initialcca}, \ref{theorem:convergence}, and Lemma~\ref{lemma:initial_stepB} together and choose the appropriate parameters and the number of iterations. 
We have the following main result.
\begin{theorem}\label{theorem:main} Suppose that Assumptions~\ref{assumption:compact:L}--\ref{assumption:AB} hold and that $N$, $\eta_A$, $\eta_B$, $\eta_{B_0}$ satisfy the conditions stated in~Theorem~\ref{theorem:convergence} and Lemma~\ref{lemma:initial_stepB}. Additionally, $N$ is large enough so that $\Xi^2\leq (1-\pi)2^{-3}\Lambda^{2}(k)$ and there exists a $k_m$  such that $\opnorm{\Lscr^{m,r}}{}=o(1/\sqrt{k_m})$. 
We set $\Amk^{(0)}$ as~\eqref{eq:cca_model}, run Algorithm~\ref{alg:initializationB} with the number of iterations $-C_1\log\{(k^2+\log s^\star)/{N}\}$ and then run Algorithm~\ref{alg:update} with $(-\log(1-\pi)+2\log\Xi-2\log R_0)/\log\pi$ iterations.
Let $\hat{\Nscr}_i$ be the estimated edges based on the outcomes from Algorithm~\ref{alg:update} with the edge selection threshold $\epsilon_0=2^{-1}\Lambda(k)$ defined in~\eqref{eq:edge_thre}. Then,  we have 
$
\PP(\hat{\Nscr}_i=\pNi; 1\leq i\leq p)\geq 1- 14\delta_0.
$
\end{theorem}
\begin{remark}
The theorem states that we can recover the edge set $\pNi$  with probability at least $1- 14\delta_0$. Note that to establish graph recovery with a reasonable sample size, the truncated latent signal $\Lambda(k)$ cannot be too small. 
As the minimum sample size scales with $\Lambda^{-2}(k)$, a sufficiently large sample size is required to meet the condition $\Xi^2\leq (1-\pi)2^{-3}\Lambda^{2}(k)$.
\end{remark}

\section{Related Work}

We review relevant research on multimodal estimation and functional graphical models.

The low-rank latent space assumption has been widely explored in multimodal estimation~\citep{zhou2015group,yang2016non}.
JIVE~\citep{lock2013joint} jointly estimates shared low-rank components and individual low-rank components and has been successfully applied to biomedical data~\citep{o2016r}. 
Our work differs in the way that we find the low-rank model of the original data generation process, while previous work finds a low-dimensional representation of the data, providing a better interpretation of the underlying mechanism. 
Multimodal integration framework has been widely used in joint prediction tasks. \citet{li2018integrative} constructed linear additive low-rank predictors for joint multitask regression. \citet{li2021integrative}~proposed a statistical inference procedure to select significant modalities for integrative linear regression models. Another methodology is to stack multimodal data in a tensor and perform low-rank tensor regression~\citep{zhou2013tensor}. The predictive power of various models was improved by using rich multimodal data in applications, including clinical diagnoses~\citep{wolfers2015estimating} and biomarker detection~\citep{mimitou2021scalable}. 

There is increasing interest in functional graphical models. \citet{qiao2019functional} established the penalized maximum log-likelihood framework for Gaussian functional data. \citet{qiao2020doubly} successively extended it to the discrete sampling setting and further considered estimation of time-varying graphs. \citet{Zapata2019functional} developed a separability condition for the multivariate covariance operator and applied it to learning functional graphical models. 
\citet{zhao2021high} developed a function-on-function regression model, as the functional data version of the neighborhood regression method. While the aforementioned work focuses on the Gaussian distribution, another active line of work develops functional graphical model estimators under the non-Gaussian setting. \citet{li2018nonparametric} and~\citet{lee2021nonparametric} used additive conditional independence~\citep{li2014additive} and proposed nonparametric estimators. 
\citet{moysidis2021joint} studied the joint estimation of multiple functional Gaussian graphical models by building a hierarchical structure on inverse covariance. Although this procedure is most closely related to our work, there are considerable differences in both methodology and theory. \citet{moysidis2021joint} extended the approach of \citet{guo2011joint} to the functional data setting. They imposed a hierarchical structure on the inverse covariance, 
and jointly estimated multiple graphs. Our method, in contrast, takes a generative perspective and studies a ``single'' latent graph shared over multiple views. 


\section{Simulations}

The simulations focus on comparisons of sparse precision matrices, where we test four types of sparse graphs. First, we compare our models with three existing methods by computing the Area under the ROC Curve (AUC). Second, we compute the convergence distance with respect to the sample size, verifying the results from Theorem~\ref{theorem:convergence}. Details of data generation processes are presented in Appendix~\ref{ssec:dgp}, where we synthesize four types of different graph structures and two noise models.

\subsection{Comparison}
\begin{figure}[t!]
    \centering    \includegraphics[width=.9\textwidth]{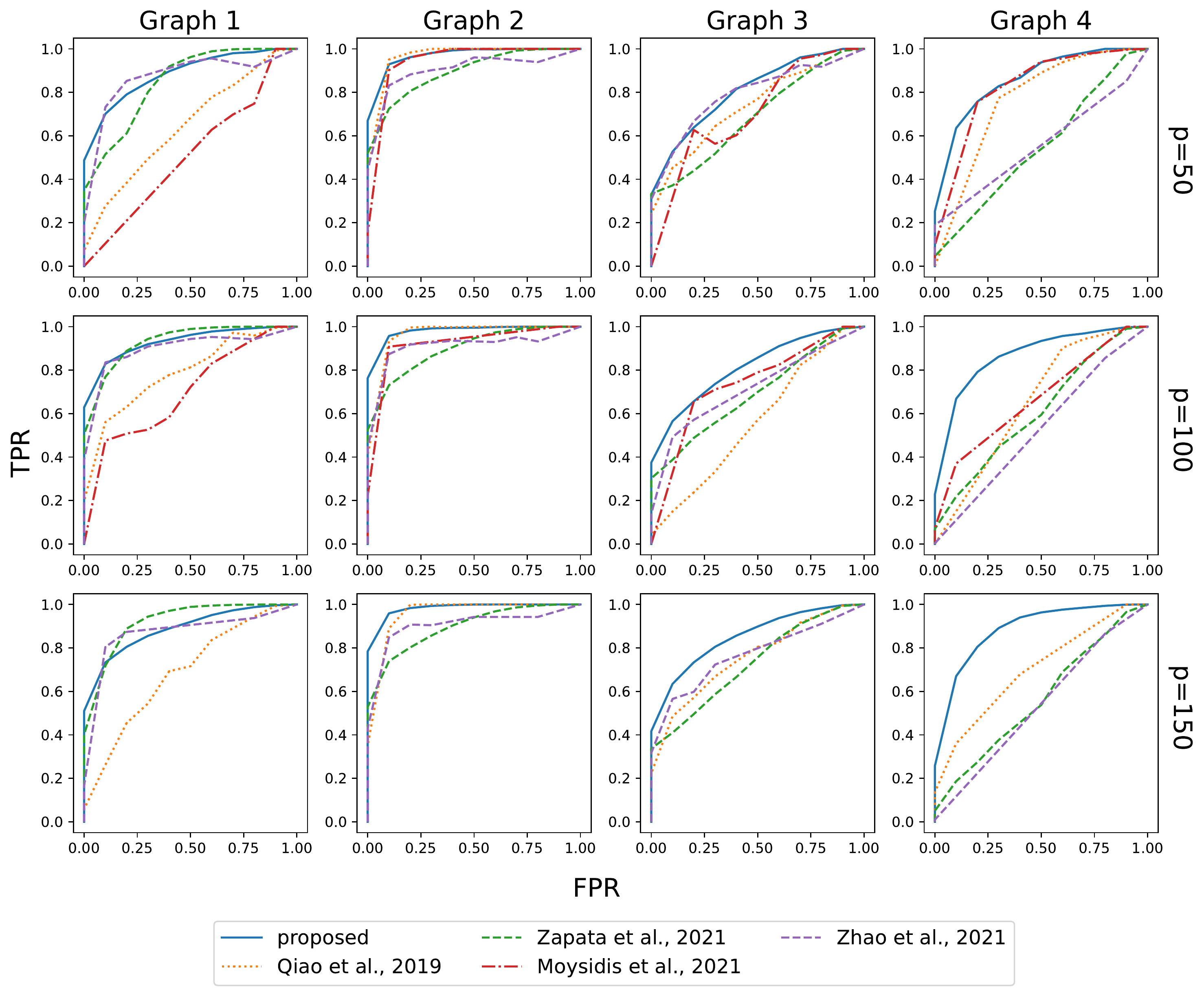}
    \caption{The ROC curves of Graph 1–4. The additive noise is generated from noise
model 2. The AUC is discussed in Table~\ref{tab:auc}. The proposed method has consistent performance across four graphs and $p=\{50,100,150\}$. }
    \label{fig:noise_model3}
\end{figure}
\begin{figure}[h!]
    \centering
\includegraphics[width=.9\textwidth]{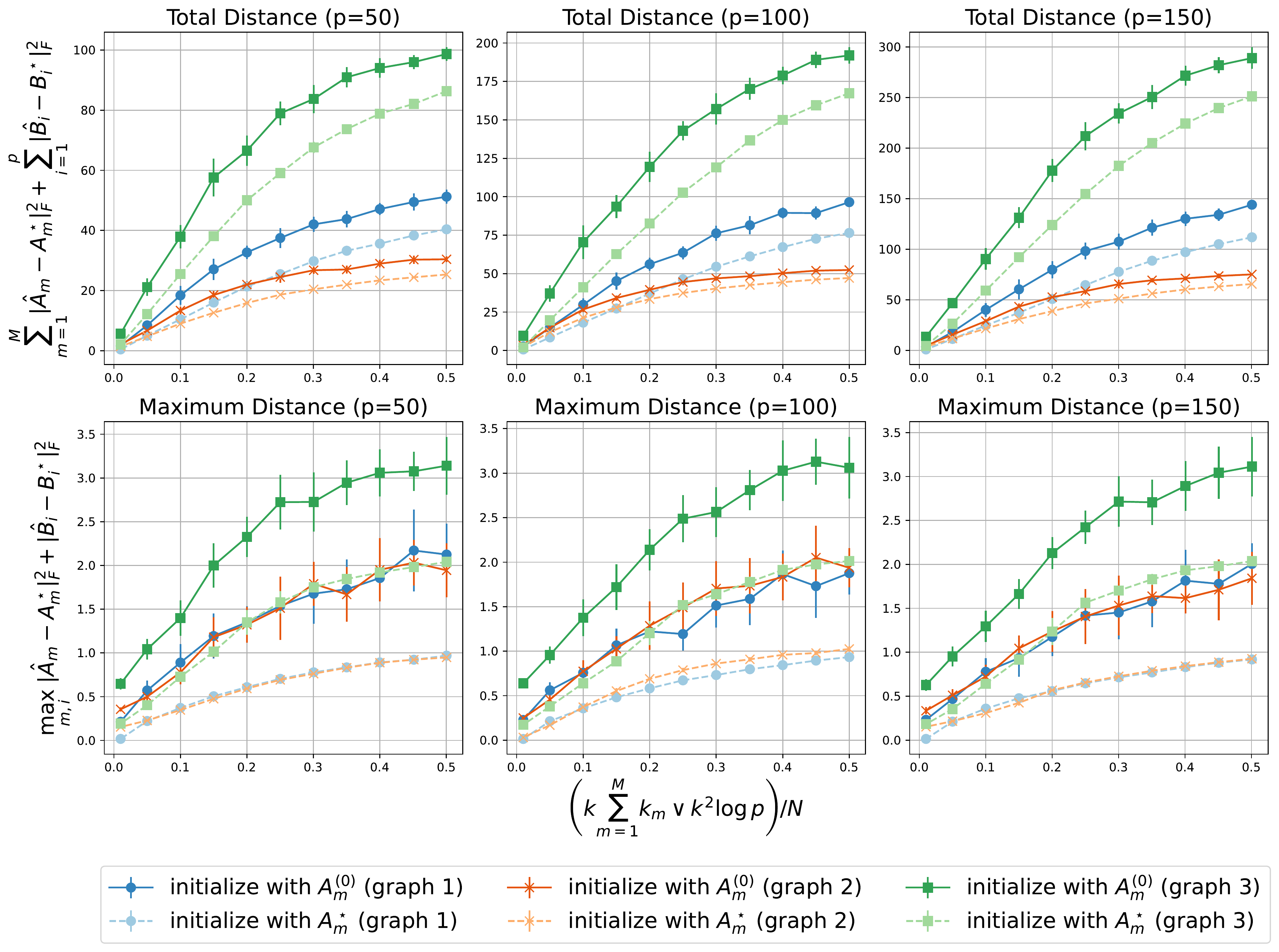}
    \caption{Distance v.s. statistical error. The solid lines denote the proposed algorithm. The dashed lines denote the case where $\Amk$ is first initiated with true $\pAmk$ for $\mseq$ and then refined by Algorithm~\ref{alg:update}--\ref{alg:initializationB} subsequently.}
    \label{fig:distance_sample}
\end{figure}
We compare the model with 3 other methods: FGGM~\citep{qiao2019functional}, PSFGGM~\citep{Zapata2019functional}, FPCA~\citep{zhao2021high}, and JFGGM~\citep{moysidis2021joint}. Since those candidates were not originally designed to learn latent graphs, we apply them to learn the graph of each modality separately and then take the intersection of the graphs as the surrogate of the latent graph. Let $\hat E^{1},\hat E^{2}$ be the edge set estimated by one of the algorithms, then we construct the latent graph as $
\hat E^{z}=\{ (i,j): (i,j)\in\hat E^{1}, (i,j)\in\hat E^{2}\}
$. 
We fix $k$, $k_m$ to be the true values and set $N=100$, $\Tcal=[0,1]$ and $p=\{50,100,150\}$. Since JFGGM requires computing an eigen-decomposition at each iterate, which is computationally expensive when $p$ is large, we omit testing JFGGM under the setting of $p=150$.

To evaluate the performance of competing methods, we vary the sparsity parameters and plot the Receiver Operating Characteristic (ROC) curves, where the x-axis denotes the
False Positive Rate (FPR) of the estimated edges and the y-axis denotes the True 
Positive Rate (TPR) of the estimated edges. Under the ``AND'' operation of edge selection and $i\neq j$, the TPR and FPR are defined respectively as:
\begin{align*}
\text{TPR}=\frac{|\{(i,j):(j,i),(i,j)\in E^{z}\cap\hat E^{z}\}|}{|\{(i,j):(i,j)\in E^{z}\}|},
\quad
\text{FPR}=
    \frac{|\{(i,j):(i,j),(j,i)\in \hat E^{z}\backslash E^{z}\}|}{|\{(i,j):(i,j)\not\in E^{z}\}|}.
\end{align*}
To plot the ROC curve, we vary the sparsity level $s$ for the proposed method and vary the sparse regularization coefficient $\lambda$ for the competing methods. Each point in the ROC plot is $(\text{FPR}(s), \text{TPR}(s))$ or $(\text{FPR}(\lambda), \text{TPR}(\lambda))$ for $1\leq s\leq p$ and $\lambda$ is chosen from a grid of values $[0,\lambda_{\max}]$, where $\lambda_{\max}=\max_{i\neq j}N^{-1}\norm{\Ykmnfree{i}(\Ykmnfree{j})^\top}_F$ is sufficient as discussed in Proposition~1 in~\citet{zhao2021high}. Furthermore, we test the proposed method with step size $\eta_A$, $\eta_B$, and $\eta_{B_0}$ chosen from $\{0.1,0.01,0.001,0.0001\}$ and select the combination of step size $(\eta_A, \eta_B, \eta_{B_0})$ that has the largest AUC. We run the simulation for $10$ runs and take the average of the results. The simulation result is plotted in Figure~\ref{fig:noise_model3} and the corresponding AUC is presented in Table~\ref{tab:auc} in Appendix~\ref{ssec:roc:noise1}. The proposed method achieves the best or comparable performance for Graph 1--4. \textcolor{black}{Appendix~\ref{sec:addsimulation} includes more simulation results along with discussion on the comparing methods.}

\subsection{Distance v.s. sample size}\label{ssec:dis_vs_sample}

Empirically, we demonstrate that the metric $\max_{\mseq}\norm{\Amk-\pAmk}_F^2+\max_{\pseq}\norm{\Bik-\pBik}_F^2$ at convergence scales almost linearly with respect to $(k\sum_{m=1}^Mk_m\vee k^2\log p)/N$, which matches the result of Theorem~\ref{theorem:convergence}.\footnote{We assume that $s^\star$ is a constant here.} We test the graph model 1--3 with the magnitude of the off-diagonal entries further scaled down by half. For additive noise, we use noise model 1. For each simulation, we fix $k$, $p$, and $k_m$ for $\mseq$ and vary $N$.
We take the average of the results for $20$ runs of simulations. From Figure~\ref{fig:distance_sample}, we see that when $(k\sum_{m=1}^Mk_{m}\vee k^2\log p)/N$ is smaller than $0.45$,  $\sum_{m=1}^M\norm{\Amk-\pAmk}_F^2+\sum_{i=1}^p\norm{\Bik-\pBik}_F^2$ and $\max_{\mseq}\norm{\Amk-\pAmk}_F^2+\max_{\pseq}\norm{\Bik-\pBik}_F^2$ scale almost linearly with $(k\sum_{m=1}^Mk_m\vee k^2\log p)/N$ 
for most tasks. Note that the solid lines converge to the dashed lines as the x-axis approaches zero. This implies that the CCA initialization~\eqref{eq:cca_model} is a consistent estimator, matching the result of Theorem~\ref{theorem:initialcca}.

\section{Brain Network Estimation with EEG-fMRI}\label{sec:experiment:brain}

We apply our proposed model to estimate the connectivity patterns of brain networks using concurrent measurements of EEG-fMRI~\citep{morillon2010neurophysiological,sadaghiani2010intrinsic}. The dataset includes $23$ test subjects in each of the two sessions. {Both EEG and fMRI are first source-localized~\citep{wirsich2020concurrent} and then parcellated according to the Desikan-Killiany cortical atlas~\citep{desikan2006automated}, which has $68$ parcels}. Each recorded session is $600$ seconds. In the first $300$ seconds, all subjects are synced to watch a short movie clip; in the following $300$ seconds, all the subjects are resting. It is believed that brain networks display similar patterns across subjects during movie-watching as compared to the resting-state~\citep{vanderwal2019movies}. Hence, we only use the first $300$ seconds for estimation and testing. Our goal is to predict the brain
networks of session 2 using the data from session 1. 
Specifically, we first learn the precision matrices using the data from session 1. Then, we apply the results to predict the brain network of session 2, where the log-likelihood is used for evaluation. The preprocessing pipeline is presented in Appendix~\ref{ssec:pp}.

Since it is difficult to directly evaluate the prediction error of the latent graph, we design a surrogate task to demonstrate the effectiveness
of the proposed model. We used the estimated edge set $\hat E^{z}$ as auxiliary information to learn the graph of EEG and fMRI.
After obtaining the estimated edge set $\hat{E}_z$ through the proposed algorithm, we use it to estimate the precision matrices of fMRI and EEG. Let $m=1$ denote the index for the fMRI modality and $m=2$ denote the index for the EEG modality. Given an edge set $E$, we denote $\Scal^+_E(k_m)$ be the set of $pk_m\times pk_m$ positive definite matrices with support associated to $E$: $\Scal^+_E(k_m):=\{{\bm\Omega}\in \Scal^+: \omega_{i,j}=0,\;(i,j)\not\in E\}$, where $\omega_{i,j}\in\RR^{k_m\times k_m}$ is the $(i,j)$-th submatrix of ${\bm\Omega}$. Let $\lambda\geq0$. We estimate the graph by solving the problem:
\begin{equation}\label{eq:graph_learning}
\hat{\bm\Omega}^{m} = \argmin_{{\bm\Omega}^{m}\in\Scal^+_E(k_m)} \tr(\hat\Sigmab^{m}{\bm\Omega}^{m})-\log\det({\bm\Omega}^{m})+\lambda\sum_{i\neq j,(i,j)\in E}\norm{\omega_{i,j}^m}_{1,1}.
\end{equation}
When $\lambda=0$, this is a standard procedure to estimate the graph after obtaining an edge set through neighborhood regression~\citep{ma2016joint}. When $\lambda>0$, this means that the submatrix $\omega_{i,j}$ is sparse, following a similarly structured assumption of the second condition of $\Kcal_B(s,\alpha)$. We then use the skggm package~\citep{laska_narayan_2017_830033} to optimize~\eqref{eq:graph_learning}. The parameter $\lambda$ is selected using 5-fold cross-validation with the BIC metric: $-2\hat L_N+s\log N$, where $\hat L_N$ is the log-likelihood and $s$ is the number of parameters. Let $\hat E_{1}$, $\hat E^{2}$ be the edge sets obtained by running the regression type algorithm, as detailed in Appendix~\ref{ssec:details_nr}, separately for fMRI and EEG data. We validate the effectiveness of $\hat E^{z}$ by testing three types of edge sets used in~\eqref{eq:graph_learning}: individual edge set, latent edge set, and fused edge set. The individual edge set is $E=\hat E^{1}$ for the fMRI modality or $E=\hat E^{2}$ for the EEG modality; the latent edge set is $E=\hat E^{z}$ for both modalities;  and the fused edge set is $E=\hat E^{1}\cup \hat E^{z}$ for the fMRI modality or $E=\hat E^{2}\cup \hat E^{z}$ for the EEG modality. Then, for the edge candidate and each data modality, we refit~\eqref{eq:graph_learning} to estimate the precision matrix. 

To address the small sample size, we employ sampling with replacement and repeat the experiments for 5 runs. The result is shown in Table~\ref{tab:experiment_result}, where we use the log-likelihood as the prediction score. The result indicates that, with the auxiliary edge information, the fused edge set has maximum in-sample log-likelihood and out-of-sample log-likelihood for both data modalities. Note that, with the latent edge set alone, we might neglect the individual graph structure for each modality, and hence the result is suboptimal. 

 
 \begin{table}[t!]
    \spacingset{1}
    \centering
    \resizebox{\textwidth}{!}{%
    \begin{tabular}{llll}
    \toprule
    \midrule
    Data modality & Edge candidate & In-sample Log-likelihood & Out-of-sample Log-likelihood \\
    \midrule
         \multirow{2}{*}{\shortstack[l]{Modality 1\\ (fMRI)} }      & Individual edge set 
     & $-118411.88(2543.47)$
     & $
 -124217.55 (315.57)
     $
     \\
        & Latent edge set  &
        $-115916.02(3515.06)$&
        $-122497.21(377.08)$\\
      &  Fused edge set &
      $-104131.66(1928.12)$
      &$-117812.87(248.49)$\\
      \midrule
    \multirow{3}{*}{\shortstack[l]{Modality 2\\ (EEG)} } & Individual edge set & $-22417.26 (1189.99)$
    &$ -34798.19(401.08)$\\
     & Latent edge set&
     $-28523.50(2059.87)$&
     $-36820.28(321.84)$
     \\
        & Fused edge set&
        $-21023.02(701.56)$
        &
        $-32842.76(447.48)$
        \\
    \midrule
    \bottomrule
    \end{tabular}%
    }
    \caption{The in-sample and out-of-sample log-likelihood with $3$ different edge set candidates. The result is the average over $5$ runs and the value in each parenthesis denotes the standard deviation.}
    \label{tab:experiment_result}
\end{table}

\section{Discussion}

We have developed a new procedure for integrating multimodal functional data to estimate the underlying latent graph, providing a new statistical solution to answer scientific questions. 
\textcolor{black}{The functions are assumed to be continuous whereas we might only have access to the discrete samples in practice. We included a discussion to address this issue in Appendix~\ref{appendix:treatment}. }
We studied the setting where data across modalities share the same set of nodes but have heterogeneous temporal characteristics. 
In practice, many applications might have a different number of nodes across modalities. For example, the spatial resolution of the fMRI data is much higher than that of the EEG data. Therefore, fMRI is known to capture richer spatial information. On the contrary, the temporal resolution of the EEG data is much higher than that of the fMRI data, offering more expressive temporal information. It is left for future work to develop new statistical models and estimation procedures that accommodate more complex spatial and temporal discrepancies.



\section{Acknowledgement}
We thank Thomas Alderson and Sepideh Sadaghiani for their help with data pre-processing and insightful discussion. 
We thank Benjamin Morillon, Katia Lehongre, and  Anne-Lise Giraud for their generosity of sharing the original EEG-fMRI raw data. 
We thank Jonathan Wirsich for providing source-localized pre-processed EEG-fMRI data.  
We thank Percy Zhai for sharing the code of the implementation of the functional Gaussian graphical model~\citep{qiao2019functional} and functional neighborhood regression~\citep{zhao2021high}. The research project is funded by the National Science Foundation (NSF) Graduate Research Fellowships Program,  NSF 2046795, 1909577, 1934986, 2216912, and NIFA award 2020-67021-32799.


\putbib[bu1]
\end{bibunit}

\appendix
\pagebreak
\begin{bibunit}[my-plainnat]
\section{Conditional Covariance Operator and Partial Covariance Operator}

This section shows the analysis of Section~\ref{ssec:latent_process}. Appendix~\ref{ssec:A:notation} introduces the notation. The main result is presented in in Appendix~\ref{ssec:main_results}. We start by discussing the univariate case and the extension to multivariate case follows similarly. Appendix~\ref{ssec:theorem:proof:partialcov} shows the proof of Theorem~\ref{theorem:partialcovariance_multi}; Appendix~\ref{ssec:proof:corollary:pcm} shows the proof of Corollary~\ref{corollary:partialcovariance_multi}; Appendix~\ref{ssec:A:auxiliary} presents the proofs of auxiliary lemmas.
\subsection{Notation}\label{ssec:A:notation}
Let $\HH$  be a separable Hilbert space and $\Bscr(\cdot)$  be the Borel set. Define $L^2(\Omega,\Fcal,\mu)$ as the space of square-integrable random element endowed with inner product $\EE[\dotp{\cdot}{\cdot}_\HH]$. Let $\mathfrak{B}(\HH,\HH)$ be the space of bounded linear operator from $\HH$ to $\HH$. Since $\HH$ is a complete normed space, it follows that $\mathfrak{B}(\HH,\HH)$ is a Banach space. Given a linear operator $\Tscr$, we say that $\Tscr$ is in the equivalence class of a zero operator if $\opnorm{\Tscr}{}=0$.
\subsection{Main Results}\label{ssec:main_results}
In this section, we connect the relation of the partial covariance operator, defined later in~\eqref{eq:partialoperator} with the covariance operator defined in~\citet{fukumizu2007kernel}. We first start with the simple case that discusses conditioning on the single random element. Consider $\Zcal_{1}$, $\Zcal_2$ be mean-squared integrable random elements in a separable Hilbert space. 
Let $\Kscr_{12}$ be the covariance operator for $\Zcal_{1}$ $\Zcal_{2}$ and $\Kscr_1$, $\Kscr_1$ be the variance operator of $\Zcal_{1}$ and $\Zcal_{2}$ respectively. It is known~\citep{baker1973joint} that there exists a unique bounded linear operator $\Vscr_{12}$ with norm $\opnorm{\Vscr_{12}}{}\leq 1$ such that $\Kscr_{12}=\Kscr_1^{1/2}\Vscr_{12}\Kscr_2^{1/2}$, $\Image(\Vscr_{12})\subset\overline{\Image(\Kscr_{1})}$ and $\ker(\Vscr_{12})^\perp\subset\overline{\Image(\Kscr_{2})}$. Then the conditional covariance operator is defined as 
 \begin{align}\label{eq:cond_op}
      \Vscr_{12\mid 3}=\Vscr_{12} - \Vscr_{13}\Vscr_{32},
 \end{align}
 which could be applied to measure conditional inference. 
\citet{fukumizu2004dimensionality} shows that $\Vscr_{12\mid 3}=0$ if and only if $\Zcal_{1}\indep\Zcal_{2}\mid\Zcal_{3}$ when $\HH$ is a RKHS. We will show that this property also holds under the Gaussian assumption in any separable Hilbert space. 
 Before establishing the theoretical results, we introduce the definition of the Gaussian random element. First, we introduce the definition of Gaussian linear space.
 \begin{definition}[Gaussian linear space, adopted from Definition~1.2 in \citet{janson1997gaussian}]
 A Gaussian linear space is a real linear space of random
variables, defined on  $(\Omega,\Fcal,\mu)$, such that each variable
in the space is centered Gaussian.
 \end{definition}
 Furthermore, by Theorem~1.5 in~\citet{janson1997gaussian}, we know that any set of random variables in linear Gaussian space is jointly normal.
 Next, we introduce the definition of Gaussian random element.
 \begin{definition}
 Let $\Zcal\in L^2(\Omega,\Fcal,\mu)$ be a centered random element taking values in $(\HH,\Bscr(\HH))$, we say $\Zcal$ is a Gaussian random element if  $\{\dotp{\Zcal}{f}_{\HH}:f\in\HH\}$ forms a Gaussian linear space.
 \end{definition}
 
Consider three centered Gaussian random elements $\Zcal_{i}\in L^2(\Omega,\Fcal,\mu)$ for $i=1,2,3$ taking values in $\HH$. Our goal is to measure the conditional independence of $\Zcal_{1}$ and $\Zcal_{2}$ given the random element $\Zcal_{3}$.
Let $L^2(\Omega,\sigma(\Zcal_3),\mu)$ be a closed linear subspace of $L^2(\Omega,\Fcal,\mu)$, where $\sigma(\Zcal_3)$ denote the smallest $\sigma$-algebra generated by $\Zcal_3$.
Let $\tilde\beta_{i3}\in\mathfrak{B}(\HH,\HH)$ be a bounded linear operator, where $\mathfrak{B}(\HH,\HH)$ is a Banach space. Since $\tilde\beta_{i3}$ is continuous, it is measurable and hence 
$\tilde\beta_{i3}\Zcal_3\in L^2(\Omega,\sigma(\Zcal_3),\mu)$. Additionally, $\Zcal_i-\tilde\beta_{i3}\Zcal_3$ is the residual of $\Zcal_{i}$ after regressing on $\Zcal_{3}$ for $i=1,2$. The regression procedure is defined as:
\begin{equation}\label{eq:definition_reg}
\tilde\beta_{i3}=\argmin_{\substack{
\beta\in\mathfrak{B}(\HH,\HH)
}}\EE\sbr{\norm{\Zcal_i-\beta\Zcal_3}^2_{\HH}},\quad i=1,2.
\end{equation}
 Since $\Zcal_i$ for $i=1,2,3$ are centered Gaussian random elements, the mean element of $\Zcal_i-\tilde\beta_{i3}\Zcal_3$ is a zero element of $\HH$ for $i=1,2$. The partial cross covariance operator $ \Kscr_{12\sbullet3}$ is defined as
 \begin{equation}\label{eq:partialoperator}
 \dotp{\Kscr_{12\sbullet3}f}{g} = \EE\sbr{
 \dotp{\Zcal_1-\tilde\beta_{13}\Zcal_3}{g}_\HH\dotp{\Zcal_2-\tilde\beta_{23}\Zcal_3}{f}_\HH}, \quad f,g\in\HH,
 \end{equation}
 which are bounded linear operators, as shown in Lemma~\ref{lemma:bdd_partial}. We define the partial cross covariance operator $\Kscr_{21\sbullet3}$ as the adjoint operator of $\Kscr_{12\sbullet3}$.
 We establish the following properties.

 \begin{theorem}\label{theorem:partialcov_formal}Let $\Zcal_i\in L^2(\Omega,\Fcal,\mu)$ for $i=1,2,3$ be centered Gaussian random elements taking values in $(\HH,\Bscr(\HH))$. $\Kscr_{12\sbullet3}$ and $\Kscr_{21\sbullet3}$ are defined in~\eqref{eq:partialoperator}. The following properties hold.
 \begin{enumerate}
     \item $ \Kscr_{12\sbullet3}$ and $\Kscr_{21\sbullet3}$ are in the equivalence class of the zero operator  if and only if $\Vscr_{12\mid 3}$ and $\Vscr_{21\mid 3}$ are in the equivalence class of the zero operator.
     \item $\Zcal_1$ and $\Zcal_2$ are conditionally independent given  $\Zcal_{3}$,  denoted as $\Zcal_{1}\indep\Zcal_{2}\mid\Zcal_{3}$,  if and only if $\Vscr_{12\mid 3}$ and $\Vscr_{21\mid 3}$ are in the equivalence class of zero operator.
 \end{enumerate}
 \end{theorem}

To evaluate the conditional independence of two Gaussian random elements in the Hilbert space, one could check the partial covariance operator, realized by computing~\eqref{eq:definition_reg}.
Theorem~\ref{theorem:partialcov_formal} shows the conditional independence properties when $\Zcal_1$, $\Zcal_2$ and $\Zcal_3$ are random elements.
Oftentimes, we encounter the situation when $\Zcal_3$ is not a single element, but a set of elements. For instance, this is the case for estimating Gaussian graphical models. Let $\Zcal=(\Zcal_{1},\ldots,\Zcal_{p})^\top\in\HH^p$ be centered Gaussian random vector with $\Zcal_{i}\in L^2(\Omega,\Fcal, \mu)$ for $i=1,\ldots,p$.  We could easily generalize Theorem~\ref{theorem:partialcov_formal} to a $p$-dimensional centered random vector  by studying the linear subspace $L^2(\Omega, \sigma(\Zcal_{i}:i=1,\ldots,p),\mu)$. The result is stated in Theorem~\ref{theorem:partialcovariance_multi}.

\subsection{Proof of Theorem~\ref{theorem:partialcov_formal}}\label{ssec:theorem:proof:partialcov}

It suffices to verify the properties of $\Kscr_{12,\sbullet 3}$, and similar properties hold for $\Kscr_{21,\sbullet 3}$. We begin with showing the first statement. Notice that since $\Image(\Vscr_{ij})\subset\overline{\Image(\Kscr_{i})}$ and $\ker(\Vscr_{ij})^\perp\subset\overline{\Image(\Kscr_{j})}$ for $i,j\in\{1,2,3\}$,  $\Vscr_{12\mid 3}=0$ if and only if $\Kscr_1^{1/2}\Vscr_{12\mid 3}\Kscr_2^{1/2}=0$. Therefore, from the result of Lemma~\ref{lemma:equivalence_of_partial_and_condition}, we could establish the following equivalence: for all $f,g\in\HH$, 
\[
 \Vscr_{12\mid 3}(f,g)=0, \; 
 \Vscr_{21\mid 3}(f,g)=0
  \Longleftrightarrow
 \Kscr_{12\sbullet3}(f,g)=0,\;
 \Kscr_{21\sbullet3}(f,g)=0,
 \]
 which proves the first statement. 
 It remains to show that, under the Gaussian setting, the following equivalence holds
 \[
 \Zcal_{1}\indep\Zcal_{2}\mid\Zcal_{3}\Longleftrightarrow
 \Kscr_{12\sbullet3},
 \Kscr_{21\sbullet3} \text{ are in the equivalent class of the zero operator}.
 \]
 
 Note that under the Gaussian assumption, by Lemma~\ref{lemma:linear_conditional}, the conditional estimator is the linear estimator. Then, 
 one could look at residual element $\Zcal_{i\sbullet 3}$ as the residual as $\Zcal_{i\sbullet 3}=\Zcal_{i}-\EE[\Zcal_{i}\mid\Zcal_{3}]$ for $i=1,2$. In addition, since $\Zcal_{i}$ for $i=1,2,3$ are zero mean, the mean elements $\mathfrak{m}_{i\sbullet 3}$ for $i=1,2$ are zero elements. Consequently, for any $f,g\in\HH$, we could write 
\begin{align}
\dotp{f}{\Kscr_{12\sbullet 3} g}_\HH
&=
    \EE\sbr{
        \dotp{\Zcal_{1}-\EE[\Zcal_{1}\mid\Zcal_{3}]}{f}_{\HH}
        \dotp{\Zcal_{2}-\EE[\Zcal_{2}\mid\Zcal_{3}]}{g}_{\HH}
    }\notag\\
&=
\EE_{\Zcal_{3}}\big[\underbrace{\EE_{\Zcal_{1},\Zcal_{2}\mid\Zcal_{3}}\sbr{
\dotp{\Zcal_{1}-\EE[\Zcal_{1}\mid\Zcal_{3}]}{f}_{\HH}
        \dotp{\Zcal_{2}-\EE[\Zcal_{2}\mid\Zcal_{3}]}{g}_{\HH}
\mid\Zcal_{3}}
}_{=:\Tscr(f,g)}\big].\label{eq:conditional_covariance}
\end{align}
By definition, if $\Zcal_{1}\indep\Zcal_{2}\mid\Zcal_{3}$, for any $f,g\in\HH$, we have $\Tscr(f,g)=0$, which implies that $\Kscr_{12\sbullet 3}=0$. Conversely, under the Gaussian setting, the conditional covariance does not depend on the value of $\Zcal_{3}$. Therefore, we can ignore the outer expectation of~\eqref{eq:conditional_covariance} and see that $\Kscr_{12\sbullet 3}=0$ implies $\Zcal_{1}\indep\Zcal_{2}\mid\Zcal_{3}$.

\subsection{Proof of Theorem~\ref{theorem:partialcovariance_multi}}

The proof step follows similarly to the proof step of the second statement of Theorem~\ref{theorem:partialcov_formal}. Noting that, we could generalize the argument of Lemma~\ref{lemma:congruence:l2andlinear} to a set of random element $\{\Zcal_1,\ldots,\Zcal_d\}$ and show that the space
\begin{multline*}
    \cbr{\sum_{j'\in\{1,\ldots,p\}\backslash\{i,j\}}\beta_{j'}\Zcal_{j'}:\beta_{j'}\in\mathfrak{B}(\HH,\HH), j'\in\{1,\ldots,p\}\backslash\{i,j\}} \\
    \subseteq
L^2(\Omega,\sigma(\Zcal_{j'};j'\in\{1,\ldots,p\}\backslash\{i,j\}),\mu),
\end{multline*}
is linear and closed, following a similar proof step of Lemma~\ref{lemma:congruence:l2andlinear}. By the orthogonality principle of Hilbert space (see~\citet[Theorem~2.6]{alma99954915694905899} for an exposition), we have 
$$
\cbr{\Zcal_i-\sum_{j'\in\{1,\ldots,p\}\backslash\{i,j\}}\tilde\beta_{ij'}\Zcal_{j'}}
\perp\Zcal_{j'},\quad j\in\{1,\ldots,p\}\backslash\{i,j\}.$$  Then, under the jointly Gaussian assumption, we could write 
\begin{align*}
    &\Zcal_i-\sum_{j'\in\{1,\ldots,p\}\backslash\{i,j\}}\tilde\beta_{ij'}\Zcal_{j'}=\Zcal_i-\EE[\Zcal_i\mid\Zcal_{-(i,j)}];\\
    &\Zcal_j-\sum_{j'\in\{1,\ldots,p\}\backslash\{i,j\}}\tilde\beta_{jj'}\Zcal_{j'}=\Zcal_i-\EE[\Zcal_j\mid\Zcal_{-(i,j)}],
\end{align*}
 following similar argument of Lemma~\ref{lemma:linear_conditional}. Then, the remaining step is the same as the second part of the proof of Theorem~\ref{theorem:partialcov_formal}, where it is left to verify the condition: for any $f,g\in\HH$
 $$
\EE\sbr{
\dotp{\Zcal_i-\EE[\Zcal_i\mid\Zcal_{-(i,j)}]}{f}_{\HH}
\dotp{\Zcal_i-\EE[\Zcal_j\mid\Zcal_{-(i,j)}]}{g}_{\HH}\mid Z_{j'},j'\in\{1,\ldots,p\}\backslash\{i,j\}}=0,
$$
similar as~\eqref{eq:conditional_covariance}. Therefore, following Definition~\ref{def:conditional_graph}, we can show that $\Kscr_{ij{\sbullet}}$ is in the equivalence class of a zero operator if and only if $\Zcal_i\indep\Zcal_j\mid\Zcal_{-(i,j)}$ for $(i,j)\not\in E$. 

\subsection{Proof of Corollary~\ref{corollary:partialcovariance_multi}}\label{ssec:proof:corollary:pcm}
The first statement is easy to verify as we can apply the orthogonality principle of the Hilbert space and obtain
\[
\Zcal_i = \sum_{j'\in\pni}\beta_{ij'}^\star\Zcal_{j'}+\Wscr_i,\quad \Zcal_j = \sum_{j'\in\{1,\ldots,p\}\backslash\{j\}}\beta_{jj'}^\star\Zcal_{j'}+\Wscr_j,
\]
where $\Wscr_i$ is independent of $\Zcal_{j'}$ for $j'\in\pni$ and $\Wscr_j$ is independent of $\Zcal_{j'}$ for $j'\in\{1,\ldots,p\}\backslash\{j\}$. Furthermore, $\Wscr_i$ and $\Wscr_j$ are centered and independent of each other. 

We then prove the second statement. Here we prove an equivalent statement: $\beta_{ij}^\star$ and $\beta_{ji}^\star$ are in the equivalent class of zero operator if and only if $\Zcal_i\indep\Zcal_j\mid\Zcal_{-(i,j)}$. 
We first show the forward direction: if $\beta_{ij}^\star$ and $\beta_{ji}^\star$ are in the equivalent class of zero operator, then
$\Zcal_i\indep\Zcal_j\mid\Zcal_{-(i,j)}$.
By Theorem~\ref{theorem:partialcovariance_multi}, we only need to show that $\Kscr_{ij\sbullet}$ is in the equivalence class of zero operator. 
For any $f,g\in\HH$, we can write 
\begin{align*}
\dotp{\Kscr_{ij{\sbullet}}f}{g}_\HH &= \EE\sbr{
\dotp{\beta_{ij}^\star \Zcal_j+\Wscr_i}{f}_\HH
\dotp{\beta_{ji}^\star\Zcal_i+\Wscr_j}{g}_\HH
}\\
&=
\EE\sbr{
\dotp{\Wscr_i}{f}_\HH
\dotp{\Wscr_j}{g}_\HH
}=0.
\end{align*}

To show the other direction, it suffices to show that $\Kscr_{ij\sbullet}$ is in the equivalence class of a zero operator implies that $\beta_{ij}^\star$ in the equivalence class of a zero operator, which is equivalent to show that $$\dotp{\EE[\Zcal_{i}\mid\Zcal_{-(i,j)}]}{f}_\HH=\dotp{\EE[\Zcal_{i}\mid\Zcal_{j\in\pni}]}{f}_\HH.$$ 

First, we denote the residuals of $\Zcal_i$ and $\Zcal_j$ on $\Zcal_{j'}$ for $j'\in\{1,\ldots,p\}\backslash\{i,j\}$ as $r_i$ and $r_j$, respectively. Then for any $f\in\HH$, we have
\begin{align*}
\dotp{\EE[\Zcal_{i}\mid \Zcal_{j\in\pni}]}{f}_\HH
&= \dotp{\EE[\Zcal_{i}\mid \EE[Z_j\mid \Zcal_{-(i,j)}]+r_j, \Zcal_{-(i,j)}]}{f}_\HH\\
&= \dotp{\EE[\Zcal_{i}\mid r_j, \Zcal_{-(i,j)}]}{f}_\HH\\
&=\dotp{\EE[\EE[\Zcal_{i}\mid\Zcal_{-(i,j)}]\mid r_j, \Zcal_{-(i,j)}]}{f}_\HH
+\dotp{\EE[r_i\mid r_j, \Zcal_{-(i,j)}]}{f}_\HH.
\intertext{Apply the fact that $\EE[\Zcal_{i}\mid\Zcal_{-(i,j)}]$ is independent of $r_j$ and $r_i$ is independent of $\Zcal_{-(i,j)}$, the above display is equivalent as}
&=\dotp{\EE[\Zcal_{i}\mid \Zcal_{-(i,j)}]}{f}_\HH
+\dotp{\EE[r_i\mid r_j]}{f}_\HH.
\intertext{Since $\Kscr_{ij\sbullet}$ is in the equivalence class of zero operator, we know that the residual $r_i$ and $r_j$ are independent. Therefore, the above display is equivalent as}
&=\dotp{\EE[\Zcal_{i}\mid \Zcal_{-(i,j)}]}{f}_\HH
+\dotp{\EE[r_i]}{f}_\HH\\
&=\dotp{\EE[\Zcal_{i}\mid \Zcal_{-(i,j)}]}{f}_\HH.
\end{align*}
Hence, we complete the proof.
\subsection{Auxiliary Lemmas}\label{ssec:A:auxiliary}
\begin{lemma}\label{lemma:equivalence_of_partial_and_condition}
 Under the conditions of Theorem~\ref{theorem:partialcov_formal}, for any  $f,g\in\HH$, the properties hold $$\dotp{\Kscr_{12\sbullet3}g}{f}_{\HH}=\dotp{\Kscr_1^{1/2}\Vscr_{12\mid 3}\Kscr_2^{1/2}g}{f}_{\HH},\quad\dotp{\Kscr_{21\sbullet3}f}{g}_{\HH}=\dotp{\Kscr_2^{1/2}\Vscr_{21\mid 3}\Kscr_1^{1/2}f}{g}_{\HH}.$$
\end{lemma}
 \begin{proof}[Proof of Lemma~\ref{lemma:equivalence_of_partial_and_condition}]
 Recall the definition of the residual element $\Zcal_{i\sbullet 3}$ in~\eqref{eq:definition_reg} and based on Lemma~\ref{lemma:partial}, we could write $\Zcal_{i\sbullet 3}$ as $\Zcal_{i}-\beta_i\Zcal_{i}$ for some  $\beta_i$ that satisfies  $\Kscr_{i3}-\beta_i\Kscr_{3}=0$  for $i=1,2$.  
Since $\Zcal_{i}$ for $i=1,2,3$ are zero-mean, it follows that the residual mean elements $\mathfrak{m}_{i\cdot 3}$ for $i=1,2$ are zero elements. Then, for any $f,g\in\HH$, write
\begin{align*}
    \dotp{f}{\Kscr_{12,\sbullet 3}g}_{\HH}
    &=
    \EE[
    \dotp{\Zcal_{1\sbullet 3}}{f}_{\HH}
    \dotp{\Zcal_{2\sbullet 3}}{g}_{\HH} 
    ]\\
    &=
    \EE[
    \dotp{\Zcal_{1}-\beta_1\Zcal_{3}}{f}_{\HH}
    \dotp{\Zcal_{2}-\beta_2\Zcal_{3}}{g}_{\HH} 
    ]\\
    &=\dotp{f}{(\Kscr_{12}+\beta_1\Kscr_{3}\beta_2^*-\beta_1\Kscr_{32}-\Kscr_{13}\beta_2^*)g}_{\HH},
    \intertext{where $\beta_2^*$ is the adjoint operator of $\beta_2$. Apply the fact that $\beta_1\Kscr_3=\Kscr_{13}$. The above display is equivalent as}
    &= \dotp{f}{(\Kscr_{12}-\beta_1\Kscr_{32})g}_{\HH}.
\end{align*}
Next, write $\Kscr_1^{1/2}\Vscr_{12\mid 3}\Kscr_2^{1/2}= \Kscr_{12}-\Kscr_1^{1/2}\Vscr_{13}\Vscr_{32}\Kscr_2^{1/2}$. Then, for any $f,g\in\HH$, we have
\begin{align*}
    \dotp{f}{\Kscr_1^{1/2}\Vscr_{12\mid 3}\Kscr_2^{1/2}g}_{\HH}
    =
    \dotp{f}{\Kscr_{12}g}_{\HH}
    -
    \dotp{f}{\Kscr_1^{1/2}\Vscr_{13}\Vscr_{32}\Kscr_2^{1/2}g}_{\HH}.
\end{align*}
Therefore, to show the first statement, it suffices to show that $\dotp{f}{\Kscr_1^{1/2}\Vscr_{13}\Vscr_{32}\Kscr_2^{1/2}g}_{\HH}=\dotp{f}{\beta_1\Kscr_{32}g}_{\HH}$ for any $f,g\in\HH$. 

Recall that $\Image(\Vscr_{32})\subset\overline{\Image(\Kscr_{3})}=\overline{\Image(\Kscr_{3}^{1/2})}$ and define the projection operator to 
$\overline{\Image(\Kscr_{3}^{1/2})}$ as $\Pscr_3=\Kscr_{3}^{1/2}(\Kscr_3^{1/2})^\dagger$. Then, we have $\Pscr_3\Vscr_{32}=\Vscr_{32}$. Therefore, we can write
\begin{align*}
    \dotp{f}{\Kscr_1^{1/2}\Vscr_{13}\Pscr_3\Vscr_{32}\Kscr_2^{1/2}g}_{\HH}
    &=
    \dotp{f}{\Kscr_1^{1/2}\Vscr_{13}\Kscr_{3}^{1/2}(\Kscr_3^{1/2})^\dagger\Vscr_{32}\Kscr_2^{1/2}g}_{\HH}\\
    &=
    \dotp{f}{\Kscr_{13}(\Kscr_3^{1/2})^\dagger\Vscr_{32}\Kscr_2^{1/2}g}_{\HH}.
    \intertext{Applying the fact that $\Kscr_{13}=\beta_1\Kscr_{3}$, the above display yields}
    &=
     \dotp{f}{\beta_1\Kscr_{3}^{1/2}\Kscr_{3}^{1/2}(\Kscr_3^{1/2})^\dagger\Vscr_{32}\Kscr_2^{1/2}g}_{\HH}\\
     &=
      \dotp{f}{\beta_1\Kscr_{3}^{1/2}\Pscr_3\Vscr_{32}\Kscr_2^{1/2}g}_{\HH}\\
      &=
      \dotp{f}{\beta_1\Kscr_{3}^{1/2}\Vscr_{32}\Kscr_2^{1/2}g}_{\HH}\\
      &=\dotp{f}{\beta_1\Kscr_{32}g}_{\HH}.
\end{align*}
Hence, we complete the proof.
 \end{proof}
\begin{lemma}\label{lemma:congruence:l2andlinear}
Let $\Zcal\in L^2(\Omega,\Fcal,\mu)$ be a centered random element taking values in the measurable Hilbert space $(\HH,\Bscr(\HH))$, where $\HH$ is a separable Hilbert space. Define the space
\[
V=\{\beta\Zcal:\beta\in\mathfrak{B}(\HH,\HH)\},
\]
where $\mathfrak{B}(\HH,\HH)$ is the space of all linear bounded operator mapping to its own Hilbert space $\HH$. 
Then, $V$ is a closed linear subspace of  $L^2(\Omega,\sigma(\Zcal),\mu)$.
\end{lemma}
\begin{proof}[Proof of Lemma~\ref{lemma:congruence:l2andlinear}]
First, we verify that $V\subseteq L^2(\Omega,\sigma(\Zcal),\mu)$. 
Since $\beta$ is a bounded linear operator from a complete normed space to itself, it follows that $\beta$ is a continuous operator and hence measurable. Therefore, the composite $\beta\Zcal$ is a measurable operator from $(\Omega,\sigma(\Zcal))$ to $(\HH,\Bscr(\HH))$. Moreover, we have $\EE[\norm{\beta\Zcal}_\HH^2]\leq\opnorm{\beta}{}^2\EE[\norm{\Zcal}_\HH^2]<\infty$. Therefore, we have $V\subseteq L^2(\Omega,\sigma(\Zcal),\mu)\subset L^2(\Omega,\Fcal,\mu)$. 

It is easy to see that $V$ forms a linear subspace and hence what is left is to verify that $V$ is closed.  Let $\{\beta_n\}_{n\in\NN}$ be a sequence of indexed bounded linear operators $\beta_n\in\mathfrak{B}(\HH,\HH)$  for $n\in\NN$ converging to $\beta$. Since $\mathfrak{B}(\HH,\HH)$ is a Banach space, it follows that $\beta\in\mathfrak{B}(\HH,\HH)$. Therefore, we have $\{\beta_n\Zcal\}_{n\geq 1}$ converging to $\beta\Zcal\in V$. Hence $V$ is closed.


\end{proof}

\begin{lemma}\label{lemma:linear_conditional}
Let $\Zcal_1,\Zcal_2\in L^2(\Omega,\Fcal,\mu)$ be centered Gaussian random elements taking values in $(\HH,\Bscr(\HH))$, where $\HH$ is a separable Hilbert space. 
Let 
 $$\hat\beta
=
\argmin_{
\beta\in\mathfrak{B}(\HH,\HH)} \EE\sbr{\norm{\Zcal_1-\beta\Zcal_2}_{\HH}^2}.$$
Then, we have
\[
\EE[\Zcal_1\mid \Zcal_2]=\hat\beta\Zcal_2.
\]
\end{lemma}
This is a classical result when $\Zcal_1,\Zcal_2$ are Gaussian random variables taking values in the real line. The proof follows by the fact that uncorrelatedness implies independence under the jointly Gaussian setting. The extension to Gaussian random element in the Hilbert space is straightforward and yet we provide the statement for clarity.
\begin{proof}[Proof of Lemma~\ref{lemma:linear_conditional}]
By Lemma~\ref{lemma:partial}, we can write
\[
\Zcal_1=e+\hat\beta\Zcal_2,
\]
where we have $\EE[\dotp{e}{\hat\beta\Zcal_2}_\HH]=0$ by the orthogonality principle of the Hilbert space. Since ${e}$ and ${\hat\beta\Zcal_2}$ are uncorrelated, they are independent under the joint Gaussian assumption. Therefore, given $\Zcal_2=f$, $\hat\beta\Zcal_2$ is deterministic and the randomness $\Zcal_2\mid\Zcal_1$ only comes from $e$. As a result, $\Zcal_1\mid\Zcal_2\sim\Ncal(\hat\beta\Zcal_2, \Kscr_e)$, where $\Kscr_e$ is the covariance operator of $e$. Note that $\hat\beta\Zcal_2$ is the mean of $\Zcal_1\mid\Zcal_2$ and hence we have $\EE[\Zcal_1\mid\Zcal_2]=\hat\beta\Zcal_2$.
\end{proof}
\begin{lemma}\label{lemma:bdd_partial}
Let $\Zcal_i\in L^2(\Omega,\Fcal,\mu)$ for $i=1,2,3$ random elements taking values in $(\HH,\Bscr(\HH))$. We define $\Zcal_{1\sbullet 3}$, $\Zcal_{2\sbullet 3}$ be the residuals of $\Zcal_{1}$ and $\Zcal_{2}$ after regressing on $\Zcal_{3}$, respectively, where recall the definition of regression in~\eqref{eq:definition_reg}. 
The mean elements of $\Zcal_{1\sbullet 3}$, $\Zcal_{2\sbullet 3}$ are defined as $\mathfrak{m}_{1\sbullet 3}$ and $\mathfrak{m}_{2\sbullet3}$ respectively.
 Define the partial cross covariance operator
 \begin{equation}\label{eq:partialoperator_nonzero}
 \dotp{\Kscr_{12\sbullet3}f}{g} = 
 \EE\sbr{
 \dotp{(\Zcal_{1\sbullet 3}-\mathfrak{m}_{1\sbullet 3})}{g}_\HH
\dotp{\Zcal_{2\sbullet 3}-\mathfrak{m}_{2\sbullet 3}}{f}_\HH},\quad f,g\in\HH,
 \end{equation}
which is linear and bounded. Let $\Kscr_{21\sbullet 3}$ be the adjoint operator of $\Kscr_{12\sbullet3}$. Then, $\Kscr_{21\sbullet 3}$ is linear and bounded.

If $\Zcal_i\in L^2(\Omega,\Fcal,\mu)$ for $i=1,2,3$ are zero-mean. Then the mean elements $\mathfrak{m}_{1\sbullet 3}$ and $\mathfrak{m}_{2\sbullet3}$ are zero elements. Hence \eqref{eq:partialoperator_nonzero} is equivalent as \eqref{eq:partialoperator}.

\end{lemma}

\begin{proof}[Proof of Lemma~\ref{lemma:bdd_partial}]
For notation simplicity, we define $\tilde{\Zcal}_i=\Zcal_{i\sbullet 3}-\mathfrak{m}_{i\sbullet 3}$, the marginal probability measure of $\tilde{\Zcal}_i$ as $\tilde{\mu}_i$ for $i=1,2$ and the joint probability measure of $\tilde{\Zcal}_{1}$ and $\tilde{\Zcal}_{2}$ as $\tilde{\mu}_{12}$.
First, for any $f,g\in\HH$, we define a bilinear functional $\tilde{G}:\HH\otimes\HH\rightarrow \RR$:
\[
\tilde{G}(f,g)
=
\int_{\Omega\times\Omega}
\dotp{
    \tilde{\Zcal}_1}{f}_{\HH}
    \dotp{
    \tilde{\Zcal}_2}{g}_{\HH}
    d\tilde{\mu}_{12}.
\]
We show that $\tilde{G}(f,g)$ is bounded for any bounded $f,g\in\HH$ by applying Cauchy-Schwarz inequality:
\begin{align}
    \abr{\tilde{G}(f,g)}^2&=
    \rbr{\int_{\Omega\times\Omega}
\dotp{
    \tilde{\Zcal}_1}{f}_{\HH}
    \dotp{
    \tilde{\Zcal}_2}{g}_{\HH}
    d\tilde{\mu}_{12}}^2\notag\\
    &\leq
    \cbr{
    \int_\Omega
    (\dotp{
    \tilde{\Zcal}_1}{f}_{\HH})^2
    d\tilde{\mu}_1
    }
    \cbr{
    \int_\Omega
    (\dotp{
    \tilde{\Zcal}_2}{g}_{\HH})^2
    d\tilde{\mu}_2
    }.\notag
    \intertext{Using the fact that $\Zcal_{i}=\tilde{\Zcal}_i+(\Zcal_{i}-\tilde{\Zcal}_i)$ and $E[\dotp{\tilde{\Zcal}_i}{\Zcal_{i}-\tilde{\Zcal}_i}_{\HH}]=0$ for $i=1,2$, we could upper bound the above display as}
    &\leq 
    \cbr{
    \int_\Omega
    (\dotp{
    {\Zcal}_1}{f}_{\HH})^2
    d{\mu}_1
    }
    \cbr{
    \int_\Omega
    (\dotp{
    {\Zcal}_2}{g}_{\HH})^2
    d{\mu}_2
    }\notag\\
    &=\norm{\Kscr_{1}^{1/2}f}_{\HH}^2
   \norm{\Kscr_{2}^{1/2}g}_{\HH}^2<\infty.\label{eq:bounded}
\end{align}
Therefore, for any fixed $g\in\HH$, by Riesz' representation theorem, there exists a unique element $q_g$ in $\HH$ such that $G(f,g)=\dotp{f}{q_g}$. Then, we define $\Kscr_{12\sbullet 3}:\HH\rightarrow\HH$ by $\Kscr_{12\sbullet 3}g=q_g$. It is clear that $\Kscr_{12\sbullet 3}$ is linear and bounded by~\eqref{eq:bounded}:
\begin{align*}
    \norm{\Kscr_{12\sbullet 3}g}_{\HH}^2=\sup_{f\in\HH}\frac{|\tilde{G}(f,g)|^2}{\norm{f}_{\HH}^2}\leq 
    \sup_{f\in\HH}\frac{\norm{\Kscr_1^{1/2}f}_{\HH}^2}{\norm{f}_{\HH}^2}\norm{\Kscr_2^{1/2}g}_{\HH}^2
    \leq \norm{\Kscr_1}\norm{\Kscr_2}\norm{g}_{\HH}^2,
\end{align*}
which implies that $\norm{\Kscr_{12\sbullet 3}}\leq \norm{\Kscr_1}\norm{\Kscr_2}<\infty$. $\Kscr_{12\sbullet 3}$ is defined everywhere in $\HH$. 
\end{proof}
\begin{lemma}\label{lemma:partial} Under the same condition of Lemma~\ref{lemma:bdd_partial}, define  covariance operator as $\Kscr_{12}$. Define
\[
\hat\beta
=
\argmin_{\substack{
\beta\in\mathfrak{B}(\HH,\HH).
}} \EE\sbr{\norm{\Zcal_1-\beta\Zcal_2}_{\HH}^2}.
\]
Then, $\hat\beta$ satisfies $\Kscr_{12}=\hat\beta\Kscr_{2}$.
\end{lemma}
since $\Image(\Kscr_{21})\subseteq\Image(\Kscr_2)$, it follows that $\hat\beta=(\Kscr_2^\dagger\Kscr_{21})^*$.
\begin{proof}[Proof of Lemma~\ref{lemma:partial}]
Define the space
$
V=\{\beta\Zcal_2:\beta\in\mathfrak{B}(\HH,\HH)\}
$. By Lemma~\ref{lemma:congruence:l2andlinear}, we see that $V$ is a closed subspace of $L^2(\Omega,\Fcal,\mu)$.

Then, by the orthogonality principle of the Hilbert space (see~\citet[Theorem~2.6]{alma99954915694905899} for an exposition), we have $\mathbb{E}[(\Zcal_1-\hat\beta\Zcal_2) \otimes \Zcal_2]=0$. This implies that for any $f,g\in\HH$, we have
\begin{align*}
    \EE[\dotp{\Zcal_1-\hat\beta\Zcal_2}{f}_\HH\dotp{\Zcal_2}{g}_{\HH}]
    &=\dotp{(\Kscr_{12}-\hat\beta\Kscr_2)g}{f}_\HH=0.
\end{align*}
Therefore, we can conclude that $\Kscr_{12}-\hat\beta\Kscr_2=0$. 
\end{proof}
\section{Truncation in Finite Dimensional Space}
In this section, we first show the proof of Proposition~\ref{prop:infinite_program}, which transforms the optimization problem~\eqref{eq:infinite_obj} in the Hilbert space into an optimization problem of sequences of real numbers. Then, in Appendix~\ref{ssec:proof:finite:program}, we derive and propose a finite dimensional optimization problem that computing resources can realize. 
\subsection{Proof of Proposition~\ref{prop:infinite_program}}
Recall that for $\mseq$ and $\pseq$ we have
\[
{\Ascr^m} = \sum_{\ell,\ell''\in\NN} a_{\ell\ell''}^m\phi^m_\ell\otimes\phi_{\ell''}^m,\quad \beta_{ij}=\sum_{\ell,\ell'\in\NN}b_{ij,\ell\ell'}\phi_{\ell}\otimes\phi_{\ell'},\quad 
{\Ycal_{i}^m}=\sum_{\ell''\in\NN} y_{i,\ell''}^m\phi_{\ell''}^m.
\]
Therefore, we can write
\begin{align*}
    &\EE\sbr{\bignorm{{\Ascr^m}{\Ycal_{i}^m}-\sum_{j\in\pni}\beta_{ij}{\Ascr^m}\Ycal_{j}^m }_\HH^2}\\
    &\quad=\EE\Bigg[\Bigg\|
    \rbr{\sum_{\ell,\ell''\in\NN} a_{\ell\ell''}^m\phi^m_\ell\otimes\phi_{\ell''}^m}
    \rbr{\sum_{\ell''\in\NN} y_{i,\ell''}^m\phi_{\ell''}^m}\\
    &\quad\quad\quad-
    \sum_{j\in\pni}
    \rbr{\sum_{\ell,\ell'\in\NN}b_{ij,\ell\ell'}\phi_{\ell}\otimes\phi_{\ell'}}
    \rbr{\sum_{\ell',\ell''\in\NN} a_{\ell'\ell''}^m\phi^m_{\ell'}\otimes\phi_{\ell''}^m}
    \rbr{\sum_{\ell''\in\NN} y_{j,\ell''}^m\phi_{\ell''}^m}
    \Bigg\|_\HH^2\Bigg]\\
    &\quad=\EE\Bigg[\Bigg\|
    \rbr{\sum_{\ell,\ell''\in\NN} a_{\ell\ell''}^my_{i,\ell''}^m\phi_\ell }-
    \sum_{j\in\pni}
    \rbr{\sum_{\ell,\ell',\ell''\in\NN}b_{ij,\ell\ell'} a_{\ell'\ell''}^m y_{j,\ell''}^m\phi_{\ell}}
    \Bigg\|_\HH^2\Bigg]\\
    &\quad=\EE\sbr{\sum_{\ell\in\NN}
 \rbr{\sum_{\ell''\in\NN}a_{\ell\ell''}^m y_{i,\ell''}^m
 - 
 \sum_{j\in\pni}\sum_{\ell',\ell''\in\NN} b_{ij,\ell\ell'}a_{\ell'\ell''}^my_{j,\ell''}^m
}^2}.
\end{align*}
Hence we complete the proof. 
\subsection{Derivation of Equation~\eqref{eq:population:finite:program}}\label{ssec:proof:finite:program}
 Similar to the latent space, we define the $ k_m$-dimensional approximation of $\chi_{i}^m$ as
 \[
 \chi_{i}^{k_m}:=\sum_{\ell=1}^{k_m}
x_{i,\ell}^m \phi_{\ell}^m,\quad\pseq,\;\mseq,
 \]
where $x_{i,\ell}^m=\dotp{\chi_{i}^m}{\phi_{\ell}^m}$ for $\ell=1,\ldots,k_m$. 
Then, we can represent the truncated $\zb_i=(z_{i,1},\ldots,z_{i,k})^\top\in\RR^{k}$ from the infinite dimensional vector $(z_{i,1},z_{i,2},\ldots)$ in the following form
  \begin{align}
 \zb_i&=\begin{pmatrix}z_{i,1}\\\vdots\\z_{i,k}\end{pmatrix}=\sum_{j\in[p]\backslash \{i\}}
 \begin{pmatrix}b_{ij,11}^\star&\ldots&b_{ij,1k}^\star\\&\ddots&\\b_{ij,k1}^\star&\ldots&b_{ij,kk}^\star\end{pmatrix}
 \begin{pmatrix}z_{j,1}\\\vdots\\z_{j,k}\end{pmatrix}
 +\rb_i+\wb_i=\sum_{j\in[p]\backslash \{i\}}\pBijk\zb_j+\rb_i+\wb_i,\label{eq:z_truncate}
 \end{align}
 where 
$ r_{i,\ell}=\sum_{j\in\pni}\sum_{\ell'=k+1}^\infty \dotp{b_{ij,\ell\ell'}^\star\phi_{\ell'}}{\Zcal_{ j}}$    
 , $w_{i,\ell}=\dotp{\Wcal_i}{\phi_{\ell}}$ for $\ell=1,\ldots,k$ 
and $\rb_{i}=(r_{i,1},\ldots,r_{i,k})^\top\in\RR^{k}$ denotes the residual vector. Let $\EE(\rb_i\rb_i^{\top})=\Sigmab_i^{\rb}$ and  $\EE(\wb_i\wb_i^{\top})=\Sigmab_i^{\wb}$ for $\pseq$.

Following~\eqref{eq:relationofA}, we can also express~\eqref{eq:z_truncate} as
 \begin{align}\label{eq:z_A}
 \Zcal^{k}_{i} = \sum_{\ell=1}^k\dotp{{\Ascr^m}(\chi_{i}^m)}{\phi_{\ell}}\phi_{\ell} 
 &=\sum_{\ell'=1}^\infty
 \sum_{\ell=1}^k
 \dotp{a_{\ell\ell'}^{m\star} \phi_{\ell'}^m}{\chi_{i}^m}\phi_{\ell}\notag\\
 &=\sum_{\ell=1}^k\rbr{  \sum_{\ell'=1}^{k_m}
 \dotp{a_{\ell\ell'}^{m\star} \phi_{\ell'}^m}{\chi_{i}^{k_m}}\phi_{\ell}
 +\sum_{\ell'=k_m+1}^{\infty}
 \dotp{a_{\ell\ell'}^{m\star}
 \phi_{\ell'}^m}{\chi_{i}^m}\phi_{\ell}}.
 \end{align}
 To simplify the analysis, we assume that the second term is $0$. 
Then, Equation~\eqref{eq:z_A} has the following equivalent expression in matrix form
 \begin{equation}
 \zb_i
 =
 \begin{pmatrix}\label{eq:z}
 z_{i,1}\\\vdots\\z_{i,k}
 \end{pmatrix}
 =
    \begin{pmatrix}
        a_{11}^{m\star}&\ldots& a_{1k_m}^{m\star}\\
        &\ddots&\\
        a_{k_m1}^{m\star}&\ldots&a_{kk_m}^{m\star}
    \end{pmatrix}
        \begin{pmatrix}
        x_{i,1}^m\\\vdots\\x_{i,k_m}^m\end{pmatrix}
    +
    \begin{pmatrix}
    u_{i,1}^m\\\vdots\\u_{i,k}^m
    \end{pmatrix}
    =\pAmk \xkmfree{i} +\ub_{i}^m,
 \end{equation}
 where
 $
 u_{i,\ell}^m=\sum_{\ell'=k+1}^\infty \dotp{a_{\ell\ell'}^{m\star}\phi_{\ell'}^m}{\chi_{i}^m}
 $ for $\ell=1,\ldots,k$
 and $\ub_{i}^m=(u_{i,1}^m,\ldots, u_{i,k}^m)^\top\in\RR^{k}$ is the truncation error in the observation space. 
 Let $\ub^m=(\ub_{1}^{m\top},\ldots,\ub_{p}^{m\top})^\top\in\RR^{kp}$ and $\EE(\ub^m\ub^{m\top})=\Sigmab^{m,\ub}$. Combining ~\eqref{eq:z_truncate} and ~\eqref{eq:z} together, we then have the relationship
\begin{equation}\label{eq:neighbour_reg}
    \pAmk \xkmfree{i} + \ub_{i}^m - \cbr{\sum_{j\in\pni}\pBijk\rbr{\pAmk \xkmfree{j}
    +\ub_{j}^m} +\rb_i+\wb_i}={\bf 0},
\end{equation}
  where $\wb_i$ follows independently from $\Ncal({\bf 0},\Sigmab^{\wb}_i)$. Let ${\qb}_{i}^m=(q_{i,1}^m,\ldots,q_{i,k_m}^m)^\top\in\RR^{k_m}$ where $q_{i,j}^m=\dotp{\xi_{i}^m}{\phi_{j}^m}$ for $j=1,\ldots k_m$, we define the finite truncation of the observed random vector as
\begin{equation}\label{eq:observation}
\yb_{i}^m = \xb_{i}^m + \qb_{i}^m,\qquad\mseq,\pseq,
\end{equation}
where $\qb_{i}^m$ is the noise in the observation space and follows the distribution $\Ncal({\bf 0}, \Sigmab^{m,\qb}_i)$. Therefore, under the assumption that $k_m$ is large enough such that $\qb_{i}^m$ are small in magnitude, we can approximate~\eqref{eq:neighbour_reg} as
\begin{equation*}
    \pAmk \ykmfree{i} + \ub_{i}^m - \cbr{\sum_{j\in\pni}\pBijk\rbr{\pAmk \ykmfree{j}
    +\ub_{j}^m} +\rb_i+\wb_i}\approx{\bf 0}.
\end{equation*}
Therefore, a reasonable choice of the finite dimensional realization of the original programming~\eqref{eq:obj_infinite} is
\[
\min_{\{\Amk\},\{\Bijk\}} \sum_{m=1}^M\sum_{i=1}^p\EE\sbr{\bignorm{\Amk\ykmfree{i}-\sum_{j\in\pni}\Bijk\Amk\ykmfree{j}}_F^2}.
\]
\section{Notations}\label{ssec:notation}

We define the following quantities
\begin{align*}
&\Upsilon_{1} = \frac{3\nu_x\nu_b\rho_x\rho_b}{16},\quad 
\Upsilon_{2}= \frac{\nu_x\nu_b+3\rho_x\rho_b}{4};\\ 
&\Upsilon_{3} = \norm{\sCov}_2(C_B+1)\max_{\mseq}\norm{\Ab^{m(0)}}_2\max_{\pseq}(\norm{\pBik }_F\vee \norm{\ptBik }_1);\\
&\Upsilon_{4} =2\sqrt{(1+\vartheta_1)(1+\vartheta_2)\tau\alpha^\star},\quad 
\Upsilon_{5}=4(1+C_A)^2(2+C_B)^2\rho_b\rho_a\max_{\mseq}\norm{\sCov}_2^2;\\
&\Upsilon_{6} = \frac{\nu_x\nu_a}{2},\quad 
\Upsilon_{7} = \rho_a^2\max_{\mseq}\norm{\sCov}_2^2;\\ 
& \Upsilon_{8}= 
    C_0\rho_x\rho_b
        \cbr{
        \rho_b\rbr{
        \max_{\mseq}
        \norm{\Sigmab^{m,\ub}}_2+
        \rho_a\norm{\Sigmab^{m,\qb}}
        }
        +
        \max_{\pseq}\rbr{
            \norm{\Sigmab_i^{\rb}}_2
            +
            \norm{\Sigmab_i^{\wb}}_2
        }
    };\\
&\Upsilon_{9}
 = 
    C_1C_\vartheta(1+\vartheta_1)\rho_x\rho_a
    \cbr{
        \rho_b\rbr{
        \max_{\mseq}\norm{\Sigmab^{m,\ub}}_2+\rho_a
        \norm{\Sigmab^{m,\qb}}
        }
        +
        \max_{\pseq}\rbr{
            \norm{\Sigmab_i^{\rb}}_2
            +
            \norm{\Sigmab_i^{\wb}}_2
        }}
    ;\\
&C_\vartheta = \rbr{1+\sqrt{\frac{2}{\vartheta_1-2}}}^2\rbr{1+\sqrt{\frac{4}{\vartheta_2-2}}}^2,
\quad C_\alpha=\frac{\Upsilon_8}{\Upsilon_1},
\quad C_\beta = \frac{\Upsilon_9}{\Upsilon_7};\\
& C_\gamma = \max_{\mseq}\rho_b^{1/2}
\norm{\sCov}_2
\cbr{\sinmax{\Amk^{(0)}}+\rho_a^{1/2}}
\sqrt{(1+\vartheta_1)s^\star};\\
&C_{\delta,i}=C_2\sqrt{(1+\vartheta_2)(1+\vartheta_1)}\max_{\mseq}\rho_x^{1/2}\rho_a^{1/2}\{\rho_b^{1/2}\norm{\Sigmab^{m,\ub}}_2^{1/2}\\
&\qquad\qquad\qquad\qquad\qquad\qquad\qquad\qquad+
(\rho_a\rho_b)^{1/2}\norm{\Sigmab^{m,\qb}}_2^{1/2}+\norm{\Sigmab_i^\rb}_2^{1/2}+\norm{\Sigmab_i^\wb}_2^{1/2}\};\\
&C_{\gamma_k,\nu,\rho} =C_3\bigg\{\gamma_k^{-1/2}(\rho_x^{1/2}\nu_x^{-2}+\nu_x^{-3/2})+\rho_x^2\nu_x^{-5/2}(\gamma_k^{-1/2}+\gamma_k^{-2})(\rho_x^{1/2}\nu_x^{-1/2}+1)\bigg\}\\
&C_{\eta_A} = \frac{3\Upsilon_1}{4\Upsilon_2},
\quad C_{\eta_B} = \frac{\Upsilon_6}{4},\\
\end{align*}
where $C_0, C_1, C_2,C_3,C_4,C_5$ are absolute constants, \textcolor{black}{which
 are independent of data parameters: $M, N, p, s^\star,\alpha$,
upper and lower bounds of the singular values: $\nu_x,\rho_x,\nu_b,\rho_b,\nu_a,\rho_a$, model parameters: $k,k_m, C_A,C_B$, and tuning parameters: $\gamma_1,\gamma_2,\tau_1,\tau_2$.}

\section{Analysis of Edge Estimation}
The section discusses the analysis of edge estimation, where the results are presented in Theorem~\ref{theorem:convergence}--\ref{theorem:main}. We begin with introducing some notations and math formulations that will assist the analysis in Section~\ref{ssec:edge_notation}. In Section~\ref{ssec:edge_preonestep}, we discuss the shrinkage of the distance metric $\max_{\mseq}\norm{\Amk-\pAmk}_F^2$ and $\max_{\pseq}\norm{\Bik-\pBik}_F^2$ per iterate of Algorithm~\ref{alg:update} along with the statistical error induced by the finite sample size setting. In section~\ref{ssec:edge_onestep}, we combine the results from Section~\ref{ssec:edge_preonestep} to state the contraction of the distance metric $\max_{\mseq}\norm{\Amk-\pAmk}_F^2+\max_{\pseq}\norm{\Bik-\pBik}_F^2$ per iterate of Algorithm~\ref{alg:update} and the underlying conditions. With results from Section~\ref{ssec:edge_onestep} in hand, we are able to apply the telescoping technique to show convergence of Algorithm~\ref{alg:update}. The analysis is presented in Section~\ref{ssec:theorem_convergence}. Section~\ref{ssec:edge_auxlemma1}--\ref{ssec:edge_auxlemma2} present the proof of lemmas used in Section~\ref{ssec:edge_preonestep}--\ref{ssec:theorem_convergence}.

\subsection{Notations}\label{ssec:edge_notation}
For analysis simplicity, we can write the objective function~\eqref{eq:f_n} as
 \begin{align}
 f(\setAk,\setBk) 
 &=
 \sum_{i=1}^p\sum_{m=1}^M\frac{1}{2N}
 \bignorm{
    \Amk \Ykmnfree{i}
    -
    \sum_{j\in\pni}
    \Bijk\Amk \Ykmnfree{j} 
}_F^2\notag\\
&=
\sum_{i=1}^p\sum_{m=1}^M\frac{1}{2N}
\norml{
    \Amk \Ykmnfree{i}
    -
    \Bik (\Ib_{p-1} \otimes \Amk )
    \Ykmnfree{\noi}
}_F^2\notag\\
&=
\sum_{i=1}^p\sum_{m=1}^M
\frac{1}{2N}
\norm{\tBik (\Ib_{p} \otimes \Amk )\Ykmn}_F^2,\label{eq:main_ver2}
\end{align}
where $\otimes$ denote the Kronecker product. Define $$
\sCov =N^{-1}\Ykmn(\Ykmn)^\top, \quad\sCovfree{\cdot}{\noi}=N^{-1}\Ykmn(\Ykmnfree{\noi})^\top.
$$
Then, for any matrix $\Delta\in\RR^{k\times k_m}$, we can write
 \begin{equation}\label{eq:dpot_gradAm}
     \dotp{\nabla_{\Amk }f(\setAk,\setBk)}{\Delta} = \sum_{i=1}^p\dotp{\tBikT \tBik (\Ib_p\otimes \Amk )\sCov}{(\Ib_p\otimes \Delta)}
 \end{equation}
 for $\mseq$. Moreover, the gradient of $f(\setAk,\setBk)$ with respect to $\Amk $ for $\mseq$ is
\begin{equation}\label{eq:gradAm}
\nabla_{\Amk }f(\setAk,\setBk)
=
\sum_{j=1}^p
(\eb_j^\top\otimes \Ib_k)
\cbr{\sum_{i=1}^p\tBikT \tBik (\Ib_p\otimes \Amk )\sCov}
(\eb_j\otimes \Ib_{k_m}).
\end{equation}

Similarly, for any matrix $\Delta\in\RR^{k\times k(p-1)}$, we can write
\begin{equation}\label{eq:dpot_gradBij}
\dotp{\nabla_{\Bb_{i}}f(\setAk,\setBk)}{\Delta} = \sum_{m=1}^M \dotp{-\tBik (\Ib_p\otimes \Amk )
\sCovfree{\cdot}{\noi}
(\Ib_{p-1}\otimes \Amk )^\top}{\Delta}
\end{equation}
for every $\pseq$ and $j\in\pni$.

\subsection{Preliminaries: One-Step Update and statistical quantities}\label{ssec:edge_preonestep}
We show Theorem~\ref{theorem:convergence} in several steps. We begin with showing the contraction of distance metric per iterate of Algorithm~\ref{alg:update}. 
First, given the current estimate $\Amk$ for $\mseq$, we define an update of $\Amk$ from Algorithm~\ref{alg:update} as
\begin{align*}
    \Ab^{m+}
    =
    \Pcal_{\tau}(\Amk-\eta_A\nabla_{\Amk} f(\setAk,\setBk)),\quad\mseq;
\end{align*}
Similarly, for every $\pseq$, we define the update of $\Bik$ as $\Bb_i^+$. Let
\begin{align}\label{eq:defsupport}
    \Scal_i^\star=\supp(\pBik ),\quad \Scal_i=\supp(\Bik ),\quad  \Scal^+_i=\supp(\Bb_i^{+}),\quad
    \bar\Scal_i = \Scal^\star_i\cup\Scal_i\cup\Scal_i^+.
\end{align}
By definition of $\pBik$, we have $\abr{\Scal_i^\star}\leq s^\star(\alpha^\star k)^2$ and $\abr{\Scal_i},\abr{\Scal_i^+}\leq s(\alpha k)^2=(\vartheta_1s^\star)(\vartheta_2\alpha^\star k)^2/8$. Taking the union bound, it follows that $
\abr{\bar\Scal_i}\leq (1+\vartheta_1)s^\star(1+\vartheta_2)^2(\alpha^\star k)^2$. 

 Similarly, we define the corresponding sets of node indices as  
 \begin{align}\label{eq:defneighbour}
     &\pNi=\{j:\norm{\pBijk}_F>\epsilon_0\},\quad
     \eNi=\{j:\norm{\Bijk}_F>\epsilon_0\};\;\notag\\
     &\eNinext=\{j:\norm{\Bb_{i,j}^+}_F>\epsilon_0\},\quad
     \barNi=\pNi\cup\eNi\cup\eNinext,
 \end{align}
    with $|\barNi|\leq (1+\vartheta_1)s^\star$ for $\pseq$. 

By definition, we have
\begin{align}
    \Bb_i^{+} 
    &=
        \Hcal_\alpha\circ\Tcal_s\rbr{
            \Bik  
            - 
            \eta_B\nabla_{\Bik }
            f(\setAk,\setBk)
        }
    \notag\\
    &= 
        \Hcal_\alpha\circ\Tcal_s\rbr{
            \Bik  
            - 
            \eta_B[\nabla_{\Bik }
            f(\setAk,\setBk)]_{\bar\Scal_i}
        },\label{eq:proj_supp}
\end{align}
for $\pseq$.

With definitions of $\Amk^{+}$ for $\mseq$ and $\Bik^{+}$ for $\pseq$, Lemma~\ref{lemma:onestep_Am} states the contraction of the distance $\norm{\Amk^{+}-\pAmk}_F^2$ with respect to $\norm{\Amk-\pAmk}_F^2$ for $\mseq$ while Lemma~\ref{lemma:onestep_Bij} states the contraction of $\norm{\Bb_i^{+}-\pBik }_F^2$ with respect to $\norm{\Bb_i-\pBik }_F^2$ for $\pseq$ per iterate of Algorithm~\ref{alg:update}. We leave the details of the proof in Section~\ref{ssec:edge_auxlemma1}.

\begin{lemma}[One-Step Update of $\Amk$]\label{lemma:onestep_Am} Recall that $\Upsilon_j$ for $j=1,\ldots,6$ are constants defined in Section~\ref{ssec:notation}.
Suppose that $k\leq\min_{\mseq} k_m$ and Assumption~\ref{assumption:cov}--\ref{assumption:AB} hold.
There exists constants $C_A$ and $C_B$ such that $\norm{\Amk-\pAmk}_2\leq C_A\norm{\pAmk}_2$, $\norm{\Bik-\pBik}\leq C_B\norm{\pBik}$ for any unitarily invariant norm and $\norm{\BikT-\pBikT}_1\leq C_B\norm{\pBikT}_1$
for $\mseq,\;\pseq$.
 Let $\defcdnm$,  $N=O(\cdnm^2 (\max_{\mseq} k_ms^\star+\log p +\log M))$  and $\eta_A\leq {2}/\{3p(s^\star+1)\Upsilon_{2}\}$, then
\begin{align}
    \max_{\mseq}
    \norm{\Ab^{m+}-\pAmk }_F^2
    &\leq4\bigg[ 
        \cbr{
            1
            -
            \eta_A\frac{3p(s^\star+1)\Upsilon_{1}}{2\Upsilon_{2}}
            +
            \eta_A\frac{2ps^\star\Upsilon_{4}^2\Upsilon_{3}^2}{\Upsilon_{6}}
        }
        \max_{\mseq}
        \norm{\Amk -\pAmk }_F^2\notag\\
    &\quad +
        k\eta_A\cbr{
            \frac{4\Upsilon_{2}}{p(s^\star+1)\Upsilon_{1}}
            +
            3\eta_A
        }
        \max_{\mseq}
        \norm{\gradAstar}_2^2\notag\\
    &\quad +
        p\eta_A\cbr{
            \frac{\Upsilon_{6}}{2}
            +
            3p(s^\star+1)\Upsilon_{5}\eta_A
        }\max_{\pseq}
        \norm{\Bik -\pBik }_F^2\bigg],
    \label{eq:onestepAm_r}
\end{align}
with probability at least $1-\max_{\mseq}6\exp\{-k_m(s^\star+1)\}$.
\end{lemma}

\begin{lemma}[One-Step Update of $\Bb_{ij}$]\label{lemma:onestep_Bij}Recall that $\Upsilon_j$ for $j=3,\ldots,7$ are constants defined in Section~\ref{ssec:notation}.
Suppose that $k\leq\min_{\mseq} k_m$ and Assumption~\ref{assumption:s-sparse}--\ref{assumption:disperse},~\ref{assumption:cov}--\ref{assumption:AB} hold.
There exists constants $C_A$ and $C_B$ such that $\norm{\Amk-\pAmk}_2\leq C_A\norm{\pAmk}_2$, $\norm{\Bik-\pBik}\leq C_B\norm{\pBik}$ for any unitarily invariant norm and $\norm{\BikT-\pBikT}_1\leq C_B\norm{\pBikT}_1$
for $\mseq,\;\pseq$.
 Let $\defcdnm$ and  $N=O(\cdnm^2 (\max_{\mseq} k_ms^\star+\log p +\log M))$.
Then, we have
\begin{align}
    \max_{\pseq}
    \norm{\Bb_i^{+}-\pBik }_F^2
    &\leq 
       C_\vartheta \bigg\{
        \rbr{
            1
            -
            \eta_B M\Upsilon_{6}
            +
            3\eta_B^2M^2\Upsilon_{7}
        }
        \max_{\pseq}
        \norm{\Bik -\pBik }_F^2\notag\\
    &\quad + 
        Ms^\star\eta_B\rbr{
            \frac{2\Upsilon_{4}^2\Upsilon_{3}^2}{\Upsilon_{6}}
            +   
            3\Upsilon_{5}M\eta_B
        }
        \max_{\mseq}
        \norm{\Amk -\pAmk }_F^2\notag\\
    &\quad+(1+\vartheta_1)s^\star\eta_B\rbr{\frac{2}{M\Upsilon_{6}}+3\eta_B}\max_{\pseq}\norm{[\gradBstar]_{\bar\Scal_i}}^2_{r(k,\infty)}
    \bigg\},\label{eq:onestepBi_r}
\end{align}
with probability at least $1-\max_{\mseq}\exp(-k_m s^\star)$.
\end{lemma}

Observe that the right hand sides of~\eqref{eq:onestepAm_r} and~\eqref{eq:onestepBi_r} has the term of $\norm{\gradAstar}_2^2$ and $\norm{[\gradBstar]_{\bar\Scal_i}}^2_{r(k,\infty)}$, respectively. By law of large number, ideally, these quantities would be negligible as the number of sample size goes to infinity. Specifically, $(\setpAk, \setpBk)$ is a stationary point of $\EE f(\setpAk, \setpBk)$. Hence, we expect that $\norm{\gradAstar}_2^2$ and $\norm{[\gradBstar]_{\bar\Scal_i}}^2_{r(k,\infty)}$ to be small when we have some reasonable number of sample size. Lemma~\ref{lemma:staterror_Am} states the upper bound of the $\norm{\gradAstar}_2$ for $\mseq$ while Lemma~\ref{lemma:staterror_Bij} provides the upper bound of $\norm{[\gradBstar]_{\bar\Scal_i}}^2_{r(k,\infty)}$ for $\pseq$.
\begin{lemma}[statistical error bound of $\gradAstar$]\label{lemma:staterror_Am} Under conditions of Lemma~\ref{lemma:onestep_Am} and assume that $N=O(C_\epsilon^2(k+k_m+\log M))$. Let $C_\epsilon=C_3\max_{\pseq}\rho_x^{1/2}\rho_b^{1/2}\big\{
\rho_b^{1/2}
        \norm{\Sigmab^{m,\ub}}_2^{1/2}
        +
            (\rho_b\rho_a)^{1/2}
            \norm{\Sigmab^{m,\qb}}_2^{1/2}
        +
            \norm{\Sigmab_i^{\rb}}_2^{1/2}
        +
            \norm{\Sigmab_i^{\wb}}_2^{1/2}\big\}$ for some absolute constant $C_3>0$.

Then, we have
\[
        \frac{1}{p (s^\star+1)}
        \norm{\gradAstar}_2
        >
        C_\epsilon\sqrt{\frac{k_m+k+\log M}{N}},
\]
    with probability smaller than $M^{-1}\exp\{-(k_m+k)\}$.
\end{lemma}

\begin{lemma}[statistical error bound of $\gradBstar$]\label{lemma:staterror_Bij} 
Let $C_4, C_5>0$ be some absolute constants and 
$$C_\zeta=C_4(1+\vartheta_2)\max_{\mseq}\rho_x^{1/2}\rho_a^{1/2}
\{
\rho_b^{1/2}\norm{\Sigmab^{m,\ub}}_2^{1/2}
+
(\rho_b\rho_a)^{1/2}\norm{\Sigmab^{m,\qb}}_2^{1/2}
+
\norm{\Sigmab_i^{\rb}}_2^{1/2}
+
\norm{\Sigmab_i^{\wb}}_2^{1/2}
\}.$$ Assume conditions of Lemma~\ref{lemma:onestep_Bij} hold and sample size $N=O(C_\zeta^2\{(\alpha^\star k)^2+\log p\})$. For any support set $\bar{\Scal}_i$ defined in~\eqref{eq:defsupport} satisfying $|\bar{\Scal}_i|\leq 
 C_5 k^2s^\star$,  we have
\[
    \frac{1}{M}
    \norm{[\gradBstar]_{\bar\Scal_i}}_{r(k,\infty)}
    >
    C_\zeta
    \sqrt{\frac{(\alpha^\star k)^2+\log p}{N}},
\]
    with probability smaller than $p^{-1}\exp\{-(\alpha^\star k)^2\}$.
\end{lemma}

\subsection{One-step Update}\label{ssec:edge_onestep}
With results from Lemma~\ref{lemma:onestep_Am}--\ref{lemma:staterror_Bij} introduced in Section~\ref{ssec:edge_preonestep} as the building blocks, we are able to construct the analysis of the one-step update of Algorithm~\ref{alg:update}. Noting that the sample complexity required in Lemma~\ref{lemma:onestep_Am}--\ref{lemma:staterror_Bij} depend on some condition numbers $\kappa_m$, $C_\epsilon$ and $C_\zeta$. We assume that  $\kappa_m^2$, $C_\epsilon^2$ and $C_\zeta^2$ are not large and see them as constants. Therefore, we could drop the terms and simplify the notations by focusing the rates of convergence with respect to $\alpha^\star$, $k$, $s^\star$,$p$, $M$ and $k_m$ for $\mseq$.
Lemma~\ref{lemma:onestep} states that with proper initial condition and step size $\eta_A$ and $\eta_B$, the distance $\max_{\pseq}\norm{\Bb_i^{+}-\pBik }_F^2
        +
        \max_{\mseq}\norm{\Ab^{m+}-\pAmk }_F^2$ is contracted at each iteration.
\begin{lemma}\label{lemma:onestep}[One-Step Update] Suppose that the conditions of Lemma~\ref{lemma:onestep_Bij} hold.
Let   
$N=O(s^\star\max_{\mseq}k_m+(\alpha^\star k)^2+\log p+ \log M)$
. Define 
    \[
        \rho_1 = 4\cbr{
                1
                -
                \frac{3p(s^\star+1)\Upsilon_{1}}{4\Upsilon_{2}}\eta_A
            },
            \quad \rho_2 =
           C_\vartheta \rbr{
            1
            -
            \frac{M\Upsilon_{6}}{4}\eta_B
        },
    \]
    and $\delta_0:=\exp\{- (s^\star+1)\min_{\mseq}k_m\}\vee\exp\{-(\alpha^\star k)^2\}$.
 Suppose that 
    \begin{align*} \alpha^\star&\leq\frac{\Upsilon_{1}\Upsilon_{6}}{32(1+\vartheta_1)(1+\vartheta_2)\tau_2\Upsilon_{2}\Upsilon_{3}^2},\\
    \eta_A&\leq\frac{1}{3p(s^\star+1)}\rbr{\frac{\Upsilon_{6}}{8\Upsilon_{5}}\wedge\frac{2}{\Upsilon_{2}}},
    \quad\eta_B\leq\frac{1}{12M}\rbr{\frac{\Upsilon_{1}}{\Upsilon_{5}\Upsilon_{2}}\wedge\frac{\Upsilon_{6}}{2\Upsilon_{7}}},
    \end{align*}
    and $\eta_A$ and $\eta_B$ are selected such that $\rho_1,\rho_2< 1$  and satisfy
    $4p\eta_A=C_\vartheta M\eta_B$.
    Then, 
    \begin{multline*}
        \max_{\pseq}\norm{\Bb_i^{+}-\pBik }_F^2
        +
        \max_{\mseq}\norm{\Ab^{m+}-\pAmk }_F^2\\
        \leq
            \rho_1\max_{\mseq}\norm{\Amk -\pAmk }_F^2,
        +
        \rho_2\max_{\pseq}\norm{\Bik -\pBik }_F^2
        +
        \frac{\Upsilon_{8}}{\Upsilon_{1}}\max_{\mseq}\frac{k(k_m+k+\log M)}{N}\\
        +
        \frac{\Upsilon_{9}}{\Upsilon_{7}}
        \frac{s^\star\{(\alpha^\star k)^2+\log p\}}{N},
\end{multline*}
    with probability at least $1-10\delta_0$.
\end{lemma}
\begin{proof}[Proof of Lemma~\ref{lemma:onestep}] First, we combine the results from Lemma~\ref{lemma:onestep_Am} and Lemma~\ref{lemma:onestep_Bij}, we can obtain
\begin{align}
        \max_{\pseq}\norm{\Bb_i^{+}-\pBik }_F^2
        &+
        \max_{\mseq}\norm{\Ab^{m+}-\pAmk }_F^2
\notag\\
    &\leq
    T_1\max_{\mseq}\norm{\Amk -\pAmk }_F^2+T_2\max_{\pseq}\norm{\Bik -\pBik }_F^2\notag\\
    &\quad+
    T_3\max_{\mseq}\norm{\gradAstar}_2^2
    \notag\\
    &\quad+
    T_4\max_{\pseq}
    \norm{[\gradBstar]_{\bar\Scal_i}}^2_{r(k,\infty)}
    \label{eq:onestep}
\end{align}    
where
\begin{align*}
    T_1
    &=
        4\cbr{1
        -
        \eta_A\frac{3p(s^\star+1)\Upsilon_{1}}{2\Upsilon_{2}}
        +
        \eta_A\frac{2ps^\star\Upsilon_{4}^2\Upsilon_{3}^2}{\Upsilon_{6}}}
        +
       C_\vartheta Ms^\star\eta_B\rbr{
            \frac{2\Upsilon_{4}^2\Upsilon_{3}^2}{\Upsilon_{6}}
            +
            3\Upsilon_{5}M\eta_B};\\
    T_2
    &=
       C_\vartheta {\rbr{
            1
            -
            \eta_B M\Upsilon_{6}
            +
            3\eta_B^2M^2\Upsilon_{7}}
        }
    + 
        4p\eta_A\rbr{
            \frac{\Upsilon_{6}}{2}
            +
            3p(s^\star+1)\Upsilon_{5}\eta_A
        };\\
    T_3
    &=
        4k\eta_A\rbr{
            \frac{4\Upsilon_{2}}{p(s^\star+1)\Upsilon_{1}}
            +
            3\eta_A
        };\\
    T_4
    &=
       C_\vartheta 
            (1+\vartheta_1)s^\star\eta_B\rbr{\frac{2}{M\Upsilon_{6}}+3\eta_B}.
\end{align*}
First, we want to simplify the term of $T_1$. With the choices of $\alpha^\star$, $\eta_B$ in the statement, we can obtain
\begin{align*}
    \frac{2ps^\star\Upsilon_{4}^2\Upsilon_{3}^2}{\Upsilon_{6}}\leq\frac{p(s^\star+1)\Upsilon_{1}}{4\Upsilon_{2}},\quad \rbr{\frac{2\Upsilon_{4}^2\Upsilon_{3}^2}{\Upsilon_{6}}+3\Upsilon_{5}M\eta_B}\leq\frac{\Upsilon_{1}}{2\Upsilon_{2}}.
\end{align*} 
Since $4p\eta_A =C_\vartheta M\eta_B$, we can conclude 
\begin{align*}
    T_1\leq 4\cbr{1-\frac{3p(s^\star+1)\Upsilon_{1}}{4\Upsilon_{2}}\eta_A}.
\end{align*}

Next, we want to simplify the term $T_2$. Similarly, with the choices of $\alpha^\star$, $\eta_B$, and $\eta_A$, we have
\begin{align*}
    3\eta_BM^2\Upsilon_{7}\leq\frac{M\Upsilon_{6}}{8}, \quad 3p(s^\star+1)\Upsilon_{5}\eta_A\leq \frac{\Upsilon_{6}}{8};
\end{align*}
Then, using the fact that $p\eta_A =C_\vartheta \eta_B$, we can conclude that 
\begin{align*}
    T_2\leq 
   C_\vartheta 
    \cbr{1 - \frac{M\Upsilon_{6}}{4}\eta_B}.
\end{align*}
With the choice of $\eta_A$, we have
\begin{align*}
    T_3 &\leq \frac{32}{3}\frac{k}{\Upsilon_1\{p(s^\star+1)\}^2} + \frac{16}{3}\frac{k}{\Upsilon_2^2\{p(s^\star+1)\}^2}.
\end{align*}
Using the fact that $\Upsilon_2^2\geq 2\Upsilon_1$, we can further bound the above display as
\begin{align*}
    T_3\leq 
    \frac{14 k}{\Upsilon_{1}\{p(s^\star+1)\}^2}
    \leq
    \frac{14 k}{\Upsilon_{1}(ps^\star)^2}.
\end{align*}
Recall the definition of $\Upsilon_8$ in Section~\ref{ssec:notation}. In fact, by Lemma~\ref{lemma:staterror_Am} and take the union bound over $\mseq$, we can verify that
\begin{align}\label{eq:even_staterrorA}
    T_3\max_{\mseq}\norm{\gradAstar
    }_2^2
    \leq   \frac{\Upsilon_{8}}{\Upsilon_{1}}        \max_{\mseq}k\frac{k+k_m+\log M}{N},
\end{align}
with probability at least $1 - \max_{\mseq}\exp\{-(k_m+k)\}\geq 1-\delta_0$, where we define this event as $\Ecal_A$. 

We find the upper bound of $T_4$ in a similar way as we do for $T_3$. Note that by Lemma~\ref{lemma:ci}, We have $\Upsilon_7\geq \Upsilon_6^2$. Then, with the choice of $\eta_B$, we have 
\begin{align*}
    &T_4
    \leq C_\vartheta(1+\vartheta_1)s^\star\rbr{\frac{1}{12M^2\Upsilon_7}+\frac{3}{576}\frac{1}{M^2\Upsilon_7}}
    \leq 
    \frac{C_\vartheta(1+\vartheta_1)s^\star}{M^2\Upsilon_{7}}. 
\end{align*}
Recall the definition of $\Upsilon_9$ in Section~\ref{ssec:notation}.
Then, apply Lemma~\ref{lemma:staterror_Bij} and take the union bound over $\pseq$, we can obtain that
\begin{equation}\label{eq:even_staterrorB}
    T_4\max_{\pseq}
    \norm{[\gradBstar]_{\bar\Scal_i}}_{r(k,\infty)}^2
    \leq
    \frac{\Upsilon_{9}}{\Upsilon_{7}}
    \frac{s^\star\{(\alpha^\star k)^2+\log p\}}{N},
\end{equation}
with probability at least $1-\exp\{-(\alpha^\star k)^2\}\geq 1-\delta_0$ and we define such event as $\Ecal_B$.



Let $\Ecal$ be the event when Lemma~\ref{lemma:onestep_Am}--\ref{lemma:onestep_Bij} hold. Therefore, by De Morgan's law,  $\Ecal\cap\Ecal_A\cap\Ecal_B$ happens with probability at least 
$1-10\delta_0$.
\end{proof}
\subsection{Proof of Theorem~\ref{theorem:convergence}}\label{ssec:theorem_convergence}
First, we apply Lemma~\ref{lemma:onestep} for one iteration and obtain that with probability at least $1-10\delta_0$:
\begin{multline*}
        \max_{\pseq}\norm{\Bb_i^{(1)}-\pBik }_F^2
        +
        \max_{\mseq}\norm{\Ab^{m(1)}-\pAmk }_F^2
        \\
        \leq
            \rho_1\max_{\mseq}\norm{\Ab^{m(0)} -\pAmk }_F^2,
        +
        \rho_2\max_{\pseq}\norm{\Bb_i^{(0)} -\pBik }_F^2
        +
        \Xi^2,
\end{multline*}
where the right hand side is bounded by $(C_1+1) (C_A^2\rho_a+C_B^2\rho_b)$ because (i) $\rho_1,\rho_2<1$ and (ii) $\Xi^2\leq C_1 (C_A^2\rho_a+C_B^2\rho_b)$ under the condition of Theorem~\ref{theorem:convergence}. Therefore, we apply Lemma~\ref{lemma:onestep} again to obtain the upper bound for $\max_{\pseq}\norm{\Bb_i^{(2)}-\pBik }_F^2
        +
        \max_{\mseq}\norm{\Ab^{m(2)}-\pAmk }_F^2$:
\begin{multline*}
        \max_{\pseq}\norm{\Bb_i^{(2)}-\pBik }_F^2
        +
        \max_{\mseq}\norm{\Ab^{m(2)}-\pAmk }_F^2
        \\
        \leq
            \rho_1\max_{\mseq}\norm{\Ab^{m(1)} -\pAmk }_F^2,
        +
        \rho_2\max_{\pseq}\norm{\Bb_i^{(1)} -\pBik }_F^2
        +
        (1+\pi)\Xi^2,
\end{multline*}
where $\pi=\rho_1\vee\rho_2$ and with probability at least $1-10\delta_0$.

Then, iterate for $L$ times and by the telescoping technique, we complete  the proof:
\begin{align*}
    \max_{\mseq}\norm{\Ab^{m(L)} - \pAmk}_F^2 
    + 
    \max_{\pseq}\norm{\Bb_i^{(L)} - \pBik}_F^2\leq
    \pi^LR_0^2+
    \frac{1}{1-\pi} \Xi^2,
\end{align*}
with probability at least $1-10\delta_0$.

\subsection{Proof of Lemma~\ref{lemma:onestep_Am}--\ref{lemma:staterror_Bij}}\label{ssec:edge_auxlemma1}
\begin{proof}[Proof of Lemma~\ref{lemma:onestep_Am}]
We study one iteration for updating $\Amk $ for $\mseq$ via Algorithm~\ref{alg:update} and we have
\begin{align*}
\norm{\Ab^{m+}-\pAmk }_F^2 
&= 
    {\norm{\Pcal_\tau\rbr{\Amk  - \eta_A\nabla_{\Amk }f(\setAk,\setBk)}-\pAmk }_F^2}.\\
    \intertext{
    First, we claim that any row of $\pAmk$ is in $\Kcal_m(\tau_1,\tau_2)$ defined by~\eqref{eq:constraint_a} with probability at least $1-5\exp(-\sum_{m=1}^2 k_m)$. Then, apply Lemma~\ref{lemma:proj_Km} and the above display could be further bounded as}
&\leq 
    4{\norm{\Amk  - \eta_A\nabla_{\Amk }f(\setAk,\setBk)-\pAmk }_F^2}\\
&=
    4\{\norm{\Amk -\pAmk }_F^2-2\eta_A\Ical_{1,m}+\eta_A^2\Ical_{2,m}\},
\end{align*}
where 
\[
\Ical_{1,m}=\dotp{\nabla_{\Amk }f(\setAk,\setBk)}{\Amk -\pAmk },\quad 
\Ical_{2,m}=\norm{\nabla_{\Amk }f(\setAk,\setBk)}_F^2.
\]

Now we justify the claim that any row of $\pAmk$ is in $\Kcal_m(\tau_1,\tau_2)$. For each row $j=1,\ldots,k$, we could apply triangle inequality and obtain that
\begin{align*}
\norm{\Ab_{j\cdot}^{m\star}}_2
&\leq 
\norm{\Ab_{j\cdot}^{m\star}-\Ab_{j\cdot}^{m(0)}}_2
+
\norm{\Ab_{j\cdot}^{m(0)}}_2\\
&\leq
\norm{\Ab_{j\cdot}^{m\star}-\Ab_{j\cdot}^{m(0)}}_2
+
\norm{\Ab^{m(0)}}_2.
\end{align*}
Under Assumption~\ref{assumption:disperse} and $N=O(\max_{\mseq} \cdnm^2 k_m)$, by Theorem~\ref{theorem:initialcca}, we know that $\norm{\Ab_{j\cdot}^{m\star}-\Ab_{j\cdot}^{m(0)}}_2$ is bounded by some small constant with probability at least $1-5\exp(-\sum_{m=1}^M k_m)$, and we define such event as $\Ecal_I$. Therefore, there exists some finite $\tau_2>0$ such that $\norm{\Ab_{j\cdot}^{m\star}}_2\leq\tau_2k_m^{-1/2}\norm{\Ab^{m(0)}}_2$ for $j=1,\ldots,k$. The existence of a valid $\tau_1$ can be justified in a similar way and hence we omit the 
step.

In order to find the upper bound of $\norm{\Ab^{m+}-\pAmk }_F^2$, it suffices to find the lower bound of $\Ical_{1,m}$ and the upper bound of $\Ical_{2,m}$, which are presented in the Lemma~\ref{lemma:I_m1} and Lemma~\ref{lemma:I_m2}, respectively.

On the event of $\Ecal_I$, combining the results of Lemma~\ref{lemma:I_m1}--\ref{lemma:I_m2} and taking the maximum over $\mseq$ on both sides, we obtain that
\begin{align}\label{eq:lA}
    \frac{1}{4}\max_{\mseq}
    \norm{\Ab^{m+}-\pAmk }_F^2
    &\leq 
        \rbr{
            1
            -
            \eta_A\frac{3p(s^\star+1)\Upsilon_{1}}{2\Upsilon_{2}}
            +
            \eta_A\frac{2ps^\star\Upsilon_{4}^2\Upsilon_{3}^2}{\Upsilon_{6}}
        }
        \max_{\mseq}
        \norm{\Amk -\pAmk }_F^2\notag\\
    &\quad +
        k\eta_A\rbr{
            \frac{4\Upsilon_{2}}{p(s^\star+1)\Upsilon_{1}}
            +
            3\eta_A
        }
        \max_{\mseq}
        \norm{\gradAstar}_2^2\notag\\
    &\quad +
        p\eta_A\rbr{
            \frac{\Upsilon_{6}}{2}
            +
            3p(s^\star+1)\Upsilon_{5}\eta_A
        }\max_{\pseq}
        \norm{\Bik -\pBik }_F^2\notag\\
    &\quad +
        \eta_A\rbr{
            -\frac{2}{p(s^\star+1)\Upsilon_{2}}
            +
            3\eta_A
        }\notag\\
    &\quad\times\cbr{
        \max_{\mseq}\norm{\nabla_{\Amk }f(\setAk,\setpBk)-\nabla_{\Amk }f(\setpAk,\setpBk)}_F^2
    }.
\end{align}
with probability at least $1-\max_{\mseq}\exp\{-k_m(s^\star+1)\}$. 
By conditions of $\eta_A$, $-{2}/\{p(s^\star+1)\Upsilon_{2}\}+3\eta_A<0$ and hence we can drop the last term of~\eqref{eq:lA}.

Then, marginalizing over the event $\Ecal_I$,~\eqref{eq:onestepAm_r} holds with probability at least
\begin{multline*}
\cbr{1-5\exp\rbr{-\sum_{m=1}^M k_m}}\sbr{1-\max_{\mseq}\exp\{-k_m(s^\star+1)\}}\\
\geq 1-\max_{\mseq}6\exp\{-k_m(s^\star+1)\}.
\end{multline*}

\end{proof}

\begin{proof}[Proof of Lemma~\ref{lemma:onestep_Bij}]
We study one iteration for updating $\Bik $ for  $i\in[p]$ via Algorithm~\ref{alg:update}. 
Recall that $s=\vartheta_1s^\star/2$ and $\alpha=\vartheta_2\alpha^\star/2$. 
Then, apply  Lemma~\ref{lemma:disperse_coef}--\ref{lemma:groupsparse} and~\eqref{eq:proj_supp}, we have
\begin{multline}\label{eq:B:update}
    \norm{\Bb_i^{+}-\pBik }_F^2 
    = 
        \norm{
            \Hcal_{\alpha}\circ\Tcal_s\cbr{
                \Bik  
                - 
                \eta_B[\nabla_{\Bik }f(\setAk,\setBk)]_{\bar\Scal_i}
            } 
            - 
        \pBik }_F^2\\
    \leq
        C_\vartheta
        \big\{
            \norm{\Bb_{i} - \Bb_{i}^\star}_F^2 
            - 
            2\eta_B \Ical_{3,i} 
            + 
            \eta_B^2 \Ical_{4,i}
            \big\},
\end{multline}
where
\[
    \Ical_{3,i} 
    = 
        \dotp{[
            \nabla_{\Bb_{i}} f(\setAk,\setBk)
        ]_{\bar\Scal_i}}{
            \Bb_{i}-\Bb_{i}^\star
        },
    \quad 
    \Ical_{4,i} 
    = 
        \norm{[
            \nabla_{\Bb_{i}} f(\setAk, \setBk)
        ]_{\bar\Scal_i}}_F^2.
\]
Next, we will apply Lemma~\ref{lemma:I_m3} for the lower bound of $\Ical_{3,i}$ and apply Lemma~\ref{lemma:I_m4} for the upper bound of $\Ical_{4,i}$ for $\pseq$.

Combining Lemma~\ref{lemma:I_m3}--\ref{lemma:I_m4} and taking the maximum over $\pseq$ yields,
\begin{align*}
    \max_{\pseq}
    \norm{\Bb_i^{+}-\pBik }_F^2
    &\leq 
        C_\vartheta\bigg\{
        \rbr{
            1
            -
            \eta_B M\Upsilon_{6}
            +
            3\eta_B^2M^2\Upsilon_{7}
        }
        \max_{\pseq}
        \norm{\Bik -\pBik }_F^2\notag\\
    &\quad + 
        Ms^\star\eta_B\rbr{
            \frac{2\Upsilon_{4}^2\Upsilon_{3}^2}{\Upsilon_{6}}
            +   
            3\Upsilon_{5}M\eta_B
        }
        \max_{\mseq}
        \norm{\Amk -\pAmk }_F^2\notag\\
    &\quad+(1+\vartheta_1)s^\star\eta_B\rbr{\frac{2}{M\Upsilon_{6}}+3\eta_B}\max_{\pseq}\norm{[\gradBstar]_{\bar\Scal_i}}^2_{r(k,\infty)}
    \bigg\},
\end{align*}
with probability at least $1-\max_{\mseq}\exp(-k_m s^\star)$.

\end{proof}

Before showing the proofs of Lemma~\ref{lemma:staterror_Am}--\ref{lemma:staterror_Bij}, we introduce the notion of Orlicz norm, where we only consider the sub-exponential and sub-Gaussian cases here. Assume that $X$ is a sub-exponential random variable, and  sub-exponential norm is defined as
\begin{align*}
    \|X\|_{\psi_1}=\inf\{t>0:\EE\exp(|X|/t)\leq 2\}.
\end{align*}

Assume that $X$ is a sub-Gaussian random variable, and  sub-Gaussian norm is defined as
\begin{align*}
    \|X\|_{\psi_2}=\inf\{t>0:\EE\exp(|X|^2/t^2)\leq 2\}.
\end{align*}

\begin{proof}[Proof of Lemma~\ref{lemma:staterror_Am}]

    Let $\Ncal_1$ be the $1/8$-net of $\mathbb{S}^{k-1}=\{\vb\in\RR^{k}:\norm{\vb}_2=1\}$ and $\Ncal_2$ be the $1/8$-net of $\mathbb{S}^{k_m-1}=\{\vb\in\RR^{k_m}:\norm{\vb}_2=1\}$. Then, by Lemma~\ref{lemma:epsilon_net}, we have
    \begin{align*}
        \norm{\gradAstar}_2&=\sup_{\ub\in\mathbb{S}^{k-1}}\sup_{\db\in\mathbb{S}^{k_m-1}} \abr{\ub^\top\gradAstar\db}\\
        &\leq 2\max_{\ub\in\Ncal_1}\max_{\db\in\Ncal_2}
        \abr{\ub^\top\gradAstar\db}
    \end{align*}
    Define
    \begin{align}
    \vb_{i}^{m,(n)}
    &=
    \ptBik (\Ib_p\otimes \pAmk )\ykmn\notag\\
    &=    \pAmk \ykmnfree{i}{n} -\pBik (\Ib_{p-1}\otimes \pAmk )\ykmnfree{\noi}{n} \notag\\
    &=
    \ptBik\cbr{{\ub^m}^{(n)}+
    \bigAs{p}{\pAmk}{\qb^m}^{(n)}}
    +\rb_i^{(n)}+\wb_i^{(n)}\label{eq:def_v},
    \end{align}
    for $\pseq$, $\mseq$ and $n=1,\ldots,N$.
    Using~\eqref{eq:gradAm}, we can write
    \begin{align*}
        \gradAstar 
        &= 
            \sum_{j=1}^p
            (\eb_j^\top\otimes \Ib_k)
            \cbr{
                \sum_{i=1}^p
                \ptBikT \ptBik 
                (\Ib_p\otimes \pAmk )
                \sCov
            }
            (\eb_j\otimes \Ib_{k_m})\\
        &= 
            \sum_{j=1}^p
            (\eb_j^\top\otimes \Ib_k)
            \sbr{
                \sum_{i=1}^p
                \ptBikT 
                \cbr{
                    \frac{1}{N}\sum_{n=1}^N    \vb_{i}^{m,(n)}(\ykmn)^\top
                }
            }
            (\eb_j\otimes\Ib_{k_m}).
    \end{align*}
    Then, combining the above two results, we observe that for any $\ub\in\Ncal_1$ and $\db\in\Ncal_2$
    \[
        \ub^\top\sbr{
            \sum_{j=1}^p
            (\eb_j^\top\otimes \Ib_k)
            \cbr{
                \sum_{i=1}^p
                \ptBikT \vb_{i}^m(\ykmfree{i})^\top
            }
            (\eb_j\otimes\Ib_{k_m})
        }\db,
    \]
    is a sub-exponential random variable. Therefore, we can obtain the upper bound of Orlicz norm as
\begin{align}
    \bignorm{
        \sum_{j=1}^p
        (\eb_j^\top\otimes \Ib_k)&
        \cbr{
            \sum_{i=1}^p
            \ptBikT
            \vb_{i}^m
            (\ykmfree{i})^\top
        }(\eb_j\otimes\Ib_{k_m})}_{\psi_1}\notag\\
    &\leq
        \sum_{i=1}^p
        \bignorm{
            \sum_{j=1}^p
            (\eb_j^\top\otimes \Ib_k)
            \cbr{
                \ptBikT \vb_{i}^m
                (\ykmfree{i})^\top
            }
            (\eb_j\otimes\Ib_{k_m})
        }_{\psi_1}.\notag\\
    \intertext{ Note that if $j\not\in\pNi\cup\{i\}$, then $(\eb_j^\top \otimes \Ib_k)\ptBikT={\bf 0}$. Therefore, the above display is equivalent as}
    &=
        \sum_{i=1}^p
        \bignorm{
            \sum_{j\in\pNi\cup\{i\}}
            (\eb_j^\top\otimes \Ib_k)
            \cbr{
                \ptBikT \vb_{i}^m(\ykmfree{i})^\top
            }
            (\eb_j\otimes\Ib_{k_m})
        }_{\psi_1}\notag.\\
    \intertext{By Lemma~2.7.7 in~\citep{vershynin2018high}, we can further bound the above display as}
    &\leq 
        \sum_{i=1}^p \sum_{j\in\pNi\cup\{i\}}
        \bignorm{
            (\eb_j^\top\otimes \Ib_k)
            \ptBikT 
            \vb_{i}^m
        }_{\psi_2}
        \bignorm{
            (\ykmfree{i})^\top
            (\eb_j\otimes \Ib_{k_m})
        }_{\psi_2}.\label{eq:orcliz_bound1}
\end{align}
Using the fact that $|\pNi\cup\{i\}|\leq s^\star+1$ and the fact that
\begin{align*}
    \norm{\ptBikT\vb_{i}^m}_{\psi_2} 
    &=    
        \bignorm{
        \ptBikT\cbr{\ptBik
        \rbr{{\ub^m}+
        \bigAs{p}{\pAmk}{\qb^m}
        }
        +\rb_i+\wb_i}
        }_{\psi_2}\\  
    &\leq
        \norm{\ptBikT\ptBik{\ub^m}}_{\psi_2}  
        +
        \norm{\ptBikT\bigAs{p}{\pAmk}{\qb^m}}_{\psi_2}
        +
            \norm{\ptBikT\rb_i}_{\psi_2}  
        +
            \norm{\ptBikT\wb_i}_{\psi_2}  \\
    &\leq
        \norm{\ptBik}_2
        (
            \norm{\ptBik}_2\norm{\Sigmab^{m,\ub}}_2^{1/2}
        +
        \norm{\ptBik}_2\norm{\pAmk}_2\norm{\Sigmab^{m,\qb}}_2^{1/2}
        +
            \norm{\Sigmab_i^{\rb}}_2^{1/2}
        +
            \norm{\Sigmab_i^{\wb}}_2^{1/2}
        ),
\end{align*}
for $\pseq$, we can further bound the right hand side of~\eqref{eq:orcliz_bound1} as
\begin{align*}
   \bignorm{
        \sum_{j=1}^p
        (\eb_j^\top\otimes \Ib_k)
        &\cbr{
            \sum_{i=1}^p
            \ptBikT
            \vb_{i}^m
            (\ykmfree{i})^\top
        }(\eb_j\otimes\Ib_{k_m})}_{\psi_1}\\
    &\leq 
        p(s^\star+1)\norm{\Sigmab^m}_2^{1/2}
        \max_{\pseq}
        \norm{\ptBik}_2
        (
            \norm{\ptBik}_2 \norm{\Sigmab^{m,\ub}}_2^{1/2}
           + \norm{\ptBik}_2\norm{\pAmk}\norm{\Sigmab^{m,\qb}}_2^{1/2}\\
           &\quad\quad\quad\quad\quad\quad\quad\quad\quad\quad\quad\quad\quad+
            \norm{\Sigmab_i^{\rb}}_2^{1/2}
            +
            \norm{\Sigmab_i^{\wb}}_2^{1/2}
        )\\
        &\leq p(s^\star+1)\rho_x^{1/2}\rho_b^{1/2}\big\{
\rho_b^{1/2}
        \norm{\Sigmab^{m,\ub}}_2^{1/2}
        +
            (\rho_b\rho_a)^{1/2}
            \norm{\Sigmab^{m,\qb}}_2^{1/2}\\
        &\quad\quad\quad\quad\quad\quad\quad\quad\quad+\max_{\pseq}(
            \norm{\Sigmab_i^{\rb}}_2^{1/2}
        +
            \norm{\Sigmab_i^{\wb}}_2^{1/2})\big\},
\end{align*}
where the last line follows by Assumption~\ref{assumption:cov}--\ref{assumption:AB}.

Next, let $C_\epsilon=C_3\max_{\pseq}\rho_x^{1/2}\rho_b^{1/2}\big\{
\rho_b^{1/2}
        \norm{\Sigmab^{m,\ub}}_2^{1/2}
        +
            (\rho_b\rho_a)^{1/2}
            \norm{\Sigmab^{m,\qb}}_2^{1/2}
        +
            \norm{\Sigmab_i^{\rb}}_2^{1/2}
        +
            \norm{\Sigmab_i^{\wb}}_2^{1/2}\big\}$ 
for some constant $C_3>0$. 
Under the condition of the sample size $N=O(C_\epsilon^2(k+k_m+\log M))$,
apply Bernstein's inequality (see Theorem 2.8.1 in\citet{vershynin2018high}) and take the union bound over $\Ncal_1$ and $\Ncal_2$, where $\abr{\Ncal_1}\leq 17^k$ and $\abr{\Ncal_2}\leq 17^{k_m}$, we can obtain
\[
    \frac{1}{p(s^\star+1)}\norm{\gradAstar}_2
    \geq 
        C_\epsilon
        \sqrt{\frac{k_m+k+\log M}{N}},
\]
with probability smaller than $M^{-1}\exp\{-(k_m+k)\}$.
\end{proof}

\begin{proof}[Proof of Lemma~\ref{lemma:staterror_Bij}]
We first expand $\gradBstar$ and obtains
\[
    \gradBstar 
    = 
        -\frac{1}{N}\sum_{m=1}^M\sum_{n=1}^N
        \{
            \pAmk \ykmnfree{i}{n} 
            -
            \pBik \bigAs{p-1}\pAmk \ykmnfree{\noi}{n}
        \}
        (\ykmnfree{\noi}{n})^\top
        \bigAs{p-1}{\pAmkT}.
\]
Replacing 
\begin{equation*}
    \pAmk \ykmnfree{i}{n} 
    - 
    \pBik \bigAs{p-1}{\pAmk}
    \ykmnfree{\noi}{n}
    =
        \vb_{i}^{m,(n)},
\end{equation*} 
for $\nseq$ and $\mseq$ yields
\[
    \gradBstar 
    = 
        -\frac{1}{N}\sum_{m=1}^M
        \sum_{n=1}^N
        \vb_{i}^{m,(n)}
        (\ykmnfree{\noi}{n})^\top
        \bigAs{p-1}{\pAmkT}.
\]
Define $\tb^{m,(n)}_{i}\in\RR^{k\times (p-1)}$ which transforms the original matrix
$
(\ykmnfree{\noi}{n})^\top\bigAs{p-1}{\pAmkT}
\in
    \RR^{1\times k(p-1)}
$ to 
\begin{equation}\label{eq:deft}
    \tb^{m,(n)}_{i}
    =
        (\pAmk \ykmnfree{1}{n}
        ,\ldots,\pAmk \ykmnfree{p}{n}
        )
        \in
            \RR^{k\times (p-1)},
\end{equation}
for $\mseq$ and $\nseq$. Then, we can write 
\[
    \bignorm{
        \frac{1}{N}\sum_{m=1}^M\sum_{n=1}^N
        [\vb_{i}^{m,(n)}
        (\ykmnfree{\noi}{n})^\top
        \bigAs{p-1}{\pAmkT}]_{\bar\Scal_i}
    }_{r(k,\infty)}
    =
        \bignorm{
            \frac{1}{N}\sum_{m=1}^M\sum_{n=1}^N
            [\vb_{i}^{m,(n)}\otimes \tb^{m,(n)}_{i}]_{\bar\Scal_i'}
        }_{2,\infty},
\]
 where we denote $\bar\Scal_i'$ as the support set whose elements are translated from $\bar\Scal_i$ due to the transformation from a matrix of dimension $k\times k(p-1)$  to ${k^2\times (p-1)}$.

\emph{Step 1}.
Consider a fixed $\bar{\Scal}_i$
and 
recall the definition of the neighborhood set $\barNi$ and the relation of $\bar\Scal_i$ and $\barNi$ defined in~\eqref{eq:defneighbour}. 
Since $|\barNi|\leq(1+\vartheta_1)s^\star$, we know that there are at most $(1+\vartheta_1)s^\star$ nonzero columns in 
$\sum_{m=1}^M\sum_{n=1}^N
            [\vb_{i}^{m,(n)}\otimes \tb^{m,(n)}_{i}]_{\bar\Scal_i'}$.
Similarly, for each column of $\sum_{m=1}^M\sum_{n=1}^N
 [\vb_{i}^{m,(n)}\otimes \tb^{m,(n)}_{i}]_{\bar\Scal_i'}$, there are at most $\{(1+\vartheta_2)\alpha^\star k\}^2$ nonzero entries. 
Define the subset of the support set 
\begin{align*}
    \Scal_{\cdot,i}'&=\{u:(u,i)\in \Scal'\};\\
    \Scal_{\cdot,(i-1)k+1:ik}&=\{(u,v):(u,v)\in \Scal,v=(i-1)k+1,\ldots,ik\}.
\end{align*}
Let 
\begin{align*}
\Ucal(\Scal')&=\{\ub\in R^{k^2}: \supp{(\ub)}\subset \{\Scal_{\cdot,i}',i=1,\ldots,(p-1)\},\norm{\ub}_2\leq 1\};\\
\Dcal(\Scal)&=\{\Db\in\RR^{k\times k}: \supp(\Db)\subset\{\Scal_{\cdot,(i-1)k+1:ik},i=1,\ldots,(p-1)\}, \norm{\Db}_F\leq 1\}.
\end{align*}
Therefore, we can express the above equation as
\begin{align}\label{eq:statB_norm}
    \bignorm{
        \frac{1}{N}\sum_{m=1}^M\sum_{n=1}^N
        [\vb_{i}^{m,(n)}\otimes \tb^{m,(n)}_{i}]_{\bar\Scal_i'}
    }_{2,\infty} 
    &= 
        \max_{j\in\barNi}
        \sup_{\ub\in\Ucal(\bar\Scal'_i)} 
        \abr{
        \frac{1}{N}
        \sum_{m=1}^M
        \sum_{n=1}^N
        \ub^\top[\vb_{i}^{m,(n)}\otimes \tb^{m,(n)}_{i}]_{\bar\Scal_i'}\eb_j}\notag\\
    &\leq
        \max_{j\in\barNi}
        \sup_{\ub\in\Ucal(\bar\Scal'_i)} 
        \abr{\frac{1}{N}
        \sum_{m=1}^M\sum_{n=1}^N
        \ub^\top
        \vb_{i}^{m,(n)}\otimes \tb^{m,(n)}_{i}
        \eb_j}
\end{align}
where $\eb_j\in\RR^{p-1}$ denotes the $j$-th standard canonical basis.

For matrices $\Eb$, $\Fb$, $\Gb$, $\Hb$, applying the fact 
$
\tr(\Eb^\top\Fb\Gb\Hb^\top)
=
    \VEC\rbr{\Eb}^\top(\Hb\otimes\Fb)\VEC\rbr{\Gb}
$ to the right hand side of~\eqref{eq:statB_norm} we can obtain
\begin{align*}
    \max_{j\in\barNi}
        \sup_{\ub\in\Ucal(\bar\Scal_i')} \frac{1}{N}&
        \sum_{m=1}^M\sum_{n=1}^N
        \ub^\top
        \vb_{i}^{m,(n)}\otimes \tb^{m,(n)}_{i}
        \eb_j\\
    &=
        \max_{j\in\barNi}
        \sup_{\Db\in\Dcal(\bar\Scal_i)} 
        \frac{1}{N}
        \sum_{m=1}^M\sum_{n=1}^N
        \tr\cbr{\Db^\top\tb^{m,(n)}_{i}\eb_j\vb_{i}^{m,(n)\top}}\\
    &=
        \max_{j\in\barNi}
        \sup_{\Db\in\Dcal(\bar\Scal_i)} 
        \frac{1}{N}\sum_{m=1}^M\sum_{n=1}^N
        (\Db\vb_{i}^{m,(n)})^\top\tb^{m,(n)}_{i}\eb_j.
\end{align*}
Define $\Ncal$ to be the $\frac{1}{2}$-net of the set $\Dcal(\bar\Scal_i)$, we have
\[
    \max_{j\in\barNi}
    \sup_{\Db\in\Dcal} \frac{1}{N}
    \sum_{m=1}^M\sum_{n=1}^N
    (\Db\vb_{i}^{m,(n)})^\top\tb^{m,(n)}_{i}\eb_j
    \leq 
        2 \max_{j\in\barNi}
        \max_{\Db\in\Ncal} \frac{1}{N}\sum_{m=1}^M\sum_{n=1}^N
        (\Db\vb_{i}^{m,(n)})^\top\tb^{m,(n)}_{i}\eb_j.
\]

We consider the case that
elements of $\{\{\ykmn\}_{m=1}^M\}_{n=1}^{N}$ are pairwise independent. This implies that $(\Db\vb_{i}^{m,(n)})^\top\tb^{m,(n)}_{i}\eb_j$ and $(\Db\vb_{i}^{m',(n')})^\top\yb^{m',(n')}_{i}\eb_j$ are pairwise independent for every $n\neq n'$. Given $n$ being fixed,  $(\Db\vb_{i}^{m,(n)})^\top\tb^{m,(n)}_{i}\eb_j$ and $(\Db\vb_{i}^{m',(n)})^\top\yb^{m',(n)}_{i}\eb_j$ might not be independent for $m\neq m'$. Observe that 
$(\Db\vb_{i}^{m,(n)})^\top\tb^{m,(n)}_{i}\eb_j$ are sub-exponentially distributed. Under this setting, we can bound the Orlicz norm $\norm{\cdot}_{\psi_1}$ of the random variable $M^{-1}\sum_{m=1}^M(\Db\vb_{i}^m)^\top\tb_{i}^m\eb_j$ as the following. By writing $\vb_{i}^m$ in terms of~\eqref{eq:def_v} and $\tb_i$ in terms of~\eqref{eq:deft}, we have 
\begin{align*}
    \bigg\|\frac{1}{M}\sum_{m=1}^M&(\Db\vb_{i}^m)^\top\tb_{i}^m\eb_j\bigg\|_{\psi_1}\\
    &\leq \max_{\mseq }\norm{\pAmk }_2
        \norm{\Sigmab^{m}}_2^{1/2}
    \big(
            \norm{\ptBik}_2
            \norm{\Sigmab^{m,\ub}}_2^{1/2}
            +
            \norm{\ptBik}_2
            \norm{\pAmk}
            \norm{\Sigmab^{m,\qb}}_2^{1/2}\\
            &\quad\quad\quad\quad\quad\quad\quad\quad\quad\quad\quad\quad
            +
            \norm{\Sigmab_i^{\rb}}_2^{1/2}
            +
            \norm{\Sigmab_i^{\wb}}_2^{1/2}
        \big)
        \\
        &\leq \max_{\mseq}\rho_x^{1/2}\rho_a^{1/2}
\{
\rho_b^{1/2}\norm{\Sigmab^{m,\ub}}_2^{1/2}
+
(\rho_b\rho_a)^{1/2}\norm{\Sigmab^{m,\qb}}_2^{1/2}
+
\norm{\Sigmab_i^{\rb}}_2^{1/2}
+
\norm{\Sigmab_i^{\wb}}_2^{1/2}
\},
\end{align*}
where we recall that $\norm{\pBik}\leq\rho_b^{1/2}$ for $\pseq$, $\norm{\pAmk}\leq \rho_a^{1/2}$ and $\norm{\Sigmab^m}\leq \rho_x$ for $\mseq$.

 \emph{Step 2}. 
 We define 
$$C_\zeta = (1+\vartheta_2)C_4\max_{\mseq}\rho_x^{1/2}\rho_a^{1/2}
\{
\rho_b^{1/2}\norm{\Sigmab^{m,\ub}}_2^{1/2}
+
(\rho_b\rho_a)^{1/2}\norm{\Sigmab^{m,\qb}}_2^{1/2}
+
\norm{\Sigmab_i^{\rb}}_2^{1/2}
+
\norm{\Sigmab_i^{\wb}}_2^{1/2}
\},$$ for some small constant $C_4>0$.
Then, under the condition of the sample size $N=O( C_\zeta^2\{(\alpha^\star k)^2+\log p\})$, apply the Bernstein's inequality with 
\[
t =
        C_\zeta
        \sqrt{\frac{(\alpha^\star k)^2+\log p}{N}},
\]
 and take the union bound over $j\in\barNi$ with $|\barNi|\leq(1+\vartheta_1)s^\star\leq p$ and $\Db\in\Ncal$ with $|\Ncal|\leq 5^{(1+\vartheta_2)^2(\alpha^\star)^2 k^2}$, we obtain 
\begin{align}\label{eq:staterrB:1}
    \PP\bigg\{\frac{1}{M}        \bignorm{
            \frac{1}{N}\sum_{m=1}^M\sum_{n=1}^N&
            [\vb_{i}^{m,(n)}\otimes \tb^{m,(n)}_{i}]_{\bar\Scal_i'}
        }_{2,\infty}>t\bigg\}\notag\\
    &\leq
    \PP\cbr{
        2 \max_{j\in\barNi}
        \max_{\Db\in\Ncal} \frac{1}{MN}\sum_{m=1}^M\sum_{n=1}^N
        (\Db\vb_{i}^{m,(n)})^\top\tb^{m,(n)}_{i}\eb_j>t
    }      \notag  \\
   &\leq \sum_{j\in\barNi}\sum_{\Db\in\Ncal}
    \PP\cbr{
        2 \frac{1}{MN}\sum_{m=1}^M\sum_{n=1}^N
        (\Db\vb_{i}^{m,(n)})^\top\tb^{m,(n)}_{i}\eb_j>t
    }\notag \\
    &\leq
   p^{-1}\exp\{-(\alpha^\star k)^2\}.
\end{align}


\end{proof}

\subsection{Lemma~\ref{lemma:I_m1}--~\ref{lemma:I_m4} and Their Proofs}
\label{ssec:edge_auxlemma2}

In this section we introduce Lemma~\ref{lemma:I_m1} --~\ref{lemma:I_m4} in order. The proof is stated following to the introduction of each lemma. We start with stating Lemma~\ref{lemma:I_m1}.

\begin{lemma}\label{lemma:I_m1}
Let $\Upsilon_1,\Upsilon_2,\Upsilon_3,\Upsilon_4,\Upsilon_6$ be constants defined in Section~\ref{ssec:notation}.
Under the conditions of Lemma~\ref{lemma:onestep_Am}, we have

\begin{align*}
    \Ical_{1,m}
    &\geq
        \rbr{
            \frac{3p(s^\star+1)\Upsilon_{1}}{4\Upsilon_{2}}
            -
            \frac{ps^\star\Upsilon_{4}^2\Upsilon_{3}^2}{\Upsilon_{6}}
        }\norm{\Amk -\pAmk }_F^2 
        - 
        \frac{\Upsilon_{6}}{4}\sum_{i=1}^p
        \norm{\Bb_{i}-\Bb_{i}^\star}_F^2\\
    &\quad - 
        \frac{2k\Upsilon_{2}}{p(s^\star+1)
        \Upsilon_{1}}\norm{\gradAstar}_{2}^2\\
    &\quad + 
        \frac{1}{p(s^\star+1)\Upsilon_{2}}
        \norm{
            \nabla_{\Amk }f(\setAk,\setpBk)
            -
            \nabla_{\Amk }f(\setpAk,\setpBk)
        }_2^2,
\end{align*}
with probability at least $\pkms$.
\end{lemma}

\begin{proof}[Proof of Lemma~\ref{lemma:I_m1}] 

Let $\Mscr_i^\star=\pNi\cup\{i\}$, where $\pNi$ is defined in~\eqref{eq:def_pNi}, with $\abr{\Mscr_i}\leq s^\star +1$ and let 
\begin{align}\label{eq:eventmi}
    \Ecal_{i}^m(\Mscr_i^\star) = \cbr{
        \frac{\nu_x}{2}
        \leq
        \sinmin{\sCovfree{\Mscr_i^\star}{\Mscr_i^\star} } 
        \leq 
        \sinmax{\sCovfree{\Mscr_i^\star}{\Mscr_i^\star} } 
        \leq 
        \frac{3\rho_x}{2}    
    }.
\end{align}
Under the conditions of the lemma, we have
\[
\PP\rbr{\bigcap_{i=1}^p\Ecal_{i}^m} \geq 1-M^{-1}\exp\{-k_m(s^\star+1)\},
\]
using Lemma~\ref{lem:covariance:concentration:neighborhoods}.
We will work on the event 
\begin{equation}\label{eq:eventm}
\Ecal^m = \bigcap_{i=1}^p\Ecal_{i}^m.
\end{equation}

By definition of $\Ical_{1,m}$  we can write
\begin{align*}
    \Ical_{1,m}
    &= 
        \dotp{\nabla_{\Amk }f(\setAk,\setBk)}{\Amk -\pAmk }\\
    &=
        \dotp{\nabla_{\Amk }f(\setAk,\setBk) -\nabla_{\Amk }f(\setAk,\setpBk)}{\Amk  - \pAmk }\\
    &\quad+ 
        \dotp{\nabla_{\Amk }f(\setAk,\setpBk) - \gradAstar}{\Amk  - \pAmk }\\
    &\quad + 
        \dotp{\gradAstar }{\Amk  - \pAmk }\\
    &= 
        \Ical_{11,m} + \Ical_{12,m} + \Ical_{13, m}.
\end{align*}
We proceed to lower-bound the terms separately. 

Using \eqref{eq:dpot_gradAm}, for $\Ical_{11,m}$ we have
\begin{align}
    \Ical_{11,m} 
    &= 
        \dotp{\nabla_{\Amk }f(\setAk,\setBk) - \nabla_{\Amk }f(\setAk,\setpBk)}{\Amk  - \pAmk }\notag\\
    & = 
        \sum_{i=1}^p\dotp{(\tBikT \tBik -\ptBikT \ptBik)(\Ib_p\otimes \Amk )\sCov}{\Ib_p\otimes (\Amk  - \pAmk )}\notag\\
    & = 
        \sum_{i=1}^p
        \dotp{\tBik -\ptBik }{
            \ptBik \bigAd{p}{\Amk-\pAmk}\sCov \bigAs{p}{\AmkT}
        }\notag\\
    &\quad +
        \sum_{i=1}^p
        \dotp{\tBik -\ptBik }{\tBik \bigAs{p}{\Amk}\sCov \bigAd{p}{\AmkT-\pAmkT}}.\label{eq:I11}
\end{align}
Let 
$
    \bar{\Scal}_{i,u\cdot}
    =
        \{v:(u,v)\in\bar{\Scal}_{i}\}
$, where $\bar{\Scal}_{i}$ is defined in~\eqref{eq:proj_supp}. By definition, for each $u=1,\ldots,k$, $\abr{\bar{\Scal}_{i,u\cdot}}\leq (1+\vartheta_1)s^\star(1+\vartheta_2)\alpha^\star k$. For the first term in~\eqref{eq:I11}, we apply the fact that the $i$-th block matrix $\Bb_{ii}=\Ib_k\in\RR^{k\times k}$ in $\tBik$ and $\Bb_{ii}^\star=\Ib_k\in\RR^{k\times k}$ in $\ptBik$ are equivalent. Hence the $i$-th block matrix $\tBik -\ptBik$ is a zero matrix. Therefore, we can write
\begin{align*}
    \langle&{\tBik -\ptBik }, {
        \ptBik\bigAd{p}{\Amk-\pAmk}\sCov
        \bigAs{p}{\AmkT}
    }\rangle\\
    &=
        \dotpl{-[\Bik -\pBik ]_{\bar\Scal_i}}{
            [\ptBik \bigAd{p}{\Amk -\pAmk}\sCovfree{\cdot}{\noi}
            \bigAs{p-1}{\AmkT}
        ]_{\bar\Scal_i}}\\
    &\leq
        \norml{\Bik -\pBik }_F\norml{
            [\ptBik \bigAd{p}{\Amk -\pAmk}\sCovfree{\cdot}{\noi}\bigAs{p-1}{\AmkT}
        ]_{\bar\Scal_i}}_F && (\text{by Cauchy-Schwarz}).
\end{align*}
Furthermore, we have
\begin{align*}
    &\norml{[\ptBik 
        \bigAd{p}{\Amk - \pAmk} 
        \sCovfree{\cdot}{\noi}
        \bigAs{p-1}{\AmkT}
    ]_{\bar\Scal_i}}_F^2\notag\\
    & =
        \sum_{(u,v)\in\bar\Scal_i}
        \left|(
        \ptBik
        \bigAd{p}{\Amk - \pAmk}
        \sCovfree{\cdot}{\noi}
        \bigAs{p-1}{\AmkT}
    )_{u,v}\right|^2. \notag\\
    &\leq
        \sum_{(u,v)\in\bar\Scal_i}
        \norml{
            \rbr{\ptBik
            \bigAd{p}{\Amk - \pAmk}
            \sCovfree{\cdot}{\noi}}_{u,\cdot}}_2^2
        \norml{
            \bigAs{p-1}{\AmkT}_{\cdot,v}
        }_2^2\notag  \\
    &\leq 
        \max_{v\in [k]}
        \norml{(\Amk )_{v,\cdot}}_2^2
        \sum_{u\in[k]}\sum_{v\in\bar\Scal_{i,u\cdot}}
        \norml{(
            \ptBik
            \bigAd{p}{\Amk - \pAmk}
            \sCovfree{\cdot}{\noi}
        )_{u,\cdot}}^2_2.\notag\\
    \intertext{Since $\Amk \in \Kcal_m(\tau_1,\tau_2)$, defined in~\eqref{eq:constraint_a}, by the projection operator $\Pcal_{m,\tau}$ in Algorithm~\ref{alg:update}, we have
    $\max_{v\in [k]}\norm{(\Amk )_{v,\cdot}}_2^2\leq
    \frac{\tau_2}{k_m}
        \norm{\Ab^{m(0)}}^2_2$.
    Furthermore, using the fact that $|\bar\Scal_{i,u\cdot}|\leq(1+\vartheta_1)(1+\vartheta_2)s^\star\alpha^\star k$, the above display can be bounded as}
    &\leq 
        \frac{\tau_2}{k_m}
        \norml{\Ab^{m(0)}}^2_2
        (1+\vartheta_1)(1+\vartheta_2)s^\star\alpha^\star k 
        \norml{
            \ptBik
            \bigAd{p}{\Amk - \pAmk}
            \sCovfree{\cdot}{\noi}
        }_F^2\notag\\    
    &\leq 
    (1+\vartheta_1)(1+\vartheta_2)s^\star\alpha^\star {\tau_2} 
        \norml{\Ab^{m(0)}}_{2}^2
        \norml{\ptBik}_F^2
        \norml{ \Amk  - \pAmk }_F^2
        \norml{\sCovfree{\cdot}{\noi}}_2^2
        && (k_m\geq k) \notag\\
    &\leq 
    (1+\vartheta_1)(1+\vartheta_2)s^\star\alpha^\star {\tau_2} 
        \norml{\Ab^{m(0)}}_{2}^2
        \norml{\ptBik}_F^2
        \norml{ \Amk  - \pAmk }_F^2
        \norml{\sCov}_2^2        
        && (\norm{\sCovfree{\cdot}{\backslash\{i\}}}_2
        \leq 
        \norm{\sCov}_2).
\end{align*}
Combining the last two results, we can bound the first term in~\eqref{eq:I11} as
\begin{align}\label{eq:I11_1f}
    \norm{[
        \ptBik
        &\bigAd{p}{\Amk - \pAmk}
        \sCovfree{\cdot}{\noi}\bigAs{p-1}{\AmkT}
    ]_{\bar\Scal_i}}_F\norml{\Bik-\pBik}_F\notag\\
    &\leq 
        \sqrt{
            (1+\vartheta_1)(1+\vartheta_2)s^\star\alpha^\star\tau_2 
        }
        \norml{\sCov}_2    \norml{\Ab^{m(0)}}_2\norml{\ptBik}_F
        \norm{\Amk -\pAmk}_F\norm{\Bik-\pBik}_F.
\end{align}

For the second term in~\eqref{eq:I11}, we have 
\begin{align*}
& \left\langle
{\tBik -\ptBik},
{\tilde\Bb_{i} \bigAs{p}{\Amk} \sCov \bigAd{p}{\AmkT - \pAmkT}}\right\rangle \\
&=
-\dotpl{[\Bik -\pBik]_{\bar{\Scal}_i}}{\tilde\Bb_{i}\bigAs{p}{\Amk}\sCovfree{\cdot}{\noi} \bigAd{p-1}{\AmkT - \pAmkT}}  &&  (\tilde\Bb_{ii}=\tilde\Bb_{ii}^\star=\Ib_k)     \\
&\leq 
        \norml{\Bik -\pBik }_F
        \norml{[
            \tilde\Bb_{i}\bigAs{p}{\Amk}\sCovfree{\cdot}{\noi}
            \bigAd{p-1}{\AmkT - \pAmkT}
        ]_{\bar{\Scal}_i}}_F
         && (\text{by Cauchy-Schwarz}).
\end{align*}
Let $\bar{\Scal}_{i,\cdot v}=\{u:(u,v)\in\bar{\Scal}_{i}\}$.
Note that if $(u,v)\in\bar\Scal_i$, then $\lceil{v/k}\rceil\in\barNi$. Let the index set $\noi=\{1,\ldots,p\}\backslash\{i\}$, we have
\begin{align}
    &\norml{[\tBik \bigAs{p}{\Amk} \sCovfree{\cdot}{\noi} 
    \bigAd{p-1}{\AmkT - \pAmkT}
    ]_{\bar{\Scal}_i}}_F^2\notag\\
    & = 
        \sum_{(u,v)\in\bar{\Scal}_i} 
        \abr{
            (\tBik \bigAs{p}{\Amk} \sCovfree{\cdot}{\barNi}
            \bigAd{|\barNi|}{\AmkT - \pAmkT}
        )_{u,v}}^2\notag\\
    & \leq 
        \sum_{(u,v)\in\bar{\Scal}_i}
        \norm{(\tBik \bigAs{p}{\Amk})_{u\cdot}}_2^2
        \norm{(\sCovfree{\cdot}{\barNi}
            \bigAd{|\barNi|}{\AmkT - \pAmkT})_{\cdot,v}}_2^2\notag\\
    & \leq 
        \norm{\bigAs{p}{\Amk^\top}\tBikT }_{2,\infty}^2
        \sum_{v\in[k(p-1)]}\sum_{u\in\bar{\Scal}_{i,\cdot v}} 
        \norm{(\sCovfree{\cdot}{\barNi}
            \bigAd{|\barNi|}{\AmkT - \pAmkT})_{\cdot,v}}_2^2\notag\\
    & \leq 
        (1+\vartheta_2)\alpha^\star k \norm{(\bigAs{p}{\AmkT}\tBikT }_{2,\infty}^2
        \norm{
            \sCovfree{\cdot}{\barNi}
            \bigAd{|\barNi|}{\AmkT - \pAmkT}
        }_F^2\notag.
    \intertext{By Lemma~F.9 in~\citet{na2021estimating}, for any two matrices $\Xb$, $\Yb$, we have $\norm{\Xb\Yb}_{2,\infty}\leq \norm{\Xb}_{2,\infty}\norm{\Yb}_1$, where $\norm{\Yb}_1$ denotes the maximum $\ell_1$-norm of the column of $\Yb$. Then, the above display is abounded as}
    &\leq
        (1+\vartheta_2)\alpha^\star k \norm{\bigAs{p}{\AmkT}}_{2,\infty}^2
        \norm{\tBikT }_{1}^2
        \norm{
            \sCovfree{\cdot}{\barNi}
            \bigAd{|\barNi|}{\AmkT -\pAmkT}
        }_F^2\notag\\
    &=
        (1+\vartheta_2)\alpha^\star k \norm{\AmkT}_{2,\infty}^2
        \norm{\tBikT }_{1}^2
        \norm{
            \sCovfree{\cdot}{\barNi}
            \bigAd{|\barNi|}{\AmkT -\pAmkT}
        }_F^2.\notag\\
        \intertext{Apply the condition of $\Kcal_m(\tau_1,\tau_2)$, we can bound the above display as}
    &\leq 
    (1+\vartheta_2)\alpha^\star\tau_2
        \norm{\Ab^{m(0)}}_2^2\norm{\tBikT }_{1}^2
        \norm{
            \sCovfree{\cdot}{\barNi}
            \bigAd{|\barNi|}{\AmkT -\pAmkT}
        }_F^2\notag.\\
    \intertext{Since $|\barNi|\leq(1+\vartheta_1)s^\star$, and hence
$\norm{\bigAd{\abr{\barNi}}{\AmkT - \pAmkT}}_F\leq \sqrt{(1+\vartheta_1)s^\star}\norm{\AmkT - \pAmkT}_F$. Along with Cauchy interlacing theorem, we can bound the above display as}
    &\leq 
        (1+\vartheta_1)(1+\vartheta_2)s^\star\alpha^\star\tau_2
        \norm{\Ab^{m(0)}}_2^2
        \norm{\tBikT}_1^2
        \norm{\Amk -\pAmk }_F^2\norm{\sCov}_2^2.\notag\\
    \intertext{Finally, under assumption of Lemma~\ref{lemma:onestep_Am}, we have  
    $\norm{\tBikT-\ptBikT}_1=\norm{\BikT-\pBikT}_1\leq C_B\norm{\pBikT}_1\leq C_B\norm{\ptBikT}_1$. Therefore, apply triangle inequality, we can further obtain the bound}
    &\leq 
        (1+\vartheta_1)(1+\vartheta_2)(1+C_B)^2s^\star\alpha^\star\tau_2
        \norm{\Ab^{m(0)}}_2^2
        \norm{\ptBikT}_1^2
        \norm{\Amk -\pAmk }_F^2\norm{\sCov}_2^2,
    \label{eq:I11_2f}
\end{align}

Taking the results of~\eqref{eq:I11_1f}--\eqref{eq:I11_2f} and plugging back to~\eqref{eq:I11}
, we can bound $\Ical_{11,m}$ with Young's inequality and obtains
\begin{align*}
    \Ical_{11,m}
    &\geq 
        -\Upsilon_{4}\Upsilon_{3}\sqrt{s^\star}
        \sum_{i=1}^p\norm{\Amk -\pAmk }_F\norm{\Bik -\pBik }_F\\
    &\geq 
        -\Upsilon_{4}\Upsilon_{3}\sqrt{s^\star }
        \sum_{i=1}^p
            \rbr{
                \frac{1}{2C_{12}}\norm{\Bik -\pBik }_F^2
            +
                \frac{C_{12}}{2}\norm{\Amk -\pAmk }_F^2
            },
\end{align*}
where we have
\begin{align*}
    \Upsilon_{3} &= \norm{\sCov}_2(C_B+1)\max_{\mseq}\norm{\Ab^{m(0)}}_2\max_{\pseq}(\norm{\pBik }_F\vee \norm{\ptBikT}_1);\\
    \Upsilon_{4}&=2\sqrt{(1+\vartheta_1)(1+\vartheta_2)\tau_2\alpha^\star},
\end{align*}
 and $C_{12}$ is a constant defined later.


Under the event of ${\Ecal^m}$ defined in~\eqref{eq:eventm},
by Lemma~\ref{lemma:scsm_coef}, we know that $f(\setAk,\setpBk)$ is $\mu$-strongly convex and $L$-smooth with respect to $\Amk $ with probability at least $\pkms$, where
\[
\mu = \frac{ p(s^\star+1)\nu_{x}\nu_b }{4},\quad L=\frac{3p (s^\star+1) \rho_{x}\rho_p}{4}.
\]
Define 
\begin{align*}
    \Upsilon_{1} &= \frac{3\nu_x\nu_b\rho_x\rho_b}{16}, \quad\Upsilon_{2}= \frac{\nu_x\nu_b+3\rho_x\rho_b}{4}.
\end{align*}
Therefore, by Lemma~\ref{lemma:convexbound}, we have
\begin{align*}
   \Ical_{12,m} 
   &= 
        \dotp{\nabla_{\Amk }f(\setAk,\setpBk) -\nabla_{\Amk }f(\setpAk,\setpBk)}{\Amk  - \pAmk }\\
   &\geq
        \frac{p(s^\star+1)\Upsilon_{1}}{\Upsilon_{2}}
        \norm{\Amk  - \pAmk }_F^2\\
   &\quad+
        \frac{1}{p (s^\star+1)\Upsilon_{2}}
        \norm{
            \nabla_{\Amk }f(\setAk,\setpBk)
            -
                \nabla_{\Amk }f(\setpAk,\setpBk)
        }_F^2,
\end{align*}
with probability at least $\pkms$, the probability that ${\Ecal^m}$ happens.

Define $\norm{\cdot}_*$ to be the nuclear norm. In the last step, we apply H\"older's inequality and obtain,
\begin{align*}
    \Ical_{13, m}
    &=
        \dotp{\gradAstar}{\Amk  - \pAmk } \\
    &\geq - 
        \norm{\gradAstar }_{2}
        \norm{\Amk  - \pAmk }_* \\
    &\geq - 
        \sqrt{2k} 
        \norm{\gradAstar}_{2}
        \norm{\Amk  - \pAmk }_F\\
    &\geq - 
        \sqrt{2k} 
        \cbr{
            \frac{1}{2C_{13}}
            \norm{\gradAstar}_{2}^2 
            + 
                \frac{C_{13}}{2}
                \norm{\Amk  - \pAmk }_F^2
        },
\end{align*}
and the second last step follows by the fact $\|\Zb\|_*\leq\sqrt{r}\|\Zb\|_F$, where $r$ is the rank of $\Zb$. Since $k\leq k_m$ for all $\mseq$ and hence the rank of $\Amk -\pAmk $ is at most $2k$. $C_{13}>0$ will be defined below.

Setting $C_{12}$ and $C_{13}$ as 
\[
C_{12}
=
    \frac{2\Upsilon_{4}\Upsilon_{3}\sqrt{s^\star}}{\Upsilon_{6}}
, \quad 
C_{13}
= 
    \frac{1}{2\sqrt{2k}}
    \frac{p (s^\star+1)\Upsilon_{1}}{\Upsilon_{2}},
\]
where $\Upsilon_{6} = \frac{\nu_x\nu_a}{2}$.

Combining results of $\Ical_{11,m}$, $\Ical_{12,m}$, and $\Ical_{13, m}$, we can obtain that
\begin{align*}
    \Ical_{1,m}
    &\geq
        \rbr{
            \frac{3p (s^\star+1)\Upsilon_{1}}{4\Upsilon_{2}}
            -
            \frac{ps^\star\Upsilon_{4}^2\Upsilon_{3}^2}{\Upsilon_{6}}
        }\norm{\Amk -\pAmk }_F^2 
        - 
            \frac{\Upsilon_{6}}{4}\sum_{i=1}^p
            \norm{\Bb_{i}-\Bb_{i}^\star}_F^2\\
    &\quad - 
        \frac{2k\Upsilon_{2}}{p (s^\star+1)\Upsilon_{1}}
        \norm{\gradAstar}_{2}^2\\
    &\quad +  
        \frac{1}{p (s^\star+1)\Upsilon_{2}}
        \norm{
                \nabla_{\Amk }f(\setAk,\setpBk)
            -
                \nabla_{\Amk }f(\setpAk,\setpBk)
        }_F^2,
\end{align*}
with probability at least $1-M^{-1}\exp\{-k_m (s^\star+1)\}$.
\end{proof}

\begin{lemma}\label{lemma:I_m2}
Let $\Upsilon_5$ be a constant defined in Section~\ref{ssec:notation}. Under conditions of Lemma~\ref{lemma:onestep_Am}, we have
\begin{align*}
    \Ical_{2,m}&\leq  3\bigg\{ \norm{\nabla_{\Amk }f(\setAk,\setpBk) -\nabla_{\Amk }f(\setpAk,\setpBk)}_F^2\\
   &\quad  + p(s^\star+1)\Upsilon_{5}\sum_{i=1}^p\norm{\Bb_{i} - \Bb_{i}^\star}_F^2 + k\norm{\gradAstar}_2^2\bigg\}.
\end{align*}
\end{lemma}

\begin{proof}[Proof of Lemma~\ref{lemma:I_m2}]
By definition of $\Ical_{2,m}$, we can write
\begin{align*}
    \Ical_{2,m} 
    &= 
        \norm{\nabla_{\Amk }f(\setAk,\setBk)}_F^2\\
    &\leq 
        3\bigg\{ 
        \norm{
            \nabla_{\Amk }f(\setAk,\setBk) 
            - 
            \nabla_{\Amk }f(\setAmk\cup \setpAnmk,\setpBk) 
        }_F^2\\
   &\quad \quad  + 
        \norm{\nabla_{\Amk }f(\setAmk\cup \setpAnmk,\setpBk) - \gradAstar}_F^2\\
   &\quad \quad + 
        \norm{\gradAstar}_F^2\bigg\}.
\end{align*}

By~\eqref{eq:gradAm} and triangle inequality, we can first decompose $\Ical_{21,m}$ as
\begin{align*}
    \Ical_{21,m} 
    &=
        \norm{
            \nabla_{\Amk }f(\setAk,\setBk) 
            - 
            \nabla_{\Amk }f(\setAmk\cup \setpAnmk,\setpBk) 
        }_F \\
    &=
        \norml{
            \sum_{i=1}^p\sum_{j=1}^p
            (\eb_j^\top\otimes\Ib_k)
            \{
                \tBikT (\tBik -\ptBik )
                +
                (\tBikT -\ptBikT )\ptBik
            \}
            \bigAs{p}{\Amk}\sCov
            (\eb_j\otimes\Ib_{k_m})
        }_F\notag.\\
        \intertext{Since $ (\eb_j^\top\otimes\Ib_k)\tBikT={\bf 0}$ and $(\eb_j^\top\otimes\Ib_k)(\tBikT -\ptBikT )={\bf 0}$ for $j\not\in\barNi\backslash\{i\}$, we can write the above display as}
    &=
        \norml{
            \sum_{i=1}^p\sum_{j\in\barNi\cup\{i\}}
            (\eb_j^\top\otimes\Ib_k)
            \{
                \tBikT (\tBik -\ptBik )
                + 
                (\tBikT -\ptBikT )\ptBik
            \}
            \bigAs{p}{\Amk}\sCov
            (\eb_j\otimes\Ib_{k_m})
        }_F\notag\\
    &\leq 
        \sqrt{(1+\vartheta_1)(s^\star+1)}
        \sum_{i=1}^p
        \norm{
            \{
                \tBikT (\tBik -\ptBik )
                +
                (\tBikT -\ptBikT )\ptBik
            \}
            (\Ib_p\otimes \Amk )\sCov
        }_F,
\end{align*}
where the last inequality follows by $|\barNi\cup\{i\}|\leq(1+\vartheta_1)(s^\star+1)$.
Now, consider single $i\in\{1,\ldots,p\}$ and apply triangle inequality once more, we have
\begin{align*}
    \Ical_{211,m,i}
    &=
        \norm{
            \{
                \tBikT (\tBik -\ptBik )
                + 
                (\tBikT -\ptBikT )\ptBik
            \}
            (\Ib_p\otimes \Amk )\sCov
        }_F\\
    &\leq 
        (
            \norm{\tBik }_2
            +
            \norm{\ptBik}_2
        )
        \norm{\bigAs{p}{\Amk}}_2
        \norm{\sCov}_2
        \norm{\tBik  - \ptBik}_F.\\
    \intertext{Using the fact that for a pair of matrices $\Xb$, $\Yb$, $\norm{\Xb\otimes\Yb}_2 = \norm{\Xb}_2\norm{\Yb}_2$, the above display is equivalent as}
    &=
        (
            \norm{\tBik }_2
            +
            \norm{\ptBik}_2)
        \norm{ \Amk }_2
        \norm{\sCov}_2
        \norm{\tBik  - \ptBik}_F\\
    &\leq 
        (1+C_A)\norm{\pAmk }_2
        (
            \norm{\tBik }_2
            +
            \norm{\ptBik}_2
        )
        \norm{\sCov}_2
        \norm{\tBik  - \ptBik}_F,
\end{align*}
where the last step follows by the assumption stated in Lemma~\ref{lemma:onestep_Am}: $\norm{\Amk-\pAmk}_2\leq C_A\norm{\pAmk}_2$. 

Note that, by triangle inequality and the definition of spectral norm , we can write 
\begin{align}
    \norm{\tBik }_2 
    \leq 
        \norm{\tBik  - \ptBik }_2
        +
        \norm{\ptBik}_2 
    & = 
        \norm{\Bik  - {\Bb}_i^\star}_2
        +
        \norm{\tBik }_2\notag\\
    &\leq 
        C_B\norm{\pBik }_2
        +
        \norm{\ptBik}_2
    \leq 
        (C_B+1)\norm{\ptBik}_2\label{eq:tildeB}.
\end{align}

Plugging~\eqref{eq:tildeB} back to $\Ical_{211,m,i}$ yields
\begin{align}
    \Ical_{211,m,i}
    &\leq 
        (1+C_A)(2+C_B)
        \norm{\sCov}_2
        \norm{\ptBik }_2
        \norm{\pAmk }_2
        \norm{{\Bb}_i-{\Bb}_i^\star}_F,\label{eq:I111}
\end{align}
where it follows by the fact that $\norm{\tBik -\ptBik }_F = \norm{{\Bb}_i-{\Bb}_i^\star}_F$.

Then, sum up $\Ical_{211,m,i}$ for $\pseq$, we have
\begin{align*}
    \Ical_{21,m}
    &\leq 
        \Upsilon_{5}^{1/2}
        \sqrt{s^\star+1}
        \sum_{i=1}^p
        \norm{\Bik -\pBik }_F,
\end{align*}
where $\Upsilon_{5}=(1+C_A)^2(2+C_B)^2\norm{\sCov}_2^2\rho_b\rho_a$.

Then, applying Cauchy-Schwarz inequality yields
\[
\Ical_{21,m}^2
\leq 
    p(s^\star+1)\Upsilon_{5}
    \sum_{i=1}^p
    \norm{\Bik  - \pBik }_F^2.
\]
Therefore, we can conclude that
\begin{align*}
    \Ical_{2,m}
    &\leq  
        3\bigg\{ 
            \norm{
                \nabla_{\Amk }f(\setAk,\setpBk) 
                -
                \nabla_{\Amk }f(\setpAk,\setpBk)
            }_F^2\\
            &\quad  +  
                \Upsilon_{5}p(s^\star+1)
                \sum_{i=1}^p
                \norm{\Bb_{i} - \Bb_{i}^\star}_F^2 
            + 
                k\norm{\gradAstar}_2^2
        \bigg\}.
\end{align*}

\end{proof}

\begin{lemma}\label{lemma:I_m3}
The constants $\Upsilon_3,\Upsilon_4,\Upsilon_6$ are defined in Section~\ref{ssec:notation}.
Under the conditions of Lemma~\ref{lemma:onestep_Bij}, we have

\begin{multline*}
    \Ical_{3,i}
    \geq
        \frac{M\Upsilon_{6}}{2}
        \norm{\Bb_{i}-\Bb_{i}^\star}_F^2
        -
        \frac{s^\star\Upsilon_{4}^2\Upsilon_{3}^2}{\Upsilon_{6}}
        \sum_{m=1}^M\norm{\Amk  - \pAmk }_F^2\\
    -
        \frac{(1+\vartheta_1)s^\star}{M\Upsilon_{6}}\norm{\gradBstar }_{r(k,\infty)}^2.
\end{multline*}
with probability at least $\ppkms$.
\end{lemma}

\begin{proof}[Proof of Lemma~\ref{lemma:I_m3}]
By definition of $\Ical_{3,i}$, we can write
\begin{align*}
    \Ical_{3,i} 
    &= 
        \dotp{[\nabla_{\Bb_{i}} f(\setAk,\setBk)]_{\bar\Scal_i}}{\Bb_{i}-\Bb_{i}^\star} \\
    & = 
        \dotp{
            [\nabla_{\Bb_{i}}
            f(\setAk,\setBk) 
            - 
            \nabla_{\Bb_{i}}
            f(\setpAk,\setBk)]_{\bar\Scal_i}}{
            \Bb_{i}-\Bb_{i}^\star
        } \\
    & + 
        \dotp{
            [\nabla_{\Bb_{i}} f(\setpAk,\setBk)
            - 
            \nabla_{\Bb_{i}} f(\setpAk,\setpBk)]_{\bar\Scal_i} }{
            \Bb_{i}-\Bb_{i}^\star
        } \\
    & + 
        \dotp{
            [\gradBstar]_{\bar\Scal_i} }{
            \Bb_{i}-\Bb_{i}^\star
        }\\
    & = 
        \Ical_{31,i} + \Ical_{32,i} + \Ical_{33,i}.
\end{align*}

First, we express $\Ical_{31,i}$ as~\eqref{eq:dpot_gradBij} and obtain
\begin{align}
    \Ical_{31,i} 
    &= 
    \dotp{[
        \nabla_{\Bb_{i}} f(\setAk,\setBk) 
        -
        \nabla_{\Bb_{i}} f(\setpAk,\setBk)
        ]_{\bar\Scal_i}}{
        \Bb_{i}-\Bb_{i}^\star
    }\notag\\
    &= 
        -\sum_{m=1}^M
        \dotp{[
            \tBik  
            \bigAs{p}{\Amk}
            \sCovfree{\cdot}{\noi} 
            \bigAs{p-1}{\AmkT}
        ]_{\bar\Scal_i}}{
            \Bik -\pBik 
        }\notag\\
    &\quad 
        +\sum_{m=1}^M
        \dotp{[
            \tBik 
            \bigAs{p}{\pAmk}
            \sCovfree{\cdot}{\noi} 
            \bigAs{p-1}{\pAmkT}
        ]_{\bar\Scal_i} }{
            \Bik -\pBik 
        }\notag\\
    &= 
        -\sum_{m=1}^M
        \dotp{[
            \tBik  
            \bigAd{p}{\Amk - \pAmk}
            \sCovfree{\cdot}{\noi} 
            \bigAs{p-1}{\pAmkT}
        ]_{\bar\Scal_i}}{
            \Bik -\pBik 
        }\notag\\
    &\quad 
        -\sum_{m=1}^M
        \dotp{[
            \tBik  
            \bigAs{p}{\Amk} 
            \sCovfree{\cdot}{\noi} 
            \bigAd{p-1}{\AmkT - \pAmkT}
        ]_{\bar\Scal_i}}{
            \Bik -\pBik 
        }.\notag\\
        \intertext{Apply Cauchy-Schwarz inequality, we can lower bound the above display as}
        &\geq 
        -\sum_{m=1}^M
        \norm{[
            \tBik  
            \bigAd{p}{\Amk - \pAmk}
            \sCovfree{\cdot}{\noi} 
            \bigAs{p-1}{\pAmkT}
        ]_{\bar\Scal_i}}_F\norm{
            \Bik -\pBik 
        }_F\notag\\
       &\quad 
        -\sum_{m=1}^M
        \norm{[
            \tBik  
            \bigAs{p}{\Amk} 
            \sCovfree{\cdot}{\noi} 
            \bigAd{p-1}{\AmkT - \pAmkT}
        ]_{\bar\Scal_i}}_F
        \norm{
            \Bik -\pBik 
        }_F     
        \label{eq:l31m_lb}.
\end{align}
We bound the above two terms individually. Following the same proof steps as in~\eqref{eq:I11_1f}, for $\mseq$, we can obtain 
\begin{align}
    \norm{[
        \tBik  
        &\bigAd{p}{\Amk - \pAmk}
        \sCovfree{\cdot}{\noi} 
        \bigAs{p-1}{\pAmkT}
    ]_{\bar\Scal_i}}_F
    \norm{\Bik -\pBik }_F\notag\\
    &\leq 
        \sqrt{(1+\vartheta_1)(1+\vartheta_2)s^\star\alpha^\star\tau_2 }
        \norm{\Ab^{m(0)}}_2
        \norm{\tBik }_F
        \norm{\sCov}_2
        \norm{\Amk -\pAmk }_F
        \norm{\Bik -\pBik }_F\notag\\
    &\leq\sqrt{(1+\vartheta_1)(1+\vartheta_2)s^\star\alpha^\star\tau_2 }(C_B+1)
        \norm{\Ab^{m(0)}}_2
        \norm{\ptBik }_F
        \norm{\sCov}_2
        \norm{\Amk -\pAmk }_F
        \norm{\Bik -\pBik }_F.\label{eq:I31_1f}
\end{align}

We can bound the second term in~\eqref{eq:l31m_lb} similar in the way shown in~\eqref{eq:I11_2f} and obtain
\begin{align}\label{eq:I31_2f}
    \norm{[
        \tBik
        &\bigAs{p}{\Amk}
        \sCovfree{\cdot}{\noi}
        \bigAd{p-1}{\AmkT - \pAmkT}
    ]_{\bar\Scal_i}}_F
    \norm{\Bik -\pBik }_F\notag\\
    &\leq 
        \sqrt{(1+\vartheta_1)(1+\vartheta_2)s^\star\alpha^\star\tau_2}
        \norm{\sCov}_2
        \norm{\Ab^{m(0)}}_2
        \norm{\tBik }_1
        \norm{\Amk -\pAmk }_F
        \norm{\Bik -\pBik }_F\notag\\
    &\leq 
        \sqrt{(1+\vartheta_1)(1+\vartheta_2)s^\star\alpha^\star\tau_2}(1+C_B)
        \norm{\sCov}_2
        \norm{\Ab^{m(0)}}_2
        \norm{\ptBik }_1
        \norm{\Amk -\pAmk }_F
        \norm{\Bik -\pBik }_F.
\end{align}

Combining results from~\eqref{eq:I31_1f}--\eqref{eq:I31_2f}, we can apply Young's inequality and lower bound $\Ical_{31,i}$ as
\begin{align*}
    \Ical_{31,i}
    &\geq - 
        \sum_{m=1}^M
        \Upsilon_{4}\Upsilon_{3}\sqrt{s^\star}\norm{\Amk  - \pAmk }_F\norm{\Bik -\pBik }_F\\
    &\geq -
        \Upsilon_{4}\Upsilon_{3}\sqrt{s^\star}
        \sum_{m=1}^M
        \rbr{\frac{1}{2C_{31}}
        \norm{\Amk  - \pAmk }_F^2 
        + 
        \frac{C_{31}}{2}
        \norm{\Bik  - \pBik }_F^2},
\end{align*}
where $\Upsilon_3$, and $\Upsilon_4$ are defined in Section~\ref{ssec:notation} and $C_{31}$ is defined below.

We define the event $\bar\Ecal_i$:
\begin{equation}\label{eq:even2}
    \bar\Ecal_i 
    := 
        \cap_{m=1}^M{\Ecal_{i}^m}(\barNi), 
\end{equation}
where ${\Ecal_{i}^m}(\cdot)$ for $\mseq$ are defined in~\eqref{eq:eventmi}. Then, taking the union bound, if $N=O(\max_{\mseq}\kappa_m^2 (k_m s^\star+\log M+\log p))$, we see that event $\bar\Ecal_i$ will happen with probability at least $$1-(pM)^{-1}\sum_{m=1}^M\exp\{-k_m s^\star\}\geq \ppkms.$$

To find the lower bound of $\Ical_{32,i}$, we use formula of $\nabla_{\Bik} f$ stated in~\eqref{eq:dpot_gradBij}. Combining with the fact that 
\[
    -(\tBik -\ptBik )\bigAs{p}{\pAmk}\Ykmn 
    = 
        (\Bik -\pBik )
        \bigAs{p-1}{\pAmk}\Ykmnfree{\noi}
\]
for $\mseq$, we can write
\begin{align*}
    \Ical_{32,i} 
    &= 
        \dotp{
            [\nabla_{\Bik} f(\setpAk,\setBk) 
            - 
            \nabla_{\Bik} f(\setpAk,\setpBk)]_{\bar\Scal_i}
        }{
            \Bik-\pBik}\\
    &=
        \sum_{m=1}^M 
        \dotp{
            -(\tBik -\ptBik)
            \bigAs{p}{\pAmk}
            \sCovfree{\cdot}{\noi}
            \bigAs{p-1}{\pAmkT}
        }{
            [\Bik-\pBik]_{\bar\Scal_i}} \\
    &= 
        \sum_{m=1}^M
        \dotp{
            [\Bik - \pBik]_{\bar\Scal_i}
            \bigAs{p-1}{\pAmk}
            \sCovfree{\cdot}{\noi}
            \bigAs{p-1}{\pAmkT}
        }{
            [\Bik-\pBik]_{\bar\Scal_i}}\\
     &= 
        \frac{1}{N}\sum_{m=1}^M
        \norm{
            [\Bik -\pBik ]_{\bar\Scal_i}
            \bigAs{p-1}{\pAmk}
            \Ykmnfree{\noi}
        }_F^2.\\
\end{align*}
Therefore, we can further lower bound $\Ical_{32,i}$ as 
\begin{align*}
     \Ical_{32,i}
     &\geq
        \sum_{m=1}^M
        \sinmintwo{\pAmk }
        \sinmin{\sCovfree{\barNi}{\barNi}}
        \norm{\Bik -\pBik }_F^2.\\
     &\geq 
        \frac{\nu_{x}}{2}
        \sum_{m=1}^M
        \sinmintwo{\pAmk }
        \norm{\Bik  - \pBik }_F^2
    =
        M\Upsilon_{6}\norm{\Bik -\pBik }_F^2,
\end{align*}
with probability at least $\ppkms$. The last inequality is followed by conditioning on the event $\bar\Ecal_i$ in~\eqref{eq:even2}.

Finally, applying Lemma~\ref{lemma:dualnormrpq} yields,
\begin{align*}
    \Ical_{33,i} 
    &= 
        \dotp{[\gradBstar]_{\bar\Scal_i} }{\Bik - \pBik}\\
    &\geq -
        \norm{[\gradBstar]_{\bar\Scal_i}}_{r(k,\infty)}
        \norm{\Bik - \pBik}_{r(k,1)}\\
    &\geq -
        \sqrt{(1+\vartheta_1)s^\star}
        \norm{[\gradBstar]_{\bar\Scal_i}}_{r(k,\infty)}
        \norm{\Bik - \pBik}_{r(k,2)}\\
    &= -
        \sqrt{(1+\vartheta_1)s^\star}
        \norm{[\gradBstar]_{\bar\Scal_i}}_{r(k,\infty)}
        \norm{\Bik - \pBik}_F\\
    & \geq - 
        \sqrt{(1+\vartheta_1)s^\star}
        \cbr{
            \frac{1}{2C_{33}}
            \norm{[\gradBstar]_{\bar\Scal_i}}_{r(k,\infty)}^2 
            + 
            \frac{C_{33}}{2}
            \norm{\Bb_{i}-\Bb_{i}^\star}_{F}^2
        }.
\end{align*}

Setting $C_{31}$, $C_{33}$ as 
\[
    C_{31} 
    = 
        \frac{\Upsilon_{6}}{2\Upsilon_{3}\Upsilon_{4}\sqrt{s^\star}}
    , \quad 
    C_{33}
    = 
        \frac{M\Upsilon_{6}}{2\sqrt{(1+\vartheta_1)s^\star}},
\]
Combining $\Ical_{31,i}$, $\Ical_{32,i}$, $\Ical_{33,i}$, $\Ical_{33,i}$ obtains
\begin{multline*}
    \Ical_{3,i}
    \geq
        \frac{M\Upsilon_{6}}{2}
        \norm{\Bik - \pBik}_F^2
        -
        \frac{s^\star\Upsilon_{4}^2\Upsilon_{3}^2}{\Upsilon_{6}}
        \sum_{m=1}^M
        \norm{\Amk  - \pAmk }_F^2\\
    -
        \frac{(1+\vartheta_1)s^\star}{M\Upsilon_{6}}
        \norm{[\gradBstar]_{\bar\Scal_i}}_{r(k,\infty)}^2.
\end{multline*}
with probability at least $\ppkms$.

\end{proof}

\begin{lemma}\label{lemma:I_m4}
Under the conditions of Lemma~\ref{lemma:onestep_Bij}, we have
\begin{multline*}
    \Ical_{4,i}
    \leq 
        3\bigg\{ 
        \Upsilon_{5} Ms^\star
        \sum_{m=1}^M\norm{\Amk  - \pAmk }_F^2 
        +
        \Upsilon_{7}M^2
        \norm{\Bb_{i} - \Bb_{i}^\star}_F^2\\
        +
        (1+\vartheta_1)s^\star
        \norm{[\gradBstar]_{\bar\Scal_i}}_{r(k,\infty)}^2
    \bigg\}.
\end{multline*}
\end{lemma}
\begin{proof}[Proof of Lemma~\ref{lemma:I_m4}]
We want to find the upper bound of $\Ical_{4,i}$. By definition of $\Ical_{4,i}$, we have
\begin{align*}
    \Ical_{4,i} 
    &= 
        \norm{[
            \nabla_{\Bb_{i}} f(\setAk,\setBk)
        ]_{\bar\Scal_i}}_F^2\\
    &\leq 
        3\bigg\{ 
            \norm{[
                \nabla_{\Bb_{i}} f(\setAk,\setBk) 
                -
                \nabla_{\Bb_{i}} f(\setpAk,\setBk)
            ]_{\bar\Scal_i}}_F^2 \\
        &\quad \quad + 
            \norm{[
                \nabla_{\Bb_{i}} f(\setpAk,\setBk)
                - 
                \gradBstar 
            ]_{\bar\Scal_i}}_F^2\\
        &\quad \quad + 
            \norm{[\gradBstar]_{\bar\Scal_i}}_F^2
        \bigg\}\\
    &\leq 
        3(\Ical_{41,i}+\Ical_{42,i}+\Ical_{43,i}).
\end{align*}
We further upper bound $\Ical_{41,i}$ and $\Ical_{42,i}$, respectively below.

First, we can write $\Ical_{41,i}$ using the result in~\eqref{eq:l31m_lb}:
\begin{align}
    \Ical_{41,i}
    &=
        \norm{[
            \nabla_{\Bb_{i}} f(\setAk,\setBk) 
            -
            \nabla_{\Bb_{i}} f(\setpAk,\setBk)
        ]_{\bar\Scal_i}}_F^2\notag\\
    &=
        \bigg\|
            \sum_{m=1}^M[
              \tBik\bigAs{p}{\Amk} \sCovfree{\cdot}{\noi} \bigAd{p-1}{\Amk-\pAmk}^\top]_{\bar\Scal_i}\notag\\
            &\quad\quad    +
            [\tBik\bigAd{p}{\Amk-\pAmk} \sCovfree{\cdot}{\noi} \bigAs{p-1}{\pAmkT}
            ]_{\bar\Scal_i}
        \bigg\|_F^2.\notag\\
    \intertext{Apply Cauchy-Schwarz inequality, the above display can be bounded as}
    &\leq 
        2M\sum_{m=1}^M
        \norm{[
            \tBik\bigAs{p}{\Amk} \sCovfree{\cdot}{\noi} \bigAd{p-1}{\AmkT-\pAmkT}
        ]_{\bar\Scal_i}}_F^2\notag\\
        &\quad
        +2M\sum_{m=1}^M
        \norm{
        [
         \tBik\bigAd{p}{\Amk-\pAmk} \sCovfree{\cdot}{\noi} \bigAs{p-1}{\pAmkT}
        ]_{\bar\Scal_i}}_F^2.\label{eq:I41_i}
\end{align}
Then, we bound the above two terms separately.

Since $\tBik\bigAs{p}{\Amk} \sCovfree{\cdot}{\noi} \bigAd{p-1}{\AmkT-\pAmkT}$ is evaluated on the support $\bar\Scal_i$, for any $(u,v)\in\bar\Scal_i$, we have $\lceil v/k \rceil\in\barNi$. Therefore, we have
\begin{align}
    \sum_{m=1}^M
        \norm{[
            \tBik\bigAs{p}{\Amk} &\sCovfree{\cdot}{\noi} \bigAd{p-1}{\AmkT-\pAmkT}
        ]_{\bar\Scal_i}}_F^2\notag\\
    &= 
        \sum_{m=1}^M
        \norm{[
            \tBik\bigAs{p}{\Amk} \sCovfree{\cdot}{\barNi} \bigAd{|\barNi|}{\AmkT-\pAmkT}
        ]_{\bar\Scal_i}}_F^2.\notag\\   
    &\leq 
         (1+\vartheta_1)s^\star \norm{\tBik}_2^2
        \sum_{m=1}^M
        \norm{\Amk}_2^2\norm{\sCov}_2^2
        \norm{\Amk  - \pAmk }_F^2\notag\\
    &\leq (1+\vartheta_1)s^\star(1+C_A)^2(2+C_B)^2\norm{\ptBik}_2^2\sum_{m=1}^M
        \norm{\pAmk}_2^2\norm{\sCov}_2^2
        \norm{\Amk  - \pAmk }_F^2,\label{eq:I41_i1}
\end{align}
where the last line follows by the condition of Lemma~\ref{lemma:onestep_Bij}.

Similarly, we can write the second term of~\eqref{eq:I41_i} as
\begin{align}
\sum_{m=1}^M
        \norm{
        [
         \tBik\bigAd{p}{\Amk-\pAmk} &\sCovfree{\cdot}{\noi} \bigAs{p-1}{\pAmkT}
        ]_{\bar\Scal_i}}_F^2\notag\\
        &=
            \sum_{m=1}^M
        \norm{
        [
         \tBikfree{\eNi}\bigAd{|\eNi|}{\Amk-\pAmk} \sCovfree{\eNi}{\noi} \bigAs{p-1}{\pAmkT}
        ]_{\bar\Scal_i}}_F^2.\notag\\
        \intertext{Using the fact that $|\eNi|\leq\vartheta_1s^\star/2\leq (1+\vartheta_1)s^\star$, we have $\norm{\bigAd{|\eNi|}{\Amk-\pAmk}}_F^2\leq(1+\vartheta_1)s^\star\norm{\Amk-\pAmk}_F^2$. Therefore, the above display can be further bounded as}
        &\leq
        (1+\vartheta_1)s^\star \norm{\tBik}_2^2
        \sum_{m=1}^M
        \norm{\pAmk}_2^2\norm{\sCov}_2^2
        \norm{\Amk  - \pAmk }_F^2\notag\\
        &\leq  (1+\vartheta_1)s^\star(2+C_B)^2 \norm{\ptBik}_2^2
        \sum_{m=1}^M
        \norm{\pAmk}_2^2\norm{\sCov}_2^2
        \norm{\Amk  - \pAmk }_F^2.\label{eq:I41_i2}
\end{align}

Plugging the results of~\eqref{eq:I41_i1}--\eqref{eq:I41_i2} back to~\eqref{eq:I41_i}, we could conclude that
\[
\Ical_{41,i}\leq M\Upsilon_5\sum_{m=1}^Ms^\star\norm{\Amk  - \pAmk }_F^2,
\]
where $\Upsilon_5=4(1+C_A)^2(2+C_B)^2\rho_b\rho_a\max_{\mseq}\norm{\sCov}_2^2$.

Similarly, expanding $\Ical_{42,ij}$ obtains
\begin{align*}
    \Ical_{42,i} 
    &\leq
        \norm{
            \nabla_{\Bb_{i}} f(\setpAk,\setBk)
            - 
            \nabla_{\Bb_{i}} f(\setpAk,\setpBk)
        }_F^2\\
    &= 
        \norml{
            \sum_{m=1}^M
            (\Bik - \pBik)
            \bigAs{p-1}{\pAmk} \sCovfree{\backslash\{i\}}{\backslash\{i\}}
            \bigAs{p-1}{\pAmkT}
        }_F^2\\
    &\leq 
        M\sum_{m=1}^M
        \norm{\pAmk}_2^4
        \norm{\sCovfree{\backslash\{i\}}{\backslash\{i\}}}_2^2
        \norm{\Bik -\pBik }_F^2
    \leq 
        \Upsilon_{7}M^2\norm{\Bik -\pBik }_F^2,
\end{align*}
where $\Upsilon_{7} =\max_{\mseq} \rho_a^2\norm{\sCov}_2^2$.

Next, we can write $\Ical_{43,i}$ as
\[
    \Ical_{43,i} 
    = 
        \norm{[\gradBstar]_{\bar\Scal_i}}_F^2
    \leq 
        (1+\vartheta_1)s^\star
        \norm{[\gradBstar]_{\bar\Scal_i}}_{r(k,\infty)}^2.
\]

Combining results of $\Ical_{41,i}$ and $\Ical_{42,i}$, we can conclude that 
\begin{multline*}
    \Ical_{4,i}
    \leq 
        3\bigg\{ 
            \Upsilon_{5} M s^\star
            \sum_{m=1}^M
            \norm{\Amk  - \pAmk }_F^2 
            + 
            M^2\Upsilon_{7}
            \norm{\Bik - \pBik}_F^2\\
            +
            (1+\vartheta_1)s^\star
            \norm{[\gradBstar]_{\bar\Scal_i}}_{r(k,\infty)}^2
    \bigg\}.
\end{multline*}
\end{proof}

\section{Analysis of the Initialization}
Section~\ref{ssec:proof_theoreminicca} discusses the proof of Theorem~\ref{theorem:initialcca}. Section~\ref{ssec:auxiliary_initialcca} introduces auxiliary lemmas for the proofs of Theorem~\ref{theorem:initialcca}. 
Section~\ref{ssec:proof_lemmastepB} shows the proof of Lemma~\ref{lemma:initial_stepB}.

\subsection{Proof of Theorem~\ref{theorem:initialcca}}
\label{ssec:proof_theoreminicca}

We prove the result for one modality. 
Let $\cpAmk=(\Lb^{\star})^\dagger$ be the pseudoinverse of $\Lb^{m\star}$, where $\Lb^{m\star}\in\RR^{k_m\times k}$ be the matrix realization of ${\Lscr^m}$ such that the $\ell\ell'$-th entry of  $\Lb^{m\star}$ is $L_{\ell\ell'}^m=\dotp{{\Lscr^m}\phi_{\ell'}^m}{\phi_{\ell}^m}_{\HH_m}$. It should be noted that $\cpAmk\in\RR^{k\times k_m}$ is different from $\pAmk\in\RR^{k\times k_m}$: the former is first truncating ${\Lscr^m}$ and then taking pseudoinverse while the later is first taking the pseudoinverse of ${\Lscr^m}$ and then conducting finite truncation. However, the difference of the two will become small as $k$ and $k_m$ are selected large enough and the remaining terms are small in magnitude.
For the estimate in \eqref{eq:cca_model}, we have the following 
\begin{equation}\label{eq:QAA}
\norm{\Amk^{(0)} -\Qb\pAmk}_F=\norm{\Qb\Amk^{(0)}-\pAmk}_F\leq\norm{\Qb\Amk^{(0)}-\cpAmk}_F+\norm{{\cpAmk-\pAmk}}_F,\qquad m=1,2,
\end{equation}
where $\Qb$ is a diagonal matrix with diagonal entries taking values in $\{-1,1\}$ that aligns columns of the matrix $\Vb^m$ with the corresponding columns of the population canonical matrix $\Vb^{m\star}$ so that $\dotp{q_{j}\vb_{j}^m}{\vb_{j}^{m\star}}\geq 0$. For the first term of~\eqref{eq:QAA}, we can apply the result of Lemma~\ref{lemma:first_error} to find a valid upper bound. For the second term, we recall that 
\begin{align*}
    &{\Lscr^{k_m,k}}=\sum_{\ell=1}^{k_m}\sum_{\ell'=1}^k\dotp{{\Lscr^m}\phi_{\ell'}}{\phi_{\ell}^m}_{\HH_m}\phi_{\ell}^m\otimes\phi_{\ell'};\\
    &\Ascr_m^{k,k_m}=\sum_{\ell=1}^{k}\sum_{\ell'=1}^{k_m}\dotp{\Ascr_m\phi_{\ell'}^m}{\phi_{\ell}}_{\HH}\phi_{\ell}\otimes\phi_{\ell'}^m.
\end{align*}
 Define ${\Lscr^{m,r}}={\Lscr^m}-{\Lscr^{k_m,k}}$ and write
\begin{align*}
    \norm{\cpAmk-\pAmk}_F
    &\leq
    (2k)^{1/2}\norm{\cpAmk-\pAmk}_2\\
    &=
    (2k)^{1/2}
    \opnorm{({\Lscr^{k_m,k}})^\dagger - \Ascr_m^{k,k_m}}{}.\\
    \intertext{Apply Lemma~\ref{lemma:truncation_error}, the above display can be bounded as}
    &\leq 2^{3/2}k^{1/2}\opnorm{{\Lscr^m}^\dagger}{}^2\opnorm{{\Lscr^{m,r}}}{}.
\end{align*}
Hence, we complete the proof.
\subsection{Proofs of Lemma~\ref{lemma:first_error}--\ref{lemma:Rxy_tail}}\label{ssec:auxiliary_initialcca}
\begin{lemma}\label{lemma:first_error}
Suppose that Assumptions~\ref{assumption:distinctcca}--\ref{assumption:cov}  hold and $N=O( \max_{\mseq}\kappa_m^2k_m)$ where $\defcdnm$. Then, let $C_{\gamma_k,\nu_x,\rho_x}>0$ be a constant depending on $\gamma_k$, $\nu_x$, and $\rho_x$, we have
\[
\norm{\Qb\Amk^{(0)} -\cpAmk}_F
\leq 
C_{\gamma_k,\nu,\rho}\rbr{1+\max_{j=1,\ldots,k}\max_{j\neq i}\frac{1}{|\gamma_j -\gamma_i|}}
\sqrt{k\frac{\sum_{k=1}^2k_m}{N}},
\]
with probability at least $1-5\exp(-\sum_{m=1}^2 k_m)$.
\end{lemma}
\begin{proof}[Proof of Lemma~\ref{lemma:first_error}]
We have 
\begin{align*}
&\left\|\Qb\Amk^{(0)} - \cpAmk\right\|_F \\
&=
\left\|\hat\Gammab^{-1/2}\hat\Vb^{m\top}(\hat\Sigmab_{11}^m)^{-1/2}- \Qb\Gammab^{\star-1/2}\Vb^{m\star\top}(\Sigmab_{11}^m)^{-1/2}\right\|_F && (\Qb, \Gammab \text{ are diagonal matrices})\\
&=
\left\|\hat\Gammab^{-1/2}\hat\Vb^{m\top}(\hat\Sigmab_{11}^m)^{-1/2}- \Gammab^{\star-1/2}\Qb\Vb^{m\star\top}(\Sigmab_{11}^m)^{-1/2}\right\|_F \\
&\leq  
\left\|\hat\Gammab^{-1/2}\hat\Vb^{m\top}\cbr{
        (\hat\Sigmab_{11}^m)^{-1/2}
        -
        (\Sigmab_{11}^m)^{-1/2}
    }\right\|_F    \\
    & \quad +
    \left\|\hat\Gammab^{-1/2}\rbr{
        \hat\Vb^m
        -
        \Vb^{m\star}\Qb
    }^\top
    (\Sigmab_{11}^m)^{-1/2}
    \right\|_F \\
    & \quad +
    \left\|
    \rbr{\hat\Gammab^{-1/2}-\Gammab^{\star-1/2}}
    \Qb\Vb^{m\star\top}(\Sigmab_{11}^m)^{-1/2}
    \right\|_F.
\end{align*}
We then apply Lemma~\ref{lemma:I81}--\ref{lemma:I83} to upper bound the above three terms individually and arrive at
\begin{align*}
\left\|\Qb\Amk^{(0)} - \cpAmk\right\|_F &\leq
C_1\norm{\hat{\Gammab}^{-1/2}}_2(\rho_x^{1/2}\nu_x^{-2}+\nu_x^{-3/2})\sqrt{
\frac{kk_m}{N}
}\\
&\quad+
C_2k^{1/2}\nu_x^{-1/2}\norm{\hat{\Gammab}^{-1/2}}_2(\gamma_k^{-3/2}+1)
\max_{j=1,\ldots,k}\max_{j\neq i}\frac{1}{|\gamma_j-\gamma_i|}\norm{\hat\Rb^{12}-\Rb^{12\star}}_2,
\end{align*}
 with probability at least $1-\exp(-k_m)$.
 
Under the condition that $N=O(\max_{\mseq}\kappa_m^2 k_m)$, we could apply Lemma~\ref{lemma:Rxy_tail} and obtain
\begin{equation*}
\norm{\hat\Rb^{12}-\Rb^{12\star}}_2^2
\leq C_3\rho_x^4\nu_x^{-4}(\rho_x^{1/2}\nu_x^{-1/2}+1)^2\frac{\sum_{m=1}^2k_m}{N}
\end{equation*}
with probability at least $1-3\exp(-\sum_{m=1}^2k_m)$ for some constant $C_3>0$.

Recall that 
$\|\hat\Gammab^{-1/2}\|_2=\norm{\hat\Rb_{12}^{-1/2}}_2$ and by Lemma~\ref{lemma:ci}, we have $\|\hat\Gammab^{-1/2}\|_2\leq (3/2) \gamma_k^{-1/2}$, with probability at least $1-\exp(-k)$.

Taking the union bound, we can conclude that
\[
\norm{\Qb\Amk^{(0)} - \cpAmk}_F
\leq 
C_{\gamma_k,\rho_x,\nu_x}\rbr{1+\max_{j=1,\ldots,k}\max_{j\neq i}\frac{1}{|\gamma_j -\gamma_i|}}
\sqrt{\frac{k\sum_{k=1}^2k_m}{N}},
\]
with probability at least $1-5\exp(-\sum_{m=1}^2 k_m)$ and 
\begin{align*}
C_{\gamma_k,\rho_x,\nu_x}&=C_4\bigg\{\gamma_k^{-1/2}(\rho_x^{1/2}\nu_x^{-2}+\nu_x^{-3/2})+\rho_x^2\nu_x^{-5/2}(\gamma_k^{-1/2}+\gamma_k^{-2})(\rho_x^{1/2}\nu_x^{-1/2}+1)\bigg\},
\end{align*}
for some universal constant $C_4>0$.
\end{proof}

\begin{lemma}\label{lemma:I81}
Under the conditions of Lemma~\ref{lemma:first_error} and an universal constant $C_1>0$, we have
\[
\left\|\hat\Gammab^{-1/2}\hat\Vb^{m\top}\cbr{
        (\hat\Sigmab_{11}^m)^{-1/2}
        -
        (\Sigmab_{11}^m)^{-1/2}
    }\right\|_F \leq 
    C_1\norm{\hat{\Gammab}^{-1/2}}_2(\rho_x^{1/2}\nu_x^{-2}+\nu_x^{-3/2})
\sqrt{
\frac{kk_m}{N}
},
\]
 with probability at least $1-\exp(-k_m)$.
\end{lemma}

\begin{proof}[Proof of Lemma~\ref{lemma:I81}]
We have
\begin{align*}
\left\|\hat\Gammab^{-1/2}\hat\Vb^{m\top}\cbr{
        (\hat\Sigmab_{11}^m)^{-1/2}
        -
        (\Sigmab_{11}^m)^{-1/2}
    }\right\|_F& = 
\left\|{\hat\Gammab^{-1/2}\hat\Vb^{m\top}\cbr{
        (\hat{\Sigmab}_{11}^m)^{-1/2}
        -
        (\Sigmab_{11}^m)^{-1/2}
    }}\right\|_F\\
&\leq k^{1/2}\left\|\hat\Gammab^{-1/2}\right\|_2\norml{ (\hat{\Sigmab}_{11}^m)^{-1/2}
        -
        (\Sigmab_{11}^m)^{-1/2}}_2,
\intertext{where the inequality follows by the fact that $\norm{\hat\Vb^{m}}_2\leq 1$. Apply Lemma~\ref{lemma:difference_sqrt}, we can further obtain that}
&
\leq 
k^{1/2}\cbr{3\norml{(\hat{\Sigmab}_{11}^m)^{-1/2}}_2\rbr{\norml{\hat{\Sigmab}_{11}^m}_2\vee\norml{\Sigmab_{11}^m}_2}+1}\\
&\quad\times\left\|\hat\Gammab^{-1/2}\right\|_2\norml{(\Sigmab_{11}^m)^{-3/2}}_2
\norml{\hat{\Sigmab}_{11}^m-\Sigmab_{11}^m}_2.
\end{align*}

Recall that $
\sinmin{\Sigmab^{m}_{11}}\geq\sinmin{\Sigmab^{m}}\geq\nu_x$ and $\rho_x\geq\sinmax{\Sigmab^{m}}\geq \sinmax{\Sigmab^m_{11}}
$.
Since $N=O( \max_{\mseq}\kappa_m^2k_m)$, we have 
\begin{equation}\label{eq:intermediate:1}
3\rho_x/2\geq\sinmin{\hat\Sigmab^m_{11}}\geq\nu_x/2,
\end{equation}
with probability at least $1-\exp(-k_m)$, following Lemma~\ref{lemma:ci}. 


Therefore, under Assumption~\ref{assumption:cov}, we have 
\begin{align}\label{eq:cauchyI_inequality}
\cbr{3\norml{(\hat{\Sigmab}_{11}^m)^{-1/2}}_2\rbr{\norml{\hat{\Sigmab}_{11}^m}_2\vee\norml{\Sigmab_{11}^m}_2}+1}\norml{(\Sigmab_{11}^m)^{-3/2}}_2\leq 9\rho_x^{1/2}\nu_x^{-2}+\nu_x^{-3/2},
\end{align}
with probability at least $1-\exp(-k_m)$.
Then, we arrive at the conclusion that
\begin{align*}
\norm{\hat\Gammab^{-1/2}\hat\Vb^{m\top}\rbr{
        (\hat{\Sigmab}_{11}^m)^{-1/2}
        -
        (\Sigmab_{11}^m)^{-1/2}
    }}_F
&\leq 
C_1\norm{\hat{\Gammab}^{-1/2}}_2(\rho_x^{1/2}\nu_x^{-2}+\nu_x^{-3/2})
\sqrt{
\frac{kk_m}{N}
},
\end{align*}
 with probability at least $1-\exp(-k_m)$.
\end{proof}

\begin{lemma}\label{lemma:I82}
Under the conditions of Lemma~\ref{lemma:first_error} and an universal constant $C_1>0$, we have
\begin{multline*}
\norm{\hat\Gammab^{-1/2}\rbr{
        \hat\Vb^m
        -
        \Vb^{m\star}\Qb
    }^\top
    (\Sigmab_{i,i}^m)^{-1/2}
    }_F\\
\leq C_1k^{1/2}\nu_x^{-1/2}\norm{\hat{\Gammab}^{-1/2}}_2\max_{j=1,\ldots,k}\max_{j\neq i}\frac{1}{|\gamma_j-\gamma_i|}
\norm{\hat\Rb^{12}-\Rb^{12\star}}_2.
\end{multline*}
\end{lemma}
\begin{proof}[Proof of Lemma~\ref{lemma:I82}]
Write
\[
\left\|\hat\Gammab^{-1/2}\rbr{
        \hat\Vb^m
        -
        \Vb^{m\star}\Qb
    }^\top
    (\Sigmab_{11}^m)^{-1/2}
    \right\|_F
    \leq \norm{\hat\Gammab^{-1/2}}_2\norm{(\Sigmab_{11}^m)^{-1/2}}_2\norm{\hat\Vb^m
        -
        \Vb^{m\star}\Qb}_F.
\]
Note that we have
\begin{align*}
    \norm{
    \hat\Vb^m
    -
    \Vb^{m\star}\Qb}_F^2
    =\sum_{i=1}^k \norm{
    \hat\vb_{i}^m
    -
    q_i\vb_{i}^{m\star}}_2^2,    
\end{align*}
where $q_i$ is the $i$th diagonal entry of $\Qb$.

Since $q_i$ is either $+1$ or $-1$ such that $\norm{
    \hat\vb_{i}^m
    -
    q_i\vb_{i}^{m\star}}_2^2$ is minimized, we know that $\dotp{\hat\vb_{i}^m}{q_i\vb_{i}^{m\star}}\geq0$. Let $\Delta=\hat\Rb^{12}-\Rb^{12\star}$ and $C_2,C_3>0$ be universal constants.
    Recall the singular decomposition of $\Rb^{12\star}$ is $\Rb^{12\star}=\sum_j^{k} \gamma_j \vb_{j}^{1\star} (\vb_{j}^{2\star})^\top$, where $\gamma_j$ could be viewed as the canonical correlation of $\yb_{1}^1$ and ${\yb_{1}^2}$ for $j=1,\ldots,k$. Without the loss of generality, let $m=1$ be the first data modality and $m'=2$ be the second data modality, we can apply Theorem 5.2.2 in~\citet{hsing2015theoretical}, restated in Lemma~\ref{lemma:perturb_sinvec}, and obtain
\begin{align*}
    \norm{
    \hat\vb_{i}^m
    -
    q_i\vb_{i}^{m\star}}_2
    &\leq 
    \bignorm{
    \sum_{j\neq i}
    \frac{
    \gamma_j\dotp{\vb_{j}^{m'}q_{j}}{\Delta^\top\vb_{i}^{m\star} q_{i}}
    +
    \gamma_i\dotp{\vb_{i}^{m'}q_{i}}{\Delta^\top\vb_{j}^{m\star} q_{j}}
    }{\gamma_j^2-\gamma_i^2}
    (\vb_{j}^{m\star} q_{j})
    }_2\\
    &\quad+
    C_2\max_{j\neq i}\frac{2}{|\gamma_j^2-\gamma_i^2|}\norm{\Delta}_2\\
    &\leq
    \max_{j\neq i}\frac{\gamma_j}{|\gamma_j^2-\gamma_i^2|}\bignorm{
    \sum_{j\neq i}
    \dotp{\vb_{j}^{m'}q_{j}}{\Delta^\top\vb_{i}^{m\star} q_{i}}
    (\vb_{j}^{m\star} q_{j})
    }_2\\
    &\quad
    +
    \max_{j\neq i}\frac{\gamma_i}{|\gamma_j^2-\gamma_i^2|}
    \bignorm{
    \sum_{j\neq i}
   {
    \dotp{\vb_{i}^{m'}q_{i}}{\Delta^\top\vb_{j}^{m\star} q_{j}}
    }(\vb_{j}^{m\star} q_{j})
    }_2\\
    &\quad+
    C_2\max_{j\neq i}\frac{2}{|\gamma_j^2-\gamma_i^2|}\norm{\Delta}_2.
    \intertext{Since $\norm{\sum_{j\neq i}
   {
    \dotp{\vb_{i}^{m'}q_{i}}{\Delta^\top\vb_{j}^{m\star} q_{j}}
    }(\vb_{j}^{m\star} q_{j})
    }_2\leq\norm{\Delta}_2$, $\norm{\sum_{j\neq i}
    \dotp{\vb_{j}^{m'}q_{j}}{\Delta^\top\vb_{i}^{m\star} q_{i}}
    (\vb_{j}^{m\star} q_{j})}_2\leq \norm{\Delta}_2$ and $\gamma_i,\gamma_j\leq 1$, the above display can be further bounded as}
    &\leq 
    \max_{j\neq i}\frac{2\norm{\Delta}_2}{|\gamma_j^2-\gamma_i^2|}
    +
    C_2\max_{j\neq i}\frac{2}{|\gamma_j^2-\gamma_i^2|}\norm{\Delta}_2\\
    &
    \leq C_3\max_{j\neq i}\frac{1}{|\gamma_j-\gamma_i|}\norm{\Delta}_2,
\end{align*}
where the last inequality follows by the fact that the magnitude of the canonical correlation is less than one $|\gamma_i|\leq 1$ for any $i$ and hence $|\gamma_j^2-\gamma_i^2|\leq |\gamma_j-\gamma_i||\gamma_j+\gamma_i|\leq 2|\gamma_j-\gamma_i|$. Then, we complete the proof.

\end{proof}

\begin{lemma}\label{lemma:I83}Under the conditions of Lemma~\ref{lemma:first_error} and an universal constant $C_1>0$, we have
\[
\norm{
    \rbr{\hat\Gammab^{-1/2}-\Gammab^{\star-1/2}}
    \Qb\Vb^{m\star\top}(\Sigmab_{11}^m)^{-1/2}
}_F
\leq 
k^{1/2}\gamma_k^{-3/2}\nu_x^{-1/2}(3\norm{\hat\Gammab^{-1/2}}+1)
\norm{\hat\Rb^{12}-\Rb^{12\star}}_2.
\]
\end{lemma}
\begin{proof}[Proof of Lemma~\ref{lemma:I83}]
Since $\Qb$ is a diagonal matrix whose diagonal entries taking values in $\{1,-1\}$ and $\norm{\Vb^{m\star\top}}_2\leq 1$, we can write
\[
\left\|
    \rbr{\hat\Gammab^{-1/2}-\Gammab^{\star-1/2}}
    \Qb\Vb^{m\star\top}(\Sigmab_{11}^m)^{-1/2}
    \right\|_F\leq 
    \|(\Sigmab_{11}^m)^{-1/2}
    \|_2\left\|{\hat\Gammab^{-1/2}-\Gammab^{\star-1/2}}\right\|_F.
\]
Applying Lemma~\ref{lemma:difference_sqrt} again with $\lambda=(\norm{\hat\Gammab}_2\vee\norm{\Gammab^{\star}}_2)\leq 1$ yields
\begin{align}\label{eq:I83_1}
\norm{\hat\Gammab^{-1/2}-\Gammab^{\star -1/2}}_F
&\leq 
k^{1/2}\norm{\Gammab^{\star-3/2}}(3\norm{\hat\Gammab^{-1/2}}+1)
\norm{\hat\Gammab-\Gammab^{\star}}_2.\notag\\
&\leq k^{1/2}\gamma_k^{-3/2}(3\norm{\hat\Gammab^{-1/2}}+1)
\norm{\hat\Gammab-\Gammab^{\star}}_2.
\end{align}
Applying Weyl's inequality (See Theorem~III.2.1 in~\citet{bhatia2013matrix} for details), we have
\[
\norm{\hat\Gammab-\Gammab^{\star}}_2
\leq 
\norm{\hat\Rb^{12}-\Rb^{12\star}}_2.
\]
Then, we complete the proof.

\end{proof}

\begin{lemma}\label{lemma:truncation_error}
Let ${\Lscr}\in\mathfrak{B}(\HH_1,\HH_2)$ be a compact operator and $\{\phi_{1,\ell}\}_{\ell\in\NN}$ and $\{\phi_{2,\ell}\}_{\ell\in\NN}$ be the CONS for $\HH_1$ and $\HH_2$, respectively. 
Let 
$$
{\Lscr^{r}}={\Lscr}-{\Lscr^{p,d}}, \quad{\Lscr^{p,d}}=\sum_{\ell=1}^p\sum_{\ell'=1}^{d}\dotp{{\Lscr}\phi_{1,\ell'}}{\phi_{2,\ell}}\phi_{2,\ell}\otimes \phi_{1,\ell'}.
$$
Define $\Ascr=\Lscr^\dagger$ and $\Ascr^{d,p}= \sum_{\ell=1}^d\sum_{\ell'=1}^{p}\dotp{{\Lscr^\dagger}\phi_{2,\ell'}}{\phi_{1,\ell}}\phi_{1,\ell}\otimes \phi_{2,\ell'}$, then we have
$$
\norm{(\Lscr^{p,d})^\dagger-\Ascr^{d,p}}_{\text{HS}} \leq2\opnorm{{\Lscr}^\dagger}{}^2\opnorm{{\Lscr^{r}}}{}.
$$
\end{lemma}
\begin{proof}[Proof of Lemma~\ref{lemma:truncation_error}] Write
\begin{align*}
    \opnorm{(\Lscr^{p,d})^\dagger-\Ascr^{d,p}}{}
    &\leq
    \opnorm{({\Lscr^{p,d}})^\dagger - \Ascr}{}\\
    &=
    \opnorm{
    ({\Lscr^{p,d}})^\dagger
    -
    ({\Lscr^{p,d}}+{\Lscr^{r}})^\dagger
    }{}.
    \intertext{Apply the result from Lemma~\ref{lemma:pseuoinv_hs}, we could further bound the above display as}
    &\leq2(\opnorm{({\Lscr^{p,d}})^\dagger}{}^2\vee\opnorm{{\Lscr}^\dagger}{}^2)\opnorm{{\Lscr^{r}}}{}\\
    &\leq 2\opnorm{{\Lscr}^\dagger}{}^2\opnorm{{\Lscr^{r}}}{}.
\end{align*}
Hence we complete the proof.
\end{proof}

\begin{lemma}\label{lemma:Rxy_tail} Given two centered Gaussian random vectors $\xb\in\RR^d,\yb\in\RR^p$ with covariance $\Sigmab_x$ and $\Sigmab_y$, respectively. We define $\Sigmab_{xy}$ be the cross covariance and $\Rb_{xy}=\Sigmab_x^{-1/2}\Sigmab_{xy}\Sigmab_y^{-1/2}$. Let $\hat{\Sigmab}_x$ be the sample covariance of $\xb$, $\hat{\Sigmab}_y$ sample covariance of $\yb$, and $\hat\Sigmab_{xy}$ be the sample cross covariance of $\xb,\yb$ with $N$ independent samples. Define $\hat\Rb_{xy}=\hat\Sigmab_x^{-1/2}\hat\Sigmab_{xy}\hat\Sigmab_y^{-1/2}$. Let $C_1> 0$ be an universal constant. Assume that $N=O (p+d)$, $\alpha=(\norm{\Sigmab_x}_2\vee \norm{\Sigmab_y}_2)$ and $\beta=(\norm{\Sigmab_x^{-1/2}}_2\vee \norm{\Sigmab_y^{-1/2}}_2)$, then
\[
\norm{\hat\Rb_{xy}-\Rb_{xy}}_2\leq
C_1\alpha^2\beta^4(3\alpha^{1/2}\beta+1)
\sqrt{\frac{(p+ d)}{N}},
\]
with probability at least $1-3\exp\{-(p+d)\}$.
\end{lemma}

\begin{proof}
The inequality could be shown by first applying the triangle inequality and then use the tail bound to find the upper bound of each term individually. We have
\begin{align*}
    \norm{\hat\Rb_{xy}-\Rb_{xy}}_2
    &\leq 
    \norm{
        (\hat\Sigmab_x^{-1/2}-\Sigmab_x^{-1/2})
        \hat\Sigmab_{xy}\hat\Sigmab_y^{-1/2}
        }_2\\
    &\quad+
    \norm{
        \Sigmab_x^{-1/2}
        (\hat\Sigmab_{xy}-\Sigmab_{xy})
        \hat\Sigmab_y^{-1/2}
    }_2\\
    &\quad+
    \norm{
    \Sigmab_x^{-1/2}\Sigmab_{xy}
    (\hat\Sigmab_y^{-1/2}-\Sigmab_y^{-1/2})
    }_2\\
    &\leq \Ical_{9,1}+\Ical_{9,2}+\Ical_{9,3}
    .
\end{align*}
Since $N=O(d+ p)$, we have
\begin{align*}
    &(i)\quad\sinmin{\hat\Sigmab_x}\geq\frac{1}{2}\sinmin{\Sigmab_x};\\
    &(ii)\quad\sinmin{\hat\Sigmab_y}\geq\frac{1}{2}\sinmin{\Sigmab_y};
    \\
    &(iii)\quad\sinmax{\hat\Sigmab_{xy}}\leq\frac{3}{2}\sinmin{\Sigmab_{xy}}\leq\frac{3}{2}(\norm{\Sigmab_x}_2\vee \norm{\Sigmab_y}_2),
\end{align*}
with probability at least $1-3\exp\{-(d+p)\}$, following Lemma~\ref{lemma:ci} and the union bound. 

Apply Lemma~\ref{lemma:difference_sqrt} and Lemma~\ref{lemma:ci} to $\Ical_{9,1}$ and we could obtain the upper bound
\begin{multline*}
\Ical_{9,1}\leq 
C_2(\norm{\Sigmab_x}_2\vee \norm{\Sigmab_y}_2)^{2}
(\norm{\Sigmab_x^{-1/2}}_2\vee \norm{\Sigmab_y^{-1/2}}_2)^{4}\\
\cbr{3
(\norm{\Sigmab_x}_2\vee \norm{\Sigmab_y}_2)^{1/2}
(\norm{\Sigmab_x^{-1/2}}_2\vee \norm{\Sigmab_y^{-1/2}}_2)
+1}\sqrt{\frac{(p+ d)}{N}},
\end{multline*}
with probability $1-3\exp\{-(d+p)\}$ for some absolute constant $C_2>0$. We obtain the same upper bound for $\Ical_{9,3}$. For $\Ical_{9,2}$, we have smaller coefficient:
\[
\Ical_{9,2}\leq C_3(\norm{\Sigmab_x}_2\vee\norm{\Sigmab_y}_2) 
(\norm{\Sigmab_x^{-1/2}}_2\vee \norm{\Sigmab_y^{-1/2}}_2)^2\sqrt{\frac{(p+ d)}{N}},
\]
with probability $1-2\exp\{-(d+p)\}$ and $C_3>0$.
Hence, the terms $\Ical_{9,1}$ and $\Ical_{9,3}$ dominate. Then, by summing up $\Ical_{9,1}$, $\Ical_{9,2}$, and $\Ical_{9,3}$, we complete the proof.

\end{proof}

\subsection{Proof of Lemma~\ref{lemma:initial_stepB}}\label{ssec:proof_lemmastepB}

The proof is consisted of three steps. In the first step, we show the condition where the solution is nontrivial. In the second and third step, we show the contraction of the iterates of Algorithm~\ref{alg:initializationB}.

{\em Step 1}. Given $\Amk^{(0)}$ for $\mseq$, we want to verify that $\Amk^{(0)}\Ykmnfree{i}\neq {\bf 0}$ for $\mseq$ and $\pseq$ with high probability. Then the solution to the minimizer of $h_i(\Bik)$ is nontrivial. Since
\[
\frac{1}{N}\norm{\Amk^{(0)}\Ykmnfree{i}}_F^2
=
\VEC(\Amk^{(0)})^\top
(\Ib_k \otimes \sCovfree{i}{i})
\VEC(\Amk^{(0)})
\geq 0
\]
and $\sinmin{\Ib_k \otimes \sCovfree{i}{i}}=\sinmin{\sCovfree{i}{i}}$, it suffices to show that $\sinmin{\sCovfree{i}{i}}>0$. We define the event as
\[
{\Ecal_{ii}^m} = \{\sinmin{\sCovfree{i}{i}}>0\},\quad\mseq. 
\]
Then, under the condition that $N=O(\kappa_m^2(k_m+\log M + \log p))$, Lemma~\ref{lemma:ci} states that event ${\Ecal_{ii}^m}$ happens with probability at least $1-(pM)^{-1}\exp(-k_m)$ for $\mseq$.
Taking the union bound again and applying De Morgan's law, we can show that the event
\[
\cap_{m=1}^M\cap_{i=1}^p{\Ecal_{ii}^m},
\]
happens with probability at least $1-\max_{\mseq}\exp(-k_m)$.
Then, we conclude that $\Amk^{(0)}\Ykmnfree{i}\neq{\bf 0}$ for $\pseq$ and $\mseq$ with probability at least 
$$
1-\max_{\mseq}\exp(-k_m)\geq 1-\max_{\mseq}\exp(-k_m s^\star).
$$

{\em Step 2}. Given that the solution is nontrivial, the next step is to prove that the projected gradient descent would converge to a local optimum with quantified error.
 By Lemma~\ref{lemma:strong_hn},  we know that $h$ is a 
$\min_{\mseq}2^{-1}\nu_x\sinmintwo{\Amk^{(0)}}$-strongly convex with respect to $\Bik$ with probability at least $1-\max_{\mseq}\exp(-k_m s^\star)$. It is easy to verify that $h$ is $\max_{\mseq}\sinmaxtwo{\Ab^{m(0)}}\norm{\sCov}$-smooth. Let $\hat\Zb_i^m=\Ab^{m(0)}\Ykmnfree{i}$ and $\Zb_i^m= \pAmk\Ykmnfree{i}$ for $\pseq$. Then, we can write
\[
h(\Bik) = \frac{1}{M} \sum_{m=1}^M\norm{\hat\Zb_i^m-\sum_{j\in\pseq}\Bijk\hat\Zb_j^m}_F^2.
\]
Let $\pi_0=C_\vartheta\cbr{1-{\sinmintwo{\Ab^{m'(0)}}\nu_x\eta_{B_0}}/2}$.
Apply Lemma~\ref{lemma:Z_update}--\ref{lemma:I73i}, we can obtain
\begin{align}
\norm{\Bik^+-\pBik}_F&\leq \pi_0\norm{\Bik-\pBik}_F\notag\\
&\quad+\eta_{B_0}C_\vartheta\frac{1}{M}\sum_{m=1}^M\big\{
    \norm{[
        -\ptBik
        \bigAd{p}{\Amk^{(0)}-\pAmk}
        \sCovfree{\cdot}{\noi}
        \bigAs{p-1}{\Amk^{(0)}}^\top
    ]_{\bar\Scal_i}}_F\notag\\
&\quad\quad\quad\quad\quad\quad\quad\quad+    
    \norm{[
        \ptBik
        \bigAs{p}{\pAmk}
        \sCovfree{\cdot}{\noi}
        \bigAd{p-1}{\Amk^{(0)}-\pAmk}^\top
    ]_{\bar\Scal_i}}_F
\big\}\notag\\
&\quad + {C_{\delta,i}C_\vartheta\eta_{B_0}}
    \sqrt{\frac{ \alpha^\star s^\star  k^2}{N}},\label{eq:h_updateB}
\end{align}
 with probability at least $1-3\max_{\mseq}\exp(-k_m s^\star)$.

Note that we can write
\begin{align}
&\norm{[
-\ptBik
\bigAs{p}{\Amk^{(0)}}
\sCovfree{\cdot}{\noi}
\bigAs{p-1}{\Amk^{(0)}}^\top
+
\ptBik
\bigAs{p}{\pAmk}
\sCovfree{\cdot}{\noi}
\bigAs{p-1}{\pAmk}^\top
]_{\bar\Scal_i}
}_F\notag\\
&\leq 
\big\{
    \norm{[
        -\ptBik
        \bigAd{p}{\Amk^{(0)}-\pAmk}
        \sCovfree{\cdot}{\noi}
        \bigAs{p-1}{\Amk^{(0)}}^\top
    ]_{\bar\Scal_i}}_F\notag\\
&\quad\quad\quad\quad\quad\quad+    
    \norm{[
        \ptBik
        \bigAs{p}{\pAmk}
        \sCovfree{\cdot}{\noi}
        \bigAd{p-1}{\Amk^{(0)}-\pAmk}^\top
    ]_{\bar\Scal_i}}_F
\big\}\notag\\
&\leq \rho_b^{1/2}
\norm{\sCov}_2
\cbr{\sinmax{\Amk^{(0)}}+\rho_a^{1/2}}
\sqrt{(1+\vartheta_1)s^\star}
\norm{\Amk^{(0)} - \pAmk}_F\notag\\
&\leq C_\gamma\norm{\Amk^{(0)} - \pAmk}_F
,\label{eq:bound_2term}
\end{align}
where $C_\gamma=\max_{\mseq}\rho_b^{1/2}
\norm{\sCov}_2
\cbr{\sinmax{\Amk^{(0)}}+\rho_a^{1/2}}
\sqrt{(1+\vartheta_1)s^\star}$.

 Combine \eqref{eq:h_updateB}--\eqref{eq:bound_2term} together, we have
\[
\norm{\Bb_i^{+}-\pBik}_F
\leq \pi_0\norm{\Bik-\pBik}_F
+   \frac{C_\gamma C_\vartheta\eta_{B_0}}{M}\sum_{m=1}^M\norm{\Ab^{m(0)} - \pAmk}_F
+   {C_{\delta,i} C_\vartheta\eta_{B_0}}
    \sqrt{\frac{ \alpha^\star s^\star  k^2}{N}},
\]
  with probability at least $1-3\max_{\mseq}\exp(-k_m s^\star)$.

{\em Step 3}. 
Starting with $\Bb_i^{(0)}={\bf 0}$, after $L$ iterations of Algorithm~\ref{alg:initializationB}, we have the following result by telescoping technique:
\begin{align*}
\norm{\Bb_i^{(L)}-\pBik}_F
\leq \pi_0^L\norm{\pBik}_F
+   \frac{C_\gamma C_\vartheta}{1-\pi_0}\frac{\eta_{B_0}}{M}\sum_{m=1}^M\norm{\Ab^{m(0)} - \pAmk}_F
+   \frac{C_{\delta,i}C_\vartheta\eta_{B_0}}{1-\pi_0}
    \sqrt{\frac{ \alpha^\star s^\star  k^2}{N}}
    ,
\end{align*}
 with probability at least $1-3\max_{\mseq}\exp(-k_m s^\star)$.
\begin{lemma}\label{lemma:Z_update}
Consider the following objective function
\[
    \hat g(\Bik)=(\Bik)=\frac{1}{2N}
        \bignorm{
            \hat \Zb_i
            -
            \sum_{j\neq i} \Bijk \hat \Zb_j
        }_F^2,
\]
where $\hat g$ is $L$-smooth and $\mu$-strongly convex with respect to $\Bik$ and $\hat \Zb_i\in\RR^{k\times N}$ for $\pseq$. Let one-step iterate of the update to be
\[
    \Bb_i^{+} 
    =
        \Hcal_\alpha\circ\Tcal_s\rbr{
            \Bik  
            - 
            \eta\nabla_{\Bik }
            \hat g(\Bik)
        },
\]
where $\eta\leq 1/L$. Let $\Zb_i\in\RR^{k\times N}$ for $\pseq$.
Define 
\begin{align*}
\hat\Kb_{\cdot\backslash i}
&=(\hat\Zb_1^\top,\ldots,\hat\Zb_{i-1}^\top,\hat\Zb_i^\top,\ldots,\hat\Zb_p^\top)^\top
(\hat\Zb_1^\top,\ldots,\hat\Zb_{i-1}^\top,\hat\Zb_{i+1}^\top,\ldots,\hat\Zb_p^\top)\in\RR^{kp\times k(p-1)}\\
\Kb_{\cdot\backslash i}
&=(\Zb_1^\top,\ldots,\Zb_{i-1}^\top,\Zb_i^\top,\ldots,\Zb_p^\top)^\top
(\Zb_1^\top,\ldots,\Zb_{i-1}^\top,\Zb_{i+1}^\top,\ldots,\Zb_p^\top)\in\RR^{kp\times k(p-1)}.
\end{align*}
Then
\[
\frac{1}{C_\vartheta}\norm{\Bb_i^{+} - \Bb_i^{\star} }_F\leq 
\rbr{1-\eta\mu}\norm{\Bik -\pBik
}_F
+
\eta\norm{[\ptBik(\hat\Kb_{\cdot\backslash i}-\Kb_{\cdot\backslash i})]_{\bar\Scal_i}}_F
+
\eta\norm{[\ptBik\Kb_{\cdot\backslash i}]_{\bar\Scal_i}}_F
\]
where $\bar\Scal_i$ is defined in~\eqref{eq:defsupport}.
\end{lemma}

\begin{proof}
First, we define
\[
        g(\Bik)=\frac{1}{2N}
        \bignorm{
            \Zb_i
            -
            \sum_{j\neq i} \Bijk \Zb_j
        }_F^2.
\]
 Apply Lemma~\ref{lemma:disperse_coef}--\ref{lemma:groupsparse}, we can write
\[
\norm{\Bik^+-\pBik}_F=\norm{\Hcal_\alpha\circ\Tcal_s\rbr{
            \Bik  
            - 
            \eta\nabla_{\Bik }
            \hat{g}(\Bik)
        }
        -\pBik}_F\leq C_\vartheta
        \norm{
            \Bik -\pBik 
            - 
            \eta[
                \nabla_{\Bik }
            \hat{g}(\Bik)]_{\bar\Scal_i}}_F,
\]
where $\bar\Scal_i$ is defined in~\eqref{eq:defsupport}. Then, by~\eqref{eq:B:update}, we can write
\begin{align*}
\norm{\Bik^+-\pBik}_F
&\leq 
   C_\vartheta\big\{
        \norm{
            \Bik -\pBik 
            - 
            \eta[
                \nabla_{\Bik}\hat{g}(\Bik)
                -
                \nabla_{\Bik}\hat g(\pBik)
            ]_{\bar\Scal_i}}_F\\
        &\quad+
        \norm{\eta[
                \nabla_{\Bik}\hat g(\pBik)
                -
                \nabla_{\Bik}g(\pBik)
            ]_{\bar\Scal_i}}_F
        +
        \norm{[\eta
            \nabla_{\Bik}g(\pBik)
        ]_{\bar\Scal_i}}_F
    \big\},
\end{align*}
where $C_\vartheta$ is a constant depending on $\vartheta_1$ and $\vartheta_2$ defined in Section~\ref{ssec:notation}.

Write
\begin{align}
\norm{
            \Bik -\pBik 
            - 
            &\eta[
                \nabla_{\Bik}\hat{g}(\Bik)
                -
                \nabla_{\Bik}\hat g(\pBik)
            ]_{\bar\Scal_i}}_F\notag\\
& \leq
\norm{
            \Bik -\pBik 
            - 
            \eta\cbr{
                \nabla_{\Bik}\hat{g}(\Bik)
                -
                \nabla_{\Bik}\hat g(\pBik)}
            }_F\notag\\
&=
\norm{
           \VEC\rbr{ \Bik -\pBik}
            - 
            \eta
                \VEC\rbr{\nabla_{\Bik}\hat{g}(\Bik)
                -
                \nabla_{\Bik}\hat g(\pBik)}
            }_2\notag\\
&=
\bignorm{\cbr{\Ib_{k^2(p-1)}
    -
    \eta\int_{t=0}^1\nabla^2 \hat{g}(\VEC\rbr{t\Bik +(1-t)\pBik})dt
}
\VEC\rbr{\Bik -\pBik}
}_2\notag\\
&\leq 
\bignorm{{\Ib_{k^2(p-1)}
    -
    \eta\int_{t=0}^1\nabla^2 \hat{g}(\VEC\rbr{t\Bik +(1-t)\pBik})dt
}}_2
\norm{\Bik -\pBik
}_F\notag\\
&\leq \rbr{1-\eta\mu}\norm{\Bik -\pBik
}_F\label{eq:I71_i}
\end{align}

Note that we can write $\nabla_{\Bik}\hat{g}(\pBik)=-\ptBik\hat\Kb_{\cdot\backslash i}$ and $\nabla_{\Bik}{g}(\pBik)=-\ptBik\Kb_{\cdot\backslash i}$. 
Then, we can write
\begin{align*}
    \norm{\eta[
        \nabla_{\Bik}\hat g(\pBik)
        -
        \nabla_{\Bik}g(\pBik)
    ]_{\bar\Scal_i}}_F
    =
    \eta\norm{[\ptBik(\hat\Kb_{\cdot\backslash i}-\Kb_{\cdot\backslash i})]_{\bar\Scal_i}}_F,
\end{align*}
where $\ptBik=(-\Bb_{i1}^{\star}\cdots-\Bb_{ip}^{\star})\in\RR^{k\times pk}$.
\end{proof}

\begin{lemma}\label{lemma:I73i}
Under the conditions of Lemma~\ref{lemma:initial_stepB} and define
\begin{multline*}
C_{\delta,i}=C_2\sqrt{(1+\vartheta_2)(1+\vartheta_1)}\\
\times\max_{\mseq}\rho_x^{1/2}\rho_a^{1/2}\{\rho_b^{1/2}\norm{\Sigmab^{m,\ub}}_2^{1/2}+(\rho_a\rho_b)^{1/2}\norm{\Sigmab^{m,\qb}}_2^{1/2}+\norm{\Sigmab_i^\rb}_2^{1/2}+\norm{\Sigmab_i^\wb}_2^{1/2}\},
\end{multline*}where $C_2>0$ is an universal constant, we have
\[
\eta_{B_0}\norm{[
    \nabla_{\Bik}h_i(\pBik)
]_{\bar\Scal_i}}_F\leq
C_{\delta,i}\eta_{B_0}\sqrt{\frac{\alpha^\star s^\star k^2}{N}},
\]
with probability at least $1-\exp\{{-\alpha^\star ks^\star}\}$.
\end{lemma}
\begin{proof}[Proof of Lemma~\ref{lemma:I73i}]
We use similar argument of Lemma~\ref{lemma:staterror_Bij} and write that
\[
[\nabla_{\Bik}h_i(\pBik)]_{\bar\Scal_i}=-\frac{1}{MN}\sbr{\sum_{m=1}^M\sum_{n=1}^N\vb_{i}^{m,(n)}
        (\ykmnfree{\noi}{n})^\top
        \bigAs{p-1}{\pAmkT}}_{\bar\Scal_i},
\]
where $\vb_{i}^{m,(n)}$ is defined in~\eqref{eq:def_v}. 

Define 
\begin{align*}
    \Ucal_1(\Scal)&=\{\ub:
    \supp({\ub})\subset\{\Scal_{\cdot,i},i=1,\ldots k(p-1)\},\norm{\ub}_1=1\},\; \Scal_{\cdot,i}=\{u:(u,i)\in\Scal\};\\
    \Ucal_2(\Scal)&=\{\wb:\supp(\wb)\subset\{\Scal_{i,\cdot},i=1,\ldots,k\},\norm{\wb}_1=1\},\;
    \Scal_{i,\cdot}=\{v;(i,v)\in\Scal\}.
\end{align*}
For notation simplicity, we define $\Ucal_1=\Ucal_1(\bar\Scal_i)$ and $\Ucal_2=\Ucal_2(\bar\Scal_i)$.
 Let $\Ncal_1$ be the $1/8$-net of $\Ucal_1$ and $\Ncal_2$ be the $1/8$-net of $\Ucal_2$.
Then, we can write
\begin{align*}
\norm{[\nabla_{\Bik}h_i(\pBik)]_{\bar\Scal_i}}_F
&\leq\sqrt{k}\norm{[\nabla_{\Bik}h_i(\pBik)]_{\bar\Scal_i}}_2\\
&=\sqrt{k}\sup_{\ub\in\Ucal_1}\sup_{\wb\in\Ucal_2}
\ub^\top\sbr{\frac{1}{N}\sum_{n=1}^N\frac{1}{M}\sum_{m=1}^M\vb_{i}^{m,(n)}
        (\ykmnfree{\noi}{n})^\top
        \bigAs{p-1}{\pAmkT}}_{\bar\Scal_i}\wb.\\
\intertext{Apply Lemma~\ref{lemma:epsilon_net}, we can upper bound the above display as}
&\leq 2\sqrt{k}\sup_{\ub\in\Ncal_1}\sup_{\wb\in\Ncal_2}
\ub^\top\sbr{\frac{1}{N}\sum_{n=1}^N\frac{1}{M}\sum_{m=1}^M\vb_{i}^{m,(n)}
        (\ykmnfree{\noi}{n})^\top
        \bigAs{p-1}{\pAmkT}}_{\bar\Scal_i}\wb,
\end{align*}
where the right hand side can be seen as sum of i.i.d. sub-exponential random variables with Orlicz norm upper bounded by $\max_{\mseq}\rho_x^{1/2}\rho_a^{1/2}\{\rho_b^{1/2}\norm{\Sigmab^{m,\ub}}_2^{1/2}+(\rho_a\rho_b)^{1/2}\norm{\Sigmab^{m,\qb}}_2^{1/2}+\norm{\Sigmab_i^\rb}_2^{1/2}+\norm{\Sigmab_i^\wb}_2^{1/2}\}$, where a similar proof can be found in Lemma~\ref{lemma:staterror_Bij}. 
Define
\begin{multline*}
C_{\delta,i}=C_2\sqrt{(1+\vartheta_2)(1+\vartheta_1)}\\
\times\max_{\mseq}\rho_x^{1/2}\rho_a^{1/2}\{\rho_b^{1/2}\norm{\Sigmab^{m,\ub}}_2^{1/2}+(\rho_a\rho_b)^{1/2}\norm{\Sigmab^{m,\qb}}_2^{1/2}+\norm{\Sigmab_i^\rb}_2^{1/2}+\norm{\Sigmab_i^\wb}_2^{1/2}\},
\end{multline*}
and $C_2>0$. Therefore, apply Bernstein's inequality with $N=O(C_{\delta,i}^2 \alpha^\star s^\star k)$ and take the union bound over all $\Ncal_1$ with $|\Ncal_1|\leq (17)^{(1+\vartheta_1)\alpha^\star k}$ and $\Ncal_2$ with $|\Ncal_2|\leq (17)^{(1+\vartheta_1)(1+\vartheta_2)\alpha^\star k s^\star}$, we have
\begin{align*}
\Ical_{73,i}
&=
\eta_{B_0}\norm{[
    \nabla_{\Bik}h_i(\pBik)
]_{\bar\Scal_i}}_F\leq \eta_{B_0}\sqrt{k}\norm{[
    \nabla_{\Bik}h_i(\pBik)
]_{\bar\Scal_i}}_2\leq
C_{\delta,i}\eta_{B_0}\sqrt{\frac{\alpha^\star s^\star k^2}{N}},
\end{align*}
with probability at least $1-\exp\{{-\alpha^\star ks^\star}\}\geq 1-\max_{\mseq}\exp(-k_m s^\star)$,

\end{proof}
\section{Proof of Theorem~\ref{theorem:main}}
The main theorem combines results from Theorem~\ref{theorem:initialcca} and Lemma~\ref{lemma:initial_stepB} with result from Theorem~\ref{theorem:convergence}.

\emph{Step 1}: The first step is to show that $\Amk^{(0)}$ computed from ~\eqref{eq:cca_model} for $\mseq$ and $\Bik$ outputed by Algorithm~\ref{alg:initializationB} for $\pseq$ satisfy Assumption~\ref{assumption:coef}.
 Since  $\norm{\Amk^{(0)}-\pAmk}_2\leq\norm{\Amk^{(0)}-\pAmk}_F$ and $\sqrt{8k}\opnorm{\Ascr^m}{}^2\opnorm{\Lscr^{m,r}}{}$ is a constant under the condition that $\opnorm{\Lscr^{m,r}}{}=o(1/\sqrt{k_m})$, Theorem~\ref{theorem:initialcca} tells us that there exists a $N=O(\max_{\mseq} k_m + \log M)$, with $M=2$, such that $\norm{\Amk^{(0)}-\pAmk}_2\geq C_A\norm{\pAmk}_2+\sqrt{8k}\opnorm{\Ascr^m}{}^2\opnorm{\Lscr^{m,r}}{}$ with probability smaller than $2^{-1}\exp(-k_1-k_2)$ for each $m=1,2$. 
 Using the condition that
 Taking the union  bound, we have  $\norm{\Amk^{(0)}-\pAmk}_2\geq  C_A\norm{\pAmk}_2$ for $\mseq$ with probability smaller than $\exp(-k_1-k_2)$. We define this event as $\Ecal_{A^0}^c$.

Similarly, since all norms in finite dimensional are equivalent, it suffices to show that $\norm{\Bik^{(0)}-\pBik}_F\leq C_B\norm{\pBik}_F$ for $\pseq$. Under the condition $N=O(\max_{\mseq}s^\star k_m+\log M+\log p)$, apply Lemma~\ref{lemma:initial_stepB} with the number of iteration $L_0=-C_1\log\{(k^2+\log s^\star)\}/{N}\}$  and Theorem~\ref{theorem:initialcca}, we have  $\norm{{\Bik^{(0)}-\pBik}}_F\geq C_B\norm{\pBik}_F$ with probability smaller than $3p^{-1}\delta_0$ for each $\pseq$. Then, taking the union bound over $\pseq$, we have $\norm{\Bik^{(0)}-\pBik}_F\geq C_B\norm{\pBik}_F$ for $\pseq$ with probability smaller than $3\delta_0$. We define such conditional event as $\Ecal_{B^{0}\cap A^{0}}^c$.
Therefore, Assumption~\ref{assumption:coef} holds with probability at least  
\[
1-\cbr{\PP(\Ecal_{B^{0}}^c)+\PP(\Ecal_{A^0}^c)}\geq 
1-\cbr{\PP(\Ecal_{B^{0}\cap A^{0}}^c)+\PP(\Ecal_{A^0}^c)}\geq 1-4\delta_0.
\]

\emph{Step 2}: Under Assumption~\ref{assumption:coef}, we can apply Theorem~\ref{theorem:convergence}. The second part of the theorem can be shown by triangle inequality. Let $L$ be any constant such that $L\geq (-\log(1-\pi)+2\log\Xi-2\log R_0)/\log\pi$. This implies that for any $\pseq$
\begin{align*}
\norm{\Bb_{i,j}^{(L)}-\pBijk}_F
&\leq
\rbr{
\max_{\mseq}\norm{\Ab_{m}^{(L)}-\pAmk}_F^2
+
\max_{\pseq}\norm{\Bb_{i,j}^{(L)}-\pBijk}_F^2}^{1/2}\\
&\leq \sqrt{\frac{2}{1-\pi}}\Xi\leq \frac{1}{2}\Lambda(k)
\end{align*}
Suppose that $j\in\pNi$, then
\[
\norm{\Bb_{i,j}^{(L)}}_F\geq \norm{\pBijk}_F
-
\norm{\Bb_{i,j}^{(L)}-\pBijk}_F
\geq \Lambda(k)-\frac{1}{2}\Lambda(k)
=
\frac{1}{2}\Lambda(k)
,
\]
with probability at least $1-10\delta_0$. Similarly, if $j\in\pNic\noi$,
\[
\norm{\Bb_{i,j}^{(L)}}_F
\leq
\norm{\Bb_{i,j}^{(L)}-\pBijk}_F
\leq  \frac{1}{2}\Lambda(k),
\]
with probability at least $1-10\delta_0$. Therefore, under Assumption~\ref{assumption:disperse}, if we select edges with the threshold $(1/2)\Lambda(k)$, we are able to recover the edge set with probability at least $1-10\delta_0$.
Finally, marginalizing over the conditional probability that Assumption~\ref{assumption:coef} holds, we have
\[
\PP(\hat{\Nscr}_i=\pNi;1\leq i\leq p)\geq (1-10\delta_0)(1-4\delta_0)\geq 1 - 14\delta_0
\]

\section{Strong Convexity and Smoothness}

In this section, we show the convexity coefficient and smoothness coefficient of $f(\cdot)$ and $h_i(\cdot)$. Lemma~\ref{lemma:scsm_coef} states that $f(\cdot,\setpBk)$ is strongly convex and smooth with respect to any $\Amk$ for $\mseq$ with high probability under proper sample size condition. Lemma~\ref{lemma:strong_hn} shows that under mild condition of $\Amk$ for $\mseq$ and sample size, $h_i(\cdot)$ is strongly convex and smooth with respect $\Bik$ with high probability. We being with reviewing the definition of $L$-smoothness and $\mu$-strong convexity.
\begin{definition}
A function $f:\Domain(f)\rightarrow\RR$ is $L$-smooth if for any $\xb,\yb\in\Domain(f)$ we have that
\[
\abr{f(\yb)-f(\xb)-\dotp{\nabla f(\xb)}{\yb-\xb}}\leq \frac{L}{2}\norm{\yb-\xb}_2^2.
\]
\end{definition}

\begin{definition}
A function $f:\Domain(f)\rightarrow\RR$ is $\mu$-strongly convex if for any $\xb,\yb\in\Domain(f)$ we have that
\[
f(\yb)\geq
f(\xb)+\dotp{\nabla f(\xb)}{\yb-\xb}+\frac{\mu}{2}\norm{\yb-\xb}_2^2.
\]
\end{definition}
\begin{lemma}\label{lemma:scsm_coef}
Assume that Assumption~\ref{assumption:cov}--\ref{assumption:AB} hold, 
$\norm{\pBik}_{r(k,0)}\leq s^\star$ for $\pseq$ and
$\defcdnm$. 
If $N=O(\kappa_m^2(k_ms^\star+\log M+\log p))$, then $f(\setAk,\setpBk)$ 
is $\{(s^\star+1) p\nu_{x}\nu_b\}/{2}$-strongly convex 
and $\{3(s^\star+1)p\rho_{x}\rho_b\}/{2}$-smooth with respect to $\Amk$ 
with probability at least $\pkms$.
\end{lemma}
\begin{proof}[Proof of Lemma~\ref{lemma:scsm_coef}]
Define the following two constants:
\begin{align*}
d_1&= ({s^\star+1}) \sum_{i=1}^p \sinmin{\sCovfree{\pNi\cup\{i\}}{\pNi\cup\{i\}}} 
\sinmintwo{\ptBik},\\
d_2 &= ({s^\star+1}) \sum_{i=1}^p \norm{\tilde\Bb_ i^\star}_2^2
    \norm{ \sCovfree{\pNi\cup\{i\}}{\pNi\cup\{i\}}}_2.
\end{align*}
By Lemma~\ref{lemma:f_scl}, we have $f$ be $d_1$-strongly convex. Under the condition $N=O(\kappa_m^2(k_ms^\star+\log M+\log p))$ event $\Ecal_m=\cap_{i=1}^p\Ecal_{m,i}(\pNi\cup\{i\})$ defined in~\eqref{eq:eventm} holds with probability at least $\pkms$. Applying Lemma~\ref{lemma:ci} yields $\sinmin{\sCovfree{\pNi\cup\{i\}}{\pNi\cup\{i\}}}\geq 2^{-1} \sinmin{\Sigmab^m_{\pNi\cup\{i\},\pNi\cup\{i\}}}$ for $\pseq$ with probability at least $\pkms$. Therefore we have
\begin{align*}
({s^\star+1}) \sum_{i=1}^p \sinmin{\sCovfree{\pNi\cup\{i\}}{\pNi\cup\{i\}}} 
\sinmintwo{\ptBik} 
&\geq \frac{s^\star+1}{2} \sum_{i=1}^p \sinmin{\Sigmab^m_{\pNi\cup\{i\},\pNi\cup\{i\}}} 
\sinmintwo{\ptBik},
\end{align*}
with probability at least $\pkms$. 
Apply Assumption~\ref{assumption:cov}--\ref{assumption:AB}, we can obtain a further lower bound and obtain $\mu$:
\[
\frac{s^\star+1}{2} \sum_{i=1}^p \sinmin{\Sigmab^m_{\pNi\cup\{i\},\pNi\cup\{i\}}}
\sinmintwo{\ptBik} \geq \frac{(s^\star+1) p \nu_x\nu_b}{2}.
\]

Next, we want to verify the smoothness condition: $f$ is $d_2$-smooth by Lemma~\ref{lemma:f_scl}.
    Apply Lemma~\ref{lemma:ci} again, with probability at least $\pkms$, we can bound 
\begin{align*}
 ({s^\star+1})\sum_{i=1}^p \norm{\tilde\Bb_ i^\star}_2^2
    \norm{ \sCovfree{\pNi\cup\{i\}}{\pNi\cup\{i\}}}_2 
    &\leq \frac{3 (s^\star+1)}{2} \sum_{i=1}^p \sinmax{\Sigmab^m_{\pNi\cup\{i\},\pNi\cup\{i\}}}
    \norm{\ptBik}_2^2.\\
    \intertext{Finally, apply Assumption~\ref{assumption:cov}--\ref{assumption:AB}, the above display is bounded as}
    &\leq \frac{3(s^\star+1) p\rho_x\rho_b}{2}\norm{\Amk'-\Amk }_F^2,
\end{align*}
with probability at least $\pkms$.
\end{proof}
\begin{lemma}\label{lemma:f_scl}
 $f(\setAk,\setpBk)$ 
is $(s^\star+1)\sum_{i=1}^p\sinmintwo{\pBik}\sinmin{\sCovfree{\pNi\cup\{i\}}{\pNi\cup\{i\}}}$-strongly convex 
and $(s^\star+1)\sum_{i=1}^p\sinmaxtwo{\pBik}\sinmax{\sCovfree{\pNi\cup\{i\}}{\pNi\cup\{i\}}}$-smooth with respect to $\Amk$.
\end{lemma}
\begin{proof}[Proof of Lemma~\ref{lemma:f_scl}]
To show strong convexity of $f(\setAk,\setpBk)$, 
we verify the following condition. For any $\Amk'\in\RR^{k\times k_m}$, we have
\begin{align*}
    f(\Ab',\setpBk)\geq f(\setAk,\setpBk) + \dotp{\nabla_{\Amk }f(\setAk,\setpBk)}{\Amk'-\Amk }+ \frac{\mu}{2}\norm{{\Amk'-\Amk }}_F^2,
\end{align*}
with specific $\mu>0$.
We first recall the definition of $f(\setAk,\setpBk)$ in~\eqref{eq:main_ver2}:
\[
f(\setAk,\setpBk) = \frac{1}{2N} \sum_{i=1}^p \sum_{m=1}^M 
\norm{\ptBik(\Ib_p\otimes\Amk )\Yb_m}_F^2.
\]
Then, we can write
\begin{multline*}
    f(\Ab',\setpBk)- f(\setAk,\setpBk) \\
    = \sum_{i=1}^p \dotp{\ptBikT \ptBik (\Ib_p\otimes \Amk )\sCov}{\Ib_p\otimes(\Amk'-\Amk )}
     + \frac{1}{2N}\sum_{i=1}^p\norm{\ptBik
     \bigAd{p}{\Amk'-\Amk}
     \Yb^m}_F^2.
\end{multline*}
Then, apply~\eqref{eq:dpot_gradAm}, the above equation is equivalent as
\begin{multline}
    f(\Ab',\setpBk)- f(\setAk,\setpBk)\\ 
    = \dotp{\nabla_{\Amk }f(\setAk,\setpBk)}{\Amk'-\Amk } 
    + \frac{1}{2N}\sum_{i=1}^p\norm{\ptBik\bigAd{p}{\Amk'-\Amk }\Yb^m}_F^2.\label{eq:mid1}
\end{multline}
Hence, it suffices to find the lower bound of the second of~\eqref{eq:mid1} 
in terms of $\norm{\Amk'-\Amk }_F^2$. 
Recall the definition of $\pNi$ in ~\eqref{eq:def_pNi2}, and therefore we have $\ptBikfree{\pNic\noi}={\bf 0}$. 
Consequently, we can write
\begin{align*}
\frac{1}{2N} \sum_{i=1}^p \norm{\ptBik \bigAd{p}{\Amk'-\Amk} &\Yb^m}_F^2 \\
 &= \frac{1}{2N} \sum_{i=1}^p 
\norm{\ptBikfree{\pNi\cup\{i\}} \bigAd{|\pNi|+1}{\Amk'-\Amk} \Yb^m_{\pNi\cup \{i\}}}_F^2\\
&\geq \frac{s^\star+1}{2} \sum_{i=1}^p \sinmin{\sCovfree{\pNi\cup\{i\}}{\pNi\cup\{i\}}} 
\sinmintwo{\ptBik} \norm{\Amk'-\Amk}_F^2.
\end{align*}

Next, we want to verify the smoothness condition. We recall the definition that $f(\setAk,\setpBk)$ is $L$-smooth with respect to $\Amk $ if for any $\Amk'$, we have
\begin{align*}
    f(\Ab',\setpBk)\leq f(\setAk,\setpBk)+\dotp{\nabla_{\Amk }f(\setAk,\setpBk)}{\Amk'-\Amk }
    + \frac{L}{2}\norm{{\Amk'-\Amk }}_F^2.
\end{align*}
To find a valid $L$ of the objective function, it suffices to find the upper bound of the second term
in~\eqref{eq:mid1}. Then, using the same proof technique as we find $\mu$, we have
\begin{align*}
\frac{1}{2N} \sum_{i=1}^p \norm{\ptBik \bigAd{p}{\Amk'-\Amk}& \Yb^m}_F^2 \\
&=
    \frac{1}{2N} \sum_{i=1}^p 
    \norm{\ptBikfree{\pNi\cup\{i\}} \bigAd{|\pNi|+1}{\Amk'-\Amk} \Yb^m_{\pNi\cup\{i\}}}_F^2 \\
    &\leq \frac{s^\star+1}{2} \sum_{i=1}^p \norm{\tilde\Bb_ i^\star}_2^2
    \norm{ \sCovfree{\pNi\cup\{i\}}{\pNi\cup\{i\}}}_2 \norm{\Amk'-\Amk}_F^2.
\end{align*}
We hence complete the proof
\end{proof}

\begin{lemma}\label{lemma:strong_hn} Assume that Assumption~\ref{assumption:cov} holds and $N= O(\max_{\mseq}\kappa_m^2(k_ms^\star+\log M))$, where $\defcdnm$. Then, 
     $h_i(\setBk)$ is $\min_{\mseq}2^{-1}\nu_x\sinmintwo{\Amk}$-strongly convex with respect to $\Bik$ with probability at least $1-\max_{\mseq}\exp(-k_m s^\star)$.
\end{lemma}
\begin{proof}[Proof of Lemma~\ref{lemma:strong_hn}]
We can write $h_i(\Bik+\Delta)$ as
\[
    h_i(\Bik+\Delta) 
    = 
        h_i(\Bik) 
        +
        2\dotp{\nabla_{\Bik}h_i(\Bik)}{\Delta}
        +
        \frac{1}{2N}\sum_{m=1}^M\frac{1}{M}\norm{
            \Delta
            \bigAs{p-1}{\Amk}
            \Ykmnfree{\noi}}_F^2.
\]
We can lower bound the last term of the above display as
\begin{align*}
    \min_{\mseq}\frac{1}{2N}&\norm{
        \Delta
        \bigAs{p-1}{\Amk}
        \Ykmnfree{\noi}}_F^2.\\
    &=     
        \min_{\mseq}\frac{1}{2N}\norm{
        \Delta
        \bigAs{|\pNi|}{\Amk}
        \Ykmnfree{\pNi}}_F^2
        +
        \frac{1}{2N}\norm{
            \Delta
            \bigAs{|\pNic|-1}{\Amk}
            \Ykmnfree{\pNic\backslash\{i\}}
        }_F^2\\
        &\geq\min_{\mseq}
        \frac{1}{2N}\norm{
        \Delta
        \bigAs{|\pNi|}{\Amk}
        \Ykmnfree{\pNi}}_F^2\\
        &\geq\min_{\mseq}
        \frac{1}{2}
        \sinmintwo{\Amk}
        \sinmin{\sCovfree{\pNi}{\pNi}}\norm{\Delta}_F^2\\
        \intertext{Under the condition $N= O(\max_{\mseq}\kappa_m^2(k_ms^\star+\log M))$ and Assumption~\ref{assumption:cov}, we can apply Lemma~\ref{lemma:ci} and obtain that $\min_{\mseq}\sinmin{\sCovfree{\pNi}{\pNi}}\geq\nu_x/2$ with probability at least $1-\max_{\mseq}\exp(-k_ms^\star)$. Then, the above display is upper bounded as}
        &\geq\min_{\mseq}
        \frac{\nu_x}{2}
        \sinmintwo{\Amk}
        \norm{\Delta}_F^2,
\end{align*}
with probability at least $1-\max_{\mseq}\exp(-k_ms^\star)$.
Here we complete the proof.
\end{proof}
\section{Projection coefficients}
In this section, we show the expansion coefficients resulted from projection to non-convex sets. Lemma~\ref{lemma:proj_Km} discusses coefficient of the general condition of projection to $\Kcal_m(\tau_1,\tau_2)$ and Lemma~\ref{lemma:disperse_coef}--\ref{lemma:groupsparse} discuss the coefficients of projection to $\Kcal_B(s,\alpha)$.
\begin{lemma}\label{lemma:proj_Km}
Let $\Kcal(\tau_1,\tau_2)=\{\Xb\in\RR^{p\times d}:\tau_1\leq\norm{\Xb_{i,\cdot}}_2\leq\tau_2,i=1,\ldots, p\}$ and $\Pcal_\tau$ be the projection operator that projects $\Xb$ to $\Kcal(\tau_1,\tau_2)$. Then, for any $\Yb\in\Kcal(\tau_1,\tau_2)$, we have
\[
\norm{\Pcal_\tau(\Xb)-\Yb}_F \leq 2\norm{\Xb -\Yb}_F.
\]
\end{lemma}
\begin{proof}
First, we consider that there is only single row in $\Xb$ and hence we can directly apply the result of Lemma~\ref{lemma:proj_Km_vec}.
The generalization to $p$ rows, with $p>1$, follows similarly:
\begin{align*}
    \norm{\Pcal_\tau(\Xb)-\Yb}_F^2
    =\sum_{j=1}^p\norm{\tilde\Pcal_\tau(\xb_j)-\yb_j}_2^2\leq 4\sum_{j=1}^p\norm{\xb_j-\yb_j}_2^2=4\norm{\Xb-\Yb}_F^2,
\end{align*}
where the last inequality follows by Lemma~\ref{lemma:proj_Km_vec}.
\end{proof}
\begin{lemma}\label{lemma:proj_Km_vec}
Let $\tilde\Kcal(\tau_1,\tau_2)=\{\xb\in\RR^{d}:\tau_1\leq\norm{\xb}_2\leq\tau_2\}$ and $\tilde\Pcal_\tau$ be the projection operator that projects $\xb$ to $\tilde\Kcal(\tau_1,\tau_2)$. Then, for any $\yb\in\tilde\Kcal(\tau_1,\tau_2)$, we have
\[
\norm{\tilde\Pcal_\tau(\xb)-\yb}_F \leq 2\norm{\xb -\yb}_F.
\]
\end{lemma}
\begin{proof}[Proof of Lemma~\ref{lemma:proj_Km_vec}]
We consider the following three cases:\\
\emph{Case 1}:
If $\norm{\xb}_2\geq \tau_2$, then projection to $\tilde\Kcal(\tau_1,\tau_2)$ is equivalent as the projection to a ball with radius $\tau_2$, which is a convex set. Then, following the property that projection to the convex set is contraction, we have $\|\tilde\Pcal_\tau(\xb)-\yb\|_F\leq \norm{\xb-\yb}_F$. \\
\emph{Case 2}: If $\tau_1\leq\norm{\xb}_2\leq\tau_2$, then it is clear that $\tilde\Pcal_\tau(\xb)=\xb$ and hence $\|\tilde\Pcal_\tau(\xb)-\yb\|_F=\norm{\xb-\yb}_F$.\\
\emph{Case 3}: If $\norm{\xb}_2\leq\tau_1$, then define $\hat{\xb}=\tilde\Pcal_\tau(\xb)$ we can apply triangle inequality and obtain 
\begin{align*}
    \norm{\hat\xb-\yb}_2
    \leq
    \norm{\xb - \yb}_2
    +
    \norm{\hat\xb - \xb}_2
    \leq 2\norm{\xb-\yb}_2,
\end{align*}
where the last inequality follows by the fact that $\norm{\hat\xb - \xb}_2\leq \norm{\xb-\yb}_2$. Therefore, we have $\norm{\tilde\Pcal_\tau(\xb)-\yb}_F\leq 2\norm{\xb - \yb}_F$. 

Taking the maximum coefficient of three cases, we complete the proof.
\end{proof}

\begin{lemma}\label{lemma:disperse_coef}[Lemma $B.3$ in~\citep{zhang2018unified}] Let $\Bb^\star\in\RR^{k\times k}$ and suppose that there are at most $\alpha^\star$-fraction of non-zero entries per row and per column of $\Bb^\star$. Then, for any $\Bb$ and $1\geq\alpha\geq\alpha^\star$  and $\Hcal_{\alpha}(\cdot)$ be the hard thresholding operator defined in Section~\ref{sec:algo}, we have
    \[
    \norm{\Hcal_{\alpha}(\Bb)-\Bb^\star}_F\leq\rbr{1+\sqrt{\frac{2\alpha^\star}{\alpha-\alpha^\star}}}\norm{\Bb-\Bb^\star}_F.    
    \]
\end{lemma}
\begin{lemma}[Group-sparse hard thresholding]\label{lemma:groupsparse} Let $s,s^\star$ be some integers such that $s>s^\star$, $\Bb,\Bb^\star\in\RR^{k\times pk}$ and $\Bb^\star$ is $s^\star$ group-sparse, i.e., $\norm{\Bb^\star}_{r(k,0)}\leq s^\star$. Recall $\Tcal_s(\cdot)$ defined in Section~\ref{sec:algo}, we have
\[
\|\Tcal_{s}(\Bb)-\Bb^\star\|_F\leq \rbr{1+\sqrt{\frac{s^\star}{s-s^\star}}}\norm{\Bb-\Bb^\star}_F.
\]
\end{lemma}
\begin{proof} First, let $\hat{\Bb}=\Tcal_s(\Bb)$, $\supp(\hat{\Bb})=\hat S$, $\supp(\Bb^\star)=S^\star$ and $I=\hat S\cup S^\star$. We have
\[
\norm{\hat{\Bb}- \Bb^\star}_F = \norm{[\hat\Bb-\Bb^\star]_I}_F\leq \norm{[\hat\Bb-\Bb]_I}_F+\norm{[\Bb-\Bb^\star]_I}_F.
\]
Let $U=S^\star/\hat S$ and $V = \hat S/S^\star$, then we have
\[
\norm{[\hat{\Bb}-\Bb]_I}_F^2=\norm{[\Bb]_U}_F^2\leq\frac{s^\star}{s-s^\star}\norm{[\Bb]_V}_F^2\leq \frac{s^\star}{s-s^\star}\norm{[\Bb-\Bb^\star]_I}_F^2.
\]
Therefore, we have
\[
\norm{\hat{\Bb}-\Bb^\star}_F\leq\rbr{1+\sqrt{\frac{s^\star}{s-s^\star}}}\norm{[\Bb-\Bb^\star]_I}_F\leq \rbr{1+\sqrt{\frac{s^\star}{s-s^\star}}}\norm{\Bb-\Bb^\star}_F.
\]
Wwe complete the analysis.
\end{proof}
\section{Auxiliary Lemmas}
This section introduces several useful properties that are used in analyses in previous sections.
\begin{lemma}\label{lemma:ci} 
 Assume that $(\xb^{(1)},\yb^{(1)}),\ldots,(\xb^{(N)},\yb^{(N)})$ are $N$ independent realizations of $(\xb,\yb)\in(\RR^p,\RR^{d})$, which are distributed as $\xb\sim\Ncal({\bf 0},\Sigmab_x)$ and $\yb\sim\Ncal({\bf 0},\Sigmab_y)$. Let $\bar\xb=N^{-1}\sum_{n=1}^N\xb^{(n)}$,  $\bar\yb=N^{-1}\sum_{n=1}^N\yb^{(n)}$ and $\hat\Sigmab_{xy} = N^{-1}\sum_{n=1}^N(\xb^{(n)}-\bar\xb)(\yb^{(n)}-\bar\yb)^\top$ be the sample covariance.
 Given $\delta>0$, there exist an universal constant $C_0>0$ such that
\[
\PP\cbr{
    \frac{\norm{
        \hat\Sigmab_{xy}
        -
        \EE\hat\Sigmab_{xy}}_2
        }{
        \norm{\EE\hat\Sigmab_{xy}}_2
        }\geq C_0\rbr{\sqrt{\frac{p+d+\log(1/\delta)}{N}} + \frac{p+d+\log(1/\delta)}{N}} }\leq\delta.
\]
Suppose that there exist $\sigma_{p\wedge d}, \sigma_1>0$ such that $\sigma_{p\wedge d}\leq\sigma_{min}(\EE\Sigmab_{xy})\leq\sigma_{max}(\EE\Sigmab_{xy})\leq\sigma_1$. If $N=O(\kappa^2 (p+d))$ with $\kappa=\sigma_1/\sigma_{p\wedge d}$, then
\[
\PP\cbr{
    \frac{\sigma_{p\wedge d}}{2}
    \leq
    \sinmin{\hat\Sigmab_{xy}}
    \leq
    \sinmax{\hat\Sigmab_{xy}}
    \leq
    \frac{3\sigma_1}{2}}\geq 1-\exp(-p-d).
\]
\end{lemma}
\begin{proof}
The first part of the proof can be found in Theorem~6.1.1 in~\citet{tropp2015introduction}, so here we only present the second part of the proof. Note that when $N=O(\kappa^2(p+d))$, we have $\norm{\hat\Sigmab_{xy}-\EE\hat\Sigmab_{xy}}_{2}
\leq
\frac{\sigma_{p\wedge d}}{2}$ with probability at least $1-\exp\{-(p+d)\}$. By Weyl's inequality, we have
\[
\sinmax{\hat\Sigmab_{xy}}
\leq
\sinmax{\EE\hat\Sigmab_{xy}}
+
\norm{
    \hat\Sigmab_{xy}
    -
    \EE\hat\Sigmab_{xy}}_{2}
\leq\frac{3\sigma_1}{2},
\]
and 
\[
\sinmin{\hat\Sigmab_{xy}}
\geq
\sinmin{\EE\hat\Sigmab_{xy}}
-
\norm{\hat\Sigmab_{xy}-\EE\hat\Sigmab_{xy}}_{2}
\geq\frac{\sigma_{p\wedge d}}{2},
\]
and hence we complete the second part of the proof.
\end{proof}
 \begin{lemma}
\label{lem:covariance:concentration:neighborhoods}
Let $\Yb\in\RR^{pd\times N}$ be a matrix whose columns are i.i.d. random vectors drawn from a distribution. Define $\hat\Sigmab=\frac{1}{N}\Yb(\Yb)^\top\in\RR^{pd\times pd }$,
$\Sigmab=\EE[\hat\Sigmab]\in\RR^{pd\times pd }$.
Let $\Mscr_i\subset\{1,\ldots,p\}$, where $\abr{\Mscr_i}\leq r$. Let $\hat\Sigmab_{\Mscr_i \Mscr_i}
    \in
        \RR^{
            \abr{\Mscr_i}d
            \times 
            \abr{\Mscr_i}d
            }$
be the sub-matrix of $\hat\Sigmab$ whose rows and columns are corresponding to nodes in $\Mscr_i$. Define the event
\begin{align*}
{\Ecal_{i}}&(\Mscr_i)
    &=\cbr{ 
        \frac{\sinmin{\Sigmab_{ \Mscr_i\Mscr_i}}}{2}
        \leq
        \sinmin{\hat\Sigmab_{\Mscr_i \Mscr_i} } 
        \leq 
        \sinmax{\hat\Sigmab_{\Mscr_i \Mscr_i} } 
        \leq 
        \frac{3\sinmax{\Sigmab_{\Mscr_i\Mscr_i}}}{2}
    }.
\end{align*}
Let $\kappa = \sinmax{\Sigmab}/\sinmin{\Sigmab}$. Under the condition that $N=O(\kappa^2 (dr+\log M + \log p))$, the event $
{\Ecal}=\cap_{i=1}^p{\Ecal_{i}}(\Mscr_i)
$ happens
with probability at least $1-M^{-1}\exp\{-dr\}$.

\end{lemma}
\begin{proof}
 Then, by Lemma~\ref{lemma:ci}, under the condition that $N=O(\kappa^2 (dr+\log M + \log p))$, we have
\[
\PP\rbr{\opnorm{
    \hat\Sigmab_{{\Mscr_i}{\Mscr_i}} 
    -
    \Sigmab_{\Mscr_i\Mscr_i}
    }{2}
    >
    \frac{
\sinmin{\Sigmab_{\Mscr_i\Mscr_i}}}{2}}
    \leq \frac{1}{pM}\exp(-dr),
\]

and we will be on the event
${\Ecal_{i}}(\Mscr_i)$
with probability at least $1-(pM)^{-1}\exp\rbr{-dr}$. Then, we define the intersection of all events ${\Ecal_{i}}(\Mscr_i)$ for $\pseq$ as
\begin{align*}
{\Ecal}=\cap_{i=1}^p{\Ecal_i}(\Mscr_i).
\end{align*}
Then, by De Morgan's laws, event ${\Ecal}$ happens
with probability at least $1-M^{-1}\exp(-dr)$.

\end{proof}
\begin{lemma}[Theorem~$2.1.12$ in~\citet{nesterov2003introductory}]\label{lemma:convexbound}
For a $L$-smooth and $m$-strongly convex function $h$, we have
\begin{align*}
    \langle\nabla h_i(X)-\nabla h_i(Y),X-Y\rangle\geq\frac{m L}{m+L}\|X-Y\|_F^2+\frac{1}{m+L}\|\nabla h_i(X)-\nabla h_i(Y)\|_F^2.
\end{align*}
\end{lemma}
\begin{lemma}[Theorem~$6.5$ in~\citet{wainwright2019high}]\label{lemma:epsilon_net} Let $\Xb\in\RR^{p\times d}$ and $\Ncal_1=\{\vb_1,\ldots,\vb_n\}$ be the $1/8$-net of $\mathbb{S}^{p-1}$ and $\Ncal_2 =\{\ub_1,\ldots,\ub_{m}\}$ be the $1/8$-net of $\mathbb{S}^{d-1}$. Then,
\[
\sup_{\vb\in \mathbb{S}^{p-1}}\sup_{\ub\in \mathbb{S}^{d-1}}\abr{\dotp{\vb}{\Ab\ub}}
\leq 2
\max_{\vb\in \Ncal_1}\max_{\ub\in \Ncal_2}\abr{\dotp{\vb}{\Ab\ub}}.
\]
\end{lemma}

\begin{lemma}\label{lemma:matrixdualnorm}
Let $\frac{1}{p}+\frac{1}{q}=1$ and $\frac{1}{m}+\frac{1}{n}=1$. Then the dual norm of $\norm{\cdot}_{p,m}$ is $\norm{\cdot}_{q,n}$.
\end{lemma}
\begin{proof}

Assume that we have two matrices $\Xb,\Yb\in\RR^{p\times d}$ and recall the definition of the dual norm
\[
\norm{\Xb}_*=\sup\{\tr(\Yb^\top\Xb)|\norm{\Yb}\leq 1\}.
\]
Then, it is easy to see that
\[
\tr(\Yb^\top\Xb)\leq\norm{\Yb}\norm{\Xb}_*,
\]
for every pair of $\Xb$, $\Yb$. 

Let $\xb_1,\ldots,\xb_d$ be columns of $\Xb$ and $\yb_1,\ldots,\yb_d$ be columns of $\Yb$. 
Apply H\"older's inequality, we have
\begin{align*}
    \tr(\Xb^\top\Yb) &= \sum_{j=1}^d\dotp{\xb_j}{\yb_j}\\
    &\leq \sum_{j=1}^d\norm{\xb_j}_p\norm{\yb_j}_q.
    \intertext{Apply the H\"older's inequality again, we can obtain}
    &\leq\rbr{\sum_{j=1}^d\norm{\xb_j}_p^m}^{\frac{1}{m}}\rbr{\sum_{j=1}^d\norm{\yb_j}_q^n}^{\frac{1}{n}}\\
    &=\norm{\Xb}_{p,m}\norm{\Yb}_{q,n}.
\end{align*}
The second step is to show that the upper bound is achievable. Let $y_{kj}$ be the $kj$-th entry of $\Yb$ such that
\[
y_{kj}=\frac{\sgn(x_{kj})|x_{kj}|^{p-1}}{\norm{\xb_j}_p^{p-m}(\sum_{j'=1}^d\norm{\xb_{j'}}_p^m)^{1/n}},\qquad k=1,\ldots p, j=1,\ldots,d.
\]
We can easily verify that $\norm{\Yb}_{q,n}=1$. Note that we have
\[
\norm{\yb_j}_q=\frac{\norm{\xb_j}_p^{q/p}}{\norm{\xb_j}_p^{p-m}(\sum_{j'=1}^d\norm{\xb_{j'}}_p^m)^{1/n}}
=\frac{\norm{\xb_j}^{m-1}_p}{(\sum_{j'=1}^d\norm{\xb_{j'}}_p^m)^{1/n}},\qquad j=1,\ldots,d.
\]
Summing over $d$ columns, we then conclude that
\[
\norm{\Yb}_{q,n}^n=\sum_{j=1}^d\norm{\yb_j}_q^n
=
\frac{\sum_{j=1}^d\norm{\xb_{j}}_p^{n(m-1)}}{\sum_{j'=1}^d\norm{\xb_{j'}}_p^m}=1.
\]
Moreover, we can verify that
\begin{align*}
    \tr(\Xb^\top\Yb)
    =
    \sum_{j=1}^d\dotp{\xb_j}{\yb_j}
    &=
    \sum_{j=1}^d\frac{\norm{\xb_{j}}_p^{p}}{\norm{\xb_j}_p^{p-m}(\sum_{j'=1}^d\norm{\xb_{j'}}_p^m)^{1/n}}\\
    &=\frac{\sum_{j=1}^d\norm{\xb_j}_p^m}{(\sum_{j'=1}^d\norm{\xb_{j'}}_p^m)^{1/n}}
    =\rbr{\sum_{j=1}^d\norm{\xb_j}_p^m}^{1/m}=\norm{\Xb}_{p,m}.
\end{align*}
Here we complete the proof.
\end{proof}

\begin{lemma}\label{lemma:dualnormrpq}
Let $\frac{1}{p}+\frac{1}{q}=1$ and recall the definition of $\norm{\cdot}_{r(k,p)}$ and $\norm{\cdot}_{r(k,q)}$ in Definition~\ref{def:rnorm}. Then the dual norm of $\norm{\cdot}_{r(k,p)}$ is $\norm{\cdot}_{r(k,q)}$.
\end{lemma}
\begin{proof}
The proof is similar to proof of Lemma~\ref{lemma:matrixdualnorm}.
Assume that $\Xb,\Yb\in\RR^{d\times ks}$ and are represented as $\Xb=(\Xb_1\cdots\Xb_s)$ and $\Yb=(\Yb_1\cdots\Yb_s)$ for $\Xb_i\in\RR^{d\times k}$ and $\Yb_i\in\RR^{d\times k}$ for $i=2,\ldots,s$. Then, apply H\"older's inequality, we can write
\begin{align*}
   \tr(\Xb^\top\Yb)&= \sum_{i=1}^m\tr(\Xb_i^\top\Yb_i)\\
   &=\sum_{i=1}^m\VEC\rbr{\Xb_i}^\top\VEC\rbr{\Yb_i} \\
   &\leq \sum_{i=1}^m\norm{\VEC\rbr{\Xb_i}}_p\norm{\VEC\rbr{\Yb_i}}_q.\\
   \intertext{Recall the definition of $\norm{\cdot}_{r(k,p)}$ and $\norm{\cdot}_{r(k,q)}$ in Definition~\ref{def:rnorm}, where we can apply Lemma~\ref{lemma:matrixdualnorm} with $m=n=2$ and arrive at the following conclusion}
   &\leq\norm{\Xb}_{r(k,p)}\norm{\Yb}_{r(k,q)}.
\end{align*}
\end{proof}
\begin{lemma}\label{lemma:projection_hs}
Let $\Ascr,\Bscr\in\mathfrak{B}(\HH_1,\HH_2)$ be two compact operators, where $\HH_1,\HH_2$ are separable Hilbert spaces.
Let $\Pscr_A$ be the projection operator to $\overline{\Image(\Ascr)}$ and $\Pscr_B$ be the projection operator to $\overline{\Image(\Bscr)}$. If $\rk(\Ascr)>\rk(\Bscr)$, then
\[
\opnorm{\Pscr_B^\perp\Pscr_A}{}\geq
\opnorm{\Pscr_B\Pscr_A^\perp}{}.
\]
\end{lemma}
Lemma~\ref{lemma:projection_hs} similarly holds for the Hilbert-Schimidt norm.
Since we only use the result for the operator norm, we do not provide the result for the Hilbert-Schimidt norm.
\begin{proof}
We start by decomposing the projection operator as $\Pscr_A=\Pscr_{A_1}+\Pscr_{A_2}$, where $\rk(\Pscr_{A_1})=\rk(\Pscr_{B})$ and $\Pscr_{A_2}\Pscr_B=0$. Then
\begin{align*}
\opnorm{\Pscr_B\Pscr_A^\perp}{}
&=\opnorm{\Pscr_{B}(I-\Pscr_{A_1}-\Pscr_{A_2})}{}
= \opnorm{\Pscr_{B}(I-\Pscr_{A_1})}{}
=\opnorm{\Pscr_B\Pscr_{A_1}^\perp}{}
\intertext{Since $\rk(\Pscr_{A_1})=\rk(\Pscr_{B})$, we have 
$\opnorm{\Pscr_B\Pscr_{A_1}^\perp}{}=\opnorm{\Pscr_B^\perp\Pscr_{A_1}}{}$~\citep{davis1970rotation}. Then, the above display is equivalent to}
&=\opnorm{\Pscr_B^\perp\Pscr_{A_1}}{}
\leq\opnorm{\Pscr_B^\perp\Pscr_A}{},
\end{align*}
which completes the proof.
\end{proof}
\begin{lemma}\label{lemma:pseuoinv_hs}
Let $\Ascr,\Bscr\in\mathfrak{B}(\HH_1,\HH_2)$ be two compact operators, where $\HH_1,\HH_2$ are separable Hilbert spaces.
Let $\Pscr_A$ be the projection operator to $\overline{\Image(\Ascr)}$ and $\Pscr_B$ be the projection operator to $\overline{\Image(\Bscr)}$. If $\rk(\Ascr)\geq\rk(\Bscr)$, then 
\begin{equation*}
\opnorm{\Ascr^\dagger-\Bscr^\dagger}{}\leq 2\rbr{\opnorm{\Ascr^\dagger}{}^2\vee\opnorm{\Bscr^\dagger}{}^2}\opnorm{\Ascr-\Bscr}{}.
\end{equation*}
\end{lemma}
Lemma~\ref{lemma:pseuoinv_hs} similarly holds for the Hilbert-Schmidt norm with constant 2 replaced by $2^{1/2}$. We only prove the result for the operator norm.
\begin{proof}
We generalize Theorem~$3.3$ in~\citet{stewart1977perturbation} to a Hilbert space. Let $\Qscr_A$ be the projection operator to $\ker(\Ascr)$ and $\Pscr_A$ be the orthogonal projection to $\overline{\Image(\Ascr)}$. Recalling the properties of a pseudo-inverse, we have $\Pscr_{A}=\Ascr\Ascr^\dagger$ and 
define $\Rscr_{A} =\Ascr^\dagger\Ascr =I-\Qscr_A$. We define $\Pscr_B$ and $\Rscr_B$ in the same way.
Then, we can write
\begin{align*}
    \Bscr^\dagger &= \Bscr^\dagger\Bscr\Bscr^\dagger\\
    &=\Bscr^\dagger\Pscr_B\\
    &=\Bscr^\dagger\Pscr_B(\Ascr\Ascr^\dagger+I-\Ascr\Ascr^\dagger)\\
    &=\Bscr^\dagger\Pscr_B(\Ascr\Ascr^\dagger\Ascr\Ascr^\dagger+I-\Ascr\Ascr^\dagger)\\
    &=\Bscr^\dagger\Pscr_B(\Ascr\Rscr_A\Ascr^\dagger+\Pscr_A^\perp).
\end{align*}
Similarly, write
\begin{align*}
    \Ascr^\dagger &= \Ascr^\dagger\Ascr\Ascr^\dagger\\
    &=\Rscr_A\Ascr^\dagger\\
    &=(\Bscr^\dagger\Bscr+I-\Bscr^\dagger\Bscr)\Rscr_A\Ascr^\dagger\\
    &=(\Bscr^\dagger\Bscr\Bscr^\dagger\Bscr
    +
    I-\Bscr^\dagger\Bscr
    )\Rscr_A\Ascr^\dagger\\
    &=(\Bscr^\dagger\Pscr_B\Bscr+\Rscr_B^\perp)\Rscr_A\Ascr^\dagger.
\end{align*}
Therefore, we have
\begin{align*}
\Bscr^\dagger - \Ascr^\dagger =
-\Bscr^\dagger\Pscr_B(\Bscr-\Ascr)\Rscr_A\Ascr^\dagger
+
\Bscr^\dagger\Pscr_B\Pscr_A^{\perp}
-
\Rscr_B^\perp\Rscr_A\Ascr^\dagger
= T_1+T_2+T_3.
\end{align*}
Note that $\overline{\Image(\Bscr^\dagger)}$ is orthogonal to the null space $\ker(\Bscr)$ and we have
 $\Image(T_1),\Image(T_2)\in\overline{\Image(\Bscr^\dagger)} $ and $\Image(T_3)\in\ker(\Bscr)$.
We have
\[
\opnorm{\Bscr^\dagger - \Ascr^\dagger}{}\leq\opnorm{T_1+T_2}{}+\opnorm{T_3}{}.
\]
Since $T_1+T_2=\Bscr^\dagger\cbr{-\Pscr_B(\Bscr-\Ascr)\Ascr^\dagger\Pscr_A+\Pscr_B\Pscr_A^\perp}$, we have
\begin{align*}
    \opnorm{T_1+T_2}{}
    &\leq\opnorm{\Bscr^\dagger}{}\rbr{\opnorm{\Pscr_B(\Bscr-\Ascr)\Ascr^\dagger}{}+\opnorm{\Pscr_B\Pscr_A^\perp}{}}
    \notag \\
    &\leq\opnorm{\Bscr^\dagger}{}\rbr{\opnorm{\Pscr_B(\Bscr-\Ascr)\Ascr^\dagger}{}+\opnorm{\Pscr_B^\perp\Pscr_A}{}}
    && (\text{Lemma~\ref{lemma:projection_hs}}) \notag\\
    &=\opnorm{\Bscr^\dagger}{}\cbr{\opnorm{\Pscr_B(\Bscr-\Ascr)\Ascr^\dagger}{}+\opnorm{\Pscr_B^\perp(\Bscr-\Ascr)\Ascr^\dagger}{}} && (\Pscr_B^\perp\Bscr=0) \notag\\
    &\leq \opnorm{\Bscr^\dagger}{}\opnorm{\Ascr^\dagger}{}
    \cbr{
        \opnorm{\Pscr_B(\Bscr-\Ascr)}{}
        +
        \opnorm{\Pscr_B^\perp(\Bscr-\Ascr)}{}
    }\notag\\
    &=\opnorm{\Bscr^\dagger}{}\opnorm{\Ascr^\dagger}{}\opnorm{\Bscr-\Ascr}{}.
\end{align*}
For $T_3$, since $\opnorm{\Rscr_B^\perp\Rscr_A}{}=\opnorm{(\Rscr_B^\perp\Rscr_A)^*}{}=\opnorm{\Rscr_A^*(\Rscr_B^\perp)^*}{}=\opnorm{\Rscr_A\Rscr_B^\perp}{}$, we have
\begin{align*}
\opnorm{\Rscr_B^\perp\Rscr_A\Ascr^\dagger}{}
    &\leq \opnorm{\Ascr^\dagger}{}\opnorm{\Rscr_A\Rscr_B^\perp}{}\notag\\
    &=\opnorm{\Ascr^\dagger}{}\opnorm{\Ascr^\dagger(\Ascr-\Bscr)\Rscr_B^\perp}{} && (\Bscr\Rscr_B^\perp=0) \notag \\
    &\leq\opnorm{\Ascr^\dagger}{}^2\opnorm{\Bscr-\Ascr}{}.
\end{align*}
Combining the last two displays, we have
\[
\opnorm{\Ascr^\dagger-\Bscr^\dagger}{}\leq 2\rbr{\opnorm{\Ascr^\dagger}{}^2\vee\opnorm{\Bscr^\dagger}{}^2}\opnorm{\Ascr-\Bscr}{},
\]
which completes the proof.
\end{proof}
\begin{lemma}\label{lemma:perturb_sinvec}[Adapted from Theorem~5.2.2 in~\citet{hsing2015theoretical} Let $\Tscr, \tilde{\Tscr}\in\mathfrak{B}(\HH_1,\HH_2)$ be two compact operators associated with singular systems $\{\lambda_j,f_{1j},f_{2j}\}_{j\in\NN}$ and
$\{\tilde\lambda_j,\tilde{f}_{1j},\tilde{f}_{2j}\}_{j\in\NN}$, respectively. Without the loss of generality, we assume that $\{f_{1j}\}_{j\in\NN}$ and $\{f_{2j}\}_{j\in\NN}$ provide CONS for $\HH_1$ and $\HH_2$, respectively. Define $\Delta=\tilde\Tscr-\Tscr$ and $\zeta_j=(2^{-1}\min_{k\neq j} |\lambda_k^2-\lambda_j^2|)^{-1}(\opnorm{\Tscr}{}+\opnorm{\Delta}{})\opnorm{\Delta}{}$. Assume that $\zeta_j<1-\epsilon$ for any number $\epsilon\in(0,1)$. Then, there exists a constant $C=C_\epsilon\in(0,\infty)$ such that
\[
\bignorm{(\tilde{f}_{1j} - f_{1j})-\sum_{k\neq j}\frac{\lambda_k\dotp{f_{2k}}{\Delta f_{1j}}_{\HH_2}+\lambda_j\dotp{f_{2j}}{\Delta f_{1k}}_{\HH_2}}{\lambda_k^2-\lambda_j^2} f_{1k}}_{\HH_1}\leq C{\gamma_j},
\]
with $\gamma_j=\zeta_j\{\zeta_j+\opnorm{\Delta}{}/(\opnorm{\Delta}{}+\opnorm{\Tscr}{})\}$.
\end{lemma}
\begin{lemma}\label{lemma:difference_sqrt}
Let $\Ab$ and $\Bb$ be $d \times d$ real symmetric matrices such that $ 0\prec\Ab,\Bb\preceq\lambda\Ib$ for some constant $\lambda>0$, then we have
\[
\norm{\Ab^{-1/2}-\Bb^{-1/2}}_2\leq 
\rbr{3\lambda^{1/2}\norm{\Ab^{-1/2}}_2
+
1
}\norm{\Bb^{-3/2}}_2\norm{\Ab-\Bb}_2.
\]
\end{lemma}
\begin{proof}
First, we write
\begin{align*}
    \norm{\Ab^{-1/2}-\Bb^{-1/2}}_2
    &=
    \norm{\Ab^{-1/2}(\Bb^{3/2}-\Ab^{3/2})\Bb^{-3/2}
    +
    (\Ab - \Bb)\Bb^{-3/2}}_2\\
    &\leq 
    \norm{\Ab^{-1/2}}_2\norm{\Bb^{-3/2}}_2\norm{\Bb^{3/2}-\Ab^{3/2}}_2
    +
    \norm{\Bb^{-3/2}}_2\norm{\Ab - \Bb}_2.
\end{align*}
Apply Lemma~8 in~\citet{fukumizu2007statistical}, we could bound $\norm{\Bb^{3/2}-\Ab^{3/2}}_2\leq 3\lambda^{1/2}\norm{\Ab-\Bb}_2$. Then we complete the proof.
\end{proof}

\section{Discussion of the Transformation Operator}
In this section, we discuss an alternative initialization approach as compared to the approach proposed in Section~\ref{ssec:initialization}. We aggregate the samples across nodes $i\in V$ to initiate $\Amk$. The detail is documented in Appendix~\ref{sec:appendix:aggregate}. In Appendix~\ref{ssec:extension}, we discuss the construction of $\Amk$ varied across nodes. In Appendix~\ref{ssec:extension}, we discuss the generalization of the initialization~\ref{eq:cca_model} to $M\geq 2$ modalities.
\subsection{Aggregation of Samples Across Nodes}\label{sec:appendix:aggregate}
The initialization procedure introduced in Section~\ref{ssec:initialization} only use the sample from node $1$. As an alternative, one could aggregate the samples across nodes to conduct joint canonical correlation analysis. That is, we compute the singular decomposition of
\[
\hat\Rb_{12}^\sharp = \frac{1}{p}\sum_{i=1}^p (\hat\Sigmab_{i,i}^1)^{-1/2}
\hat\Sigmab_{i,i}^{12}(\hat\Sigmab_{i,i}^2)^{-1/2}.
\]
We denote $\hat\Vb^\sharp_1$ as the matrix whose columns are the top-k right singular vectors of $\hat\Rb_{12}^\sharp$, $\hat\Vb^\sharp_2$ as the matrix whose columns are the top-k left singular vectors of $\hat\Rb_{12}^\sharp$ and $\hat\Gammab^\sharp$ is a diagonal matrix whose diagonal entries are top-k singular values of $\hat\Rb_{12}^\sharp$. Then, we initiate $\Amk$ as 
\[
\Amk^\sharp=(\hat\Gammab^\sharp)^{-1/2}\hat\Vb_m^{\sharp\top}\rbr{p^{-1}\sum_{i=1}^p\sCovfree{i}{i}}^{-1/2}.
\]
The experimental setup is described in Section~\ref{ssec:dis_vs_sample} and we test the experiments for $20$ runs and take the average of the results. In the simulation, we vary $N$ and measure the distance $\sum_{m=1}^M\sigma_A^{-1}\norm{\Amk-\pAmk}_F^2$. The result is displayed in Figure~\ref{fig:tridiag3_sj_compare}. The resulting figure indicates that the normalized distance is consistent with respect to the sample size $N$.
\begin{figure}
    \centering
    \includegraphics[width=.9\textwidth]{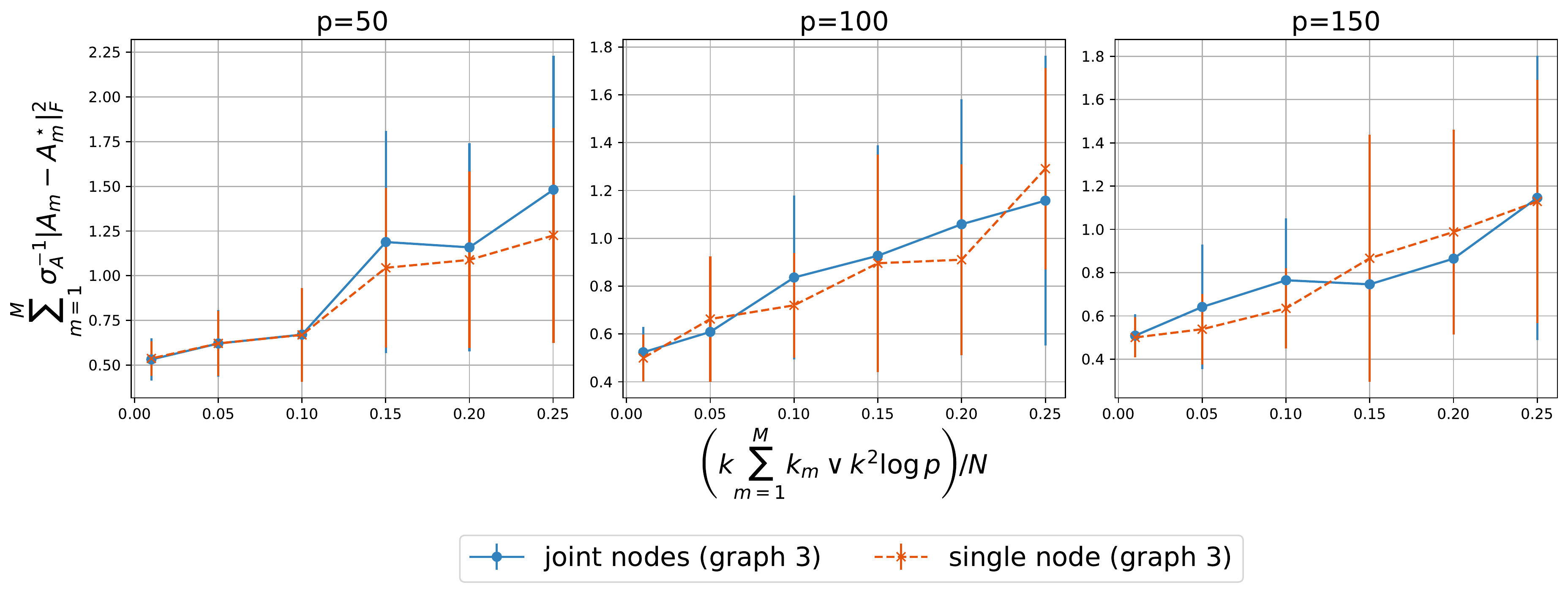}
    \caption{We compare two approaches under $p=\{50,100,150\}$, $\sigma_A=\max_{\mseq}\norm{\pAmk}_2$. The plots indicate that both approaches have comparable performance, in particular in the case when $N$ is large leading to small $(k\sum_{m=1}^Mk_m\vee k^2\log p)/N$.}
    \label{fig:tridiag3_sj_compare}
\end{figure}

\subsection{Extension to Distinct Operators}\label{ssec:extension}
Our model is easily extended to accommodate a distinct transformation operator for each node $i$, i.e., $\Ab^m_i$. Then, the objective function~\eqref{eq:f_n} becomes:
 \begin{align}\label{eq:f_n:A}
 f(\setAk,\setBk) &= \sum_{i=1}^p\sum_{m=1}^M\frac{1}{2N}\bignorm{\Ab_i^m \Ykmnfree{i}-\sum_{j\in\pni}\Bijk\Ab_j^m \Ykmnfree{j}}_F^2.
 \end{align}

 The initialization for $\{\Ab_i^m\}_{i=1,\ldots,p}$ can be obtained by simply repeating~\eqref{eq:cca_model} for each node. Similarly, we can tweak Algorithm~\ref{alg:update} to individually update $\Ab_i^m$ for this more flexible model by alternately minimizing the objective~\eqref{eq:f_n:A}.

\FloatBarrier

\subsection{Generalization to $M$ modalities}\label{ssec:generalization}
{In this section, we discuss a simple extension that enables generalization to $M\geq 2$ modalities. Note that Algorithm~1 can be applied to any $M$, and hence we will focus on the generalization of the initialization method~(4.3). The initialization method~(4.3) is the result of the maximum log-likelihood estimator of the Gaussian linear latent model~\citep{bach2005probabilistic}. It follows that the generalization to $M$ modalities may be solved using the maximum log-likelihood of the $M$-modal generative model. There are multiple potential ways to carry out the estimation, e.g., the EM algorithm. Here, we provide a simple extension from~\eqref{eq:cca_model}.\\
Given $m\in\{1,\ldots,M\}$, let $\backslash m=\{1,\ldots, M\}\backslash\{m\}$. 
Note that we can generalize~\eqref{eq:y_z} to 
\begin{align*}
    \yb^m_i\mid \zb_i&\sim\Ncal(\Lb^{m\star}\zb_i,\Sigmab_{i,i}^{m,\bf{q}}+\Sigmab_{i,i}^{m,\mub});\\
    \yb^{\backslash m}_i\mid \zb_i&\sim\Ncal(\Lb^{\backslash m\star}\zb_i,\Sigmab_{i,i}^{\backslash m,\bf{q}}+\Sigmab_{i,i}^{\backslash m,\mub}), 
\end{align*}
where $\yb^{\backslash m}_i=(\yb_i^{1\top},\ldots,\yb_i^{m-1\top},\yb_i^{m+1\top},\ldots,\yb_i^{M\top})^\top$, $\Lb^{\backslash m\star}=(\Lb^{1\star\top}\cdots \Lb^{m-1\star\top}\Lb^{m+1\star\top}\cdots \Lb^{M\star\top})^\top$. Here $\Sigmab_{i,i}^{\backslash m,\bf{q}}$ is a block-diagonal matrix whose $j$-th block is $\Sigmab_{i,i}^{j,\bf{q}}$ for $j<m$ and 
$\Sigmab_{i,i}^{j+1,\bf{q}}$ for $j\geq m$; $\Sigmab_{i,i}^{\backslash m,\mub}$ is a matrix of the covariance of $(\mub_{i}^{1\top},\ldots,\mub_{i}^{m-1\top},\mub_{i}^{m+1\top},\ldots, \mub_{i}^{M\top})^\top$. 
Then, we can view $(\yb^m_i,\yb^{\backslash m}_i)$ as two modalities and apply Theorem~\ref{theorem:initialcca} with $k_2 = \sum_{i\in \backslash m}k_i$. We generalize the result in the following paragraph. The idea is that 
since Theorem~\ref{theorem:convergence} holds for any $M\geq 2$, if the generalized initialization method can fulfill Assumption~8, then Theorem~\ref{theorem:main} can be generalized to any $M\geq 2$.\\
Define $\Rb^{m\backslash m}=(\Sigmab^{m })^{-1/2}\Sigmab^{m\backslash m}(\Sigmab^{\backslash m})^{-1/2}$. The following Assumption is a variant of Assumption~\ref{assumption:distinctcca}.
\begin{customasm}{5'}\label{asumption:v5}
Suppose that there exists an $m\in\{1,\ldots, M\}$ such that $\min(k_m,\sum_{i\in\backslash m}k_i)\geq k$ and the top-k singular values of $\Rb^{m\backslash m}$ satisfy: $\gamma_{1}> \gamma_{2}>\ldots>\gamma_{k}>0$.
\end{customasm}
\begin{customthm}{5.1'}[Generalization of Theorem~\ref{theorem:initialcca}]\label{theorem:ini_gen}
 Let $\defcdnm$. Suppose there exists a $m$ satisfies Assumption~\ref{asumption:v5} and Assumptions~\ref{assumption:cov} holds. Suppose that $N=O(\kappa_m^2 k_m\vee \max_{i\in\backslash m}\kappa_i^2\sum_{i\in\backslash m}k_i)$. Let $\Lscr^{m,r}=\sum_{(\ell,\ell')\in I}\dotp{\Lscr^m\phi_{\ell'}}{\phi_{\ell}^m}_{\mathbb{H}_m}\phi_{\ell}^m\otimes\phi_{\ell'}$, with $I=\{(\ell,\ell')\in\NN\times\NN:(\ell,\ell')\neq (a,b), a=1,\ldots,k_m, b=1,\ldots,k\}$,  denote the truncation term. Then for $m'=1,\ldots, M$
\begin{equation}
\norm{\Ab^{m'(0)} -\Qb\Ab^{m'\star}}_F
\leq 
C_{\gamma_k,\nu,\rho}\rbr{1+\max_{j\neq i}\frac{1}{|\gamma_j -\gamma_i|}}
\sqrt{k\frac{\sum_{i=1}^M k_i}{N}}+\sqrt{8k}\opnorm{\Ascr^{m'}}{}^2\opnorm{\Lscr^{m',r}}{},
\end{equation}
with probability at least $1-5\exp(-\sum_{i=1}^M k_i)$, where
$C_{\gamma_k,\nu,\rho}>0$ is a constant.
\end{customthm}
\begin{proof}[Proof of Theorem~\ref{theorem:ini_gen}]
    Note that if Assumption~\ref{asumption:v5} is satisfied, this implies that there exists a $m$ such that $\Rb^{m\backslash m}$ satisfies Assumption~\ref{assumption:distinctcca}. Then we can apply Theorem~\ref{theorem:initialcca} and obtain the result.
\end{proof}
 }

\section{Additional Simulation Results}\label{sec:addsimulation}

In this section, we present additional simulation results. 
Section~\ref{ssec:dgp} introduces the graph and data generation processes.
Section~\ref{ssec:datailsofgraph4} introduces the simulation process of the edge sets used in Graph~4. Section~\ref{ssec:roc:noise1}--\ref{ssec:pvsn} present additional simulation results with different sample size and noise model.
\subsection{Data Generation Procedures}\label{ssec:dgp}
We first introduce the precision structures followed by the data generation processes. 

{\bf Construction of inverse covariance ${\bm\Omega}$}: 
Instead of generating graphs from $\beta^\star_{ij}$, we directly construct sparse inverse covariance operators. In simulations, we require the ranks of the operators to be finite. Assume that the true latent space is $r$-dimensional. 
Note that $r$ and $k$ are different: $r$ is the true dimension that generates the process and $k$ is the estimated (low-rank) dimension via Section~\ref{sec:parameterselection}. We follow graph generation processes introduced in~\citet{Zapata2019functional} and ~\citet{qiao2019functional}. Let $\Psi\in\RR^{r\times r}$ be a tri-diagonal matrix such that $\Psi_{i,i}=1$ for $i=1,\ldots,r$ and $\Psi_{i,i+1}=\Psi_{i+1,i}=0.5$ for $i=1,\ldots,r-1$. 
Let ${\bm\Omega}\in\RR^{pr\times pr}=({\bm\Omega}_{i,j})$ be the precision matrix with $r=9$ and ${\bm\Omega}_{i,j}\in\RR^{r\times r}$ for $i,j=1,\ldots,p$. Since ${\bm\Omega}_{i,j}={\bf 0}$ if and only if $\Bb_{ij}^{r\star}={\bf 0}$ for $i\neq j$, we consider consider the following structures on ${\bm\Omega}$:
\begin{itemize}
\item Graph 1: This model is similar to Graph 1 in Section~$5.1$ in \citep{qiao2019functional}. The diagonal blocks have and ${\bm\Omega}_{i,i}=\Ib$ for $\pseq$. For any $i=1,\ldots, p-1$, the off-diagonal blocks have ${\bm\Omega}_{i,i+1}={\bm\Omega}_{i+1,i}=0.4\Psi$. For any $i=1,\ldots,p-2$, the off-diagonal blocks have ${\bm\Omega}_{i,i+2}={\bm\Omega}_{i+2,i}=0.2\Psi$. For all other off-diagonal blocks, we have ${\bm\Omega}_{i,j}={\bf 0}$.
\item Graph 2: This model has similar structure as the Graph 2 in Section~$5.1$ in \citep{qiao2019functional} with the assumption that $p$ must be a constant multiple of $10$. For $t=1,\ldots, p/10$ and let ${\bm\Omega}_{t}=({\bm\Omega}_{i,j})_{i=10(t-1)+1,\ldots,10t, j=10(t-1)+1,\ldots,10t}$ be a $10r\times 10 r$ sub-matrix of ${\bm\Omega}$. If $t$ is an odd number, ${\bm\Omega}_{t}$ comes from Graph 1 with $p=10$. If $t$ is an even number, then ${\bm\Omega}_{t}=\Ib$.
    \item Graph 3: This model is the same as Graph 1 except that we have an additional structure: ${\bm\Omega}_{i,i+3}={\bm\Omega}_{i+3,i}=0.1\Psi$ for $i=1,\ldots,p-3$. 
    \item Graph 4: The structure is similar to example used in~\citet{Zapata2019functional}, which violates the partial separability structure~\citep{Zapata2019functional}. We adopt the modified graph structure, Model D in~\citet{zhao2021high} as the graph candidate for Graph 4. First, each node has the number of neighbors following a power law distribution $f(y)=y^{-2}$ and the candidates of neighbors are selected uniformly. Then, we partition the edge set $F$ into $r$ edge sets, $F_1,\ldots, F_r$, where the construction is described in Appendix~A.4 of~\citep{Zapata2019functional} and we restate the simulation procedure in Appendix~\ref{ssec:datailsofgraph4} for completeness. Given $F_l$, we construct a $\tilde{\Omegab}_l\in\RR^{p\times p}$ precision matrix for $l=1,\ldots,r$. The $ij$-th entry of $\tilde{\Omegab}_l$ is constructed as follows
    \begin{equation*}
        \tilde{\Omegab}_{l,ij} = \left\{\begin{array}{ll}
            1, &  i=j;\\
            0, &  (i,j)\not\in F_l\text{ or } i<j;\\
            \unif\rbr{\sbr{-\frac{2}{3},-\frac{1}{3}}\cup\sbr{\frac{1}{3},\frac{2}{3}}}, & 
            (i,j)\in F_l.
        \end{array}\right.
    \end{equation*}
    We then normalize $\tilde{\Omegab}_{l}$ such that each row of $\tilde{\Omegab}_{l}$ has unit norm. Then, we symmetrize $\tilde{\Omegab}_{l}$ by computing $(\tilde{\Omegab}_{l}+\tilde{\Omegab}_{l}^\top)/2$ and setting the diagonal entries to be $1$. We define $\Sigmab_{\text{ps}}=\diag(\Sigmab_{1},\ldots,\Sigmab_{r})$ where $\Sigmab_{l}=3l^{-1.8}\tilde{\Omegab}_l^{-1}$ for $l=1,\ldots, r$. Define $\bar\Omegab\in\RR^{pr\times pr}$ be a precision matrix whose $l$-th $p\times p$ block diagonal is $\bar\Omegab_{l,l}=\tilde{\Omegab}_l$ and off-diagonal blocks are $\bar{\Omegab}_{l,l+1}=\bar{\Omegab}_{l+1,l}=\{
    \tilde{\Omegab}_l-\diag(\tilde{\Omegab}_l)
    +
    \tilde{\Omegab}_{l+1}-\diag(\tilde{\Omegab}_{l+1})
    \}/2$. We then obtain
    \[
    \Omegab = \diag(\Sigmab_{\text{ps}})^{-1/2}
    {\diag(\bar\Omegab)^{-1/2}\bar\Omegab\diag(\bar\Omegab)^{-1/2}}
    \diag(\Sigmab_{\text{ps}})^{-1/2}.
    \]
\end{itemize}

{\bf Construction of ${\Ascr^m}$}: 
 Let $\Ab^{m,r}=(\Lb^{m,r})^\dagger\in\RR^{ r\times r_m}$ be the matrix representation of ${\Ascr^m}$: for each entry $a_{\ell\ell'}^{m,r}$ in $\Ab^{m,r}$, we have $a_{\ell\ell'}^{m,r}=\dotp{{\Ascr^m}\phi_{\ell'}^m}{\phi_{\ell}}$. Noting that $\pAmk$ is a sub-matrix of $\Ab^{m,r}$. 
To construct $\Ab^{m,r}$ for $\mseq$, we first generate sparse orthonormal rows, where the ratio of the non-zero entries are $1/3$. Then we scale the magnitude of row $i$ in $\Ab^{m,r}$ with $0.2(i+1)+1$ for $i=1,\ldots, r_m$. The  covariance matrix of $\xb^m$ for $m=1,2$ is hence $\{(\Ib_p\otimes\Lb^{m,r}){\bm\Omega}^{-1}(\Ib_p\otimes(\Lb^{m,r})^\top)\}$.\\\par
{\bf Construction of Noise models}: 
We consider two covariance structures of noise models.
\begin{itemize}
    \item Noise Model 1: $\Sigmab^{m,\qb}=\sigma\Ib$, where  $\Ib\in\RR^{pr\times pr}$ for $\nseq$ and $\mseq$. In the simulation, we set $\sigma=0.05$
    \item Noise Model 2: $\Sigma^{m,\qb}$ is a block-diagonal matrix with $p/10$ blocks. First, we generate a block-diagonal matrix ${\Fb^m}$ with $p/10$ blocks and each row in a block is orthonormal to the other rows in the same block. For $m=2$, we rotate ${\Fb^m}$ $90$ degree clockwise.
    Let $\lambda_{m,1}\geq\ldots\geq\lambda_{m,pr}$ be the eigenvalues of $\tilde{\Fb}^m = ({\Fb^m}+{\Fb^m})/2$.    Then, we make $\tilde{\Fb}^m$ positive definite by taking the low-rank parts of $\tilde{\Fb}^m$ such that the corresponding eigenvalues are greater than zero. Then, we normalize the remaining eigenvalues $\lambda_{m,i}$ by $0.01\lambda_{m,i}\{\max_j(\lambda_{m,j})\}^{-1}\sinmax{\Sigmab^{m,\qb}}$. 

\end{itemize}

The major difference between Noise Model 1 and Noise Model 2 is that data from all modalities corrupted with Noise Model 1 have identical graph structures in the observational space. Moreover, when magnitude of the noise is small, i.e., $\sigma$ is small, the graph structures of the observed graphs and the latent graph will look almost identical.
In contrast, data corrupted with Noise Model 2 has very different observed graph structures between modalities. In addition, the latent graph has distinct graph structure from the observed graphs. We construct Noise Model 2 to mimic the real world situation that data from different modalities are corrupted with distinct structured noise and to test the robustness of the proposed model.  
\\
{\bf Data Simulation Process:}
We generate $N$ samples from the latent space, $\zb^{(1)},\ldots\zb^{(N)}\in\RR^{p r}$ under the distribution $\Ncal({\bf 0},{\bm\Omega}^{-1})$, where $\bm\Omega$ is constructed by one of the Graph 1--4. Then, we compute $\xkmn$ as $\xkmn=(\Ib_p\otimes \Lb^{m,r})\zb^{(n)}$ for $\nseq$ and $\mseq$. The observed samples are 
\[
\ykmn=\xkmn+\qb^{m,(n)},\quad \nseq,\;\mseq,
\]
where $\qb^{m,(n)}\sim\Ncal({\bf 0},\Sigmab^{m,\qb})$ is generated independently from either of Noise Model 1 or Noise Model 2.

\subsection{Simulation Details of Graph 4}\label{ssec:datailsofgraph4}
\begin{algorithm}[t!]
\spacingset{1}
    \caption{Partition of the edge sets $F_1,\ldots,F_r$}\label{alg:edge_partition}
\begin{algorithmic}
    \State Input: {$F$, $\tau\in[0,1]$}
    \State $F_c\leftarrow$ randomly select $\tau\abr{F}$ of the edges from $F$
    \For{$l=1,\ldots, r$}
    \State $F_l\leftarrow F_c$
    \State $l\leftarrow 1$, $c\leftarrow 1$
    \EndFor
    \For{$e\in E\backslash E_c$}
     \State $F_l\leftarrow F_l\bigcup e$
     \State $l\leftarrow l+1$
     \If{$l>c$}
        \State $l\leftarrow 1$
        \State $c\leftarrow (c+1) \text{mod } r$
     \EndIf
   \EndFor
\end{algorithmic}
\end{algorithm}

This section discusses the details of the simulation process of the edge set $F_1,\ldots,F_r$ of Graph 4. We follow the same simulation process introduced in Section~A4.1 in~\citet{Zapata2019functional} and restate it here for clarity. First, we generate a graph with the edge set $F$ whose edges follow the power law distribution $f(y)=y^{-2}$. Then, we partition the edge set $F$ into $r$ edge sets $F_1,\ldots,F_r$ such that $F=\bigcup_{l=1}^r F_i$. The partition procedure is described in Algorithm~\ref{alg:edge_partition}.

\subsection{Additional Results for Graph Simulations}\label{ssec:roc:noise1}
\begin{figure}[!h]
    \centering    \includegraphics[width=\textwidth]{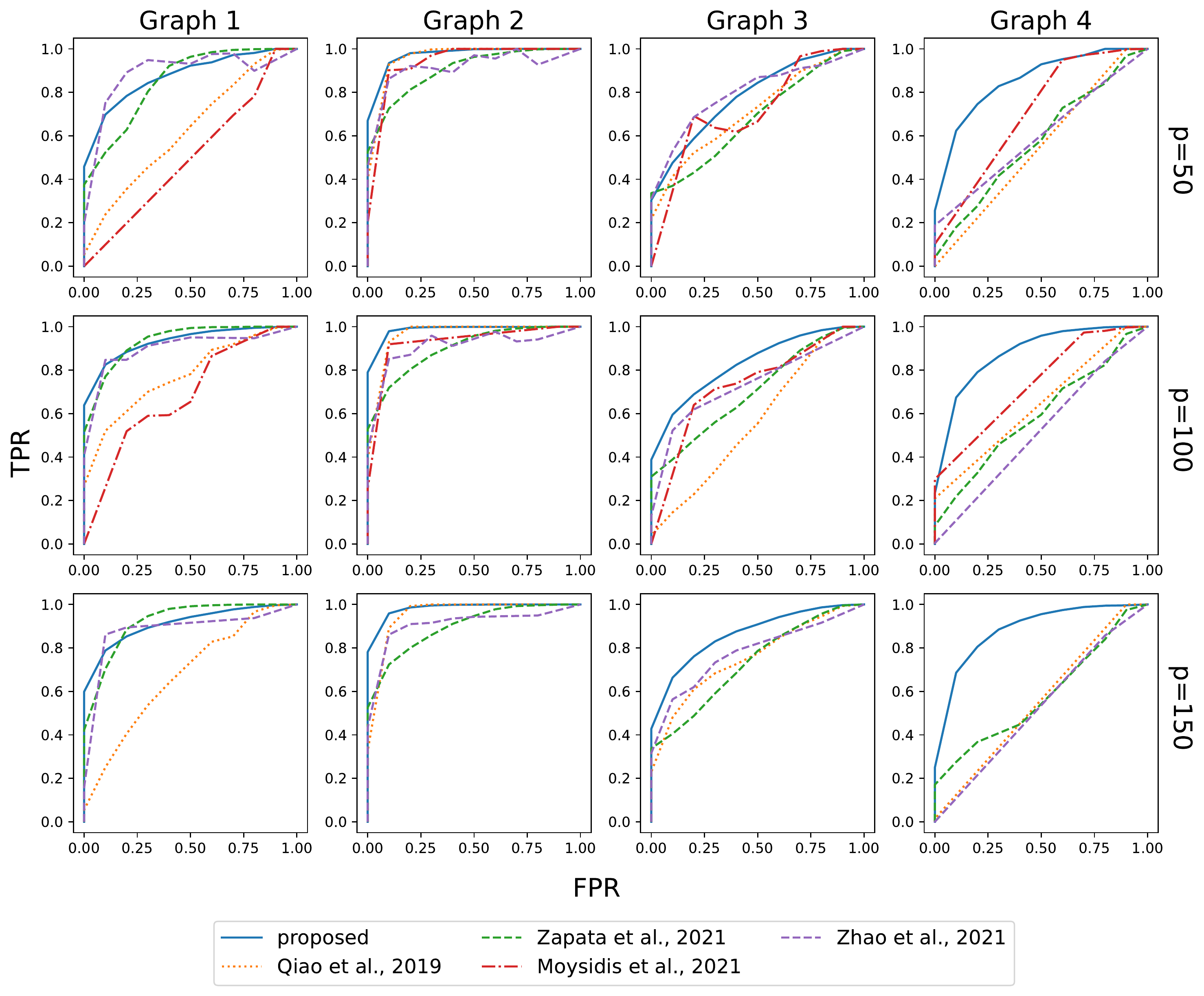}
    \caption{The ROC curves of Graph 1–4. The additive noise is generated from noise
model 1.The AUC is discussed in Table~\ref{tab:auc_noise1}.  The proposed method has consistent performance across four graphs and $p=\{50,100,150\}$. }
    \label{fig:noise_model1}
\end{figure}
{We run the experiments for noise model~1. Figure~\ref{fig:noise_model1} shows the ROC curves under noise model 1. Table~\ref{tab:auc}--\ref{tab:auc_noise1} indicate that the proposed method has gained in performance in both the AUC and AUC15 under different graph settings and dimension $p$. It is worth noting that the neighborhood regression approach on single data modality proposed by~\citet{zhao2021high} has smaller AUC and AUC15 compared to our method, suggesting that integrating data modality might improve the performance. Graph~2 has a simpler and sparser graph structure compared to the other three graphs, and hence most estimators perform well in this setting. In the case of Graph~1 and Graph~3 which have much more complicated graph structures, ~\citet{qiao2019functional} and~\citet{moysidis2021joint} achieve much lower AUC. Finally, while most methods fail in Graph~4, our method still retains good performance under $p=50,100,150$.}

\begin{table}[!h]
\spacingset{1}
    \centering
    \resizebox{\textwidth}{!}{%
    \begin{tabular}{llllllll}
    \toprule
    \midrule\\
    &&\multicolumn{3}{c}{AUC}&\multicolumn{3}{c}{AUC15}\\
     &&\multicolumn{6}{c}{Dimension (p)}\\
    Graph & Method&$50$&$100$&$150$&$50$&$100$&$150$\\
    \midrule
    \multirow{4}{*}{Graph 1}    & Proposed
    &$0.88(0.02)$	&$0.93(0.01)$	&$0.90(0.01)$	&$0.64(0.03)$	&$0.77(0.03)$	&$0.70(0.02)$
    \\
     & FGGM 
    &$0.55(0.05)$	&$0.63(0.13)$	&$0.57(0.07)$	&$0.12(0.05)$	&$0.26(0.19)$	&$0.11(0.04)$\\
    & PSFGGM
    &$0.83(0.01)$	&$0.92(0.01)$	&$0.90(0.00)$	&$0.47(0.01)$	&$0.64(0.02)$	&$0.56(0.01)$\\
    &FPCA
    &$0.84(0.29)$
    &$0.89(0.12)$
    &$0.87(0.15)$
    &$0.51(0.03)$
    &$0.70(0.01)$
    &$0.68(0.02)$
    \\
    & JFGGM
    &$0.47(0.02)$	&$0.64(0.01)$
    &\multicolumn{1}{c}{--}
    &$0.07(0.00)$	&$0.15(0.02)$
    &\multicolumn{1}{c}{--}
    \\
    [1ex]
    \multirow{4}{*}{Graph 2}    & Proposed
    &$0.97(0.01)$	&$0.98(0.01)$	&$0.98(0.00)$ 	&$0.82(0.03)$	&$0.90(0.03)$	&$0.88(0.01)$\\
    & FGGM 
    &$0.96(0.00)$	&$0.96(0.01)$	&$0.95(0.00)$	&$0.74(0.03)$	&$0.81(0.02)$	&$0.75(0.02)$\\
    & PSFGGM
    &$0.89(0.01)$	&$0.89(0.02)$	&$0.89(0.01)$	&$0.63(0.03)$	&$0.65(0.03)$	&$0.65(0.02)$\\
    &FPCA
    &$0.89(0.01)$
    &$0.89(0.01)$
    &$0.91(0.01)$
    &$0.74(0.02)$
    &$0.74(0.04)$
    &$0.77(0.01)$\\
    & JFGGM
    &$0.83(0.03)$	&$0.91(0.01)$
    &\multicolumn{1}{c}{--}
    &$0.28(0.07)$	&$0.62(0.03)$
    &\multicolumn{1}{c}{--}\\
    [1ex]
    \multirow{4}{*}{Graph 3}    & Proposed 
    &$0.79(0.03)$	&$0.80(0.03)$	&$0.88(0.02)$	&$0.39(0.03)$	&$0.48(0.04)$	&$0.56(0.03)$\\
    & FGGM 
    &$0.68(0.02)$	&$0.52(0.00)$	&$0.73(0.01)$	&$0.36(0.02)$	&$0.08(0.00)$	&$0.40(0.02)$ \\
    & PSFGGM
    &$0.65(0.01)$	&$0.67(0.01)$	&$0.70(0.01)$	&$0.35(0.00)$	&$0.34(0.00)$	&$0.36(0.00)$\\
    &FPCA
    &$0.76(0.01)$
    &$0.68(0.02)$
    &$0.75(0.01)$
    &$0.44(0.02)$
    &$0.41(0.02)$
    &$0.44(0.01)$
    \\
    & JFGGM
    &$0.63(0.04)$	&$0.68(0.01)$
    &\multicolumn{1}{c}{--}
    &$0.12(0.03)$	&$0.16(0.01)$
    &\multicolumn{1}{c}{--}\\
    [1ex]
    \multirow{4}{*}{Graph 4}    & Proposed
    &$0.83(0.00)$	&$0.85(0.00)$ 	&$0.86(0.00)$	&$0.47(0.00)$	&$0.46(0.00)$	&$0.46(0.00)$\\
    & FGGM
    &$0.72(0.00)$
    &$0.63(0.00)$
    &$0.67(0.00)$
    &$0.15(0.00)$
    &$0.11(0.00)$	&$0.27(0.00)$\\
    & PSFGGM
    &$0.54(0.00)$	&$0.56(0.00)$	&$0.54(0.00)$	&$0.09(0.00)$	
    &$0.17(0.00)$	&$0.23(0.00)$\\
    &FPCA
    &$0.62(0.00)$
    &$0.49(0.00)$
    &$0.49(0.00)$
    &$0.35(0.01)$
    &$0.08(0.00)$
    &$0.07(0.00)$
    \\
    & JFGGM
    &$0.68(0.02)$
    &$0.69(0.00)$
    &\multicolumn{1}{c}{--}
    &$0.20(0.01)$	
    &$0.29(0.00)$
    &\multicolumn{1}{c}{--}\\
    \midrule
    \bottomrule
    \end{tabular}%
    }
    \caption{The average AUC of Graph 1--4 over 10 runs. The value inside the parentheses denotes the standard deviation. The additive noise is generated from noise model 1, where the AUC plot is shown in Figure~\ref{fig:noise_model1}. The proposed method consistently achieves the largest AUC and AUC15 across four different graphs and $p=\{50,100,150\}$.}
    \label{tab:auc_noise1}
\end{table}

\begin{table}[!h]
    \spacingset{1}
    \centering
    \resizebox{\textwidth}{!}{%
    \begin{tabular}{llllllll}
    \toprule
    \midrule\\
    &&\multicolumn{3}{c}{AUC}&\multicolumn{3}{c}{AUC15}\\
     &&\multicolumn{6}{c}{Dimension (p)}\\
    Graph & Method&$50$&$100$&$150$&$50$&$100$&$150$\\
    \midrule
    \multirow{4}{*}{Graph 1}    & Proposed
    &$0.87(0.02)$	
    &$0.92(0.02)$	
    &$0.87(0.01)$	
    &$0.59(0.02)$	
    &$0.72(0.04)$	
    &$0.62(0.02)$
    \\
     & FGGM 
     &$0.60(0.02)$	&$0.76(0.01)$	&$0.64(0.01)$	&$0.17(0.02)$	&$0.45(0.02)$	&$0.15(0.02)$\\
    & PSFGGM
    &$0.83(0.01)$
    &$0.91(0.01)$	
    &$0.90(0.01)$
    &$0.46(0.01)$
    &$0.63(0.02)$	&$0.56(0.01)$\\
    & FPCA
    &$0.84(0.01)$
    &$0.89(0.00)$
    &$0.85(0.02)$
    &$0.52(0.03)$
    &$0.70(0.01)$
    &$0.65(0.02)$
    \\
    & JFGGM
    &$0.48(0.01)$
    &$0.66(0.01)$
    &\multicolumn{1}{c}{--}
    &$0.07(0.00)$
    &$0.18(0.02)$
    &\multicolumn{1}{c}{--}
    \\
    [1ex]
    \multirow{4}{*}{Graph 2}    & Proposed
    &$0.96(0.01)$
    &$0.98(0.01)$
    &$0.98(0.00)$
    &$0.82(0.02)$
    &$0.88(0.02)$	&$0.82(0.02)$\\
    & FGGM 
    &$0.95(0.01)$
    &$0.96(0.01)$
    &$0.96(0.01)$
    &$0.71(0.02)$
    &$0.78(0.03)$	&$0.76(0.03)$\\
    & PSFGGM
    &$0.88(0.02)$
    &$0.89(0.00)$
    &$0.89(0.01)$
    &$0.63(0.03)$
    &$0.64(0.02)$	&$0.65(0.02)$\\
    &FPCA
    &$0.89(0.02)$
    &$0.90(0.01)$
    &$0.90(0.02)$
    &$0.70(0.02)$
    &$0.76(0.03)$
    &$0.75(0.02)$
    \\
    & JFGGM
    &$0.85(0.02)$
    &$0.89(0.03)$
    &\multicolumn{1}{c}{--}
    &$0.33(0.07)$
    &$0.64(0.04)$
    &\multicolumn{1}{c}{--}\\
    [1ex]
    \multirow{4}{*}{Graph 3}    & Proposed 
    &$0.77(0.05)$
    &$0.80(0.48)$
    &$0.83(0.03)$
    &$0.38(0.04)$
    &$0.48(0.05)$	&$0.52(0.05)$\\
    & FGGM 
    &$0.71(0.02)$
    &$0.52(0.00)$
    &$0.73(0.01)$
    &$0.36(0.02)$
    &$0.09(0.01)$
    &$0.40(0.02)$ \\
    & PSFGGM
    &$0.66(0.01)$
    &$0.67(0.01)$
    &$0.70(0.01)$
    &$0.35(0.00)$
    &$0.33(0.01)$
    &$0.37(0.01)$ \\
    &FPCA
    &$0.76(0.02)$
    &$0.69(0.01)$
    &$0.74(0.00)$
    &$0.45(0.01)$
    &$0.41(0.01)$
    &$0.44(0.01)$
    \\
    & JFGGM
    &$0.66(0.03)$
    &$0.69(0.02)$
    &\multicolumn{1}{c}{--}
    &$0.15(0.02)$
    &$0.18(0.01)$
    &\multicolumn{1}{c}{--}\\
    [1ex]
    \multirow{4}{*}{Graph 4}    & Proposed
    &$0.84(0.00)$
    &$0.84(0.00)$
    &$0.86(0.00)$
    &$0.47(0.00)$
    &$0.45(0.00)$
    &$0.48(0.00)$\\
    & FGGM
    &$0.72(0.00)$
    &$0.63(0.00)$
    &$0.67(0.00)$
    &$0.15(0.00)$
    &$0.11(0.00)$	&$0.27(0.00)$\\
    & PSFGGM
    &$0.53(0.00)$
    &$0.56(0.00)$
    &$0.51(0.00)$
    &$0.10(0.00)$
    &$0.14(0.00)$
    &$0.12(0.00)$\\
    &FPCA
    &$0.61(0.00)$
    &$0.48(0.00)$
    &$0.51(0.00)$
    &$0.34(0.00)$
    &$0.07(0.00)$
    &$0.08(0.00)$\\
    & JFGGM
    &$0.75(0.04)$
    &$0.63(0.00)$
    &\multicolumn{1}{c}{--}
    &$0.23(0.03)$	
    &$0.24(0.00)$
    &\multicolumn{1}{c}{--}\\
    \midrule
    %
    \bottomrule
    \end{tabular}%
    }
    \caption{The average AUC of Model 1--4 over 10 runs. The value inside the parentheses denotes the standard deviation. The additive noise is generated from noise model 2, where the AUC plot is shown in Figure~\ref{fig:noise_model3}. The rightmost three columns denote the average AUC for FPR between $[0,0.15]$, normalized to have a maximum area $1$. The proposed method achieves the largest AUC and AUC15 across four different graphs and $p=\{50,100,150\}$.}
    \label{tab:auc}
\end{table}

\FloatBarrier

\begin{figure}[!h]
    \centering
    \includegraphics[width=\textwidth]{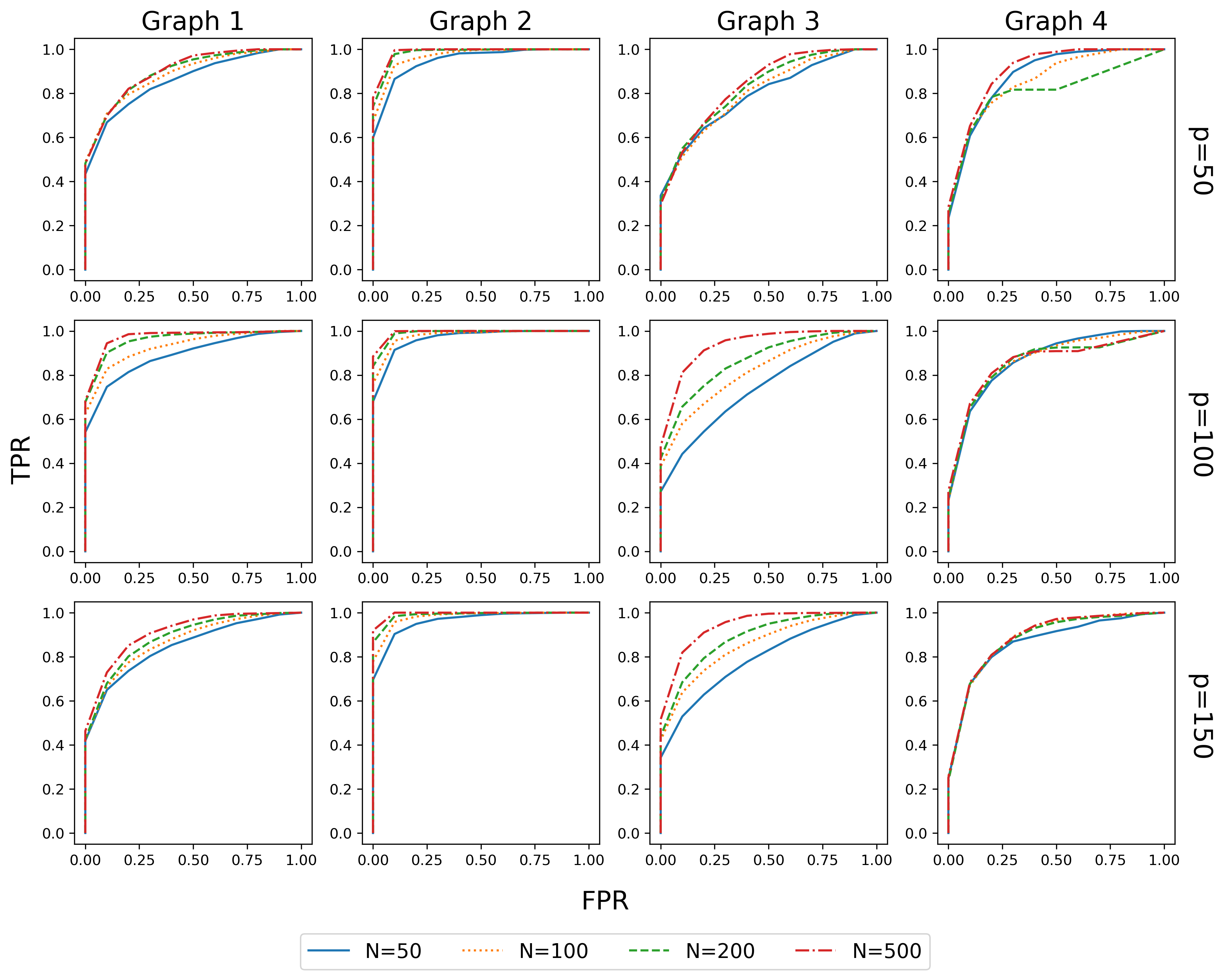}
    \caption{The AUC of proposed model where data is corrupted by noise model 2. Left to right: Graph 1--Graph 4 }
    \label{fig:simulation1}
\end{figure}
\subsection{Experiment 1: p v.s. N}\label{ssec:pvsn}
We plot the ROC curve of various graph with $p\in\{50,100,150\}$ and $\in\{50,100,200,500\}$. The results are displayed in Figure~\ref{fig:simulation1}--\ref{fig:simulation2}. The corresponding AUC is documented in Table~\ref{tab:auc_samplesize20} and Table~\ref{tab:auc_samplesize21}, respectively. For Graph 1--3, the AUC consistently increases as the sample size increases. However, the sample size has mild effect on Graph 4 in both cases. 

\begin{table}[!h]
\spacingset{1}
    \centering
    \resizebox{\textwidth}{!}{%
    \begin{tabular}{llllllll}
    \toprule
    \midrule\\
    &&\multicolumn{3}{c}{AUC}&\multicolumn{3}{c}{AUC15}\\
     &&\multicolumn{6}{c}{Dimension (p)}\\
    Graph & Method&$50$&$100$&$150$&$50$&$100$&$150$\\
    \midrule
    \multirow{4}{*}{Graph 1}    & $N=50$
    &$0.84(0.01)$  &$0.87(0.01)$  &$0.83(0.02)$  &$0.55(0.04)$  &$0.64(0.03)$  &$0.53(0.04)$ \\
    & $N=100$
    &$0.87(0.02)$  &$0.92(0.02)$  &$0.85(0.02)$  &$0.59(0.02)$  &$0.72(0.04)$  &$0.53(0.06)$ \\
    & $N=200$
    &$0.88(0.02)$  &$0.95(0.01)$  &$0.86(0.01)$  &$0.59(0.02)$  &$0.78(0.05)$  &$0.53(0.03)$ \\
    & $N=500$
    &$0.89(0.01)$  &$0.96(0.00)$  &$0.89(0.02)$  &$0.59(0.02)$  &$0.79(0.01)$  &$0.56(0.03)$ \\
    [1ex]
    \multirow{4}{*}{Graph 2}    & $N=50$
    &$0.94(0.02)$  &$0.96(0.01)$  &$0.96(0.01)$  &$0.75(0.03)$  &$0.81(0.02)$  &$0.81(0.03)$ \\
    & $N=100$
    &$0.96(0.01)$  &$0.98(0.01)$  &$0.98(0.00)$  &$0.82(0.02)$  &$0.88(0.02)$  &$0.88(0.02)$ \\
    & $N=200$
    &$0.98(0.00)$  &$0.99(0.00)$  &$0.99(0.00)$  &$0.87(0.01)$  &$0.93(0.01)$  &$0.94(0.01)$ \\
    & $N=500$
    &$0.98(0.00)$  &$0.99(0.00)$  &$1.00(0.00)$  &$0.90(0.01)$  &$0.96(0.00)$  &$0.97(0.00)$ \\
    [1ex]
    \multirow{4}{*}{Graph 3}    & $N=50$
    &$0.76(0.04)$  &$0.71(0.02)$  &$0.76(0.02)$  &$0.43(0.03)$  &$0.35(0.02)$  &$0.43(0.03)$ \\
    & $N=100$
    &$0.77(0.05)$  &$0.80(0.03)$  &$0.83(0.03)$  &$0.38(0.04)$  &$0.48(0.05)$  &$0.52(0.05)$ \\
    & $N=200$
    &$0.80(0.05)$  &$0.84(0.03)$  &$0.87(0.04)$  &$0.40(0.04)$  &$0.53(0.05)$  &$0.54(0.09)$ \\
    & $N=500$
    &$0.81(0.05)$  &$0.92(0.01)$  &$0.92(0.01)$  &$0.39(0.06)$  &$0.62(0.04)$  &$0.64(0.05)$ \\
    [1ex]
    \multirow{4}{*}{Graph 4}    & $N=50$
    &$0.86(0.00)$  &$0.84(0.00)$  &$0.84(0.00)$  &$0.44(0.00)$  &$0.44(0.00)$  &$0.45(0.00)$ \\
    & $N=100$
    &$0.84(0.00)$  &$0.84(0.00)$  &$0.86(0.00)$  &$0.47(0.00)$  &$0.45(0.00)$  &$0.48(0.00)$ \\
    & $N=200$
    &$0.76(0.00)$  &$0.83(0.00)$  &$0.86(0.00)$  &$0.47(0.00)$  &$0.46(0.00)$  &$0.46(0.00)$ \\    
    & $N=500$
    &$0.76(0.00)$  &$0.83(0.00)$  &$0.86(0.00)$  &$0.47(0.00)$  &$0.46(0.00)$  &$0.46(0.00)$ \\
    \midrule
    \bottomrule
    \end{tabular}%
    }
    \caption{The average AUC of Graph 1--4 over 10 runs with different sample size. The corresponding plots are displayed in Figure~\ref{fig:simulation1}. The value inside the parentheses denotes the standard deviation. The additive noise is generated from noise model 2.}
    \label{tab:auc_samplesize20}
\end{table}

\begin{figure}
    \centering
    \includegraphics[width=\textwidth]{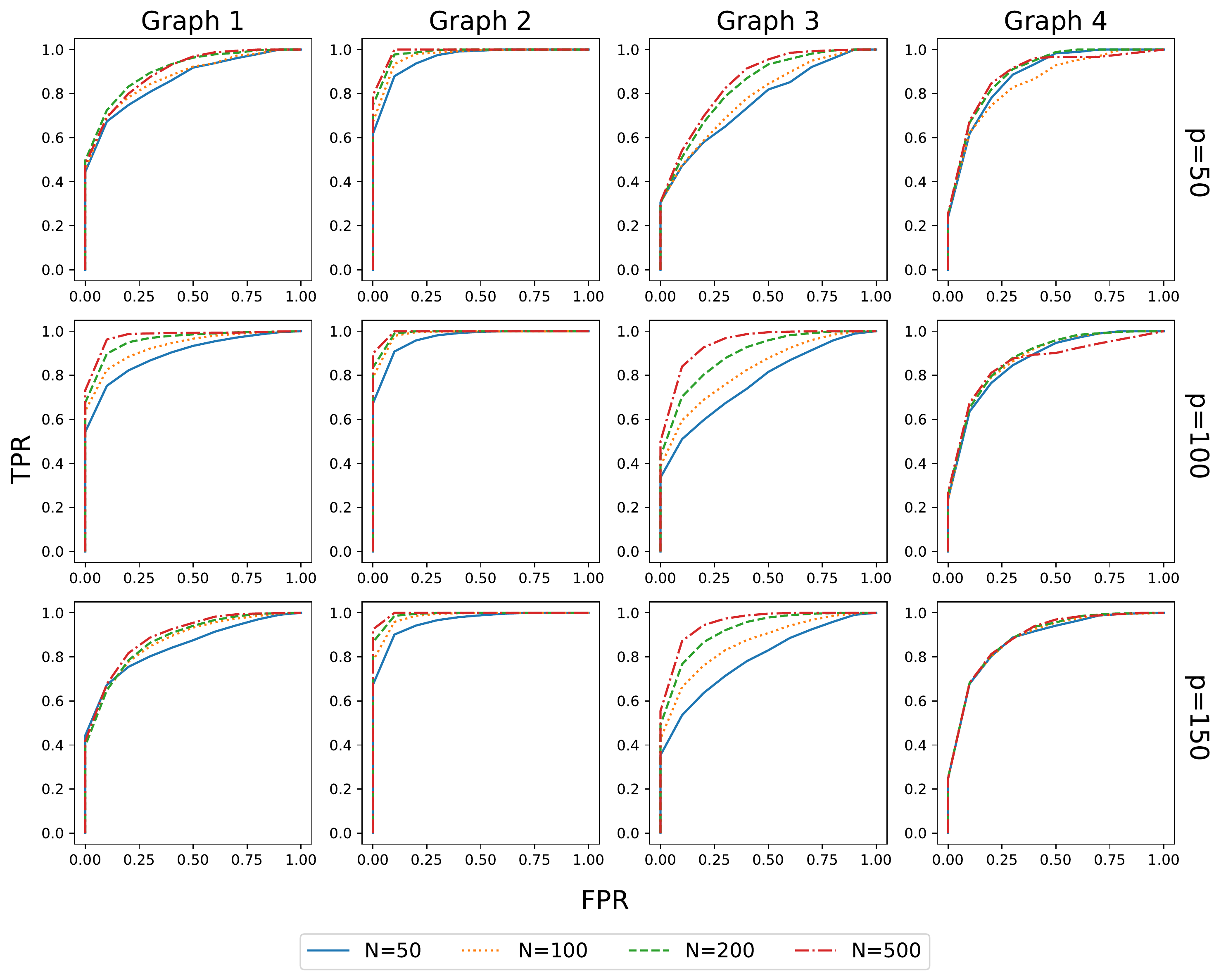}
    \caption{The AUC of proposed model where data is corrupted by noise model 1. Left to right: Graph 1--Graph 4 }
    \label{fig:simulation2}
\end{figure}

\begin{table}[!h]
\spacingset{1}
    \centering
    \resizebox{\textwidth}{!}{%
    \begin{tabular}{llllllll}
    \toprule
    \midrule\\
    &&\multicolumn{3}{c}{AUC}&\multicolumn{3}{c}{AUC15}\\
     &&\multicolumn{6}{c}{Dimension (p)}\\
    Graph & Method&$50$&$100$&$150$&$50$&$100$&$150$\\
    \midrule
    \multirow{4}{*}{Graph 1}    
    &$N=50$  &$0.84(0.02)$  &$0.88(0.02)$  &$0.82(0.03)$  &$0.57(0.03)$  &$0.66(0.04)$  &$0.57(0.03)$ \\
    &$N=100$  &$0.85(0.01)$  &$0.92(0.01)$  &$0.84(0.01)$  &$0.58(0.02)$  &$0.74(0.02)$  &$0.48(0.02)$ \\
    &$N=200$  &$0.89(0.01)$  &$0.95(0.01)$  &$0.86(0.01)$  &$0.60(0.02)$  &$0.78(0.04)$  &$0.51(0.03)$ \\
    &$N=500$  &$0.88(0.01)$  &$0.97(0.01)$  &$0.87(0.01)$  &$0.59(0.01)$  &$0.83(0.03)$  &$0.52(0.03)$ \\[1ex]
    \multirow{4}{*}{Graph 2}   
    &$N=50$  &$0.95(0.01)$  &$0.96(0.01)$  &$0.95(0.01)$  &$0.76(0.03)$  &$0.80(0.02)$  &$0.79(0.02)$ \\
    &$N=100$  &$0.97(0.01)$  &$0.98(0.01)$  &$0.98(0.00)$  &$0.82(0.03)$  &$0.90(0.03)$  &$0.88(0.01)$ \\
    &$N=200$  &$0.98(0.00)$  &$0.99(0.00)$  &$0.99(0.00)$  &$0.88(0.01)$  &$0.93(0.01)$  &$0.94(0.01)$ \\
    &$N=500$  &$0.98(0.00)$  &$0.99(0.00)$  &$1.00(0.00)$  &$0.90(0.01)$  &$0.96(0.00)$  &$0.97(0.00)$ \\[1ex]
    \multirow{4}{*}{Graph 3}   
    &$N=50$  &$0.72(0.04)$  &$0.75(0.03)$  &$0.77(0.02)$  &$0.38(0.06)$  &$0.42(0.04)$  &$0.44(0.04)$ \\
    &$N=100$  &$0.77(0.05)$  &$0.80(0.03)$  &$0.85(0.02)$  &$0.38(0.04)$  &$0.48(0.04)$  &$0.55(0.04)$ \\
    &$N=200$  &$0.81(0.02)$  &$0.87(0.02)$  &$0.90(0.01)$  &$0.38(0.03)$  &$0.52(0.06)$  &$0.60(0.03)$ \\
    &$N=500$  &$0.84(0.03)$  &$0.92(0.02)$  &$0.94(0.01)$  &$0.40(0.04)$  &$0.61(0.04)$  &$0.69(0.03)$ \\[1ex]
    \multirow{4}{*}{Graph 4}    
    &$N=50$  &$0.86(0.00)$  &$0.84(0.00)$  &$0.85(0.00)$  &$0.45(0.00)$  &$0.45(0.00)$  &$0.45(0.00)$ \\
    &$N=100$  &$0.83(0.00)$  &$0.85(0.00)$  &$0.86(0.00)$  &$0.47(0.00)$  &$0.46(0.00)$  &$0.46(0.00)$ \\
    &$N=200$  &$0.88(0.00)$  &$0.86(0.00)$  &$0.86(0.00)$  &$0.48(0.00)$  &$0.46(0.00)$  &$0.46(0.00)$ \\
    &$N=500$  &$0.87(0.00)$  &$0.82(0.02)$  &$0.86(0.00)$  &$0.49(0.00)$  &$0.48(0.00)$  &$0.45(0.00)$ \\
    \midrule
    \bottomrule
    \end{tabular}%
    }
    \caption{The average AUC of Graph 1--4 over 10 runs with different sample size. The corresponding plots are displayed in Figure~\ref{fig:simulation2}. The value inside the parentheses denotes the standard deviation. The additive noise is generated from noise model 1.}
    \label{tab:auc_samplesize21}
\end{table}

\FloatBarrier
\begin{figure}[hbt!]
    \centering
    \includegraphics[width=.9\textwidth]{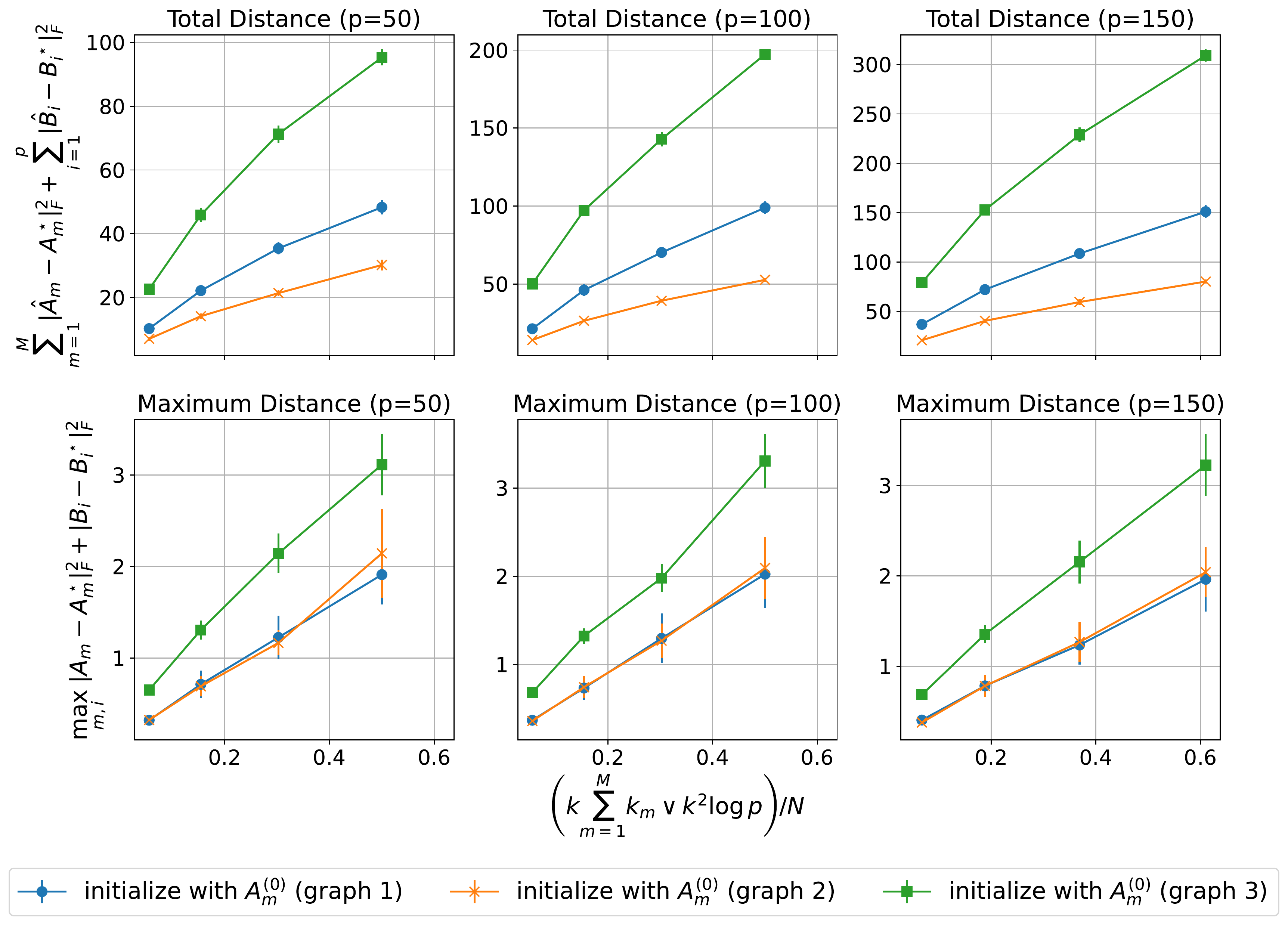}
    \caption{Distance v.s. statistical error. We fix $N=324$ for $p=50,100$ and $N=289$ for $p=150$ and vary $k={3,5,7,9}$. When the value on the $x$-axis is small, both the average error and maximum error are small in all graphs. The errors increase as the value on the $x$-axis increases. }
    \label{fig:vary_k}
\end{figure}
\subsection{Distance v.s. $k$}
{In addition to the sample complexity experiment discussed in Section~7.2, we verify Theorem~5.3 with varying $k$ as well. We choose $k=\{3,5,7,9\}$, where we set $k=k_m$ and $N=324$ here. 
We run the simulation for $20$ independent simulated datasets and the average result is shown in Figure~\ref{fig:vary_k}. Note that the lines are nearly linear for three different graphs and $p=50,100,150$, supporting the result from Theorem~\ref{theorem:convergence}.}

\subsection{Sensitivity of the Tuning Parameters}

        \begin{figure}[hbt!]
            \centering
            \includegraphics[width=.9\textwidth]{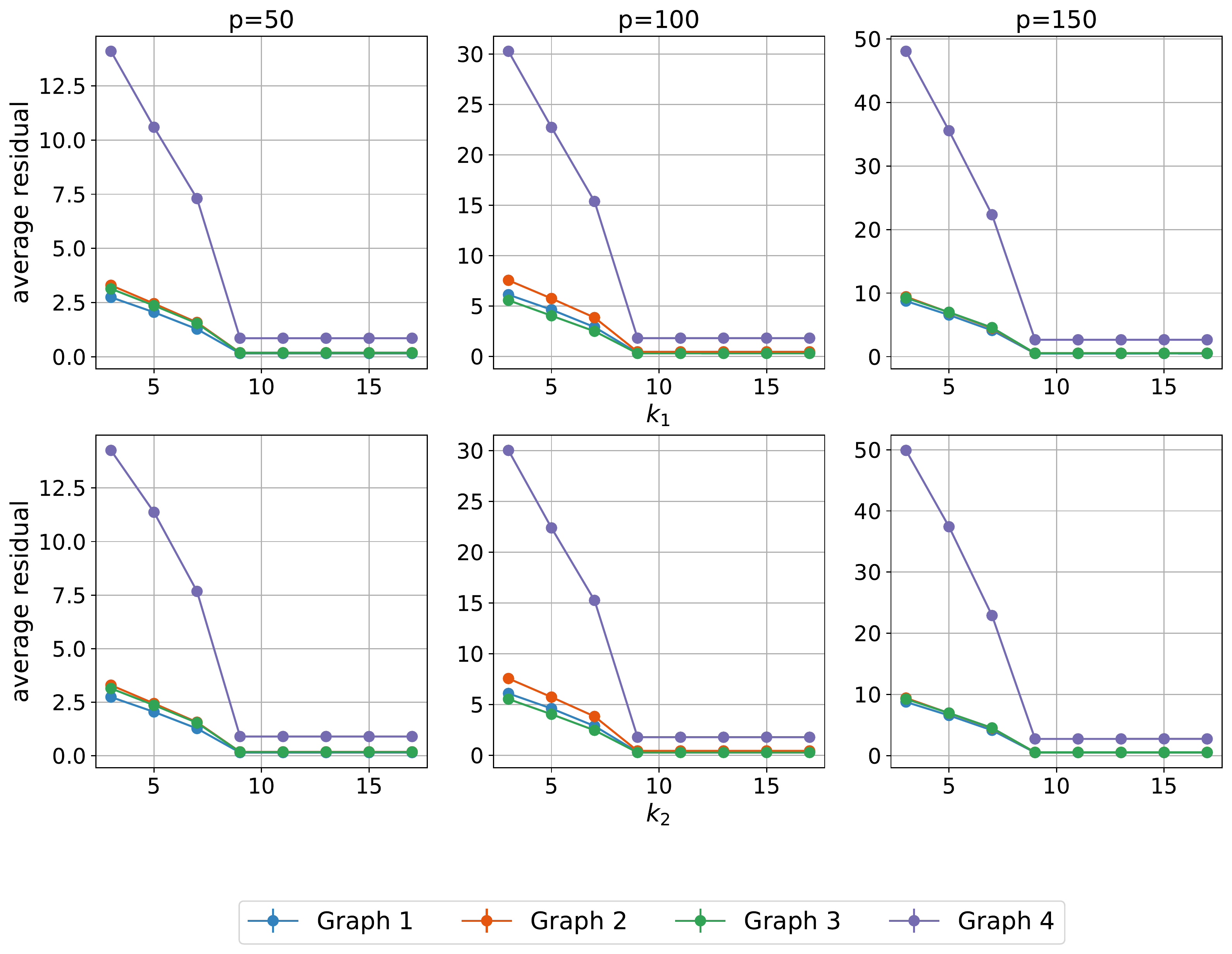}
            \caption{{\bf Top row}: The average residuals of signals from data modality 1 after projecting to $k_1$ number of basis functions, where the ground truth is $k_1=9$. 
            {\bf Bottom row}: The average residuals of signals from data modality 2 after projecting to $k_2$ number of basis functions, where the ground truth is $k_2=9$. Each line is plotted as the average result of $N=100$ subjects and over $10$ independent runs. The standard deviation of each point is at the scale of $0.1$. Both figures indicate that there is a turning point in the residual when $k_1$, $k_2$ exceed the $9$, suggesting that the elbow method could give us relatively accurate estimates. 
            }
            \label{fig:kkm}
        \end{figure}
\begin{figure}[h!]
    \centering
    \includegraphics[width=.9\textwidth]{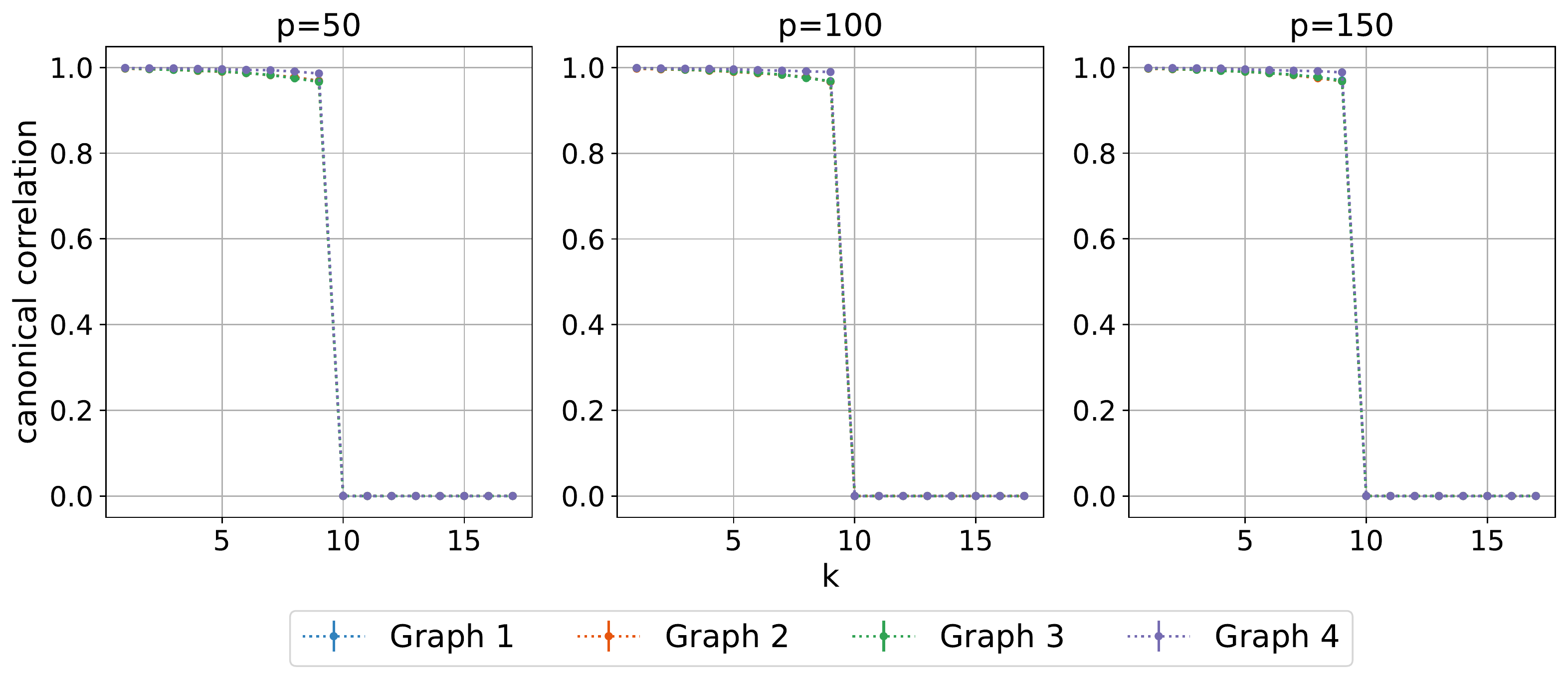}
    \caption{ The top-$k$ canonical correlations
            averaged over $10$ independent runs, where the true $k=9$. The canonical correlation plot has an abrupt drop to $0$ at $k=10$, when $k$ exceeds the ground truth value. }
    \label{fig:kkm2}
\end{figure}
\begin{figure}[h!]
    \centering
    \includegraphics[width=.9\textwidth]{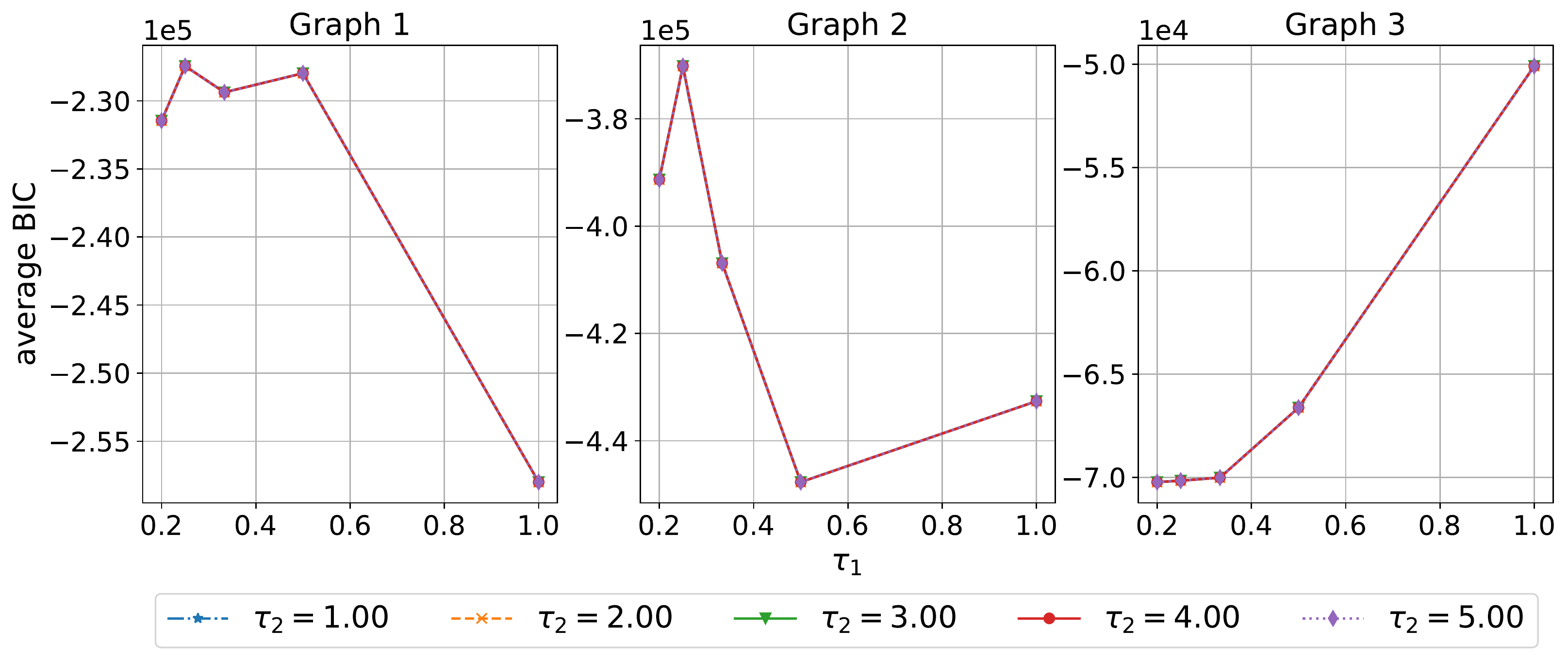}
    \includegraphics[width=.9\textwidth]{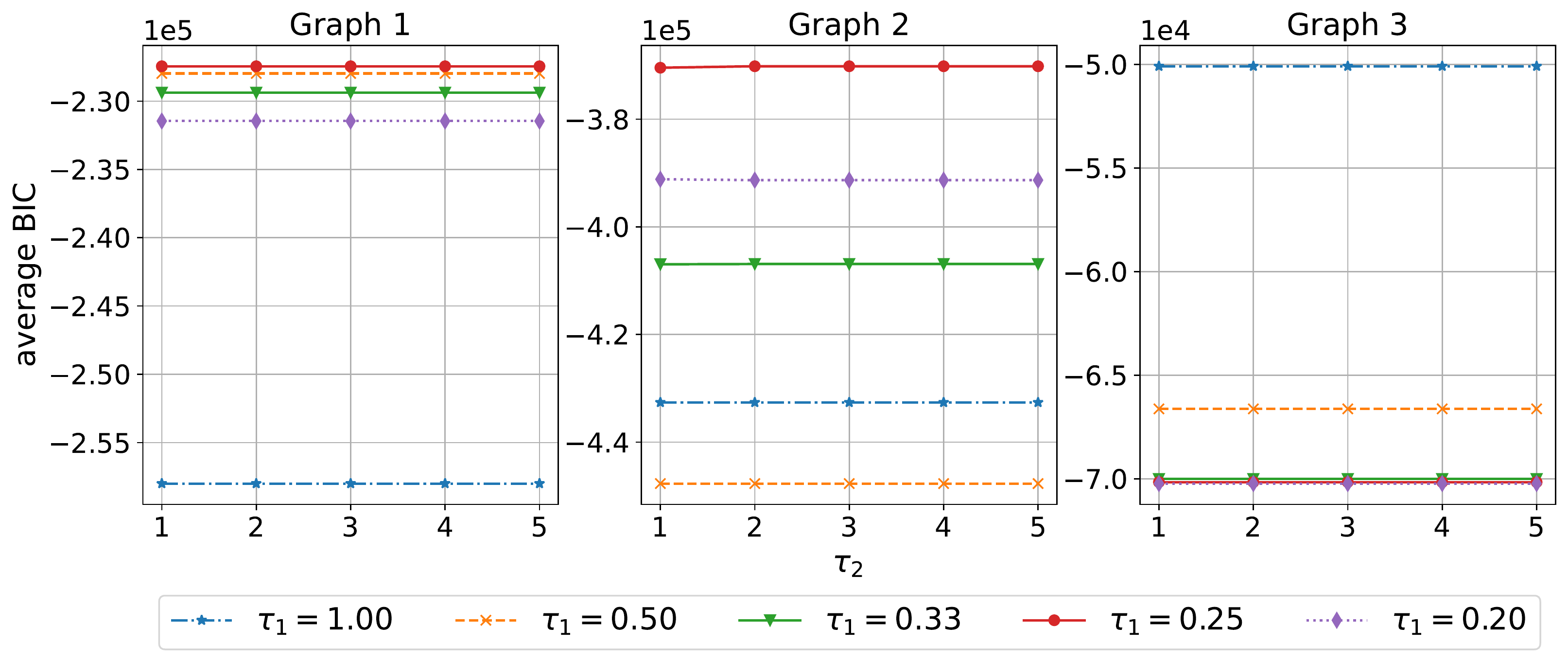}
    \caption{The BIC score plot with $5$-fold cross-validation. In this simulation, we choose $N=100$, $p=100$, $s=4$, and $\alpha=0.44$. Both the plots from the top row and bottom row suggest that the score is sensitive to the values of $\tau_1$ but indifferent to the choices of $\tau_2$.}
    \label{fig:tau_plot}
\end{figure}

\begin{figure}[hbt!]
    \centering
    \includegraphics[width=.9\textwidth]{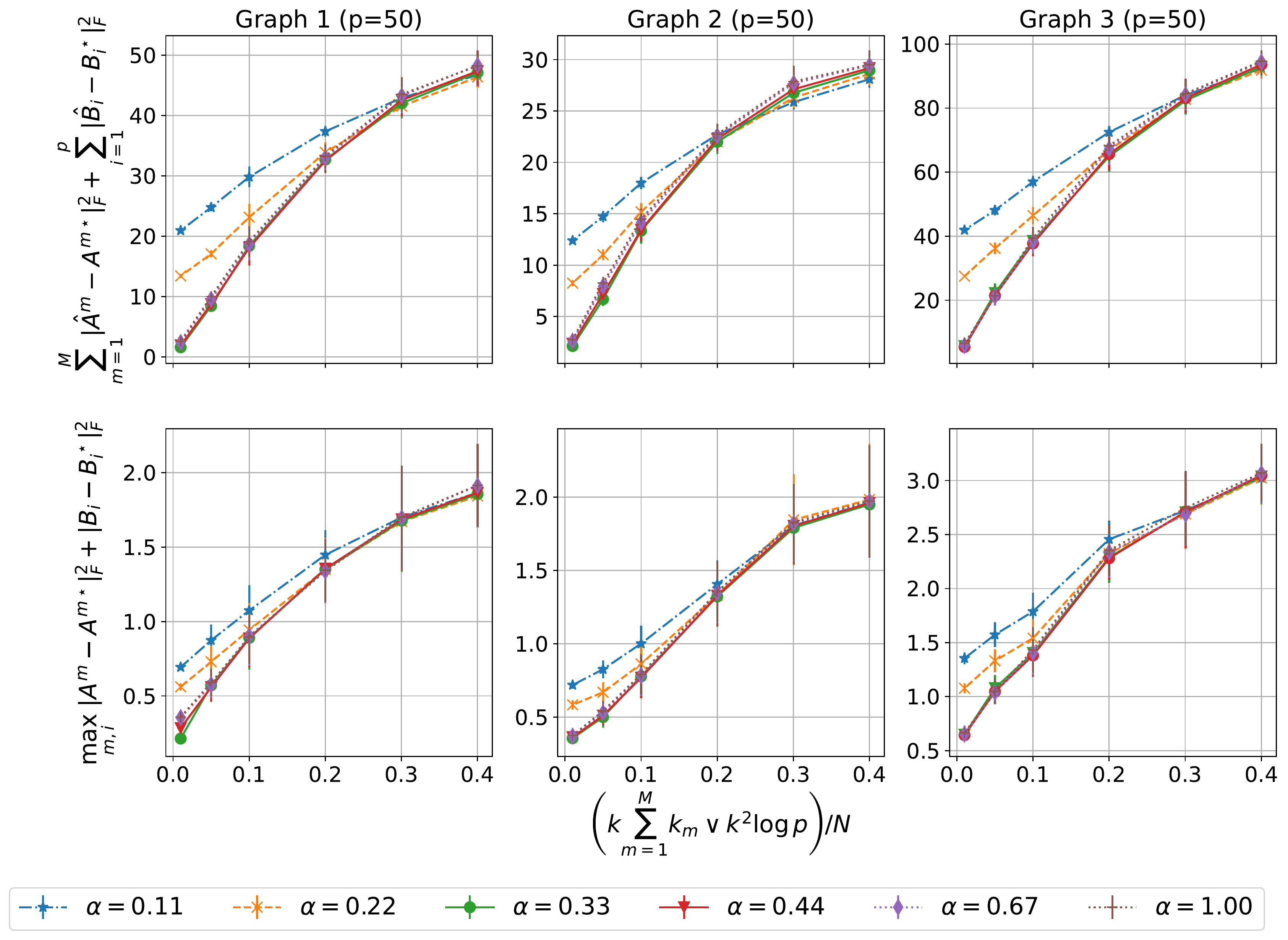}
    \caption{The sample complexity plot with varying $\alpha$. We select other tunning parameters $\{s,\tau_1,\tau_2\}$ using the procedure discussed in Section~\ref{sec:parameterselection} and apply $(4.3)$ to initialize $\Ab^m$. The true $\alpha^\star$ is $0.33$ for Graph~1--2, and $0.44$ for Graph~3. As the plots indicate, under-selection of $\alpha$ leads to a larger error when $x$ approaches zero, where in this case we have a large sample size $N$. In contrast, large $\alpha$ has a weaker influence on the plot. We find that the correct choice of $\alpha=\alpha^\star$ will lead to optimal results.}
    \label{fig:vary_alpha}
\end{figure}

\begin{figure}[hbt!]
    \centering
    \includegraphics[width=.9\textwidth]{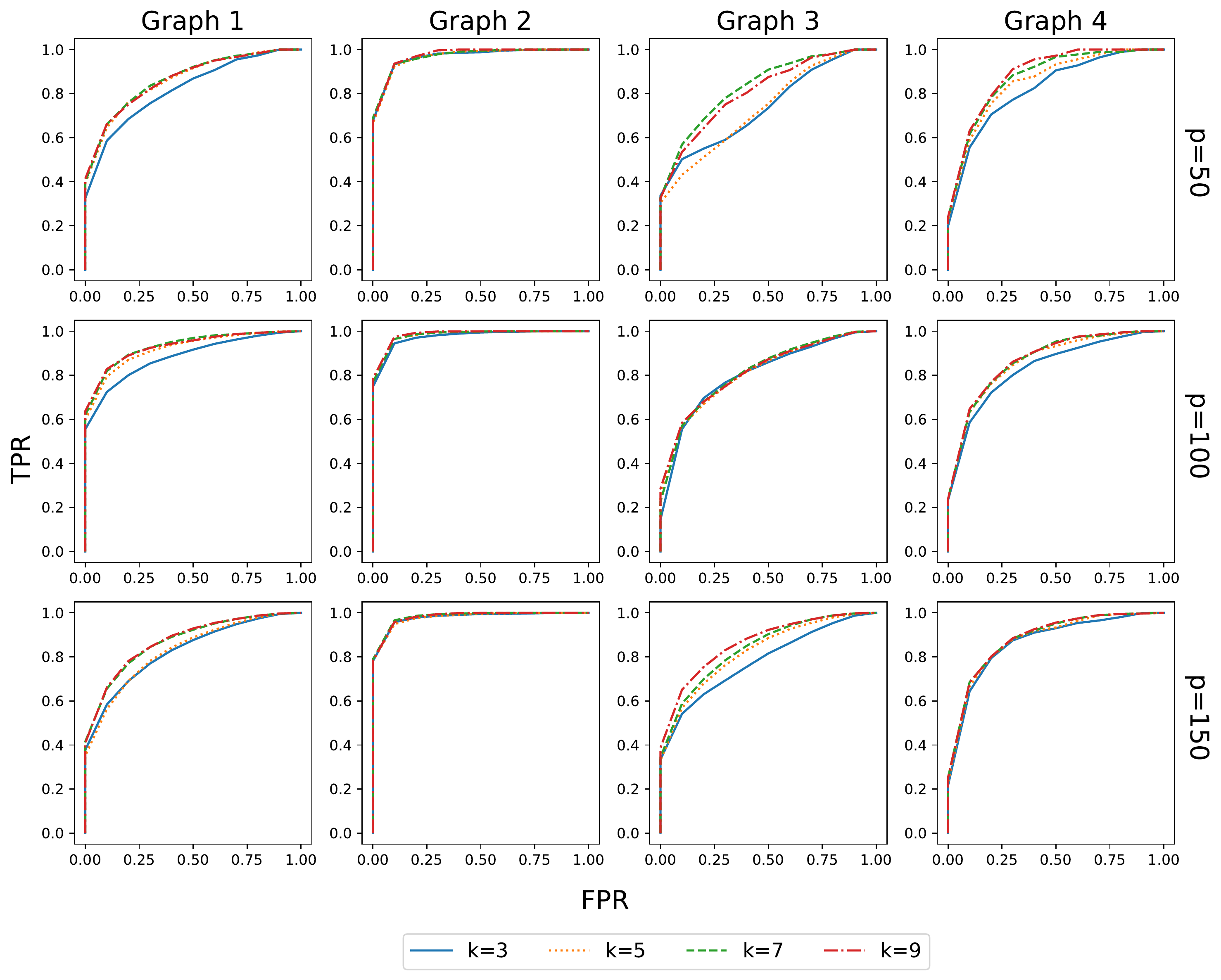}
    \caption{The ROC plot of varying $k$ (ground truth $k=9$ and $k_1=k_2=9$). We use $N=100$ and each plot is the average result over $10$ independent runs. The underlying AUCs are smaller when $k$ and $k_m$ are greatly under-selected. We found that when $k$ is slightly under-selected, i.e. $k=7$, the AUCs are close to the AUCs of $k=9$ in most settings.  }
    \label{fig:vary_k_roc}
\end{figure}

\begin{figure}[hbt!]
    \centering
    \includegraphics[width=.9\textwidth]{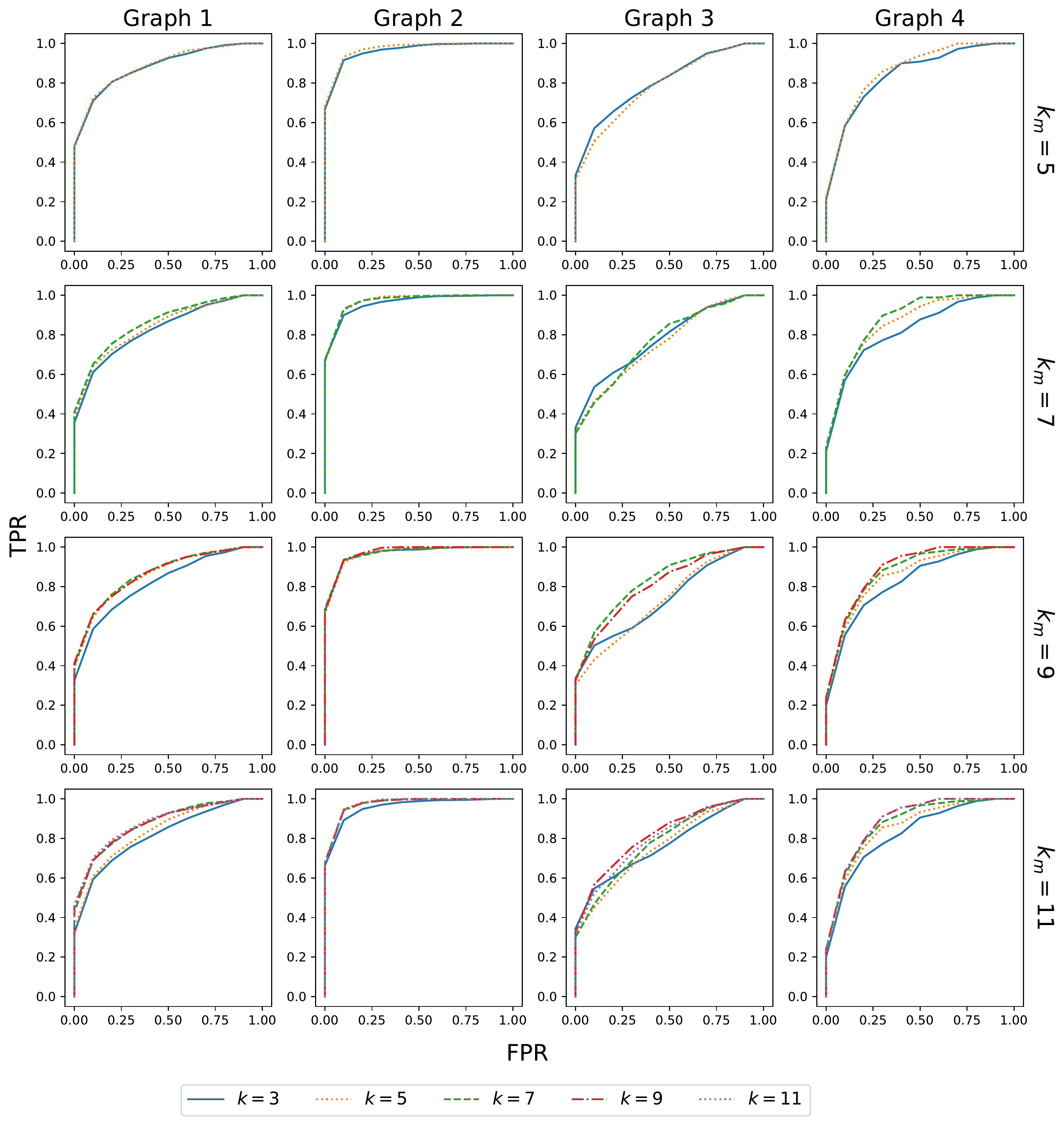}
    \caption{The ROC plot of varying $k$ and $k_m$ with $k_m\geq k$. The ground truth values of $k$ and $k_m$ are $9$. We use $N=100$, $p=50$, and noise model 1. Each plot is the average result over $10$ independent runs. {\bf Top row}: when both $k_1$ and $k_2$ are significantly under-selected, the choice of $k$ is insensitive to the ROC. {\bf Centered two rows}: when both $k_1$ and $k_2$ are slightly under-selected or equal to the true value, larger $k$ has a larger AUC. {\bf Bottom row}: when both $k_1$ and $k_2$ are over-selected, the ROC curves of $k=9$ and $k=11$ almost overlap. It only has a suboptimal curve in Graph~3.}
    \label{fig:vary_kkm_roc}
\end{figure}

{
In this section, we analyze the variable selection method introduced in Section~\ref{sec:parameterselection}. First, we discuss the practice of using elbow method to select $k$ and $k_m$. We then discuss how mis-specification of $k$ and $k_m$ affects the results. Finally, we discuss the sensitivity of choosing $\tau_1,\tau_2$ and $\alpha$ in the following.\\
First, we run the elbow algorithm discussed in Section~\ref{sec:parameterselection} to select $k$ and $k_m$. The result is displayed in Figure~\ref{fig:kkm}--\ref{fig:kkm2}, where there are clear turning points in all cases. Furthermore, the elbow points match the true values. Figure~\ref{fig:vary_k_roc} shows the ROC plot when $k$ is under-selected, and we select $k_m$ based on the elbow method. It shows that the AUCs are smaller when $k$ is significantly smaller than the true $k=9$. In the case when $k=7$, the AUCs are close to the case when $k=9$, as Table~\ref{tab:auc_varyk} shows. In Figure~\ref{fig:vary_kkm_roc}, we vary $k_m=\{5,7,9,11\}$ and $k=\{3,5,7,9,11\}$. Furthermore, we restrict $k_m\geq k$ due to the technical constraints of CCA. The corresponding AUC is documented in Table~\ref{tab:auc_varykmk}. The result indicates that if the difference between $k_m$ and $k$ is small, i.e., $k_m-k=2$, the underlying AUC is close to the case when $k_m=k$. However, if $k$ is much smaller than $k_m$, we see the decay in the underlying AUC. 
        When we over-select $k_m=11$, the corresponding AUCs of $k=9$ and $k=11$ are close to the case when $k_m=k=9$. However, under-selection of $k$ results in smaller AUCs.\\
We select the candidates of $\tau_1$ to be $\{1/n:n=1,\ldots,5\}$ and the candidates for $\tau_2$ to be $\{n:n=1,\ldots,5\}$. Then, we apply 5-fold cross-validation with BIC metric to select the optimal values. In the following simulation, we fix $s=4$ and $\alpha=4/9$ and plot the BIC score with respect to $\tau_1$ and $\tau_2$, shown in Figure~\ref{fig:tau_plot}. The plot indicates that the score is insensitive to the upper bound, the choice of $\tau_2$, and more sensitive to the lower bound, the choice of $\tau_1$. \\
Finally, we find varying $\alpha$ is less sensitive to the result so long as $\alpha\geq\alpha^\star$. We present the simulation result in Figure~\ref{fig:vary_alpha}. The simulation details are discussed in Section~\ref{ssec:dis_vs_sample}.
}

\begin{table}
    \spacingset{1}
    \centering
    \resizebox{\textwidth}{!}{%
    \begin{tabular}{llllllll}
    \toprule
    \midrule\\
    &&\multicolumn{3}{c}{AUC}&\multicolumn{3}{c}{AUC15}\\
     &&\multicolumn{6}{c}{Dimension (p)}\\
    Graph & Method&$50$&$100$&$150$&$50$&$100$&$150$\\
    \midrule
    \multirow{4}{*}{Graph 1}    & $k=3$
    &$0.79(0.02)$	
    &$0.87(0.04)$	
    &$0.80(0.02)$	
    &$0.44(0.02)$	
    &$0.62(0.04)$	
    &$0.47(0.04)$
    \\
     & $k=5$ 
     &$0.84(0.01)$		
     &$0.91(0.02)$	
     &$0.80(0.02)$	
     &$0.50(0.03)$	
     &$0.69(0.06)$	
     &$0.45(0.05)$\\
    & $k=7$
    &$0.85(0.03)$
    &$0.92(0.01)$	
    &$0.85(0.02)$
    &$0.55(0.03)$
    &$0.72(0.02)$	
    &$0.52(0.06)$\\
    & $k=9$
    &$0.85(0.02)$
    &$0.92(0.01)$
    &$0.85(0.02)$
    &$0.55(0.03)$
    &$0.74(0.03)$
    &$0.52(0.06)$\\\midrule
    \multirow{4}{*}{Graph 2}    & 
    $k=3$
    &$0.96(0.01)$
    &$0.97(0.01)$
    &$0.97(0.01)$
    &$0.81(0.04)$
    &$0.86(0.02)$	
    &$0.88(0.02)$\\
    & $k=5$
    &$0.96(0.01)$
    &$0.98(0.00)$
    &$0.97(0.01)$
    &$0.81(0.02)$
    &$0.89(0.01)$	
    &$0.88(0.02)$\\
    & $k=7$
    &$0.96(0.01)$
    &$0.98(0.01)$
    &$0.98(0.00)$
    &$0.83(0.03)$
    &$0.88(0.03)$	
    &$0.89(0.02)$\\
    &$k=9$
    &$0.97(0.00)$
    &$0.98(0.00)$
    &$0.98(0.00)$
    &$0.82(0.02)$
    &$0.90(0.01)$
    &$0.88(0.01)$
    \\\midrule
    \multirow{4}{*}{Graph 3}    & $k=3$
    &$0.72(0.04)$
    &$0.78(0.02)$
    &$0.75(0.04)$
    &$0.44(0.02)$
    &$0.37(0.04)$	
    &$0.43(0.03)$\\
    & $k=5$
    &$0.71(0.05)$
    &$0.78(0.01)$
    &$0.80(0.04)$
    &$0.37(0.05)$
    &$0.36(0.01)$
    &$0.44(0.05)$ \\
    & $k=7$
    &$0.81(0.04)$
    &$0.79(0.01)$
    &$0.81(0.02)$
    &$0.42(0.04)$
    &$0.40(0.03)$
    &$0.45(0.03)$ \\
    & $k=9$
    &$0.78(0.06)$
    &$0.79(0.02)$
    &$0.84(0.04)$
    &$0.41(0.05)$
    &$0.43(0.03)$
    &$0.51(0.06)$\\
    \midrule
    \multirow{4}{*}{Graph 4}    & $k=3$
    &$0.80(0.00)$
    &$0.81(0.00)$
    &$0.84(0.00)$
    &$0.38(0.00)$
    &$0.41(0.00)$
    &$0.41(0.00)$\\
    & $k=5$
    &$0.83(0.00)$
    &$0.84(0.00)$
    &$0.85(0.00)$
    &$0.43(0.00)$
    &$0.45(0.00)$	
    &$0.44(0.00)$\\
    & $k=7$
    &$0.85(0.00)$
    &$0.85(0.00)$
    &$0.86(0.00)$
    &$0.45(0.00)$
    &$0.45(0.00)$
    &$0.46(0.00)$\\
    &$k=9$
    &$0.86(0.00)$
    &$0.85(0.00)$
    &$0.86(0.00)$
    &$0.46(0.00)$
    &$0.46(0.00)$
    &$0.46(0.00)$\\
    \midrule
    \bottomrule
    \end{tabular}%
    }
    \caption{The average AUC of Model 1--4 over 10 runs. The ground truth $k=k_m=9$. We vary $k=\{3,5,7,9\}$. The value inside the parentheses denotes the standard deviation. The additive noise is generated from noise model 1, where the AUC plot is shown in Figure~\ref{fig:vary_k_roc}. The rightmost three columns denote the average AUC for FPR between $[0,0.15]$ divided by $0.15$. The division is a normalization such that the maximum area will be $1$. When $k$ is significant under selected, $k=3$, the AUC and AUC15 are smaller. In contrast, the AUC and AUC15 increase as $k$ increases close to the true value $9$.}
    \label{tab:auc_varyk}
\end{table}

\begin{table}[hbt!]
    \spacingset{1}
    \centering
    \resizebox{\textwidth}{!}{%
    \begin{tabular}{llllllllll}
    \toprule
    \midrule\\
    &&\multicolumn{4}{c}{AUC}&\multicolumn{4}{c}{AUC15}\\
     &&\multicolumn{8}{c}{$k_m$}\\
    Graph & $k$&$5$&$7$&$9$&$11$&$5$&$7$&$9$&$11$\\
    \midrule
\multirow{5}{*}{Graph 1}
&$3$
&$0.87(0.02)$
&$0.80(0.03)$
&$0.79(0.02)$
&$0.79(0.02)$
&$0.60(0.02)$
&$0.48(0.05)$
&$0.44(0.02)$
&$0.45(0.05)$
\\
&$5$
&$0.87(0.02)$
&$0.82(0.01)$
&$0.84(0.01)$
&$0.81(0.02)$
&$0.61(0.04)$
&$0.51(0.01)$
&$0.50(0.03)$
&$0.47(0.04)$
\\
&$7$
&\multicolumn{1}{c}{--}
&$0.84(0.02)$
&$0.85(0.03)$
&$0.85(0.02)$
&\multicolumn{1}{c}{--}
&$0.54(0.03)$
&$0.55(0.03)$
&$0.56(0.05)$
\\
&$9$
&\multicolumn{1}{c}{--}
&\multicolumn{1}{c}{--}
&$0.85(0.02)$
&$0.86(0.02)$
&\multicolumn{1}{c}{--}
&\multicolumn{1}{c}{--}
&$0.55(0.03)$
&$0.57(0.05)$
\\
&$11$
&\multicolumn{1}{c}{--}
&\multicolumn{1}{c}{--}
&\multicolumn{1}{c}{--}
&$0.86(0.02)$
&\multicolumn{1}{c}{--}
&\multicolumn{1}{c}{--}
&\multicolumn{1}{c}{--}
&$0.57(0.05)$
\\
\midrule
\multirow{5}{*}{Graph 2}
&$3$
&$0.96(0.01)$
&$0.95(0.01)$
&$0.96(0.01)$
&$0.95(0.01)$
&$0.81(0.03)$
&$0.80(0.05)$
&$0.81(0.04)$
&$0.80(0.03)$
\\
&$5$
&$0.97(0.01)$
&$0.97(0.01)$
&$0.96(0.01)$
&$0.97(0.00)$
&$0.82(0.03)$
&$0.81(0.03)$
&$0.81(0.02)$
&$0.82(0.02)$
\\
&$7$
&\multicolumn{1}{c}{--}
&$0.96(0.01)$
&$0.96(0.01)$
&$0.97(0.01)$
&\multicolumn{1}{c}{--}
&$0.80(0.05)$
&$0.83(0.03)$
&$0.82(0.02)$
\\
&$9$
&\multicolumn{1}{c}{--}
&\multicolumn{1}{c}{--}
&$0.97(0.00)$
&$0.97(0.01)$
&\multicolumn{1}{c}{--}
&\multicolumn{1}{c}{--}
&$0.82(0.02)$
&$0.82(0.03)$
\\
&$11$
&\multicolumn{1}{c}{--}
&\multicolumn{1}{c}{--}
&\multicolumn{1}{c}{--}
&$0.97(0.00)$
&\multicolumn{1}{c}{--}
&\multicolumn{1}{c}{--}
&\multicolumn{1}{c}{--}
&$0.82(0.03)$
\\
\midrule
\multirow{5}{*}{Graph 3}
&$3$
&$0.78(0.03)$
&$0.76(0.02)$
&$0.72(0.04)$
&$0.75(0.03)$
&$0.45(0.04)$
&$0.42(0.03)$
&$0.44(0.03)$
&$0.45(0.06)$
\\
&$5$
&$0.76(0.04)$
&$0.74(0.04)$
&$0.71(0.05)$
&$0.73(0.04)$
&$0.39(0.04)$
&$0.39(0.03)$
&$0.37(0.05)$
&$0.38(0.03)$
\\
&$7$
&\multicolumn{1}{c}{--}
&$0.74(0.03)$
&$0.81(0.04)$
&$0.76(0.05)$
&\multicolumn{1}{c}{--}
&$0.38(0.04)$
&$0.42(0.04)$
&$0.38(0.04)$
\\
&$9$
&\multicolumn{1}{c}{--}
&\multicolumn{1}{c}{--}
&$0.78(0.06)$
&$0.79(0.04)$
&\multicolumn{1}{c}{--}
&\multicolumn{1}{c}{--}
&$0.41(0.05)$
&$0.41(0.04)$
\\
&$11$
&\multicolumn{1}{c}{--}
&\multicolumn{1}{c}{--}
&\multicolumn{1}{c}{--}
&$0.77(0.04)$
&\multicolumn{1}{c}{--}
&\multicolumn{1}{c}{--}
&\multicolumn{1}{c}{--}
&$0.42(0.03)$
\\
\midrule
\multirow{5}{*}{Graph 4}
&$3$
&$0.81(0.00)$
&$0.80(0.00)$
&$0.80(0.00)$
&$0.80(0.00)$
&$0.40(0.00)$
&$0.40(0.00)$
&$0.38(0.00)$
&$0.38(0.00)$
\\
&$5$
&$0.84(0.00)$
&$0.84(0.00)$
&$0.83(0.00)$
&$0.83(0.00)$
&$0.42(0.00)$
&$0.43(0.00)$
&$0.43(0.00)$
&$0.43(0.00)$
\\
&$7$
&\multicolumn{1}{c}{--}
&$0.86(0.00)$
&$0.85(0.00)$
&$0.85(0.00)$
&\multicolumn{1}{c}{--}
&$0.44(0.00)$
&$0.45(0.00)$
&$0.45(0.00)$
\\
&$9$
&\multicolumn{1}{c}{--}
&\multicolumn{1}{c}{--}
&$0.86(0.00)$
&$0.86(0.00)$
&\multicolumn{1}{c}{--}
&\multicolumn{1}{c}{--}
&$0.46(0.00)$
&$0.46(0.00)$
\\
&$11$
&\multicolumn{1}{c}{--}
&\multicolumn{1}{c}{--}
&\multicolumn{1}{c}{--}
&$0.86(0.00)$
&\multicolumn{1}{c}{--}
&\multicolumn{1}{c}{--}
&\multicolumn{1}{c}{--}
&$0.46(0.00)$
\\
\midrule
    \bottomrule
    \end{tabular}%
    }
    \caption{The average AUC of Graph 1--4 over 10 runs. The ground truth $k=k_m=9$. We vary $k=\{3,5,7,9,11\}$ and $k_m=\{5,7,9,11\}$. Since $k$ must be smaller or equal to $k_m$ due to the initialization method, we only show the results for $k\leq k_m$. The value inside the parentheses denotes the standard deviation. The additive noise is generated from noise model 1, where the AUC plot is shown in Figure~\ref{fig:vary_kkm_roc}. The rightmost three columns denote the average AUC for FPR between $[0,0.15]$, normalized to have a maximum area $1$. Fix $k_m$, both AUC and AUC15 are optimal when $k$ is close to $k_m$.  }
    \label{tab:auc_varykmk}
\end{table}
\FloatBarrier
\section{Treatment to Discrete Observed Data}
\label{appendix:treatment}
{
Throughout the paper we have assumed that observations across modalities are continuous functions. As a result, this article focuses mainly on the construction of the latent model, and we assume continuous observations for simplicity. In practice, data from different modalities are expected to be recorded with different temporal resolutions and can be viewed as discrete data. 
When the observations are discrete, we can still compute the function score by projecting the discrete samples to the discretized basis functions. Recall the basis $\{\phi_1^m,\ldots,\phi^m_{k_m}\}$ and let $t_{i1}^{m, (n)},\ldots,t_{iq}^{m, (n)}$ be the sampling time points of subject $n$ of the modality $m$ at node $i$. We can obtain the function score by solving a least-squares problem~(cf. Section 3.1 of~\citet{zhao2022fudge}) and obtain
\[
\tilde{\yb}^{m, (n)}_i=(\Cb_{i}^{m,(n)\top} \Cb_{i}^{m,(n)})^{-1}\Cb_{i}^{m,(n)\top} \hb_{i}^{m,(n)}, 
\]
where
\[
\Cb_{i}^{m,(n)} = \begin{bmatrix}
\phi_1^m(t_{i1}^{m,(n)})&\ldots&\phi_{k_m}^m(t_{i1}^{m,(n)})\\
\vdots&\ddots&\vdots\\
\phi_1^m(t_{iq}^{m,(n)})&\ldots&\phi_{k_m}^m(t_{iq}^{m,(n)})
\end{bmatrix},\quad
\hb_i^{m,(n)}=\begin{bmatrix}
    \mathcal{Y}_i^{m,(n)}(t_{i1}^{m,(n)})\\
    \vdots\\
    \mathcal{Y}_i^{m,(n)}(t_{iq}^{m,(n)}).
\end{bmatrix}
\]
Then, we can replace $\Yb_i^m$ with $\tilde{\Yb}_i^m=(\tilde{\yb}_i^{m,(1)},\ldots,\tilde{\yb}_i^{m,(N)})$ in~\eqref{eq:f_n}. The question is then how well can we estimate the parameters with the new objective function? \citet{zhao2022fudge} analyzed the conditions when the covariance of $\tilde{\yb}_i^m$ is close to the covariance of ${\yb}_i^m$. This provides insight that we might be able to quantify the error due to discretization under some regularity conditions, i.e., when the basis functions have smooth structures and the samples are evenly spaced.  
}

{
The family of basis functions is often unknown in practice. Hence, one approach uses functional PCA (fPCA) to estimate the basis functions. Existing literature has studied this setting from both theoretical and methodological perspectives~\citep{yao2005functional,li2010uniform,cai2011optimal,amini2012sampled,zhang2016sparse}. Motivated by~\cite{qiao2020doubly} who studied functional graphical model under the setting of the discrete sample, we briefly outline an extension of the methodology from~\citet{yao2005functional} to our model and propose a simple treatment.
In this first stage, we apply the algorithm proposed by~\citet{yao2005functional} to estimate the functional score $\yb_i^{m,(n)}$, denoted as $\hat{\yb}^{m,(n)}_i=(\hat{y}_{i,1}^{m,(n)},\ldots, \hat{y}_{i,k_m}^{m,(n)})^\top$ for $\pseq$ and the estimated basis functions $\{\hat{\phi}_1^m,\ldots,\hat{\phi}_{k_m}^{m}\}$ individually for each modality $\mseq$. 
Assume we observe $J$ discrete samples randomly sampled at  $T_{1}^{m,(n)},\ldots, T_{J}^{m,(n)}$:
\[
\Ycal_i\rbr{T_j^{m,(n)}}=\chi_i\rbr{T_j^{m,(n)}}+\xi_i\rbr{T_j^{m,(n)}},
\]
where $\Ycal_i({T_j^{m,(n)}})$ is a shorthand for $\Ycal_i^{m,(n)}({T_j^{m,(n)}})$, $\chi_i({T_j^{m,(n)}})=\chi_i^{m,(n)}({T_j^{m,(n)}})$ and $\xi_i({T_j^{m,(n)}})=\xi_i^{m,(n)}({T_j^{m,(n)}})$. 
Define $\Kscr^{m}_i(u,v)$ as the covariance estimator of $\Ycal^m_i(u)$ and $\Ycal^m_i(v)$. The first step is to estimate $\Kscr^{m}_i(u,v)$ with discrete observations. Let $h>0$ be a smoothing constant and denote $K_h(\cdot)=h^{-1}K(\cdot/h)$ be a smoothing kernel function. Given $u,v$, we consider the minimization of the following function with respect to $(\beta_0,\beta_1,\beta_2)$:
\begin{multline}\label{eq:smoothobj}
    \sum_{n=1}^N\sum_{T_j^{m,(n)}\neq T_{j'}^{m,(n)}}\cbr{\Ycal_i\rbr{T_j^{m,(n)}}\Ycal_i\rbr{T_{j'}^{m,(n)}}-\beta_0-\beta_1\rbr{T_{j}^{m,(n)}-u}-\beta_2\rbr{T_{j'}^{m,(n)}-v}}^2\\
\times K_h\rbr{T_{j}^{m,(n)}-u}K_h\rbr{T_{j'}^{m,(n)}-v}.
\end{multline}
Define $\hat{\beta}_0$ be the optimal solution of $\beta_0$ in~\eqref{eq:smoothobj}. The estimator is obtained as $\hat{\Kscr}_i^m(u,v)=\hat\beta_0$. Let $\{(\hat{\lambda}_{i,\ell},\hat{\phi}^m_{i,\ell})\}_{\ell=1}^{k_{m,i}'}$ be the eigen-pairs of $\hat{\Kscr}_i^m(\cdot,\cdot)$. Define $\{\hat{\phi}_1^m,\ldots,\hat{\phi}_{k_m}^m\}$ as the set of orthonormal basis functions that spans $\{\hat{\phi}_{1,1}^m,\ldots,\hat\phi_{1,k_{m,1}'}^m,\ldots,\hat{\phi}_{p,1}^m,\ldots,\hat\phi_{p,k_{m,p}'}^m\}$.\\ Let
\[\tilde{\Sigmab}^{m,(n)}_i=[\hat{\Kscr}_i^m(u,v)+\hat{\sigma}^2I(u=v)]_{u,v\in\{T_{1}^{m,(n)},\ldots, T_{J}^{m,(n)}\}}\in\RR^{J\times J},\]
where the selection of $\hat\sigma$ has been discussed in~\citet{yao2005functional},
\[
\tilde{\yb}_i^{m,(n)}=(\Ycal_i({T_1^{m,(n)}}),\ldots, \Ycal_i({T_J^{m,(n)}}))^\top\in\RR^J.
\] 
Then, the estimator~\citep{yao2005functional} for ${y}_{i,\ell}^{m,(n)}$ under the discrete setting is
\[
\hat{y}_{i,\ell}^{m,(n)}=\hat{\db}_{i,\ell}^{m,(n)\top}(\tilde{\Sigmab}^{m,(n)}_i)^{-1}\tilde{\yb}_i^{m,(n)},\quad\ell=1,\ldots,k_m,
\]
where $\hat{\db}_{i,\ell}^{m,(n)}=(\hat{\db}_{i,\ell}^{m}(T_1^{m,(n)}),\ldots, \hat{\db}_{i,\ell}^{m}(T_J^{m,(n)}))^\top\in\RR^J$ and
 \[\hat{\db}_{i,\ell}^{m}(T_j^{m,(n)})=\int_u\hat{\Kscr}_i^m(T_j^{m,(n)},u)\hat{\phi}^m_\ell(u)du,\quad j=T_1^{m,(n)},\ldots,T_J^{m,(n)}.\]\\
 After obtaining the estimates $\hat{\Yb}_i^m=(\hat{\yb}_i^{m,(1)},\ldots,\hat{\yb}_i^{m,(N)})\in\RR^{k_m\times N}$ from Stage~1, we can replace $\Yb_i^m$ with $\hat{\Yb}_i^m$ in~\eqref{eq:f_n}. Then, the rest of the estimation procedures follow the proposal. } 
\section{Experiments on Real Data}\label{appendix:exp_real}
This section discusses the implementation details of the concurrent EEG-fMRI recordings~\citep{sadaghiani2010intrinsic}. Section~\ref{ssec:pp} introduces the data preprocessing pipeline. 
Section~\ref{ssec:details_nr} discusses the details of the regression procedure introduced in Section~\ref{sec:experiment:brain}. Section~\ref{ssec:visualization} shows the visualization of the precision matrices.

\subsection{Preprocessing Pipeline}\label{ssec:pp}
First, we conduct the z-transform of the time-series for both EEG and fMRI data, which is a standard preprocessing step.
In the second step, we project the data to different bases as outlined in the following.  
We first regress out the global signal and standardize each time-series. 
The truncated fMRI time-series has $143$ time points and the original EEG time-series has $71680$ time points, after removing a few time points in the beginning~\citep{poldrack2011handbook}. We span the fMRI data using the Fourier basis  functions. Then, we down-sample the EEG data evenly to $1024$ time points so that we can project it to wavelet family basis functions, a common basis family used to decompose EEG signals~\citep{gandhi2011comparative}.
 The candidates of the wavelet bases are Daubechies’ extremal phase wavelets, 
Daubechies’ “least-asymmetric” wavelets, and Coiflets wavelets, provided by the `wavethresh' R-package. To select the best wavelet basis, we use the Shannon entropy type function as the evaluation metric. The steps are as the following. For each subject $n$ and each region $i$, we compute the wavelet coefficients $w^{m,(n)}_{i,\ell}$ for $\ell=1,2,\ldots,1024$. We normalize each coefficient as  $\bar{w}_\ell^m = \sum_{i,n}w^{m,(n)}_{i,\ell}/\max_{\ell}(\sum_{i,n}w^{m,(n)}_{i,\ell})$.  We then select the basis family that has the smallest entropy $-\sum_{\ell} (\bar{w}_i^m)^2\log (\bar{w}_\ell^m)^2$. The rest of the tuning parameters are selected using the procedure discussed in Section~\ref{sec:parameterselection}.

\subsection{Details of the Neighborhood Regression Procedure}\label{ssec:details_nr}
The functional neighborhood regression method for data of single modality follows closely from~\eqref{eq:f_n}, except we do not need to estimate the transformation operator $\Amk$ here. Hence, given $\ykmnfree{i}{1},\ldots,\ykmnfree{i}{N}$ independent samples and let $\Ykmnfree{i}=(\ykmnfree{i}{1},\ldots, \ykmnfree{i}{N})\in\RR^{k_m\times N}$ for $\pseq$, we define
$$
g_{\Yb}(\setBk)=\sum_{i=1}^p\frac{1}{2N}\bignorm{\Ykmnfree{i}-\sum_{j\in\pni}\Bijk \Ykmnfree{j}}_F^2.$$
Then, we optimize the objective function:
\begin{align}\label{eq:regressin_singlemodality}
&\argmin_{\setBk} g_{\Yb}(\setBk) \notag\\
&\text{s.t.}\quad\Bik\in\Kcal_B(s,\alpha),\quad\pseq.
 \end{align}
 The optimization problem~\eqref{eq:regressin_singlemodality} is carried out by projected gradient descent, as shown in Algorithm~\ref{alg:update_single}. Next, we use the estimated $\hBik$ for $\pseq$ to construct the edge set via AND operation~\eqref{eq:edge_thre}. Finally, we estimate the inverse covariance matrix by solving~\eqref{eq:graph_learning}.
 \begin{algorithm}[t!]
  \caption{Sparse Neighborhood Regression}\label{alg:update_single}
    \begin{algorithmic}
\State Input: { $\{\Yb^m\}_{\mseq}$}
\For{$\pseq$}
 \State $\Bb_i^{(0)}\leftarrow {\bf 0}$\;
\EndFor
 \While{Not converged}
 \For{$\pseq$}
 \State $\Bb_i^{(t+.25)}\leftarrow \Bb_i^{(t)} - \eta_B\nabla_{\Bik}g_{\Yb}$\;
\State $\Bb_i^{(t+.5)}\leftarrow \Tcal_{\vartheta_1s^\star/2}(\Bb_i^{(t)})$\;   
     \For{$j\in\pni$}
     \State $\Bijk^{(t+1)}\leftarrow\Hcal_{\vartheta_2\alpha^\star/2}(\Bijk^{(t+.5)})$\;
   \EndFor
 \EndFor
 \EndWhile
  \end{algorithmic}
\end{algorithm}

\begin{figure}
    \centering
    \includegraphics[width=0.45\textwidth]{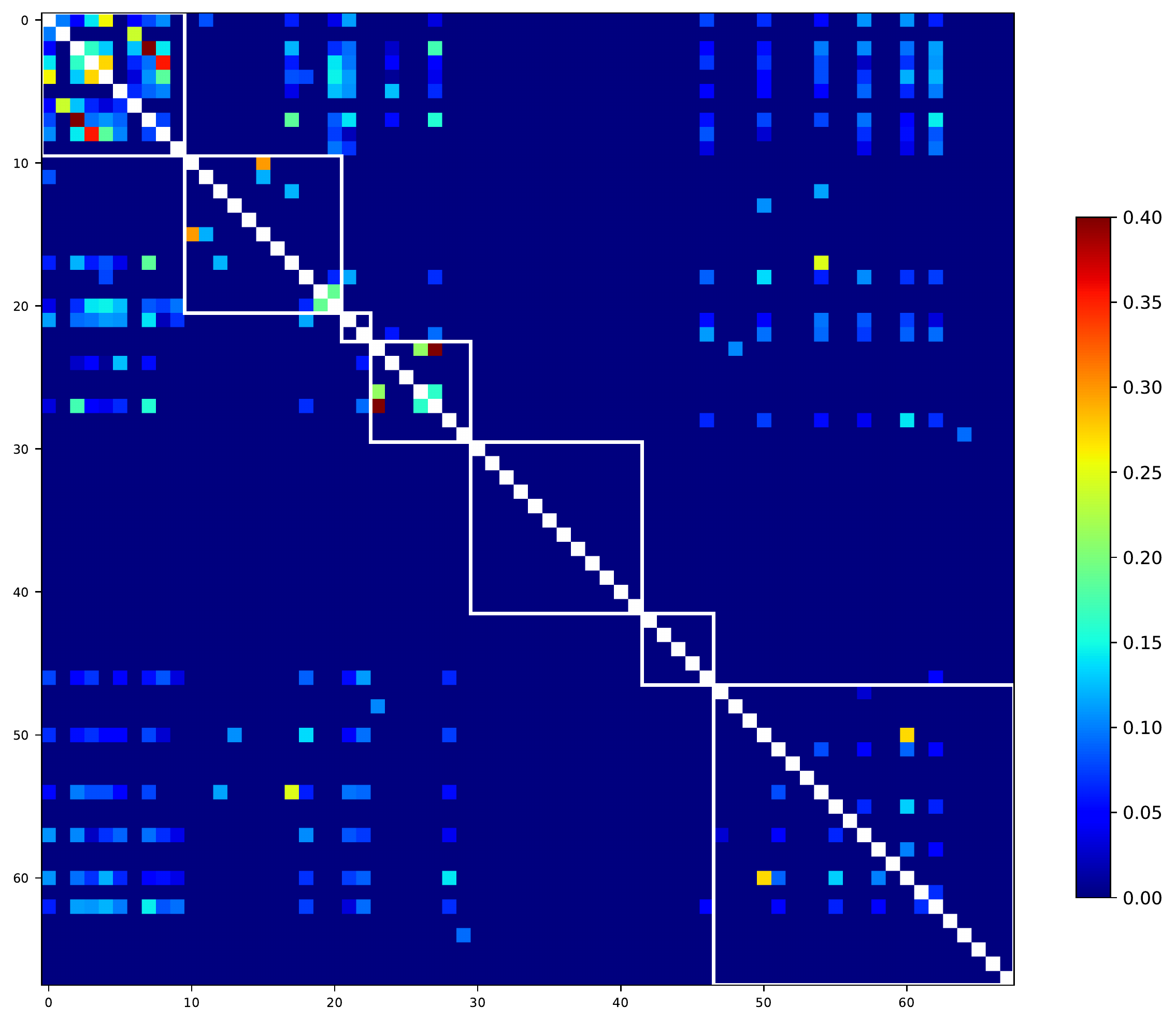}
    \includegraphics[width=0.45\textwidth]{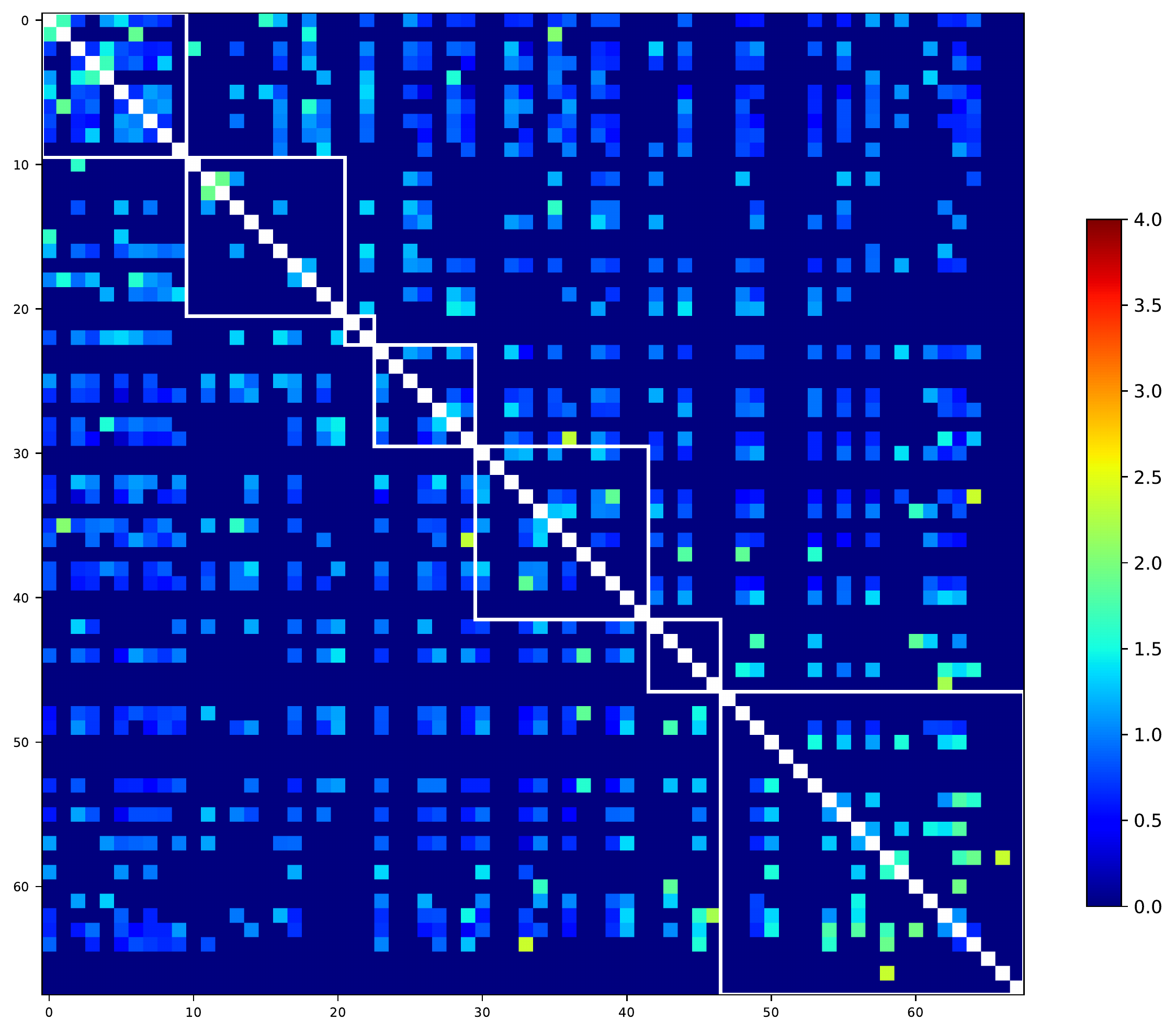}
    \includegraphics[width=0.45\textwidth]{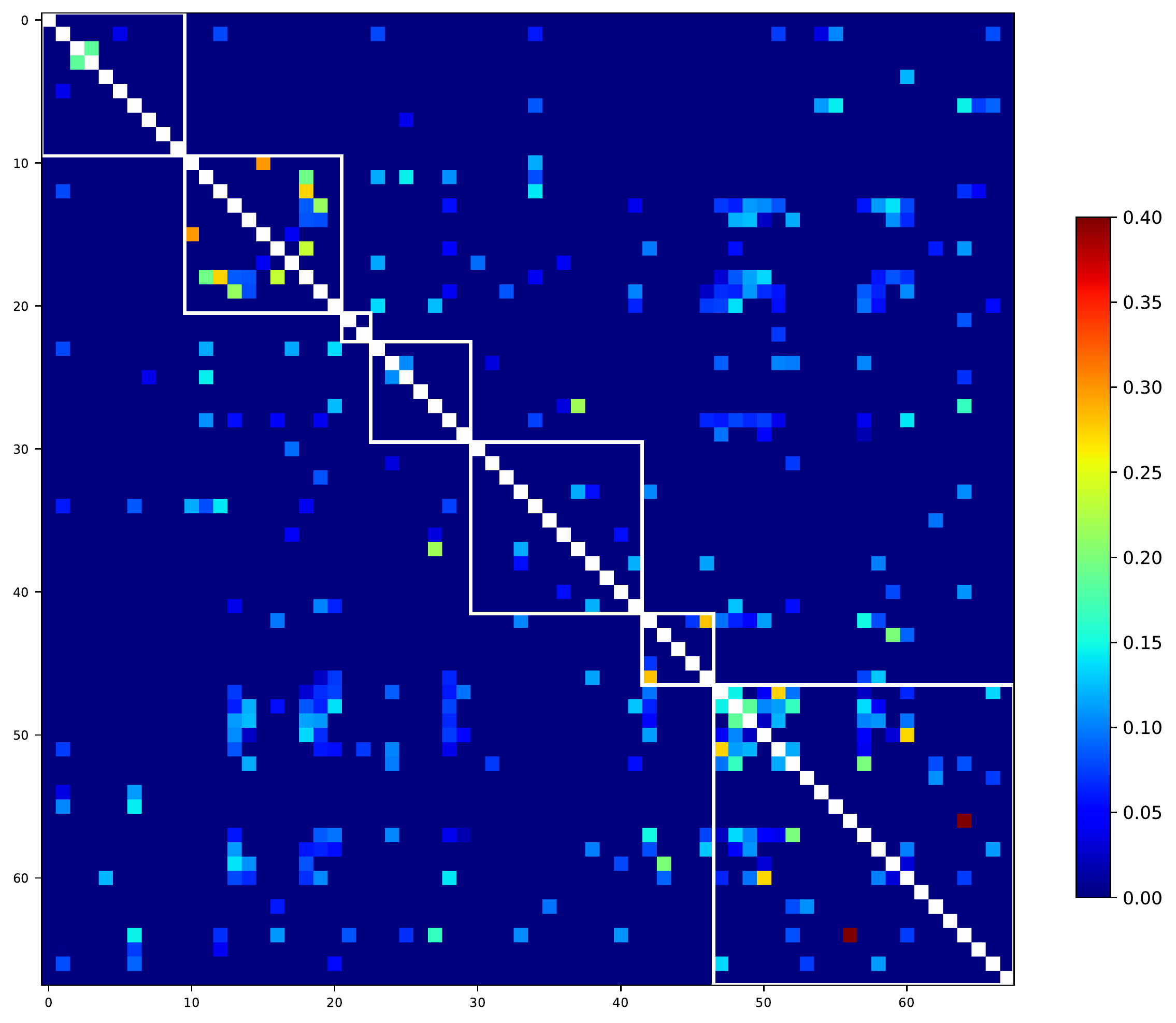}
    \includegraphics[width=0.45\textwidth]{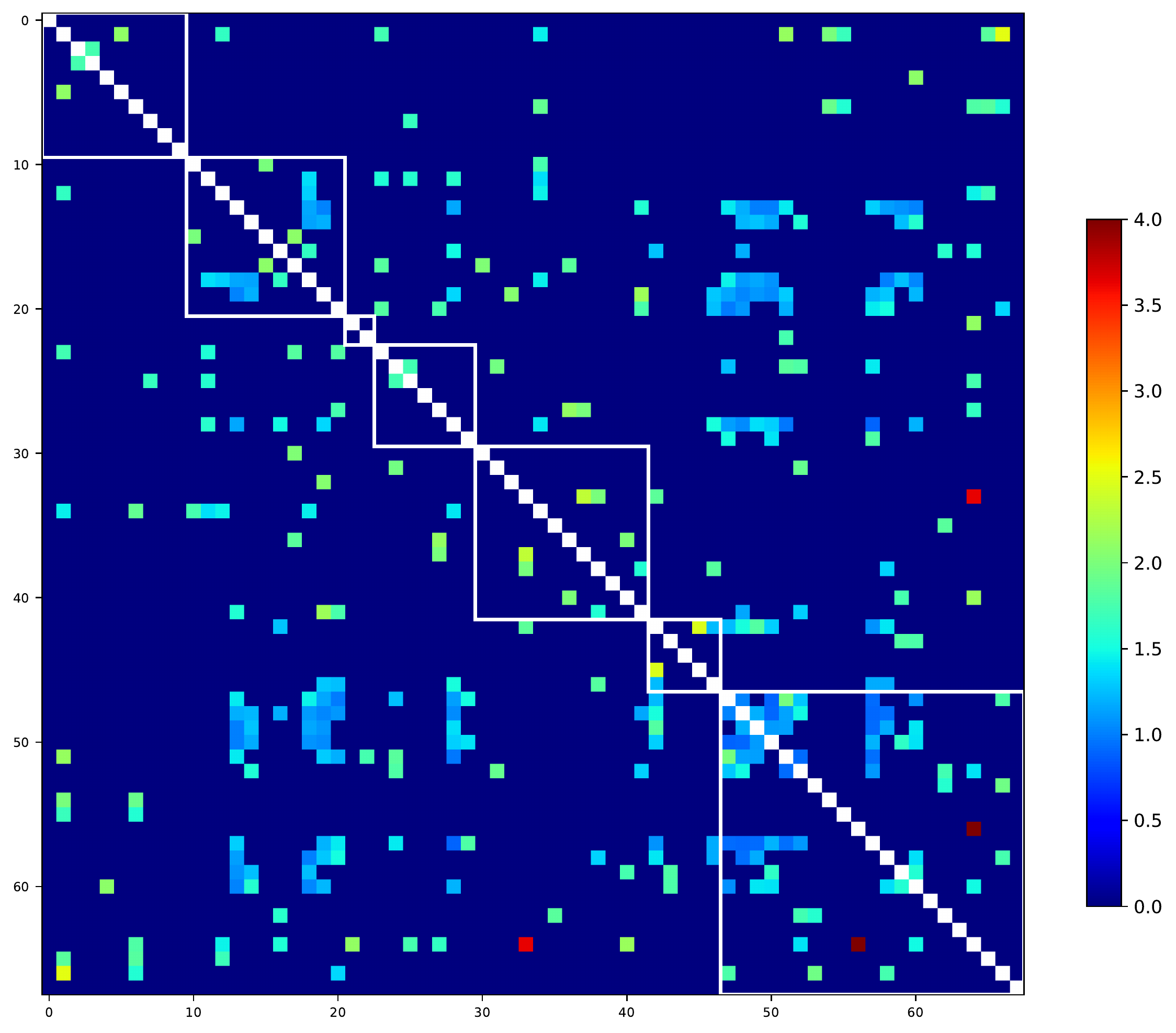}
    \includegraphics[width=0.45\textwidth]{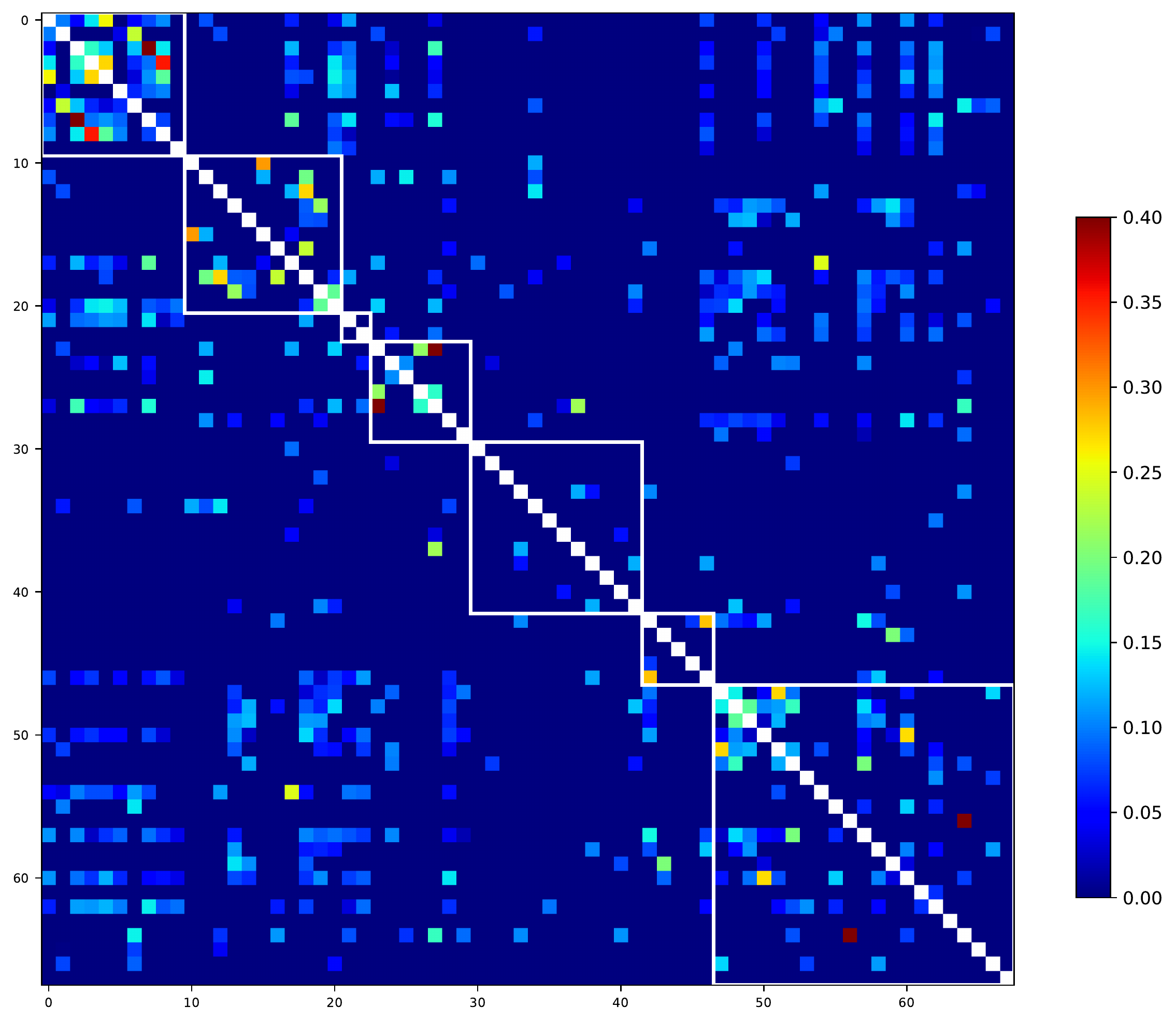}
    \includegraphics[width=0.45\textwidth]{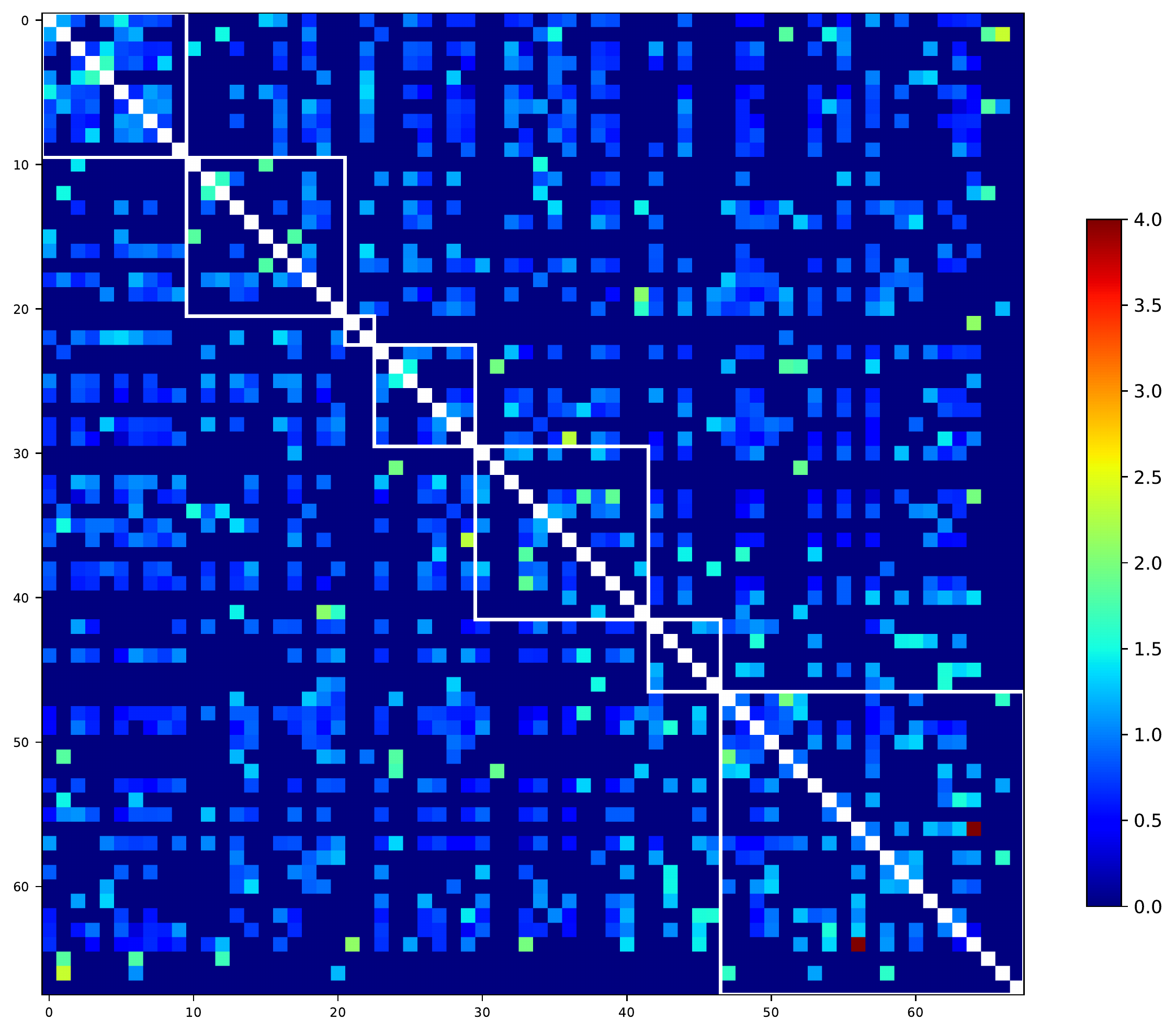}
    \caption{
    Estimated graph structure from three edge set candidates. 
    {\bf Left} to {\bf right}: the graphs of data modality 1 and the graphs of data modality 2. {\bf Top} to {\bf bottom}: the individual edge set, the latent edge set, and the fused edge set.
    The block diagonal box presents the partition using Yeo 7 network~\citep{yeo2011organization}. 
    }
    \label{fig:connectivity}
\end{figure}

\subsection{Visualization of Precision Matrices}
\label{ssec:visualization}

 As the number of bases used to span two data modalities is different, the precision matrices $\hat\Omega^{1}\in\RR^{pk_1\times pk_1}$ and $\hat\Omega^{2}\in\RR^{pk_2\times pk_2}$ have different dimensions. This makes it hard to compare the precision matrices of two graphs. As an alternative, given estimated $\hat\Omega^{m}$ estimated from~\eqref{eq:graph_learning} with $E\in\{\hat E^{m}, \hat E^{z}, \hat E^{m}\cup\hat E^{z}\}$ for $m=1,2$, we compute the magnitude of the $(i,j)$th submatrix $\hat\omega_{i,j}^m$ of $\hat\Omega^{m}$, in Frobenious norm.
The resulting plots are shown Figure~\ref{fig:connectivity}, where the $(i,j)$th is the magnitude for $\omega_{i,j}^m$ and the label is documented in Table~\ref{tab:DK_label}.
 The figure indicates that the latent dependency graph includes some edge connections that were not originally detected by any modality individually.

\subsection{Computation Complexity}\label{ssec:comutation}
{\color{black}{In this section, we analyze the computation time for running the experiments. The sample size, dimension, and number of time points are documented in Section~\ref{sec:experiment:brain}. 
        The computing time of the model selection of $\{s,\alpha,\tau_1,\tau_2\}$ introduced in Section~\ref{sec:parameterselection} is shown in the step~5 of Table~\ref{tab:experiment_runtime}. We run grid search on $s=\{3,6,9,12,15,18,21,24,27,30,33\}$, $\alpha=\{3,7,10,13\}$, $\tau_1=\{0.5,0.25\}$ and $\tau_2=\{2,4\}$. 
        The computing time of edge selection (proposed algorithm) is shown in Step~7. 
        In the experiment, we use learning rate $10^{-4}$ to initialize $\Bb$, the learning rate $10^{-3}$ for updating $\Bb$ and $10^{-4}$ for updating $\Ab$ in Algorithm~\ref{alg:update}.
        Then, we select $\lambda=\{0.1,0.2,0.3,0.4,0.5,0.6,0.7,0.8,0.9,1.0\}$ in~\eqref{eq:graph_learning} using 5-fold cross-validation with the BIC metric. Then, we use $skggm$ package to implement the graph estimation with selected $\lambda$. The computing time of model selection of $\lambda$ along with the graph estimation is shown in Step~8--9 of Table~\ref{tab:experiment_runtime}. Noting that Step~5, 8--9 take much longer time to finish as the running time depends on the size of the tuning parameter candidates, other steps are executed only once and could be done within a few minutes as shown in Table~\ref{tab:experiment_runtime}. }
}
        \begin{table}[!h]
\spacingset{1}
            \centering
            \resizebox{\textwidth}{!}{%
            \begin{tabular}{lll}
            \toprule
            \midrule
            Task && Running Time (s)\\
            \midrule
            Step 1:& standard preprocessing &$202.60$\\
            \midrule
            Step 2:& model Selection: basis family, $k$, $k_m$  & $201.64$ \\
            Step 3:& project data to selected basis& $16.91$\\
            Step 4:& initialization of $\Ab$& $2.374$\\
            Step 5:& model selection: $s,\alpha,\tau_1,\tau_2$ & $23610.33^*$\\
            Step 6:& initialization of $\Bb$ (Algorithm~2)& $251.17$\\
            Step 7:& Algorithm~1&$277.19$\\
            \midrule
            Step 8:& estimate graph of fMRI &$2777.08^*$\\
            Step 9:& estimate graph of EEG&$4289.77^*$\\
            \midrule
            \bottomrule
            \end{tabular}
            }
            \caption{Running time of the fMRI-EEG experiment. We use Intel Xeon Processor E5-2620 v3 @ 2.40 GHz to run all the steps. The * implies that we parallel the experiment with 24 CPUs and the remaining experiments are executed with single CPU.  }
            \label{tab:experiment_runtime}
        \end{table}

\FloatBarrier
\spacingset{1}
\begin{longtable}{ p{.10\textwidth}p{.2\textwidth}  p{.15\textwidth}p{.1\textwidth}
p{.35\textwidth} } 
\toprule
\midrule
Matrix Index & Desikan-Killiany Atlas Index & Abbreviation & Yeo 7 Network & Name  \\
\midrule
1 & 4 & lCUN & VIS & L cuneus \\
2 & 6 & lFUS & VIS & L fusiform \\
3 & 10 & lLOG & VIS & L lateral occipital gyrus \\
4 & 12 & lLING & VIS & L lingual \\
5 & 20 & lperiCAL & VIS & L pericalcarine \\
6 & 38 & rCUN & VIS & R cuneus \\
7 & 40 & rFUS & VIS & R fusiform \\
8 & 44 & rLOG & VIS & R lateral occipital gyrus \\
9 & 46 & rLING & VIS & R lingual \\
10 & 54 & rperiCAL & VIS & R pericalcarine \\
11 & 16 & lparaC & SM & L paracentral \\
12 & 21 & lpostC & SM & L postcentral \\
13 & 23 & lpreC & SM & L precentral \\
14 & 29 & lSTG & SM & L superior temporal gyrus \\
15 & 33 & lTT & SM & L transverse temporal \\
16 & 50 & rparaC & SM & R paracentral \\
17 & 55 & rpostC & SM & R postcentral \\
18 & 56 & rPCC & SM & R posterior cingulate cortex \\
19 & 57 & rpreC & SM & R precentral \\
20 & 63 & rSTG & SM & R superior temporal gyrus \\
21 & 67 & rTT & SM & R transverse temporal \\
22 & 28 & lSPL &  DA & L superior parietal lobule \\
23 & 62 & rSPL &  DA & R superior parietal lobule \\
24 & 2 & lcACC & VA & L caudal anterior cingulate \\
25 & 17 & lpOPER & VA & L pars opercularis \\
26 & 30 & lSMAR & VA & L supramarginal gyrus \\
27 & 34 & lINS & VA & L insula \\
28 & 36 & rcACC & VA & R caudal anterior cingulate \\
29 & 64 & rSMAR & VA & R supramarginal gyrus \\
30 & 68 & rINS & VA & R insula \\
31 & 5 & lENT & L & L entorhinal \\
32 & 8 & lITG & L & L inferior temporal gyrus \\
33 & 11 & lLOF & L & L lateral orbitofrontal \\
34 & 13 & lMOF & L & L medial orbitofrontal \\
35 & 31 & lFP & L & L frontal pole \\
36 & 32 & lTP & L & L temporal pole \\
37 & 39 & rENT & L & R entorhinal \\
38 & 42 & rITG & L & R inferior temporal gyrus \\
39 & 45 & rLOF & L & R lateral orbitofrontal \\
40 & 47 & rMOF & L & R medial orbitofrontal \\
41 & 65 & rFP & L & R frontal pole \\
42 & 66 & rTP & L & R temporal pole \\
43 & 26 & lrMFG & FP & L rostral middle frontal gyrus \\
44 & 37 & rcMFG & FP & R caudal middle frontal gyrus \\
45 & 51 & rpOPER & FP & R pars opercularis \\
46 & 53 & rpTRI & FP & R pars triangularis \\
47 & 60 & rrMFG & FP & R rostral middle frontal gyrus \\
48 & 1 & lBSTS & DMN & L bank of the superior temporal sulcus \\
49 & 3 & lcMFG & DMN & L caudal middle frontal gyrus \\
50 & 7 & lIPL & DMN & L inferior parietal lobule \\
51 & 9 & liCC & DMN & L isthmus cingulate cortex \\
52 & 14 & lMTG & DMN & L middle temporal gyrus \\
53 & 18 & lpORB & DMN & L pars orbitalis \\
54 & 19 & lpTRI & DMN & L pars triangularis \\
55 & 22 & lPCC & DMN & L posterior cingulate cortex \\
56 & 24 & lPCUN & DMN & L precuneus \\
57 & 25 & lrACC & DMN & L rostral anterior cingulate cortex \\
58 & 27 & lSFG & DMN & L superior frontal gyrus \\
59 & 35 & rBSTS & DMN & R bank of the superior temporal sulcus \\
60 & 41 & rIPL & DMN & R inferior parietal lobule \\
61 & 43 & riCC & DMN & R isthmus cingulate cortex \\
62 & 48 & rMTG & DMN & R middle temporal gyrus \\
63 & 52 & rpORB & DMN & R pars orbitalis \\
64 & 58 & rPCUN & DMN & R precuneus \\
65 & 59 & rrACC & DMN & R rostral anterior cingulate cortex \\
66 & 61 & rSFG & DMN & R superior frontal gyrus \\
67 & 15 & lPARH & DMN & L parahippocampal \\
68 & 49 & rPARH & DMN & R parahippocampal \\
\midrule
\bottomrule
\caption{ The index table of Desikan-Killiany atlas~\citep{desikan2006automated} reordered to Yeo 7 network~\citep{yeo2011organization}. VIS = visual; SM = somatomotor; DA = dorsal attention;
VA = ventral attention; L = limbic; FP = fronto-parietal; DMN = default mode network.
}
    \label{tab:DK_label}
\end{longtable}

\clearpage
\putbib[bu2]
\end{bibunit}
\end{document}